\documentclass[11pt]{article}

\usepackage[margin=1in]{geometry}
\usepackage{setspace}
\usepackage{amsmath, amssymb, amsthm, bm}
\usepackage{mathtools}

\usepackage{graphicx}
\usepackage{booktabs}
\usepackage{subcaption}
\usepackage{float}
\usepackage{tikz}
\usetikzlibrary{shapes.geometric, arrows.meta, positioning}

\usepackage{algorithm}
\usepackage{algorithmic}

\usepackage[sort, round]{natbib}

\usepackage{xcolor}

\usepackage[colorlinks=true, linkcolor=blue, citecolor=blue, urlcolor=blue]{hyperref}
\usepackage{enumitem}
\usepackage{etoc}
\usepackage{dsfont}

\theoremstyle{plain}
\newtheorem{theorem}{Theorem}
\newtheorem{lemma}{Lemma}
\newtheorem{corollary}{Corollary}
\newtheorem{proposition}{Proposition}
\newtheorem{assumption}{Assumption}

\theoremstyle{definition}
\newtheorem{definition}{Definition}

\newtheorem{example}{Example}

\newcommand{\E}{\mathbb{E}}
\newcommand{\Var}{\operatorname{Var}}
\newcommand{\Cov}{\operatorname{Cov}}
\newcommand{\indep}{\perp \!\!\! \perp}
\newcommand{\sumi}{\sum_{i=1}^n}
\newcommand{\tauxhat}{\hat{\boldsymbol\tau}_x}

\title{\textbf{Rerandomization under Interference}}
\author{Ziang Yuan \and Xinran Li \and Shuangning Li\thanks{We thank Wenxuan Guo, as well as Dennis Shen, Colin Fogarty, and other participants of the ACM SIGMETRICS Causal Inference Workshop 2026 and the Causal Inference in Econometrics and Statistics Workshop in UChicago, for helpful comments. We used OpenAI's ChatGPT and Anthropic's Claude to assist with proofreading and editing during the preparation of this manuscript. All errors are our own.}}
\date{University of Chicago}

\begin{document}
\etocdepthtag.toc{main}

\maketitle

\begin{abstract}
Covariates are widely used in randomized experiments to improve precision. However, in the presence of interference, where outcomes may depend on the treatment assignments of other units, standard covariate adjustment methods may fail to preserve desirable properties such as the no-harm property, meaning that incorporating covariates does not worsen estimator performance. Existing approaches that incorporate covariates under interference typically rely on specifying a particular interference structure, and their guarantees can be sensitive to misspecification. In this paper, we study how to incorporate covariate information in a way that preserves the no-harm property while remaining largely agnostic to the underlying interference structure. We focus on the estimation of the expected average treatment effect (EATE) using the Hájek estimator and study rerandomization under interference, a design-stage procedure that restricts the assignment space to allocations with sufficiently small covariate imbalance. Our framework also allows the covariates used for rerandomization to depend on the treatment assignment itself, such as the proportion of treated neighbors, which naturally arises in settings with interference. We show that, even under interference, rerandomization can improve estimation precision asymptotically relative to unrestricted Bernoulli randomization, relying only on mild conditions on the dependence structure across units. When a conservative dependence graph is available, we further develop an optimization-based conservative variance estimator for inference under rerandomization.
\end{abstract}

\section{Introduction}

Causal inference is central to many scientific, economic, and policy questions. Researchers are often interested in understanding the effect of an intervention on outcomes of interest, such as the impact of a new pricing policy, a change in platform design, or a medical treatment. Randomized experiments are widely regarded as the gold standard for answering such questions, as random assignment removes confounding and enables transparent identification of causal effects.

In many experimental settings, researchers observe not only treatment assignments and outcomes, but also rich covariate information describing the experimental units. Even though randomization ensures (sometimes asymptotically) unbiased estimation of treatment effects, covariates can be leveraged to improve precision. Broadly speaking, there are two main approaches to incorporating covariates into experimental analysis. First, covariates can be used at the design stage. Methods such as stratification, blocking, and rerandomization aim to ensure that treatment groups are balanced with respect to observed covariates \citep{morgan2012rerandomization, li2018asymptotic, morgan2015rerandomization, kernan1999stratified, imai2009essential, imbens2011experimental}. These approaches modify the randomization procedure itself so that the realized treatment assignment exhibits desirable balance properties. Second, covariates can be incorporated at the analysis stage. These approaches adjust for covariates after the experiment has been conducted, including methods such as regression adjustment and post-stratification \citep{lin2013agnostic, negi2021revisiting, rosenbaum2002covariance, fogarty2018regression, su2021model, zhao2022reconciling, wang2023model, chang2024exact, wang2024model, zhao2024covariate, miratrix2013adjusting, bloniarz2016lasso, guo2023generalized, cohen2024no, qu2025randomization}.

A key property that has received significant attention in this literature is the so-called \emph{no-harm} property: incorporating covariates should not lead to worse performance than ignoring them. Under standard assumptions, many commonly used procedures enjoy this guarantee, providing a strong justification for adjusting for covariates in practice \citep{lin2013agnostic}.

The above developments largely rely on the Stable Unit Treatment Value Assumption (SUTVA) \citep{rubindonald1980comment}, which states that each unit's outcome depends only on its own treatment assignment. Yet in many real-world applications, from economics to public health to education, this assumption is untenable: the treatment of one unit can influence the outcomes of others \citep{cai2015social, hudgens2008toward, sacerdote2001peer, sobel2006randomized}. For instance, the vaccination status of one individual may affect the health outcomes of their friends. This phenomenon, known as interference, poses a central challenge to the standard causal inference paradigm. Estimators that are valid under SUTVA may become biased or lose their usual interpretation. Even when unbiasedness can be maintained, variance estimation becomes significantly more complex due to the dependence induced by interference. A growing literature addresses these challenges by developing estimators and corresponding variance estimators under different forms of interference \citep{sobel2006randomized, hudgens2008toward, tchetgen2012causal, toulis2013estimation, Samii2012EstimatingAC, leung2020treatment, li2022random, leung2022causal, cortez2023exploiting}.

At the same time, the role of covariates in the presence of interference is less well understood. While covariate adjustment is known to improve precision under SUTVA, its behavior can differ under interference. In particular, \citet{gao2025causal} show that regression adjustment methods that satisfy the no-harm property under SUTVA may fail to do so when interference is present. The key challenge is that, even under Bernoulli assignment where treatments are independent across units, realized outcomes are generally dependent due to interference. As a result, variance reduction through covariate adjustment must account for cross-unit covariances. These covariance terms are difficult to handle without imposing additional structure on the form of interference.

Existing approaches that incorporate covariates under interference and retain no-harm guarantees often rely on specifying a particular interference structure. For example, \citet{lu2024adjusting} study covariate adjustment under approximate neighborhood interference \citep{leung2022causal}, \citet{fan2025causal} consider the random graph interference framework of \citet{li2022random}, and \citet{wang2025covariate} focus on low-order interaction outcome models \citep{cortez2023exploiting}. As a result, their guarantees can be sensitive to misspecification of the underlying interference mechanism.

This leads to a natural question:
\begin{quote}
\textit{ Can we incorporate covariate information in a way that preserves the no-harm property, while remaining largely agnostic to the underlying interference structure? }
\end{quote}

This paper provides an affirmative answer to this question. Our goal is to develop a design-stage approach that leverages covariates to improve precision without requiring detailed knowledge of the interference structure.

Specifically, we focus on the estimation of the expected average treatment effect (EATE), a commonly used estimand under interference that averages unit-level effects over the distribution of treatment assignments \citep{savje2021average}. To incorporate covariate information, we consider rerandomization, a design procedure that repeatedly samples treatment assignments and only accepts those assignments for which a prespecified covariate imbalance measure is sufficiently small. In practice, this is typically implemented by computing a balance statistic (such as a Mahalanobis distance) and restricting attention to assignments that satisfy a threshold constraint. Rerandomization has been shown to reduce variance under SUTVA without introducing asymptotic bias, and it provides a natural way to incorporate covariates without relying on outcome modeling \citep{morgan2012rerandomization, li2018asymptotic}.

In this paper, we study rerandomization in settings with interference. Our analysis shows that rerandomization satisfies an asymptotic no-harm property under interference: relative to Bernoulli randomization, it can improve the precision of the Hájek estimator for the EATE. Importantly, our results do not require correct specification of the interference mechanism. The procedure remains largely model agnostic, with guarantees relying only on mild restrictions on the dependence structure, such as limited or localized dependence. Technically, we establish these results by proving a joint central limit theorem for the treatment effect estimator and the covariate imbalance vector, using tools from the literature on limit theorems for network-dependent random variables.

This perspective also highlights an important contrast with the classical SUTVA setting. Under SUTVA, rerandomization is closely connected to regression adjustment. In particular, \citet{liding2020rerandomization} show that rerandomization and regression adjustment can achieve similar precision improvements for treatment effect estimation: both adjust for covariate imbalance, but rerandomization does so at the design stage whereas regression adjustment does so at the analysis stage. Under interference, however, regression adjustment is more delicate. The optimal adjustment generally depends not only on marginal relationships between covariates and outcomes, but also on cross-unit covariance terms induced by interference. These covariance terms are difficult to characterize or estimate without imposing additional structure on the interference mechanism or the underlying network. In this sense, the benefits of rerandomization become more apparent in settings with interference, because it can improve precision through the design while remaining largely agnostic about the form of interference.

Our framework also allows the balancing criterion in the rerandomization procedure to incorporate both standard pre-treatment covariates and treatment-dependent exposure summaries that may be relevant for the outcome. For example, in the presence of network interference, the number or proportion of treated friends may be predictive of a unit's realized outcome. Such quantities are not fixed baseline covariates, but they can be recomputed for each candidate assignment and included in the rerandomization criterion.

In addition, we develop a conservative variance estimator for inference under rerandomization. Starting from any jointly conservative covariance estimator that is valid under the baseline assignment mechanism before rerandomization, we construct a variance estimator that accounts for the rerandomization constraint by solving a tractable optimization problem. The resulting estimator remains conservative for the variance under rerandomization, while often being substantially tighter than the baseline variance estimator.

\subsection{Related work}
\label{subsection:related_work}
\paragraph{Rerandomization and restricted randomization.}
Rerandomization is a design-stage approach for improving covariate balance in randomized experiments. Its origins are connected to Fisher's discussion of random assignment and balance \citep{fisher1992arrangement,fisher1935design}, and to the broader literature on restricted randomization in experimental design \citep{bailey1987restricted,simon1979restricted,cox1982randomization}. The modern formulation of rerandomization in treatment-control experiments was given by \citet{morgan2012rerandomization}, who proposed repeatedly drawing assignments and accepting only those satisfying a prespecified balance criterion, such as a Mahalanobis-distance criterion. Subsequent work developed finite-sample and asymptotic theory, including the truncated-normal limiting distribution and variance-reduction formula \citep{li2017general,li2018asymptotic}, power analysis \citep{branson2024power}, extensions to other covariate balance criteria \citep{morgan2015rerandomization, zhao2024no, cohen2022gaussian, liu2025bayesian, branson2021ridge, zhang2024pca} and experiments \citep{lidingrubin2020rerandomization, johansson2022rerandomization, wang2023rerandomization, lu2023design, yang2023rejective}, and high-dimensional or vanishing-acceptance-probability regimes \citep{wang2022rerandomization}; see also \citet{junior2025does} for a recent review. Related design-stage approaches choose treatment assignments by explicitly optimizing balance, precision, or robustness criteria \citep{kallus2018optimal,kapelner2021harmonizing}. We extend rerandomization to interference settings, incorporate treatment-dependent covariates in the balance, and show that the no-harm guarantee can still be achieved.

Among the literature on rerandomization and restricted randomization, two papers are especially close to ours. \citet{yang2026design} study rerandomization under interference and propose network interference balancing rerandomization and network maximized matching randomization for estimating direct treatment effects in network experiments. Our work differs from theirs in the modeling assumptions, estimand, and nature of the theoretical guarantee. \citet{yang2026design} work under a specific exposure structure in which interference is summarized by the number of treated neighbors, and target an estimand for which the standard estimator can be biased under Bernoulli randomization; their designs are constructed to reduce both bias and variance. In contrast, we allow general potential outcomes subject only to limited cross-unit dependence. The Hájek estimator for the EATE is asymptotically unbiased under Bernoulli randomization and remains asymptotically unbiased after rerandomization, and we establish an asymptotic, model-agnostic no-harm guarantee under which rerandomization can reduce variance.

\citet{basse2018model} also use restricted randomization to exploit network information at the design stage. Their focus, however, is on network-correlated outcomes in the absence of interference: they posit a working model for outcome correlation and use it to derive network-balance criteria. In contrast, our paper allows interference and studies rerandomization for the Hájek estimator of the EATE. Our guarantees do not rely on a correctly specified outcome-correlation model or interference model, but instead on weak dependence conditions that yield asymptotic normality and unbiasedness, together with a no-harm variance reduction result.

\paragraph{Interference.}
Causal inference under interference departs from the standard SUTVA framework by allowing a unit's potential outcome to depend on the treatment assignments of other units. This issue arises in many applications, including social interactions, peer effects, public health, and network experiments \citep{cai2015social, hudgens2008toward, sacerdote2001peer, sobel2006randomized}. Foundational work developed frameworks for defining causal effects when interference is present \citep{sobel2006randomized, hudgens2008toward}, and a large subsequent literature has proposed estimands, estimators, experimental designs, and inference procedures under various interference structures \citep[among many others]{sobel2006randomized, hudgens2008toward, tchetgen2012causal, toulis2013estimation, eckles2017design, athey2018exact, Samii2012EstimatingAC, leung2020treatment, savje2021average, li2022random, leung2022causal}. In particular, the expected average treatment effect (EATE) of \citet{savje2021average} provides an estimand that remains meaningful under interference by averaging direct effects over the randomization distribution. We build on this line of work by studying design-stage precision improvement for the H\'{a}jek estimator of the EATE.

\paragraph{Covariates and regression adjustment.}
A separate literature studies how covariates can improve the analysis of randomized experiments. Under SUTVA, regression adjustment and related analysis-stage procedures are well understood, and several results show that appropriate adjustment does not harm, and can often improve, asymptotic precision relative to the unadjusted estimator \citep{lin2013agnostic, negi2021revisiting, rosenbaum2002covariance, fogarty2018regression, su2021model, zhao2022reconciling, wang2023model, chang2024exact, wang2024model, zhao2024covariate}. Under interference, the role of covariates is more delicate. Because outcomes of different units can be statistically dependent through the treatment assignment and the network, the variance of an estimator contains cross-unit covariance terms in addition to the usual marginal variance terms. As a result, a regression adjustment that reduces individual-level residual variance need not reduce the overall variance of the treatment effect estimator. Indeed, \citet{gao2025causal} show that regression adjustment methods with no-harm guarantees under SUTVA can fail to retain such guarantees under interference. Recent work has therefore developed covariate-adjusted estimators under more specific interference models, including approximate neighborhood interference \citep{leung2022causal, lu2024adjusting}, random-graph interference \citep{li2022random, fan2025causal}, and low-order interaction outcome models \citep{cortez2023exploiting, wang2025covariate}. These approaches demonstrate the value of covariates under interference, but their guarantees typically depend on correctly specifying the relevant interference structure or outcome model. In contrast, our approach uses covariates at the design stage through rerandomization and establishes a no-harm-type variance reduction result under weak dependence conditions, without requiring a correctly specified interference structure or outcome model.

\paragraph{Network-dependent data and weak dependence.}
Our asymptotic analysis is related to the literature on dependent data and central limit theorems for network-indexed random variables. Classical asymptotic theory for dependent observations often relies on strong mixing conditions, such as Rosenblatt’s $\alpha$-mixing framework \citep{rosenblatt1956central}, but such conditions can be difficult to verify and may fail even in simple autoregressive settings \citep{Andrews1984NonstrongMA}. An alternative line of work measures weak dependence of random variables \citep{doukhan1999new}. Some related approaches have also been developed for dependent sequences and random fields \citep[e.g.,][]{dedecker1998central,dedecker2001exponential,dedecker2007weak,el2013central}, and more recently for observations indexed by networks. \citet{lee2019stable} developed conditional neighborhood dependence for networks, while \citet{kojevnikov2021limit} established LLNs, CLTs, and HAC inference for network-dependent random variables under graph-distance-based dependence decay. Related decay ideas also appear in causal inference under interference, for example in approximate neighborhood interference \citep{leung2022causal}. We adapt these ideas to prove joint asymptotic normality of the H\'{a}jek estimator and the covariate imbalance vector, which is the key step for analyzing rerandomization as conditioning on an acceptable-balance event.

\section{Problem Setup and Rerandomization under Interference}
\label{sec:setup}

We consider a population of $n$ units indexed by $i = 1, \dots, n$. Each unit receives a binary treatment, and we denote the treatment assignment vector by $\mathbf Z = (Z_1, \dots, Z_n) \in \{0,1\}^n$, where $Z_i$ indicates the treatment assigned to unit $i$.

We adopt the potential outcomes framework \citep{neyman1923application, rubin1974estimating, imbens2015causal}. In the presence of interference, a unit's outcome may depend on the entire assignment vector. Accordingly, each unit $i$ is associated with a potential outcome function $Y_i : \{0,1\}^n \to \mathbb{R}$, where $Y_i(\mathbf{z})$ denotes the outcome that would be realized for unit $i$ under assignment $\mathbf{z}$. Let $\mathbf{z}_{-i}$ denote the subvector of $\mathbf{z}$ excluding the $i$th component, so that we may equivalently write $Y_i(\mathbf{z}) = Y_i(z_i, \mathbf{z}_{-i})$. We work in a design-based framework, treating the potential outcomes $\{Y_i(\mathbf{\cdot})\}$ as fixed. The observed outcome for unit $i$ is $Y_i = Y_i(\mathbf Z) = Y_i(Z_i, \mathbf Z_{-i})$.

We are interested in estimating the Expected Average Treatment Effect (EATE) \citep{savje2021average}, defined as
\[
\tau = \frac{1}{n}\sum_{i=1}^n \E_{Z_j \stackrel{\operatorname{i.i.d.}}{\sim}
\operatorname{Bern}(\pi)} \big[ Y_i(1;\mathbf Z_{-i}) - Y_i(0;\mathbf Z_{-i}) \big].
\]
This estimand captures the average effect of changing unit $i$'s own treatment from $0$ to $1$, while averaging over the treatment assignments of all other units according to a specified distribution. In particular, it reflects the expected direct effect of treatment under a stochastic intervention on the rest of the population. The EATE is also referred to as the average direct effect (ADE) by \citet{hu2022average}. Under SUTVA, when there is no interference, the EATE reduces to the standard average treatment effect (ATE).

A key feature of the EATE is that it depends on the reference assignment distribution used to average over the potential outcomes. Throughout this paper, we define the EATE with respect to an i.i.d.\ Bernoulli$(\pi)$ assignment distribution. This choice is natural because it corresponds to the canonical randomized experiment and yields a particularly convenient interpretation. Specifically, under an i.i.d.\ Bernoulli$(\pi)$ reference assignment distribution, the EATE coincides with the average distributional shift effect (ADSE) \citep{savje2021average, hudgens2008toward}, $\tau_{\mathrm{ADSE}} = n^{-1} \sum_{i=1}^n \left( \E\left[Y_i \mid Z_i=1\right] - \E\left[Y_i \mid Z_i=0\right] \right). $ Thus, throughout this paper, the EATE can be interpreted as the average difference in expected outcomes between treated and control units under i.i.d.\ Bernoulli randomization.

To estimate the EATE defined above, a natural approach is to conduct the same Bernoulli experiment used to define the reference assignment distribution. Formally, we consider an assignment vector $\mathbf Z = (Z_1, \dots, Z_n)$ with $Z_i \stackrel{\text{i.i.d.}}{\sim} \mathrm{Bern}(\pi)$. Under this design, \citet{savje2021average} show that standard estimators developed under SUTVA for the ATE remain valid for estimating the EATE under limited interference. In particular, the Hájek estimator (also referred to as the difference-in-means estimator) is given by \footnote{Throughout the paper, we adopt the convention that whenever a denominator vanishes, the corresponding ratio is defined to be zero.

In addition, when no units are assigned to the treatment or control group, we set the estimator to zero. Note that this does not affect the estimator’s asymptotic properties, since the probability of a zero denominator vanishes in large samples.}

\[
\hat{\tau} = \frac{\sum_{i=1}^n Z_i Y_i}{\sum_{i=1}^n Z_i} - \frac{\sum_{i=1}^n (1 - Z_i)
Y_i}{\sum_{i=1}^n (1 - Z_i)}.
\]
They show that this estimator is consistent for $\tau$ under Bernoulli experiments, provided that interference is sufficiently limited, in the sense that the number of units dependent with any given unit grows at a controlled rate with $n$.

In this work, we focus on the Hájek estimator due to its simplicity and interpretability. We study whether additional covariate information can be used to design improved experiments, with the goal of reducing the variance of the Hájek estimator relative to the standard Bernoulli design. Importantly, our target estimand remains $\tau$ defined under Bern$(\pi)$ assignment. Thus, rather than changing the estimand, our goal is to design assignment mechanisms that yield more precise estimation of the same quantity.

\subsection{Rerandomization}
\label{sec:rerand}

Although randomized experiments guarantee asymptotically unbiased estimation of treatment effects, a given randomization may still yield substantial imbalance in observed covariates between treatment groups. Such imbalance can lead to imprecise estimates and undermine the credibility of experimental findings. A natural remedy, dating back to \citet{fisher1992arrangement} and \citet{cox1982randomization} and formalized by \citet{morgan2012rerandomization}, is \emph{rerandomization}: repeatedly drawing treatment assignments and retaining only those that achieve adequate covariate balance.

We consider a Bernoulli randomized experiment as the baseline assignment mechanism, under which each unit is independently assigned to treatment with probability $\pi$. Rerandomization modifies this design by conditioning on a covariate balance criterion. Suppose each unit $i$ is associated with a $p$-dimensional vector of pre-treatment covariates $\mathbf X_i = (X_{i1}, X_{i2}, \dots, X_{ip})^\top$. To implement rerandomization, one specifies a measure of covariate balance based on discrepancies between treatment and control groups. A natural choice is the covariate imbalance vector $\tauxhat = (\hat{\tau}_{x1}, \dots, \hat{\tau}_{xp})^\top$, where each component measures the difference in covariate means between the treated and control groups:
\begin{equation}
\label{eq:cov_adj}
\hat{\tau}_{xj} = \frac{\sumi Z_i X_{ij}}{\sumi Z_i} - \frac{\sumi (1 - Z_i) X_{ij}}{\sumi (1 -
Z_i)}.
\end{equation}

To quantify imbalance, we use the squared Mahalanobis distance
\begin{equation}
\label{eq:mahal_general}
M_n = (\tauxhat-\boldsymbol\tau_x)^\top \boldsymbol\Sigma_{xx}^{-1} (\tauxhat-\boldsymbol\tau_x),
\end{equation}
where $\boldsymbol\tau_x=\E[\tauxhat]$ and $\boldsymbol\Sigma_{xx}=\Var(\tauxhat)$ denote the expectation and covariance matrix of $\tauxhat$ under independent Bernoulli assignment. We can verify that $\boldsymbol\tau_x=\boldsymbol 0,$\footnote{This holds under the convention stated earlier for handling zero denominators: we set $\tauxhat = \boldsymbol 0$ when no units are assigned to the treatment or control group.} so that \eqref{eq:mahal_general} simplifies to
\begin{equation}
\label{eq:mahal}
M_n = \tauxhat^\top \boldsymbol\Sigma_{xx}^{-1} \tauxhat.
\end{equation}
Following \citet{morgan2012rerandomization}, we use this squared Mahalanobis distance as the balance criterion and accept only assignments with $M_n \le a$, for a prespecified threshold $a$.

The rerandomization criterion accepts a treatment assignment if and only if $M_n \le a$, where $a > 0$ is a prespecified tolerance threshold. Figure~\ref{fig:rerand_flow} illustrates this procedure. In practice, the threshold $a$ can be chosen to achieve a desired level of precision or variance reduction. In general, a larger $a$ leads to a higher acceptance probability but weaker covariate balance, whereas a smaller $a$ leads to stronger balance but requires more draws on average. At one extreme, when $a = \infty$, rerandomization reduces to unrestricted Bernoulli randomization. At the other extreme, choosing $a$ too small can overly restrict the set of acceptable assignments and make the design less robust to unobserved covariates or outcome components not included in the criterion, partially offsetting the benefits of rerandomization \citep{kapelner2021harmonizing}.

Under SUTVA, rerandomization has been studied extensively, and its theoretical properties are well established. We refer readers to Section \ref{subsection:related_work} for a detailed discussion of the literature.

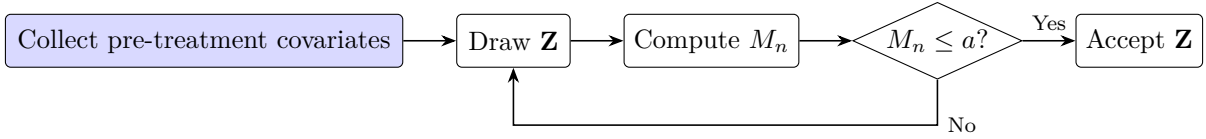
\begin{figure}
\centering
{\begin{tikzpicture}[
    node distance=0.7cm and 0.7cm,
    box/.style={rectangle, rounded corners=2pt, minimum width=1.5cm,
        minimum height=0.7cm, text centered, draw=black, font=\small},
    decision/.style={diamond, minimum width=1.6cm, minimum height=0.6cm,
        text centered, draw=black, inner sep=1pt, aspect=2.2, font=\small},
    arrow/.style={-{Stealth[length=2mm]}, semithick},
]
\node (pre)     [box, text width=5cm, fill=blue!15]
    {Collect pre-treatment covariates};
\node (draw)    [box, right=of pre] {Draw $\mathbf Z$};
\node (compute) [box, right=of draw] {Compute $M_n$};
\node (decide)  [decision, right=of compute] {$M_n \le a$?};
\node (accept)  [box, right=of decide] {Accept $\mathbf Z$};

\draw [arrow] (pre) -- (draw);
\draw [arrow] (draw) -- (compute);
\draw [arrow] (compute) -- (decide);
\draw [arrow] (decide) -- node[above, font=\scriptsize] {Yes} (accept);
\draw [arrow] (decide.south) -- ++(0,-0.6)
    node[right, font=\scriptsize] {No}
    -| (draw.south);
\end{tikzpicture}}
\caption{Rerandomization procedure: treatment assignments are repeatedly drawn until the Mahalanobis distance $M_n$ falls below the threshold~$a$.}
\label{fig:rerand_flow}
\end{figure}

\subsection{Interference-aware covariates}
\label{sec:interference_covariates}

In the presence of interference, the treatment assignments of other units can affect the outcome of unit~$i$. Consequently, certain features of the assignment vector $\mathbf Z$ may themselves be informative about the outcome and can therefore play the role of covariates in rerandomization.

To build intuition, consider the effect of a vaccine on an individual's health status. Receiving the vaccine is the treatment, and health status is the outcome. The outcome depends not only on whether the individual is vaccinated (the effect of interest), but also on pre-treatment characteristics such as baseline health, as well as on the vaccination status of others, for example, the proportion of their friends who are vaccinated. To obtain accurate estimates, it is helpful to balance both types of factors. Notably, the latter quantity is not fixed prior to the experiment; it depends on the treatment assignment vector $\mathbf Z$. This feature is absent under SUTVA.

Motivated by this, we construct a $p$-dimensional covariate vector $\mathbf X_i = (X_{i1}, X_{i2}, \dots, X_{ip})^\top$ for each unit~$i$ that incorporates two types of information. The first type consists of \textit{pre-treatment covariates}, which are fixed prior to treatment assignment and do not depend on $\mathbf Z$, such as demographic characteristics or baseline measurements. The second type consists of \textit{treatment-dependent covariates}, which are summaries derived from the assignment vector $\mathbf Z$. A canonical example is the proportion of units connected to unit~$i$ (under a posited interaction network) that are assigned to treatment. Note that the posited interaction network need not coincide with the true interference network for the theoretical results in Section~\ref{sec:asymptotic} (e.g., the no-harm guarantee) to hold. Naturally, more accurate guesses of the relevant interaction structure may lead to greater variance reduction. Figure~\ref{fig:rerand_flow_dynamic} illustrates the rerandomization procedure with treatment-dependent covariates.

Including treatment-dependent covariates introduces two key distinctions from the classical setting. First, under rerandomization with interference, each draw of $\mathbf Z$ requires recomputing $\mathbf X$ before evaluating the balance criterion. Second, because the covariates depend on the treatment, the covariate imbalance vector $\tauxhat$ may no longer be centered, i.e., $\mathbb{E}[\tauxhat] \neq \boldsymbol 0$. As a result, the Mahalanobis distance $M_n$ no longer takes the standard form in \eqref{eq:mahal}; instead, we use the more general form in \eqref{eq:mahal_general}, where both the centering and covariance are defined through expectations over the assignment mechanism. In practice, since $\tauxhat$ is a known function of $\mathbf Z$, both its mean and covariance can be approximated via Monte Carlo prior to the experiment.

\begin{figure}
\centering
\begin{tikzpicture}[
    node distance=0.7cm and 0.45cm,
    box/.style={rectangle, rounded corners=2pt, minimum width=1.5cm,
        minimum height=0.7cm, text centered, draw=black,
        font=\small},
    decision/.style={diamond, minimum width=1.6cm, minimum height=0.6cm,
        text centered, draw=black, inner sep=1pt, aspect=2.2, font=\small},
    arrow/.style={-{Stealth[length=2mm]}, semithick},
]
\node (pre)      [box, fill=blue!15, text width=2cm]
                 {Collect \textbf{pre-treatment} covariates};
\node (draw)     [box, right=of pre] {Draw $\mathbf Z$};
\node (computex) [box, right=of draw, fill=red!20,text width=3.8cm]
                 {Compute \textbf{treatment-dependent} covariates};
\node (compute)  [box, right=of computex] {Compute $M_n$};
\node (decide)   [decision, right=of compute] {$M_n \le a$?};
\node (accept)   [box, right=of decide] {Accept $\mathbf Z$};

\draw [arrow] (pre)      -- (draw);
\draw [arrow] (draw)     -- (computex);
\draw [arrow] (computex) -- (compute);
\draw [arrow] (compute)  -- (decide);
\draw [arrow] (decide)   -- node[above, font=\scriptsize] {Yes$\quad$} (accept);
\draw [arrow] (decide.south) -- ++(0,-0.7)
    node[right, font=\scriptsize] {No}
    -| (draw.south);
\end{tikzpicture}
\caption{Rerandomization procedure with treatment-dependent covariates. Pre-treatment covariates are collected before assignment. Treatment assignments are then repeatedly drawn, and treatment-dependent covariates are computed accordingly, until the Mahalanobis distance ($M_n$) falls below the threshold $a$.}
\label{fig:rerand_flow_dynamic}
\end{figure}
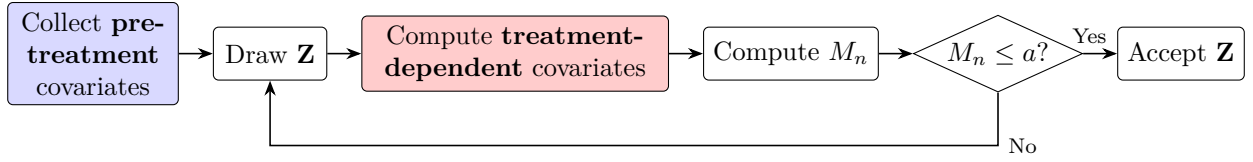

\section{Asymptotic Variance Reduction under Rerandomization}
\label{sec:asymptotic}

In Section \ref{sec:setup}, we introduced a generalized rerandomization framework that incorporates both pre-treatment and treatment-dependent covariates to capture network interference. While the algorithmic procedure is straightforward, a fundamental theoretical question remains: does restricting the randomization space based on these interference-aware covariates actually improve the precision of the H\'ajek estimator?

In this section, we provide an affirmative answer by establishing the asymptotic variance reduction achieved by rerandomization under interference. The primary theoretical challenge in our setting lies in the complex cross-unit dependence induced by both the interference structure and the treatment-dependent covariates. Section~\ref{sec:regularities} introduces the regularity conditions required for our results. Section~\ref{subsec:network-locality} then imposes a simple sparse dependency-graph condition under which exact independence holds across sufficiently separated blocks of units. Finally, Section~\ref{sec:asymptotic_vr} establishes a joint central limit theorem for the treatment effect estimator and the covariate imbalance vector, culminating in Theorem~\ref{thm:rerand_dist} and Corollary~\ref{thm:variance_reduction}, which quantify the variance reduction achieved by rerandomization. Appendix~\ref{app:approximate-locality} gives a more general alternative that permits weak dependence between every pair of units, provided that the strength of the dependence decays sufficiently rapidly along a reference graph.

\subsection{Regularity conditions}
\label{sec:regularities}
Write $\boldsymbol{\phi}_i(\mathbf z) = \bigl(Y_i(\mathbf z),\, \mathbf X_i(\mathbf z)^\top\bigr)^\top \in \mathbb R^{d_\phi}$ for the outcome--covariate vector of unit $i$, where $d_\phi=p+1$ denotes its dimension.

\begin{assumption}[Variance stability]\label{ass:variance_stability}
There exist sequences $\{\sigma_{n,y}\}_n, \{\sigma_{n,x_1}\}_n, \dots, \{\sigma_{n,x_p}\}_n$ such that, as $n\to\infty$, the following conditions hold under the i.i.d. Bernoulli design $Z_{i} \stackrel{\operatorname{i.i.d.}}{\sim} \operatorname{Bern}(\pi)$.

\begin{enumerate}[label=(\roman*)]
\item Variance convergence.
\begin{align*}
\Var(\hat{\tau}-\tau)\,/\,\sigma_{n,y}^2 &\to 1, \\
\Var(\hat{\tau}_{x_j}-\tau_{x_j})\,/\,\sigma_{n,x_j}^2 &\to 1, \quad \forall\, j \in [p].
\end{align*}

\item Covariance convergence. There exist constants $c_j$ and $c_{jk}$ such that
\begin{align*}
\Cov(\hat{\tau}-\tau,\;\hat{\tau}_{x_j}-\tau_{x_j}) \,/\, (\sigma_{n,y}\,\sigma_{n,x_j}) &\to c_j,
\quad \forall\, j \in [p], \\
\Cov(\hat{\tau}_{x_j}-\tau_{x_j},\;\hat{\tau}_{x_k}-\tau_{x_k}) \,/\,
(\sigma_{n,x_j}\,\sigma_{n,x_k}) &\to c_{jk}, \quad \forall\, j,k \in [p].
\end{align*}

\item Non-degeneracy. Define the $d_\phi \times d_\phi$ limiting correlation matrix
\[
\mathbf V = \begin{pmatrix} 1 & c_1 & c_2 & \cdots & c_p \\
c_1 & 1 & c_{12} & \cdots & c_{1p} \\
c_2 & c_{12} & 1 & \cdots & c_{2p} \\
\vdots & \vdots & \vdots & \ddots & \vdots \\
c_p & c_{1p} & c_{2p} & \cdots & 1 \end{pmatrix}.
\]
The matrix $\mathbf V$ is positive definite.

\end{enumerate}
\end{assumption}

Conditions~(i) and~(ii) require that the second-moment structure of the estimators stabilizes after appropriate normalization. Condition~(iii) is a standard non-degeneracy requirement: it ensures that no coordinate of the covariate imbalance vector is a deterministic linear function of the remaining coordinates, and that the variation in the treatment effect estimator cannot be perfectly explained by the observed covariates.

\begin{assumption}[Uniform boundedness]
\label{ass:boundedness}
There exists a constant $C_{\phi}<\infty$ such that
\[
\sup_{n\ge 1}\, \max_{\mathbf z\in\{0,1\}^n}\, \max_{i\in[n]} \bigl\| \boldsymbol{\phi}_i(\mathbf z)
\bigr\|_{\infty} \le C_{\phi}.
\]
\end{assumption}

\begin{assumption}[Lower bound on variance scales]
\label{ass:variance_lower_bound}
The sequences in Assumption~\ref{ass:variance_stability} satisfy $\sigma_{n,y} \geq \frac{C}{\sqrt{n}}, \sigma_{n,x_j} \geq \frac{C}{\sqrt{n}}, \forall j \in [p]$ for some constant $C > 0$.
\end{assumption}

Assumption \ref{ass:boundedness} is natural in many applications, since outcomes and covariates often represent physical, biological, or socio-economic quantities, such as blood pressure or test scores, that are inherently bounded or at least uniformly controlled.

Assumption~\ref{ass:variance_lower_bound} requires that the standard deviations do not decay faster than the classical parametric rate. Under SUTVA, variances typically scale at the order $\mathcal O(1/n)$, so the corresponding standard deviations are of order $\mathcal O(n^{-1/2})$. Under network interference, the variance is generally expected to be larger, not smaller, than this benchmark. This lower bound therefore only rules out pathological cases in which strong negative dependence causes the variance to vanish at a super-parametric rate.

\subsection{Limited dependence through a dependency graph}
\label{subsec:network-locality}

We next formalize limited dependence among the unit-level random vectors. For each unit $i$, define the augmented random vector
\[
\mathbf W_{n,i}
=
\begin{pmatrix}
Z_i\\
\boldsymbol\phi_i(\mathbf Z)
\end{pmatrix}
\in\mathbb R^{d_\phi+1}.
\]

\begin{definition}[Dependency graph]
\label{def:dependency-graph}
An undirected graph $\mathcal G_n^{\mathrm{dep}}=([n],E_n^{\mathrm{dep}})$ with a self-loop at every vertex is a dependency graph for $\{\mathbf W_{n,i}:i\in[n]\}$ if, for every pair of disjoint subsets $\mathcal I,\mathcal J\subseteq[n]$, the absence of any edge joining a vertex in $\mathcal I$ to a vertex in $\mathcal J$ implies $\{\mathbf W_{n,i}:i\in\mathcal I\} \indep \{\mathbf W_{n,j}:j\in\mathcal J\}$.
\end{definition}
A dependency graph is not unique, because adding edges preserves the defining property. For notational convenience, let $\mathbf A_n^{\mathrm{dep}}=(A_{n,ij}^{\mathrm{dep}})_{i,j\in[n]}$ denote its binary adjacency matrix, so that $A_{n,ii}^{\mathrm{dep}}=1$ for every $i\in[n]$. Using the adjacency-row-sum convention, under which the self-loop contributes one, define the degree of unit $i$ by $d_{n,i}^{\mathrm{dep}}=\sum_{j=1}^n A_{n,ij}^{\mathrm{dep}}$.

\begin{assumption}[Sparse dependency graph]
\label{ass:network_sparsity}
Under the i.i.d. Bernoulli design $Z_i\stackrel{\operatorname{i.i.d.}}{\sim}\operatorname{Bern}(\pi)$, there exist a constant $\delta>0$ and a sequence of dependency graphs $\{\mathcal G_n^{\mathrm{dep}}\}_{n\ge1}$ for $\{\mathbf W_{n,i}:i\in[n]\}$ such that
\[
\max_{i\in[n]}d_{n,i}^{\mathrm{dep}}
=
\mathcal O\bigl(n^{1/2-\delta}\bigr).
\]
\end{assumption}

Assumption~\ref{ass:network_sparsity} requires the size of every dependency neighborhood to grow more slowly than $\sqrt n$ by a polynomial margin. The dependency graph is a theoretical description of the dependence induced by the Bernoulli assignment through the potential outcomes and the treatment-dependent covariates. It need not be known to the experimenter and need not coincide with the posited interaction network used to construct the covariates in the rerandomization criterion.

\begin{example}[Neighborhood interference]
Suppose there exists an interference graph $\mathcal H_n$ such that, for each unit $i$, the potential outcome $Y_i(\mathbf z)$ depends on the assignment $\mathbf z$ only through the treatments in its closed one-hop neighborhood $\mathcal N_{\mathcal H_n}[i]$, which includes unit $i$ itself. That is, for some fixed function $f_{n,i}$, $\boldsymbol\phi_i(\mathbf z) = f_{n,i}\!\left(\mathbf z_{\mathcal N_{\mathcal H_n}[i]}\right)$. Suppose the treatment-dependent covariates of unit $i$ have the same local dependence structure. Let $\mathcal H_n^2$ denote the self-loop-augmented square of $\mathcal H_n$, in which every vertex has a self-loop and two distinct vertices are adjacent if their distance in $\mathcal H_n$ is at most two. Under the i.i.d.\ Bernoulli design, $\mathcal H_n^2$ is a valid dependency graph for $\{\mathbf W_{n,i}:i\in[n]\}$.

Let $h_n=\max_{i\in[n]}\deg_{\mathcal H_n}(i)$ denote the maximum degree of $\mathcal H_n$, counting the self-loop at each vertex. Since the maximum degree of $\mathcal H_n^2$ is at most $1 + (h_n-1) + (h_n-1)(h_n-2) =  1+(h_n-1)^2$, Assumption~\ref{ass:network_sparsity} holds if $h_n=\mathcal O(n^{1/4-\epsilon})$ for some $\epsilon\in(0,1/4)$. More generally, it suffices that the largest closed two-hop neighborhood in $\mathcal H_n$ has size $\mathcal O(n^{1/2-\delta})$. In particular, the assumption holds when $\mathcal H_n$ has uniformly bounded maximum degree.
\end{example}

Assumption \ref{ass:network_sparsity} above is deliberately kept simple for practical considerations. One drawback is it does not apply to weak long-range interference under which every unit may be affected by every treatment assignment, but in which case, it turns out that our results below still hold. Appendix~\ref{app:approximate-locality} gives a more general graph-based weak-dependence condition, under which, assignments outside a growing graph neighborhood may still have nonzero effects, provided that their residual influence decays sufficiently rapidly relative to neighborhood growth. The same joint central limit theorem and the resulting rerandomization conclusions continue to hold under this alternative condition.

\subsection{Asymptotic variance reduction}
\label{sec:asymptotic_vr}

\begin{theorem}[Joint asymptotic normality]
\label{thm:joint-normality}
Suppose that $Z_i\overset{\mathrm{i.i.d.}}{\sim}\operatorname{Bernoulli}(\pi)$. Under Assumptions~\ref{ass:variance_stability}--\ref{ass:network_sparsity}, recall the variance scales $\sigma_{n,y},\sigma_{n,x_1},\ldots,\sigma_{n,x_p}$ from Assumption~\ref{ass:variance_stability}, and define the diagonal scaling matrix
\begin{equation}
\mathbf D_n
=
\operatorname{diag}\bigl(\sigma_{n,y}^{-1},\sigma_{n,x_1}^{-1},\ldots,\sigma_{n,x_p}^{-1}\bigr).
\label{eq:scaling-matrix}
\end{equation}
Then
\[
\mathbf D_n
\begin{pmatrix}
\hat\tau-\tau\\
\hat{\boldsymbol\tau}_x-\boldsymbol\tau_x
\end{pmatrix}
\overset{d}{\to} N(\boldsymbol 0,\mathbf V),
\]
where $\mathbf V$ is the $d_\phi\times d_\phi$ limiting correlation matrix defined in Assumption~\ref{ass:variance_stability}.
\end{theorem}

Theorem~\ref{thm:joint-normality} characterizes the joint asymptotic distribution of the Hájek estimator and $\tauxhat$ under Bernoulli treatment assignment. In particular, it shows that these quantities are jointly asymptotically Gaussian. This result is the key input for the rerandomization analysis, since the distribution of the Hájek estimator under rerandomization is obtained by conditioning on the event $M_n \le a$, which is determined by the covariate imbalance vector.

The following theorem follows from Theorem~\ref{thm:joint-normality} and characterizes the asymptotic distribution of the H\'ajek estimator under rerandomization, that is, its limiting distribution conditional on the acceptance event $M_n \le a$.

\begin{theorem}[Asymptotic distribution under rerandomization]
\label{thm:rerand_dist}
Partition $\mathbf V$ conformably with $(\hat{\tau}, \tauxhat^\top)^\top$ as
\begin{equation}
\label{eqn:partition}
\mathbf V =
\begin{pmatrix}
V_{11} & \mathbf V_{12}\\
\mathbf V_{21} & \mathbf V_{22} \end{pmatrix}.
\end{equation}
Suppose the conditions of Theorem~\ref{thm:joint-normality} hold, and define the asymptotic squared multiple correlation
\begin{equation}
\label{eqn:R2_definition}
R^{2}\;=\;\frac{\mathbf V_{12}\,\mathbf V_{22}^{-1}\,\mathbf V_{21}}{V_{11}}.
\end{equation}
Let $\varepsilon_0\sim N(0,1)$, and let $L_{p,a}$ be independent of $\varepsilon_0$ with $L_{p,a}\overset{d}{=}J_1\mid\{\mathbf J^\top \mathbf J\le a\}$, where $\mathbf J=(J_1,\ldots,J_p)^\top\sim N(\boldsymbol 0,\mathbf I_p)$. Then
\begin{equation}
\label{eqn:limit_unscaled}
\sigma_{n,y}^{-1}\,(\hat\tau-\tau)\;\big|\;M_n\le a \;\xrightarrow{\;d\;}\;
\sqrt{1-R^{2}}\;\varepsilon_0 \;+\; \sqrt{R^{2}}\;L_{p,a}.
\end{equation}
\end{theorem}

As consequences of Theorem~\ref{thm:rerand_dist}, Corollaries~\ref{cor:unbiased} and~\ref{thm:variance_reduction} characterize the asymptotic mean and variance of the H\'ajek estimator under rerandomization. For this purpose, we introduce the following notation for asymptotic moments. Let $\mathbb{E}_{\operatorname{asymp}}(\cdot)$, $\Var_{\operatorname{asymp}}(\cdot)$, and $\Cov_{\operatorname{asymp}}(\cdot)$ denote the asymptotic expectation, variance, and covariance, respectively. Specifically, let $T_n$ be a sequence of estimators for a parameter $\theta$. Suppose there exists a deterministic sequence $c_n \to \infty$ such that $c_n (T_n - \theta) \xrightarrow{d} \mathcal{N}(0, V)$. We define $\mathbb{E}_{\operatorname{asymp}}(T_n) = \theta$, $\Var_{\operatorname{asymp}}(T_n) = c_n^{-2} V$. More generally, consider a pair $(T_n, A_n)$, where $A_n$ is an event. Suppose that the conditional distribution of $c_n (T_n - \theta)$ given $A_n$ converges weakly to a limiting distribution with mean 0 and finite variance $\Sigma$. We define the asymptotic conditional expectation and variance as $\mathbb{E}_{\operatorname{asymp}}(T_n \mid A_n) = \theta$, and $\Var_{\operatorname{asymp}}(T_n \mid A_n) = c_n^{-2}\Sigma.$

\begin{corollary}[Asymptotic unbiasedness under rerandomization]
\label{cor:unbiased}
Under the conditions of Theorem~\ref{thm:joint-normality},
\[
\mathbb{E}_{\operatorname{asymp}}(\hat{\tau} \mid M_n \le a) = \tau.
\]
\end{corollary}

Corollary~\ref{cor:unbiased} shows that rerandomization preserves the asymptotic unbiasedness of the H\'ajek estimator. One difference from the SUTVA setting is the following: under SUTVA, $\boldsymbol\tau_x = \boldsymbol 0$, so using zero as the reference mean in $M_n$ is correct. With treatment-dependent covariates, however, $\boldsymbol\tau_x$ may be nonzero. If $M_n$ is constructed using an incorrect mean (e.g., zero when $\boldsymbol\tau_x \neq \boldsymbol 0$), then in general $\mathbb{E}_{\operatorname{asymp}}(\hat{\tau} \mid M_n \le a) \neq \tau$.

\begin{corollary}[Asymptotic variance reduction under rerandomization]
\label{thm:variance_reduction}
Let $\chi^2_p$ denote a chi-square random variable with $p$ degrees of freedom, and define
\[
v_p = \frac{\Pr(\chi^2_{p+2} \le a)}{\Pr(\chi^2_p \le a)}.
\]
Under the conditions of Theorem~\ref{thm:joint-normality}, the asymptotic variance of $\hat{\tau}$ under rerandomization satisfies
\[
\Var_{\operatorname{asymp}}(\hat{\tau} \mid M_n \le a) = \Var_{\operatorname{asymp}}(\hat{\tau})
\bigl[1 - (1 - v_p) R^2\bigr].
\]
\end{corollary}

Corollary~\ref{thm:variance_reduction} formally quantifies the efficiency gain delivered by rerandomization. Specifically, the asymptotic variance of the Hájek estimator under rerandomization, $\Var_{\operatorname{asymp}}(\hat{\tau}\mid M_n\le a)$, is no larger than that under the original Bernoulli design $\Var_{\operatorname{asymp}}(\hat{\tau})$. Thus, rerandomization enjoys an asymptotic no-harm property even in the presence of interference. Moreover, this guarantee does not rely on correct specification of the interference mechanism. The procedure is largely model agnostic, requiring only mild assumptions on the dependence structure.

Corollary~\ref{thm:variance_reduction} also shows that the magnitude of the variance reduction depends on the strictness of the balance criterion, determined by the threshold $a$, and the predictive power of the covariates, summarized by $R^2$. Under Bernoulli assignment,
\begin{equation}
\label{eqn:R2_definition0}
R^2 = \frac{ \Cov_{\operatorname{asymp}}(\hat{\tau}, \tauxhat)^\top
\Var_{\operatorname{asymp}}(\tauxhat)^{-1} \Cov_{\operatorname{asymp}}(\hat{\tau}, \tauxhat) }{
\Var_{\operatorname{asymp}}(\hat{\tau}) }
\end{equation}
corresponds to the squared multiple correlation between $\hat{\tau}$ and $\tauxhat$, capturing the fraction of the asymptotic variance of $\hat{\tau}$ that can be linearly explained by $\tauxhat$.

To further interpret the value of $R^2$, we consider two simple examples. The derivations are given in Appendix~\ref{app:r2-examples}.

\begin{example}[SUTVA with one baseline covariate]\label{ex:r2-sutva}
Suppose treatments are assigned independently, $Z_i \stackrel{\mathrm{i.i.d.}}{\sim} \mathrm{Bernoulli}(\pi)$. Let $X_i \stackrel{\mathrm{i.i.d.}}{\sim} N(0,\sigma_X^2)$ and $\varepsilon_i \stackrel{\mathrm{i.i.d.}}{\sim} N(0,\sigma_\varepsilon^2)$. Assume that the collections $\{X_i\}_{i=1}^n$, $\{\varepsilon_i\}_{i=1}^n$, and $\mathbf Z$ are mutually independent. Consider the SUTVA model
\[
Y_i(\mathbf{z})=\alpha+\tau z+\beta X_i+\varepsilon_i, \qquad z\in\{0,1\}.
\]
The EATE is $\tau$. If rerandomization balances the baseline covariate $X_i$, then the $R^2$ in \eqref{eqn:R2_definition} is
\begin{equation}\label{eq:r2-sutva-main} R^2 = \frac{\beta^2\sigma_X^2}
{\beta^2\sigma_X^2+\sigma_\varepsilon^2}.
\end{equation}
Thus, in this SUTVA linear model, the $R^2$ governing the rerandomization variance coincides with the marginal unit-level predictive $R^2$ for the prognostic component of the outcome:
\[
\frac{\Var\{\E(Y_i(\mathbf{z})\mid X_i)\}}{\Var\{Y_i(\mathbf{z})\}} = \frac{\beta^2\sigma_X^2}
{\beta^2\sigma_X^2+\sigma_\varepsilon^2}.
\]
\end{example}

\begin{example}[Pair interference]\label{ex:r2-pair}
Suppose $n$ is even and units form $n/2$ disjoint pairs. Let $m(i)$ denote the pair-mate of unit $i$. Treatments are assigned independently, $Z_i \stackrel{\mathrm{i.i.d.}}{\sim} \mathrm{Bernoulli}(\pi)$. Let $X_i \stackrel{\mathrm{i.i.d.}}{\sim} N(0,\sigma_X^2)$ and $\varepsilon_i \stackrel{\mathrm{i.i.d.}}{\sim} N(0,\sigma_\varepsilon^2)$, and assume that the collections $\{Z_i\}_{i=1}^n$, $\{X_i\}_{i=1}^n$, and $\{\varepsilon_i\}_{i=1}^n$ are mutually independent. Define the pair exposure $H_i(\mathbf z)=z_{m(i)}$ and $H_i=H_i(\mathbf Z)=Z_{m(i)}$. Consider the potential-outcome model
\[
Y_i(\mathbf{z}) = \alpha+\tau z_i+\beta X_i+\lambda H_i(\mathbf z)+\varepsilon_i .
\]
The EATE is $\tau$. Let $\tauxhat=(\hat\tau_X,\hat\tau_H)^\top$ denote the H\'ajek imbalance vector for $X_i$ and $H_i$, centered in the Mahalanobis criterion at its Bernoulli expectation. If rerandomization balances this two-dimensional vector, then
\begin{equation}\label{eq:r2-pair-main}
R^2 = \frac{\beta^2\sigma_X^2+2\pi(1-\pi)\lambda^2}
{\beta^2\sigma_X^2+2\pi(1-\pi)\lambda^2+\sigma_\varepsilon^2}.
\end{equation}

For comparison, consider the marginal unit-level predictive $R^2$ obtained by focusing only on a single unit's prognostic component, $Y_i(\mathbf Z)-\tau Z_i = \alpha+\beta X_i+\lambda H_i+\varepsilon_i$, and asking how much of its marginal variance is explained by $(X_i,H_i)$. Since $X_i$ is independent of $H_i$ and $\Var(H_i)=\pi(1-\pi)$, this unit-level quantity is
\begin{equation}\label{eq:r2-pair-unit} R^2_{\operatorname{unit}} =
\frac{\beta^2\sigma_X^2+\pi(1-\pi)\lambda^2}
{\beta^2\sigma_X^2+\pi(1-\pi)\lambda^2+\sigma_\varepsilon^2}.
\end{equation}
The $R^2$ in \eqref{eq:r2-pair-main} is therefore not the same as the marginal unit-level predictive quantity in \eqref{eq:r2-pair-unit}: the exposure term is $2\pi(1-\pi)\lambda^2$ in \eqref{eq:r2-pair-main}, but only $\pi(1-\pi)\lambda^2$ in \eqref{eq:r2-pair-unit}. Thus, balancing the treatment-dependent exposure can reduce the variance of the treatment-effect estimator more than its marginal unit-level predictive power alone would suggest.
\end{example}

These examples illustrate that the $R^2$ in Corollary~\ref{thm:variance_reduction} is a design-based measure of the association between the estimator $\hat\tau$ and the imbalance vector $\tauxhat$, rather than merely a marginal measure of how well the covariates predict a single unit's outcome. Under SUTVA with a baseline covariate, the two notions coincide: balancing $X_i$ removes the same component of variation that $X_i$ explains in the unit-level outcome. Under pair interference, however, the exposure $H_i = Z_{m(i)}$ is itself a function of the treatment assignment. Although $H_i$ has marginal variance $\pi(1-\pi)$, the relevant quantity for rerandomization is the treated-minus-control imbalance in $H_i$. To first order, the two units in the same pair contribute the same centered term to this imbalance. This within-pair alignment doubles the asymptotic variance contribution of the exposure imbalance, relative to an independent baseline covariate with the same marginal variance. This is the source of the factor $2\pi(1-\pi)\lambda^2$ in \eqref{eq:r2-pair-main}.

Finally, we make a few remarks on Corollary~\ref{thm:variance_reduction}. First, the proportional variance reduction is $(1 - v_p) R^2$, so rerandomization is most effective when the observed covariates are highly predictive of the outcome. Second, choosing a smaller threshold $a$ improves covariate balance and reduces $v_p$, but it also lowers the acceptance probability, defined as $\mathbb{P}(M_n \le a)$, the probability that a random Bernoulli assignment satisfies the rerandomization criterion. When $a$ becomes too small, this probability can be very low, increasing the computational cost of obtaining an acceptable assignment. Moreover, if $a$ vanishes too quickly relative to $n$, the standard rerandomization asymptotics may no longer apply and a separate analysis is required \citep{wang2022rerandomization}. In practice, $a$ should balance statistical efficiency with computational feasibility rather than be chosen arbitrarily close to zero. Third, the variance reduction factor $v_p$ depends on the covariate dimension $p$. For a fixed acceptance probability and a fixed $R^2$, the efficiency gain becomes smaller as $p$ increases. This highlights the importance of selecting covariates that are genuinely predictive, rather than including all available pre-treatment variables indiscriminately, especially in higher-dimensional settings.

\section{Conservative Variance Estimator}
\label{sec:conservative}

In this section, we construct conservative variance estimators for the Hájek estimator under rerandomization in the presence of interference. Under SUTVA, variance estimation for rerandomized experiments is relatively well understood. A common approach is to estimate the variance under the baseline randomization mechanism, together with the squared multiple correlation $R^2$ between the treatment-effect estimator and the covariate imbalance, and then plug these quantities into the rerandomization variance formula \citep{morgan2012rerandomization, li2018asymptotic}. In the presence of interference, this strategy becomes more delicate. As discussed after Corollary~\ref{thm:variance_reduction}, the quantity $R^2$ no longer has the same simple interpretation as in the SUTVA setting: it depends on the joint asymptotic covariance between the outcome estimator and the covariate-imbalance estimator, which may involve cross-unit dependence and, in our setting, treatment-dependent covariates. Consequently, estimating $R^2$ can be substantially more challenging. This is closely related to the difficulty of constructing the optimal regression adjustment, since both problems rely on estimating the same covariance structure between the treatment-effect estimator and the covariate imbalance. Under interference, this covariance structure becomes considerably more complicated, making both optimal adjustment and consistent estimation of $R^2$ substantially more difficult.

This motivates a construction that avoids the need for an explicit plug-in estimator of $R^2$. Recall that the asymptotic variance under rerandomization can be viewed as the conditional variance of the estimator under the baseline Bernoulli assignment mechanism, given that the balance criterion $M_n \le a$ is satisfied. This perspective suggests a two-step procedure. We first study variance estimation under Bernoulli assignment in Section~\ref{sec:conservative_var}, where we propose an estimator that is asymptotically conservative and often substantially less conservative than existing alternatives. We then turn to Section~\ref{sec:opt}, where we show how to transform a conservative variance estimator under Bernoulli assignment into a valid variance estimator under rerandomization. This transformation is based on an optimization formulation whose solution yields the desired conservative bound. Importantly, the development in Section~\ref{sec:opt} is general: it does not rely on the particular estimator proposed in Section~\ref{sec:conservative_var}, and can be applied to any estimator satisfying a suitable joint conservativeness property.

We note that in the presence of interference, variance estimation is typically more challenging than estimating the target quantity itself. In this section, we assume that the experimenter has access to some additional information about the dependence structure of the potential outcomes, which allows for the construction of conservative variance estimators.

\subsection{Variance estimation under Bernoulli assignment}
\label{sec:conservative_var}
Let $\hat{\boldsymbol{\eta}} \in \mathbb{R}^{d_\phi}$ denote the centered vector of treatment-effect estimator and covariate imbalance
\[
\hat{\boldsymbol{\eta}} =
\begin{pmatrix}
\hat{\tau} - \tau \\
\hat{\tau}_{x_1} - \tau_{x_1} \\
\vdots \\
\hat{\tau}_{x_p} - \tau_{x_p} \end{pmatrix}.
\]
In this subsection, we construct a conservative covariance estimator for $\mathbf{U}_n=\Var(\hat{\boldsymbol{\eta}})$ (here the subscript $n$ comes from sample size $n$) under Bernoulli assignment.

Let
\[
\bar{\boldsymbol{\phi}}_1 = \frac{\sum_{i=1}^n Z_i \boldsymbol{\phi}_i}{\sum_{i=1}^n Z_i}, \qquad
\bar{\boldsymbol{\phi}}_0 = \frac{\sum_{i=1}^n (1-Z_i)\boldsymbol{\phi}_i}{\sum_{i=1}^n (1-Z_i)},
\]
and define the sample within-arm centered vector
\begin{equation}
\label{eq:centered_outcomes}
\hat{\boldsymbol{\phi}}_i = \boldsymbol{\phi}_i - Z_i \bar{\boldsymbol{\phi}}_1 -
(1-Z_i)\bar{\boldsymbol{\phi}}_0.
\end{equation}

To motivate our construction, consider first the classical setting with no interference and only pre-treatment covariates. Writing $n_1=\sum_{i=1}^n Z_i$, $n_0=\sum_{i=1}^n (1-Z_i)$, the usual Neyman covariance estimator is
\[
\hat{\Var}_{\mathrm{Neyman}}(\hat{\boldsymbol{\eta}}) = \frac{\hat{\boldsymbol{\Sigma}}_1}{n_1}
+ \frac{\hat{\boldsymbol{\Sigma}}_0}{n_0},
\]
where
\[
\hat{\boldsymbol{\Sigma}}_1 = \frac{1}{n_1-1}\sum_{i:Z_i=1}
\bigl(\boldsymbol{\phi}_i-\bar{\boldsymbol{\phi}}_1\bigr)
\bigl(\boldsymbol{\phi}_i-\bar{\boldsymbol{\phi}}_1\bigr)^\top, \qquad \hat{\boldsymbol{\Sigma}}_0 =
\frac{1}{n_0-1}\sum_{i:Z_i=0} \bigl(\boldsymbol{\phi}_i-\bar{\boldsymbol{\phi}}_0\bigr)
\bigl(\boldsymbol{\phi}_i-\bar{\boldsymbol{\phi}}_0\bigr)^\top.
\]
Under Bernoulli assignment, the asymptotic covariance of $\hat{\boldsymbol{\eta}}$ is
\[
\Var_{\mathrm{asymp}}(\hat{\boldsymbol{\eta}}) = \frac{1}{n} \left(
\frac{\boldsymbol{\Sigma}_1}{\pi} + \frac{\boldsymbol{\Sigma}_0}{1-\pi} -
\boldsymbol{\Sigma}_{\Delta} \right),
\]
where $\boldsymbol{\Sigma}_1$ and $\boldsymbol{\Sigma}_0$ are the finite-population covariance matrices of $\boldsymbol{\phi}_i(1)$ and $\boldsymbol{\phi}_i(0)$, and $\boldsymbol{\Sigma}_{\Delta}$ is the finite-population covariance matrix of the treatment-effect vector $\boldsymbol{\phi}_i(1)-\boldsymbol{\phi}_i(0)$. Thus, the Neyman estimator is conservative because it estimates the first two terms and omits the negative semidefinite correction $-\boldsymbol{\Sigma}_{\Delta}/n$. A straightforward algebraic manipulation shows that the Neyman estimator $\hat{\Var}_{\mathrm{Ney}}(\hat{\boldsymbol{\eta}})$ is asymptotically equivalent to
\[
\hat{\mathbf{U}}_n^{\operatorname{SUTVA}} = \frac{1}{n^2} \sum_{i=1}^n \left( \frac{Z_i}{\pi^2} +
\frac{1-Z_i}{(1-\pi)^2} \right) \hat{\boldsymbol{\phi}}_i\hat{\boldsymbol{\phi}}_i^\top,
\]
since $n_1/n \to \pi$ and $n_0/n \to 1-\pi$. It follows that $\hat{\mathbf{U}}_n^{\operatorname{SUTVA}}$ is also asymptotically conservative for $\Var_{\mathrm{asymp}}(\hat{\boldsymbol{\eta}})$ in the absence of interference.

However, in our setting with interference and treatment-dependent covariates, this construction no longer directly applies. In particular, both cross-unit dependence and treatment-induced variation in the covariates introduce additional covariance components that are not captured by $\hat{\boldsymbol{\phi}}_i \hat{\boldsymbol{\phi}}_i^\top$. This motivates the need for a more refined construction that accounts for these dependencies while retaining conservativeness.

To account for dependence across units, we assume that the experimenter has access to an undirected graph $\mathcal G_n=([n],E_n)$ with self-loops, which provides information of the dependence between units. Let $\mathbf A_n$ denote its binary adjacency matrix, so that $A_{n,ii}=1$ for every $i\in[n]$. This graph encodes the dependence structure among units, with $A_{n,ij}=1$ indicating that units $i$ and $j$ may be dependent. Below, we require $\mathcal G_n$ to be a dependency graph in the sense of Definition~\ref{def:dependency-graph}. Following the same adjacency-row-sum convention, let $d_{n,i}=\sum_{j=1}^n A_{n,ij}$ denote the degree of unit $i$ in $\mathcal G_n$. To incorporate this dependence, we inflate the variance estimator according to the degree of each unit.
\begin{definition}[Covariance estimator under Bernoulli assignment]
\label{def:Un_estimator}
Define the estimator $\hat{\mathbf{U}}_n$ as
\begin{equation} \label{eq:estimator_Un}
\hat{\mathbf{U}}_n = \frac{1}{n^2} \sum_{i=1}^n d_{n,i} \left( \frac{Z_i}{\pi^2} +
\frac{(1-Z_i)}{(1-\pi)^2} \right) \hat{\boldsymbol{\phi}}_i \hat{\boldsymbol{\phi}}_i^\top.
\end{equation}
\end{definition}

We now establish the asymptotic conservativeness of $\hat{\mathbf U}_n$.

\begin{theorem}[Asymptotic conservativeness of the covariance estimator]
\label{thm:conservative_matrix}
Under the i.i.d. Bernoulli design $Z_i\stackrel{\operatorname{i.i.d.}}{\sim}\operatorname{Bernoulli}(\pi)$, suppose that Assumptions~\ref{ass:variance_stability}--\ref{ass:variance_lower_bound} hold, $\mathcal G_n$ is a dependency graph in the sense of Definition~\ref{def:dependency-graph} and that, for some constant $\delta>0$,
\[
\max_{i\in[n]}d_{n,i}
=
\mathcal O\bigl(n^{1/2-\delta}\bigr).
\]
Then, for every $\varepsilon>0$,
\begin{equation}
\label{eq:psd_conservative}
\mathbb{P}\!\left\{ \mathbf
D_n\hat{\mathbf{U}}_n\mathbf D_n - \mathbf V + \varepsilon \mathbf I_{d_\phi} \succeq \mathbf{0}
\right\} \to 1.
\end{equation}
\end{theorem}

Theorem~\ref{thm:conservative_matrix} shows that, after the natural coordinate-wise scaling by $\mathbf D_n$, the estimator $\hat{\mathbf U}_n$ dominates the limiting correlation matrix $\mathbf V$ in the positive semidefinite order up to an arbitrarily small asymptotic slack. In particular, for any fixed vector $\mathbf a \in \mathbb{R}^{d_\phi}$, $\mathbf a^\top \hat{\mathbf{U}}_n\mathbf a$ is an asymptotically conservative estimator for the variance of $\mathbf a^\top \hat{\boldsymbol{\eta}}$. Thus, the theorem yields more than marginal variance bounds for $\hat{\tau}$ and the coordinates of $\hat{\boldsymbol{\tau}}_x$: it provides a joint covariance bound for all linear combinations of these quantities.

The graph $\mathcal G_n$ need not be the sparsest possible dependency graph. It may contain additional edges, provided that its maximum degree continues to satisfy the sparsity condition in Theorem~\ref{thm:conservative_matrix}.

Compared with existing conservative variance estimators \citep[e.g.,][]{Samii2012EstimatingAC, savje2021average}, our construction has two notable features. First, we use the group-centered vectors $\hat{\boldsymbol{\phi}}_i$ rather than the raw vectors $\boldsymbol{\phi}_i$. Using the raw vectors would introduce the magnitude of the raw vector itself and thus would typically lead to a more conservative estimator. Second, the dependence graph enters locally through the unit-specific degrees $d_{n,i}$ inside the summand. This differs from approaches that inflate a conventional estimator by a single worst-case graph summary, such as the maximum degree or the spectral radius. By allowing the inflation to vary across units, our estimator avoids applying the same worst-case penalty to every observation and can therefore be substantially less conservative when dependence is heterogeneous across units. The trade-off is that this refinement requires more detailed knowledge of the dependence graph. The simulations in Section~\ref{subsection:sim_var_comp} illustrate this gain in practice.

\subsection{Variance estimation under rerandomization}
\label{sec:opt}

We now adapt the above result to variance estimation under rerandomization. Relative to Section~\ref{sec:conservative_var}, two features change. First, the observed data, including outcomes, covariates, and treatment assignments, are generated under the rerandomization distribution, namely the Bernoulli design conditional on the acceptance event $\{M_n \le a\}$ but not the original Bernoulli distribution. Second, the target quantity changes: instead of the unconditional covariance of $\hat{\boldsymbol{\eta}}$, we seek to estimate the variance of $\hat{\tau}$ under rerandomization. We address these two issues in turn.

We begin with the first point. Although Theorem~\ref{thm:conservative_matrix} is stated for Bernoulli assignment, it continues to provide a valid covariance bound after conditioning on the rerandomization acceptance event.

\begin{corollary}[Conditional conservativeness under rerandomization]
\label{cor:conditional_psd_conservative}
Under the conditions of Theorem~\ref{thm:conservative_matrix}, for every $\varepsilon > 0$,
\[
\mathbb{P}\!\left( \mathbf D_n\hat{\mathbf{U}}_n\mathbf D_n - \mathbf V + \varepsilon \mathbf
I_{d_\phi} \succeq \mathbf{0} \,\middle|\, M_n \le a \right) \to 1,
\]
where $\mathbf V$ is defined in Assumption~\ref{ass:variance_stability}.
\end{corollary}

We now turn to the second point: how to use a joint covariance bound for $\hat{\boldsymbol{\eta}}$ to construct a variance estimator for $\hat{\tau}$ under rerandomization. Corollary~\ref{thm:variance_reduction} implies that
\[
\Var_{\operatorname{asymp}}(\hat{\tau} \mid M_n \le a) = \sigma_{n,y}^2 \left\{ V_{11} -
(1-v_p)\mathbf V_{12}\mathbf V_{22}^{-1}\mathbf V_{21} \right\},
\]
where the partition of $\mathbf V$ is defined in \eqref{eqn:partition}. Recall that $\mathbf{U}_n = \Var(\hat{\boldsymbol{\eta}})$. Partition $\mathbf{U}_n$ conformably with $(\hat{\tau}, \tauxhat^\top)^\top$ as
\[
\mathbf{U}_n =
\begin{pmatrix}
U_{11} & \mathbf{U}_{12} \\
\mathbf{U}_{21} & \mathbf{U}_{22} \end{pmatrix}.
\]
By Assumption~\ref{ass:variance_stability}, $\mathbf V$ is the limit of $\mathbf D_n\mathbf U_n\mathbf D_n$, where $\mathbf D_n =\operatorname{diag}(\sigma_{n,y}^{-1},\sigma_{n,x_1}^{-1},\dots,\sigma_{n,x_p}^{-1})$. The standardization cancels in the Schur-complement term. Therefore,
\[
\Var_{\operatorname{asymp}}(\hat{\tau} \mid M_n \le a)
= U_{11}-(1-v_p)\mathbf U_{12}\mathbf U_{22}^{-1}\mathbf U_{21}
+o(\sigma_{n,y}^2).
\]

We have already constructed a conservative estimator $\hat{\mathbf{U}}_n$ for $\mathbf{U}_n$. For $\mathbf S_{22}\succ0$, the map $\mathbf S\mapsto S_{11}-(1-v_p)\mathbf S_{12}\mathbf S_{22}^{-1}\mathbf S_{21}$ is monotone with respect to the positive semidefinite order: it is a nonnegative weighted sum of $S_{11}$ and the Schur complement $S_{11}-\mathbf S_{12}\mathbf S_{22}^{-1}\mathbf S_{21}$, with weights $v_p$ and $1-v_p$. Thus, directly plugging in $\hat{\mathbf U}_n$ yields an asymptotically conservative variance bound.

We can sharpen this bound by exploiting additional structure. The lower-right block of $\mathbf U_n$ is known exactly: $\mathbf{U}_{22} = \boldsymbol\Sigma_{xx},$ where $\boldsymbol\Sigma_{xx} = \Var(\tauxhat)$ is the covariance matrix appearing in the Mahalanobis balance criterion. This suggests we can construct a conservative estimator for the conditional variance by searching over all positive semidefinite matrices that are dominated by $\hat{\mathbf{U}}_n$ and whose lower-right block is equal to $\boldsymbol\Sigma_{xx}$, and then taking the largest rerandomization variance compatible with these constraints. This construction ensures that, if the true covariance matrix $\mathbf{U}_n$ is dominated by $\hat{\mathbf{U}}_n$, then the optimal value provides a conservative estimator of the variance under rerandomization.

Formally, below we present our proposed variance estimator for the Hájek estimator $\hat{\tau}$ under rerandomization.
\begin{definition}[Variance estimator under rerandomization]
\label{def:rerand-variance-estimator}
Partition $\hat{\mathbf U}_n$ conformably as
\[
\hat{\mathbf U}_n
=
\begin{pmatrix}
\hat U_{11}&\hat{\mathbf U}_{12}\\
\hat{\mathbf U}_{21}&\hat{\mathbf U}_{22}
\end{pmatrix}.
\]
Define the variance estimator $\hat v_n$ as the solution to the following optimization problem:
\begin{equation}
\label{eqn:optimization_problem}
\begin{aligned}
\hat v_n = \max_{\boldsymbol\Omega}\quad
&\Omega_{11}-(1-v_p)\boldsymbol\Omega_{12}
\boldsymbol\Sigma_{xx}^{-1}\boldsymbol\Omega_{21}\\
\text{subject to}\quad
&\mathbf0\preceq\boldsymbol\Omega\preceq\hat{\mathbf U}_n,\\
&\boldsymbol\Omega_{22}=\boldsymbol\Sigma_{xx}.
\end{aligned}
\end{equation}
Here $\boldsymbol\Omega$ ranges over symmetric $d_\phi\times d_\phi$ matrices partitioned conformably with $(\hat\tau,\hat{\boldsymbol\tau}_x^\top)^\top$. If the feasible set is empty and $\hat{\mathbf U}_{22}\succ\mathbf0$, set
\[
\hat v_n=\hat U_{11}-(1-v_p)\hat{\mathbf U}_{12}
\hat{\mathbf U}_{22}^{-1}\hat{\mathbf U}_{21}.
\]
If $\hat{\mathbf U}_{22}$ is singular, set $\hat v_n=\hat U_{11}$.
\end{definition}

Computationally, the above optimization problem is convex with semidefinite constraints and can therefore be solved efficiently using standard convex optimization methods. Moreover, Theorem~\ref{thm:opt_solution} below provides an explicit solution, making the proposed estimator even more straightforward to compute.

\begin{theorem}[Feasibility and solution of the variance optimization problem]
\label{thm:opt_solution}
Assume that $\boldsymbol\Sigma_{xx}\succ0$. The feasible set in \eqref{eqn:optimization_problem} is nonempty if and only if $\hat{\mathbf U}_{22}\succeq\boldsymbol\Sigma_{xx}$. On this event, define
\[
\mathbf M(t)
=\boldsymbol\Sigma_{xx}
+t(\hat{\mathbf U}_{22}-\boldsymbol\Sigma_{xx})
\]
and
\[
G(t)
=\hat{\mathbf U}_{12}\mathbf M(t)^{-1}
\left\{\boldsymbol\Sigma_{xx}
+t^2(\hat{\mathbf U}_{22}-\boldsymbol\Sigma_{xx})\right\}
\mathbf M(t)^{-1}\hat{\mathbf U}_{21}.
\]
Then the optimal value is
\[
\hat v_n
=\hat U_{11}-\hat{\mathbf U}_{12}
\mathbf M(t^\star)^{-1}\left\{
(1-v_p)\boldsymbol\Sigma_{xx}+(t^\star)^2
(\hat{\mathbf U}_{22}-\boldsymbol\Sigma_{xx})
\right\}\mathbf M(t^\star)^{-1}\hat{\mathbf U}_{21},
\]
where
\[
t^\star
=\min\left\{t\in[1-v_p,1]:G(t)\le\hat U_{11}\right\}.
\]
In particular, when $G(1-v_p)\le\hat U_{11}$, the optimal value simplifies to
\[
\hat v_n
=\hat U_{11}
-(1-v_p)\hat{\mathbf U}_{12}
\{v_p\boldsymbol\Sigma_{xx}+(1-v_p)\hat{\mathbf U}_{22}\}^{-1}
\hat{\mathbf U}_{21}.
\]
\end{theorem}

The next result establishes that the proposed estimator is asymptotically conservative for the variance of the H\'{a}jek estimator under rerandomization.
\begin{theorem}[Asymptotic conservativeness]
\label{thm:asymp_conservativeness}

Let
\[
v_n = \Var_{\operatorname{asymp}} (\hat{\tau}\mid M_n\le a) = \sigma_{n,y}^{2} \left\{ V_{11} -
(1-v_p) \mathbf V_{12} \mathbf V_{22}^{-1} \mathbf V_{21} \right\}
\]
denote the asymptotic conditional variance under rerandomization. Under the conditions of Theorem~\ref{thm:conservative_matrix}, for every $\epsilon>0$,
\[
\mathbb P \left( \hat v_n \geq v_n-\epsilon\sigma_{n,y}^{2}
\,\middle|\, M_n\le a \right)
\rightarrow 1 .
\]
\end{theorem}

There is one subtlety with the above estimator. The optimization problem in~\eqref{eqn:optimization_problem} may be infeasible with non-vanishing probability, in which case we use the plug-in fallback specified in Definition~\ref{def:rerand-variance-estimator}. When $\hat{\mathbf U}_{22}\succ\mathbf0$, this fallback can still reflect the precision gain from rerandomization, but it does not exploit the known covariate covariance $\boldsymbol\Sigma_{xx}$ through the optimization constraint. Infeasibility can occur because Theorem~\ref{thm:conservative_matrix} and Corollary~\ref{cor:conditional_psd_conservative} establish positive semidefinite dominance only up to an arbitrarily small slack after standardization, rather than the exact inequality $\hat{\mathbf U}_n\succeq\mathbf U_n$.

To address this issue in practice, we introduce a small ridge correction to $\hat{\mathbf U}_n$ that improves feasibility while preserving asymptotic validity. Formally, define
\begin{equation}
\label{eq:scale_adaptive_ridge}
\hat{\mathbf S}_n
=\operatorname{diag}\!\left(\hat U_{11},
\Sigma_{xx,11},\ldots,\Sigma_{xx,pp}\right),
\qquad
\hat{\mathbf U}_{n,\delta}
=\hat{\mathbf U}_n+\delta\hat{\mathbf S}_n,
\end{equation}
where $\delta\ge0$ is a correction level, which may be fixed or data-adaptive. For this corrected upper bound, define $\hat v_{n,\delta}$ by
\begin{equation}
\label{eq:ridge-optimization-problem}
\begin{aligned}
\hat v_{n,\delta}=\max_{\boldsymbol\Omega}\quad
&\Omega_{11}-(1-v_p)\boldsymbol\Omega_{12}
\boldsymbol\Sigma_{xx}^{-1}\boldsymbol\Omega_{21}\\
\text{subject to}\quad
&\mathbf0\preceq\boldsymbol\Omega\preceq\hat{\mathbf U}_{n,\delta},\\
&\boldsymbol\Omega_{22}=\boldsymbol\Sigma_{xx},
\end{aligned}
\end{equation}
and, if the corrected program is infeasible, set
\[
\hat v_{n,\delta}=
\begin{cases}
\hat U_{11,\delta}-(1-v_p)\hat{\mathbf U}_{12,\delta}
\hat{\mathbf U}_{22,\delta}^{-1}\hat{\mathbf U}_{21,\delta},
&\hat{\mathbf U}_{22,\delta}\succ\mathbf0,\\
\hat U_{11,\delta},&\text{otherwise},
\end{cases}
\]
where the subscripts denote the corresponding blocks of $\hat{\mathbf U}_{n,\delta}$. Thus, the corrected estimator uses the same fallback rule as Definition~\ref{def:rerand-variance-estimator}, and $\hat v_{n,0}=\hat v_n$.

Appendix~\ref{app:feasibility-refinement} studies the choice of the correction level $\delta$. It first gives the guarantee for every fixed $\delta>0$, and then proposes a minimal data-adaptive correction that makes the optimization feasible for every realization and is $o_p(1)$. In practice, a small fixed value such as $\delta=0.01$ remains a simple alternative, corresponding to a one-percent inflation along each marginal variance scale. We use this fixed correction in our simulations.

Finally, we make two remarks. First, the optimization framework developed in this section is not specific to our particular variance estimator~$\hat{\mathbf{U}}_n$. It applies whenever one has access to any joint covariance estimator satisfying the positive semidefinite conservativeness condition in~\eqref{eq:psd_conservative}. In this sense, the optimization step is a modular component that can be combined with different upstream variance estimators. Moreover, when multiple joint covariance estimators satisfying~\eqref{eq:psd_conservative} are available, one may apply the optimization procedure to each estimator and take the minimum of the resulting optimal values. For any fixed number of such estimators, this minimum remains asymptotically conservative while potentially yielding a tighter variance estimate.

Second, a related optimization-based approach to variance estimation was proposed by \citet{harshaw2026optimized}, who also formulate the construction of conservative variance estimators as a convex program. The two formulations differ in their orientation. \citet{harshaw2026optimized} search over the class of conservative estimators and minimize a measure of conservativeness, yielding the \emph{least conservative} among all valid bounds. Our formulation instead searches over covariance matrices compatible with the known covariate block and the estimated upper bound, and maximizes the implied conditional variance over that set. This worst-case perspective is tailored to the rerandomization setting.

\subsection{Confidence intervals}
\label{sec:CI}

In this subsection, we construct confidence intervals using the conservative variance estimators developed in the previous subsection. The construction requires some care because the limiting distribution of the H\'ajek estimator under rerandomization is generally non-Gaussian. Recall that $v_n = \Var_{\operatorname{asymp}}(\hat\tau\mid M_n\le a) = \sigma_{n,y}^2\{1-(1-v_p)R^2\}$ denotes the asymptotic variance of the H\'ajek estimator under rerandomization. By Theorem~\ref{thm:rerand_dist}, the studentized H\'ajek estimator satisfies
\begin{equation}
\label{eqn:limit_studentized}
\frac{\hat\tau-\tau}{\sqrt{v_n}} \,\Big|\, M_n\le a \xrightarrow{d}
\sqrt{\frac{1-R^2}{1-(1-v_p)R^2}}\,\varepsilon_0 + \sqrt{\frac{R^2}{1-(1-v_p)R^2}}\,L_{p,a}.
\end{equation}
Consequently, the usual Wald confidence interval based on a Gaussian critical value is generally not asymptotically valid.

Ideally, one would construct confidence intervals using the quantiles of the limiting distribution in~\eqref{eqn:limit_studentized}. However, these quantiles depend on the unknown quantity $R^2$, which is generally difficult to estimate reliably, as discussed at the beginning of this section. To avoid estimating $R^2$, we instead use the worst-case critical value over all possible values of $R^2$.

Specifically, for $R^2\in[0,1]$, define
\[
Q(R^2;p,a) = \sqrt{1-R^2}\,\varepsilon_0 + \sqrt{R^2}\,L_{p,a},
\]
and let $\omega_{1-\alpha/2}(R^2;p,a)$ denote its $(1-\alpha/2)$-quantile. The corresponding critical value for the limiting distribution in~\eqref{eqn:limit_studentized} is
\begin{equation}
\label{eq:rerand_critical_value}
q_{1-\alpha/2}(R^2;p,a) = \frac{\omega_{1-\alpha/2}(R^2;p,a)} {\sqrt{1-(1-v_p)R^2}}.
\end{equation}
We then define the worst-case critical value
\begin{equation}
\label{eq:worst_case_q}
\bar q_{1-\alpha/2}(p,a) = \sup_{R^2\in[0,1]} q_{1-\alpha/2}(R^2;p,a).
\end{equation}
This leads to the confidence interval
\begin{equation}
\label{eq:CI_v}
\mathcal I_{1-\alpha}
= \left[
\hat\tau-\bar q_{1-\alpha/2}(p,a)(\hat v_n)^{1/2},
\ \hat\tau+\bar q_{1-\alpha/2}(p,a)(\hat v_n)^{1/2}
\right].
\end{equation}

Theorem~\ref{thm:ci_coverage} establishes the asymptotic validity of the above confidence interval.

\begin{theorem}[Asymptotic validity of confidence intervals]
\label{thm:ci_coverage}
Under the conditions of Theorem~\ref{thm:conservative_matrix}, for every $\alpha\in(0,1)$, $\mathcal I_{1-\alpha}$ satisfies
\[
\liminf_{n\to\infty}
\mathbb P\left\{\tau\in\mathcal I_{1-\alpha}\,\middle|\,M_n\le a\right\}
\ge 1-\alpha.
\]
\end{theorem}

The same asymptotic coverage guarantee holds for the ridge-corrected bounds by Theorem~\ref{thm:adaptive-ridge-conservativeness} in Appendix~\ref{app:feasibility-refinement}.

Figure~\ref{fig:quantiles} illustrates the behavior of the critical value $q_{1-\alpha/2}(R^2;p,a)$. Across the parameter settings considered, $q_{1-\alpha/2}(R^2;p,a)$ often decreases as $R^2$ increases. Although such monotonicity does not hold universally for all $(R^2,p,a)$, the figure suggests that the worst-case critical value $\bar q_{1-\alpha/2}(p,a)$ is typically attained when $R^2$ is close to zero. When $R^2$ is not large, the critical value $q_{1-\alpha/2}(R^2;p,a)$ remains close to the Gaussian benchmark. In practice, $R^2$ is typically not very large, as a result, the standard normal critical value $z_{1-\alpha/2}$ can be used as a convenient numerical approximation when the small discrepancy displayed in the figure is acceptable. For the no-ridge estimator, $\hat v_n\leq\hat U_{11}$ by construction. Thus, when the same Gaussian critical value is used for both designs, the rerandomization interval is no longer than the corresponding Bernoulli Wald interval. With a positive ridge, the analogous comparison is with the ridge-adjusted marginal bound $\hat U_{11,\delta}$; a fixed ridge can therefore introduce a small additional inflation.

\begin{figure}
\centering \includegraphics[width=0.6\textwidth, trim={15cm 0cm 0cm 11cm},
clip]{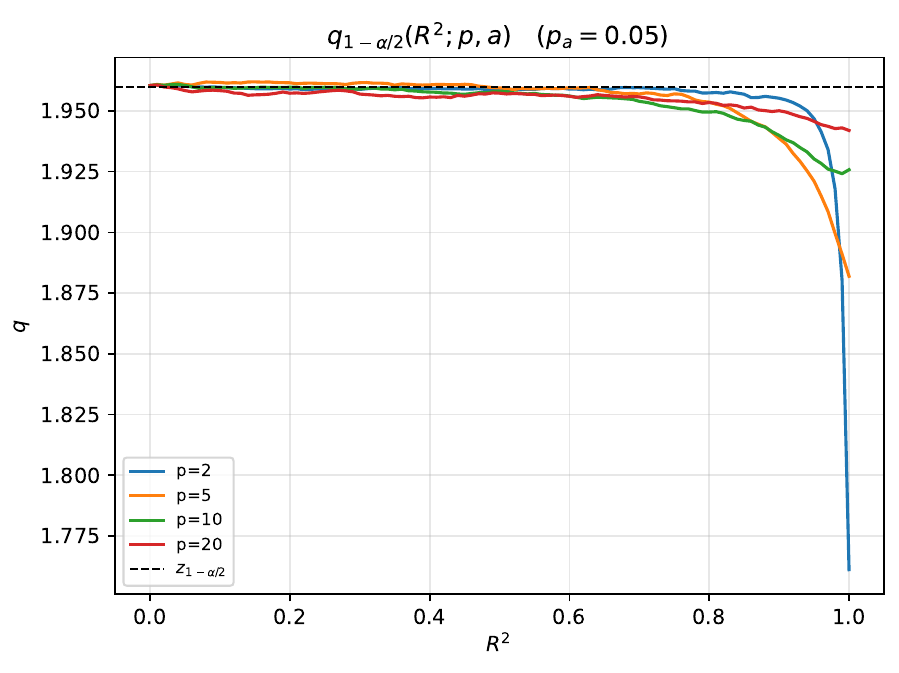} \caption{The figure plots $q_{1-\alpha/2}(R^2;p,a)$ in~\eqref{eq:rerand_critical_value} as a function of $R^2$, for $p\in\{2,5,10,20\}$ and $\alpha=0.05$. For each value of $p$, the threshold $a$ is chosen so that the rerandomization acceptance probability $p_a=\Pr(\chi^2_p\le a)$ equals $0.05$. The dashed horizontal line denotes the usual Gaussian critical value $z_{1-\alpha/2}$. Across the parameter settings shown, the rerandomization critical value remains very close to the Gaussian benchmark for most values of $R^2$, and tends to decrease only when $R^2$ is large, especially for smaller values of $p$. }
\label{fig:quantiles}
\end{figure}

Finally, we compare our inference strategy with that of \citet{savje2021average}. Their framework allows limited but otherwise arbitrary and unknown interference. Under these weak restrictions, inference is based on conservative variance bounds together with distribution-free tail inequalities such as Chebyshev’s inequality, since a Gaussian approximation cannot generally be justified. By contrast, the joint Gaussian limit in Theorem~\ref{thm:joint-normality} yields the rerandomization-specific limiting law in \eqref{eqn:limit_studentized}. This additional structure allows us to use design-specific critical values while preserving asymptotic coverage.

For completeness, we discuss several alternative conservative variance estimators and confidence interval constructions in Appendix~\ref{app:alternative_inference}. These alternatives provide useful intuition and, in some cases, simpler computational procedures. Among these alternatives, the optimization-based procedure developed in Sections~\ref{sec:opt} and~\ref{sec:CI} strikes a balance between statistical efficiency and computational tractability, which motivates our recommendation of this approach in practice.

\section{Simulation}
\label{appendx:simulation_main}

We conduct simulation studies to evaluate the finite-sample performance of the proposed framework. The simulations focus on (i) the variance of the H\'{a}jek estimator under rerandomization and (ii) the performance of the proposed variance estimator. For (i), we study, across different outcome models and different sets of covariates used in rerandomization, how much variance reduction rerandomization can achieve. For (ii), we assess whether the proposed variance estimator is conservative, how much it improves relative to its counterpart under the Bernoulli design, and how much it improves upon variance estimators based on existing baseline estimators from the literature. We present the main design choices and key findings here, with full implementation details deferred to Appendix~\ref{appendx:simulation}.

\subsection{Simulation setup}

We generate a Barab\'{a}si--Albert preferential-attachment network \citep{albert1999emergence} with $n = 2{,}000$ units and $m = 5$ edges attached per new node. The resulting graph has a power-law degree distribution: a small number of high-degree hub units coexist with a long tail of low-degree units, yielding a heterogeneous dependency structure characteristic of real social, information, and biological networks.

Treatment is assigned via $Z_i \overset{\text{i.i.d.}}{\sim} \operatorname{Bern}(0.5)$, subject to a rerandomization acceptance criterion. The target estimand, the EATE, is defined with respect to the Bernoulli$(0.5)$ design:
\[
\tau = \frac{1}{n}\sum_{i=1}^n \E_{Z_j \stackrel{\operatorname{i.i.d.}}{\sim}
\operatorname{Bern}(0.5)} \big[ Y_i(1;\mathbf Z_{-i}) - Y_i(0;\mathbf Z_{-i}) \big].
\]

Each unit has a covariate vector consisting of six pre-treatment covariates, along with treatment-dependent covariates constructed from the treatment vector $\mathbf Z$ and the network structure. These include the neighbor treatment proportion $H_{i,\text{prop}} = \frac{1}{|\mathcal{N}_i|}\sum_{j \in \mathcal{N}_i} Z_j$, and the treated-neighbor covariate sum normalized by degree $\mathbf H_{i,\text{wX}} = \frac{1}{|\mathcal{N}_i|} \sum_{j \in \mathcal{N}_i} Z_j \mathbf X_j$, where $\mathcal{N}_i$ denotes the set of $i$'s direct neighbors in the network, excluding $i$ itself. We assume that these treatment-dependent covariates are constructed using the true network.

To separate the sources of variation, we write our outcome as
\[
Y_i = \beta_0 + \tau Z_i + \beta_{\text{base}} \cdot f(\mathbf X_i) + X_{i,\text{hidden}}
\beta_{\text{hidden}} + \beta_{\text{interf}} \cdot g_i(\mathbf Z, \mathbf X) + \varepsilon_i,
\]
where $f(\mathbf X_i)$ denotes the baseline-covariate component and $g_i(\mathbf Z, \mathbf X)$ denotes the interference component.

To systematically control the signal composition, let $\text{std}_{\text{base}}$ denote the finite-population standard deviation of the baseline-covariate component $f(\mathbf X_i)$ across all $n$ units. Since $g_i(\mathbf Z, \mathbf X)$ also depends on the random treatment assignment, let $\text{std}_{\text{interf}}$ denote its standard deviation across units, averaged over the \emph{unconstrained} Bernoulli distribution of $\mathbf Z$ (so that the signal ratio characterizes the data-generating process and does not depend on the choice of rerandomization criterion). We parameterize each model by the \emph{signal ratio}
\[
\kappa = \frac{\beta_{\text{interf}} \cdot \text{std}_{\text{interf}}}{\beta_{\text{base}} \cdot
\text{std}_{\text{base}}}.
\]
Varying $\kappa$ shifts the dominant source of outcome variation from pre-treatment covariates ($\log_2 \kappa \ll 0$) to interference effects ($\log_2 \kappa \gg 0$), while holding the total signal strength $(\beta_{\text{base}} \cdot \text{std}_{\text{base}})^2 + (\beta_{\text{interf}} \cdot \text{std}_{\text{interf}})^2$ fixed.

We consider two outcome specifications and five additional settings are reported in Appendix~\ref{appendx:simulation}.

\textbf{Linear setting.} The outcome depends linearly on pre-treatment covariates and the neighbor treatment proportion:
\begin{equation}
\label{eq:outcome_linear}
Y_i = \beta_0 + \tau Z_i + \beta_{\text{base}} \, \mathbf X_i^\top \boldsymbol\beta_{\text{xbase}} +
\beta_{\text{hidden}}\, X_{i,\text{hidden}} + \beta_{\text{interf}} \, \beta_{\text{prop}} \,
H_{i,\text{prop}} + \varepsilon_i,
\end{equation}
where $\varepsilon_i \overset{\text{i.i.d.}}{\sim} \mathcal{N}(0,1)$.

\textbf{Nonlinear setting.} Both the pre-treatment and interference components are transformed nonlinearly:
\begin{equation}
\label{eq:outcome_nonlinear}
Y_i = \beta_0 + \tau Z_i + \beta_{\text{base}} \exp\!\big(\mathbf X_i^\top
\boldsymbol\beta_{\text{xbase}}\big) + \beta_{\text{hidden}}\, X_{i,\text{hidden}} +
\beta_{\text{interf}} \, \beta_{\text{prop}} \, \exp\!\big(H_{i,\text{prop}}\big) + \varepsilon_i,
\end{equation}
where $\varepsilon_i \overset{\text{i.i.d.}}{\sim} \mathcal{N}(0,1)$.

Here, $\tau$ represents the direct treatment effect, $\mathbf X_i$ are observed pre-treatment covariates, and $X_{i,\text{hidden}}$ is an unobserved pre-treatment covariate. The scalars $\beta_{\text{base}}$ and $\beta_{\text{interf}}$ control the relative strength of the baseline-covariate component and interference component. In both settings, we fix $\beta_0 = 0.5$ and $\tau = 2.0$.

The linear specification represents a simple interference model in which outcomes respond linearly to the fraction of treated neighbors. The nonlinear specification assesses the effectiveness of balancing first moments when outcomes depend nonlinearly on the observed covariates and exposures.

We compare three rerandomization strategies, each calibrated to a $5\%$ acceptance rate: Pre-Treatment (balancing only pre-treatment covariates $\mathbf X_i$), Treatment-Dependent (balancing only treatment-dependent covariates), and Joint (balancing both).

To assess how the choice of treatment-dependent covariates affects performance, we consider two balancing sets for the treatment-dependent component: (i) \textbf{NP}, which uses only the neighbor treatment proportion $H_{i,\text{prop}}$; (ii) \textbf{NP+NWX}, which additionally includes the neighbor-weighted covariate averages $\mathbf H_{i,\text{wX}}$. The outcome models in~\eqref{eq:outcome_linear}--\eqref{eq:outcome_nonlinear} are the same across the two settings; only the set of treatment-dependent covariates used by the Treatment-Dependent and Joint strategies differs. In the figures that follow, each outcome specification (linear/nonlinear) is reported under both balancing sets, yielding four panels per subsection.

\subsection{Variance reduction relative to Bernoulli experiment}

In Figure~\ref{fig:vr_main}, we empirically evaluate
\begin{equation}
\label{eq:vr_main}
\text{Variance Ratio} = \frac{\Var(\hat{\tau} \mid M_n \le a)}{\Var(\hat{\tau})},
\end{equation}
where $\Var(\hat{\tau})$ denotes the variance of the H\'{a}jek estimator under Bernoulli randomization, and $\Var(\hat{\tau} \mid M_n \le a)$ denotes its variance under rerandomization. Both quantities are estimated via Monte Carlo simulation. A smaller value of variance ratio indicates a greater precision gain from rerandomization.

\begin{figure}
\centering
\begin{subfigure}[b]{0.48\textwidth}
\centering
\includegraphics[width=\linewidth]{"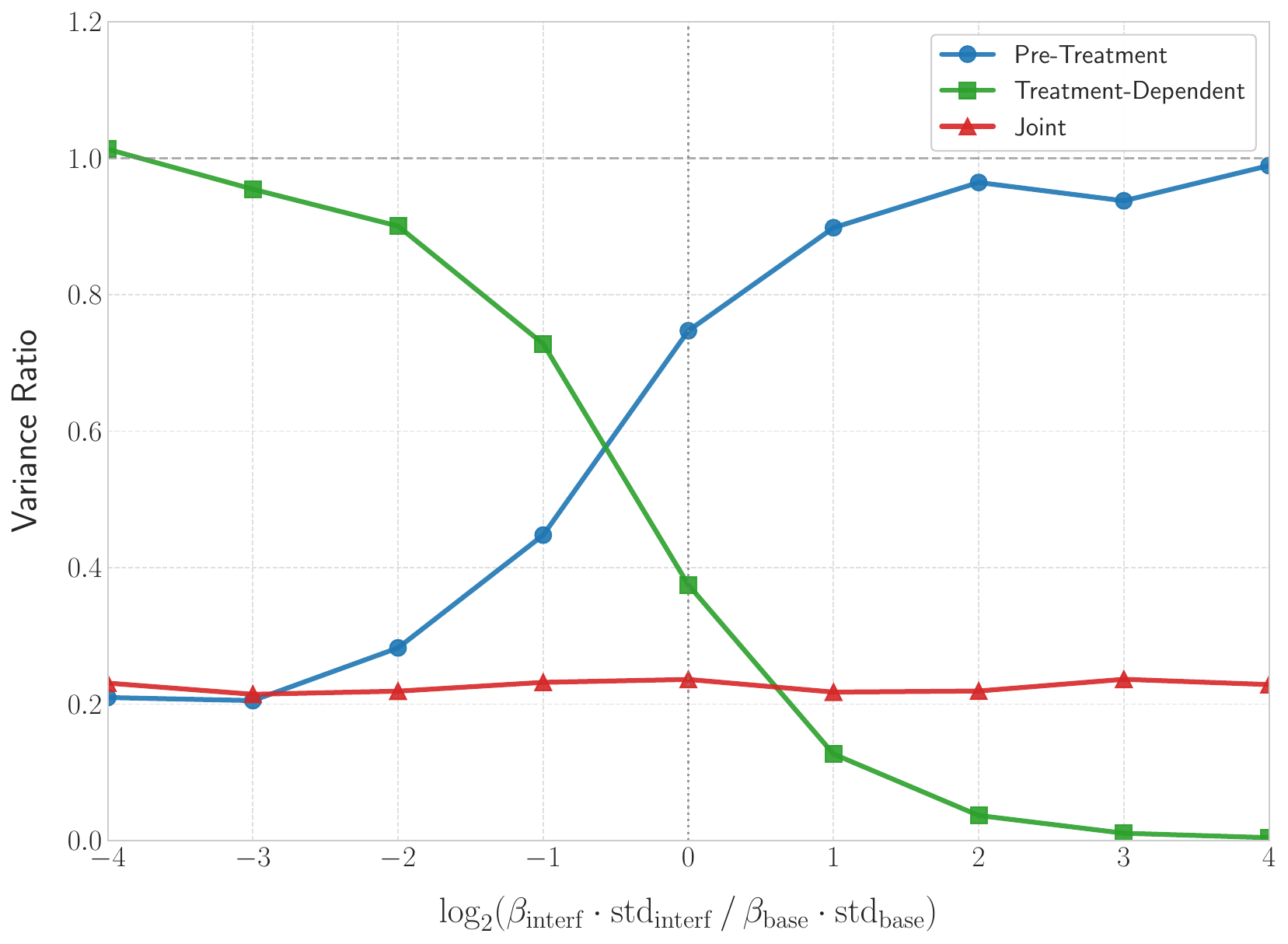"}
\caption{Linear, balancing set NP}
\label{fig:vr_type1_np_main}
\end{subfigure}
\hfill
\begin{subfigure}[b]{0.48\textwidth}
\centering
\includegraphics[width=\linewidth]{"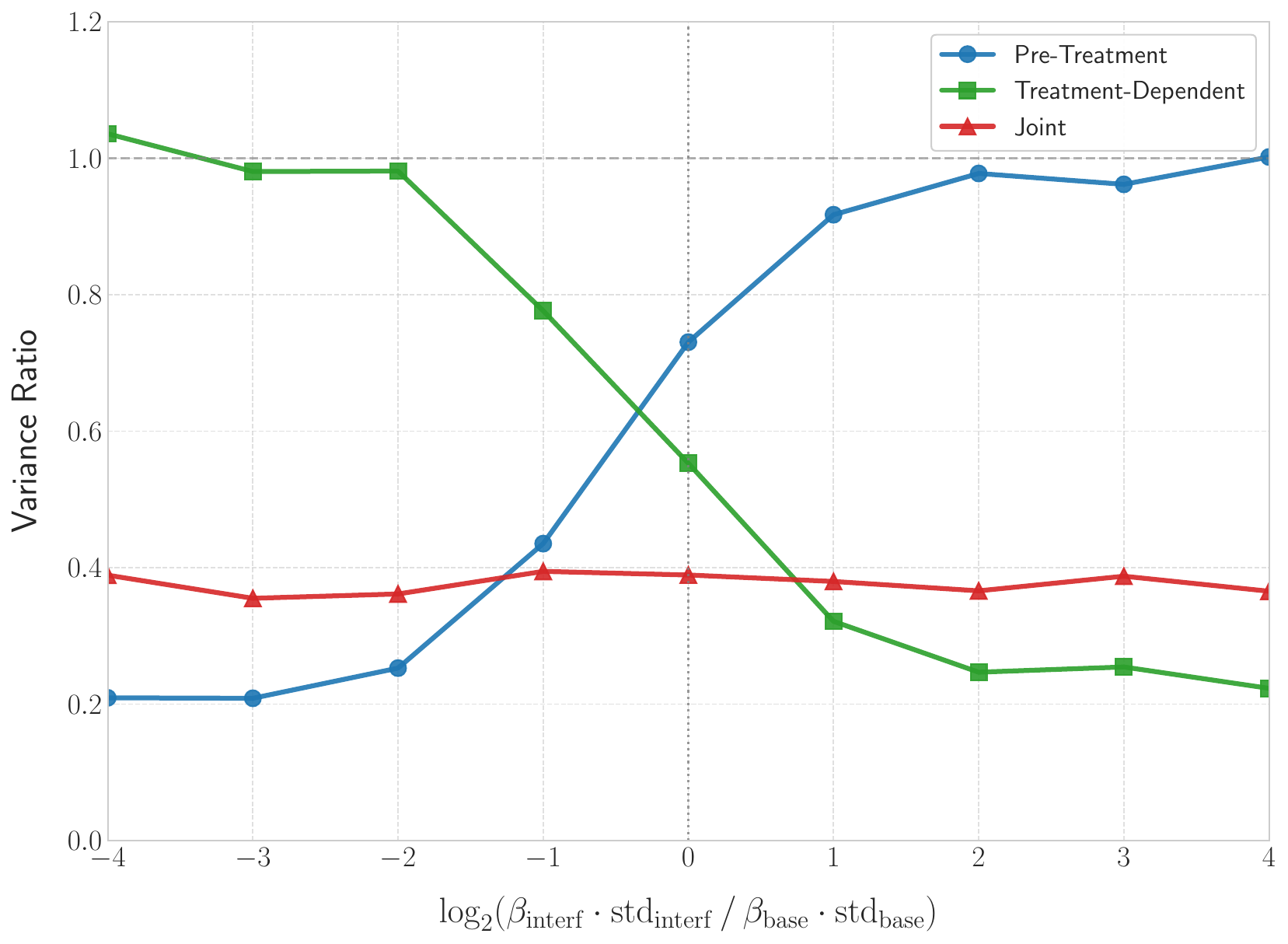"}
\caption{Linear, balancing set NP+NWX}
\label{fig:vr_type1_npnwx_main}
\end{subfigure}

\vspace{0.5em}

\begin{subfigure}[b]{0.48\textwidth}
\centering
\includegraphics[width=\linewidth]{"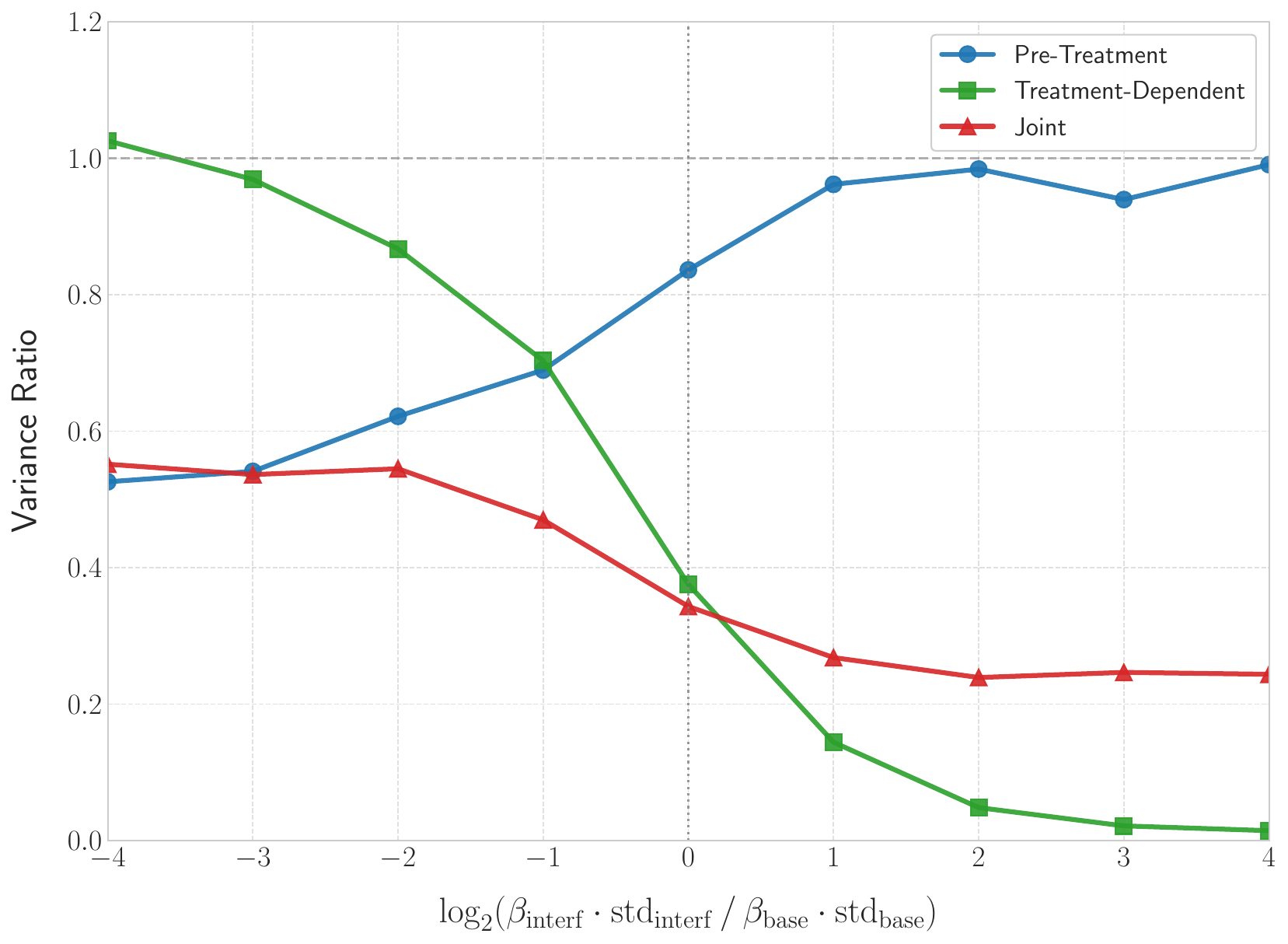"}
\caption{Nonlinear, balancing set NP}
\label{fig:vr_type5_np_main}
\end{subfigure}
\hfill
\begin{subfigure}[b]{0.48\textwidth}
\centering
\includegraphics[width=\linewidth]{"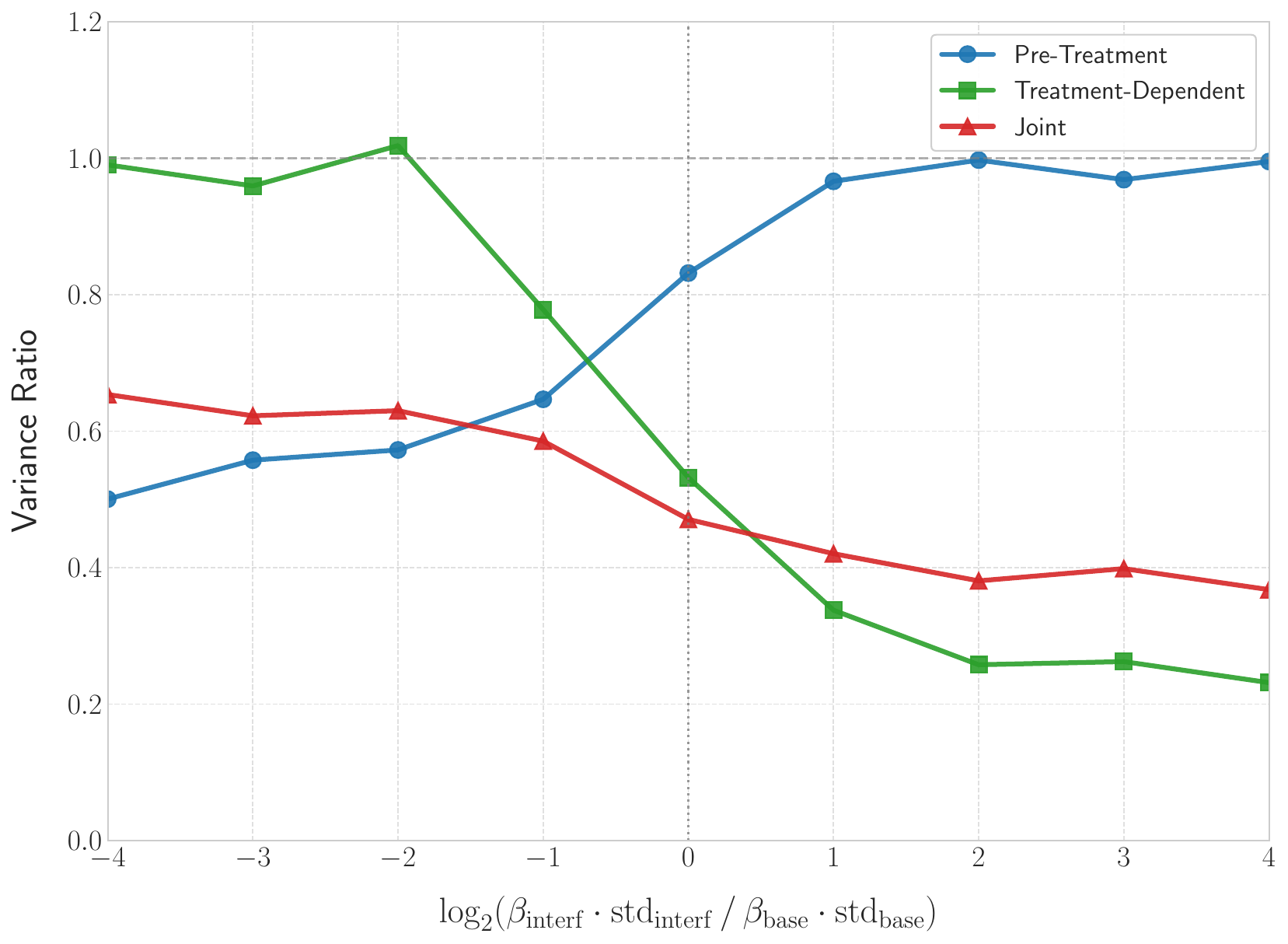"}
\caption{Nonlinear, balancing set NP+NWX}
\label{fig:vr_type5_npnwx_main}
\end{subfigure}
\caption{Variance ratio of the Hájek estimator under rerandomization relative to Bernoulli randomization, as defined in~\eqref{eq:vr_main}, for $\hat\tau$ as a function of the log signal ratio $\log_2 \kappa$. Rows correspond to outcome settings (top: linear; bottom: nonlinear); columns correspond to balancing sets (left: NP; right: NP+NWX). Blue: pre-treatment strategy; green: treatment-dependent strategy; red: joint strategy.}
\label{fig:vr_main}
\end{figure}

In the linear setting (top row), the pre-treatment strategy (blue) reduces the variance to roughly $0.21\times$ the Bernoulli variance in the pre-treatment-dominated regime ($\log_2\kappa \le -2$) but loses essentially all of its advantage in the interference-dominated regime ($\log_2\kappa \ge 2$). The treatment-dependent strategy (green) exhibits the mirror image: in the interference-dominated regime it shrinks the variance ratio to $0.04$ or below, while at $\log_2\kappa \le -2$ it provides no benefit. The crossover is sharp because, in this setting, $H_{i,\text{prop}}$ is exactly the source of the interference signal, so balancing on it removes nearly all of the assignment-induced interference variance. The nonlinear setting (bottom row) preserves the same qualitative crossover, but the magnitudes are attenuated: balancing first moments of the observed covariates need not fully balance the nonlinear transformations that enter the outcome model. The pre-treatment strategy bottoms out near $0.5$ rather than $0.2$, consistent with a weaker association between the treatment-effect estimator and the balanced first-moment contrasts. The reduction remains large, illustrating that rerandomization can provide substantial precision gains even when outcome-covariate relationships are nonlinear.

The joint strategy offers a practical default when the relative importance of pre-treatment covariates and interference is unknown. In our simulations, it delivers substantial variance reductions across signal regimes and avoids the marked loss of gains that can occur when balancing only the less informative block. A strategy tailored to the dominant source of variation can achieve greater precision, but joint balancing performs well without requiring that source to be identified in advance.

Comparing the NP and NP+NWX columns illustrates the effect of adding covariates to the balance criterion: the joint strategy under NP balancing achieves a near-uniform ratio of about $0.22$ across all signal regimes for the linear setting and about $0.24$--$0.55$ for the nonlinear setting, whereas under NP+NWX balancing the same joint strategy yields ratios of $0.36$--$0.39$ (linear) and $0.37$--$0.65$ (nonlinear). Adding $\mathbf H_{i,\text{wX}}$ to the criterion increases the dimension $p$ and inflates the factor $v_p$ in Corollary~\ref{thm:variance_reduction}, but contributes little additional explanatory power for these outcome models; neither~\eqref{eq:outcome_linear} nor~\eqref{eq:outcome_nonlinear} contains a $\mathbf H_{i,\text{wX}}$ term.

\subsection{Confidence interval length}
\label{simulation:variance_estimator}

We next examine whether the proposed inference procedure yields shorter confidence intervals than the Bernoulli benchmark. We compare average confidence interval lengths across Monte Carlo replications under the two designs. For rerandomization, let
\begin{equation}
\label{eq:ci_half_length}
\hat L_{\mathrm{ReM}}
=z_{1-\alpha/2}\,(\hat v_{n,\delta})^{1/2}
\end{equation}
denote the simulated half-length, where $\hat v_{n,\delta}$ is the fixed-ridge version of the optimization-based variance estimator with $\delta=0.01$. In these numerical comparisons we use the Gaussian approximation to the worst-case critical value discussed in Section~\ref{sec:CI}. For the Bernoulli benchmark, we use the usual Gaussian half-length
\[
\hat L_{\mathrm{Bern}} = z_{1-\alpha/2}\,\hat V_{\mathrm{Bern}}^{1/2},
\]
where $\hat V_{\mathrm{Bern}}$ is the conservative variance estimator under Bernoulli randomization from Section~\ref{sec:conservative_var}. We report the ratio of the average confidence interval length under rerandomization to that under Bernoulli randomization:
\begin{equation}
\label{eq:ci_length_ratio}
\rho = \frac{M^{-1}\sum_{m=1}^{M}\hat L_{\mathrm{ReM}}^{(m)}}
{M^{-1}\sum_{m=1}^{M}\hat L_{\mathrm{Bern}}^{(m)}}.
\end{equation}
A value $\rho \le 1$ indicates that rerandomization yields, on average, a confidence interval no longer than the Bernoulli benchmark.

\begin{figure}
\centering
\begin{subfigure}[b]{0.48\textwidth}
\centering
\includegraphics[width=\linewidth]{"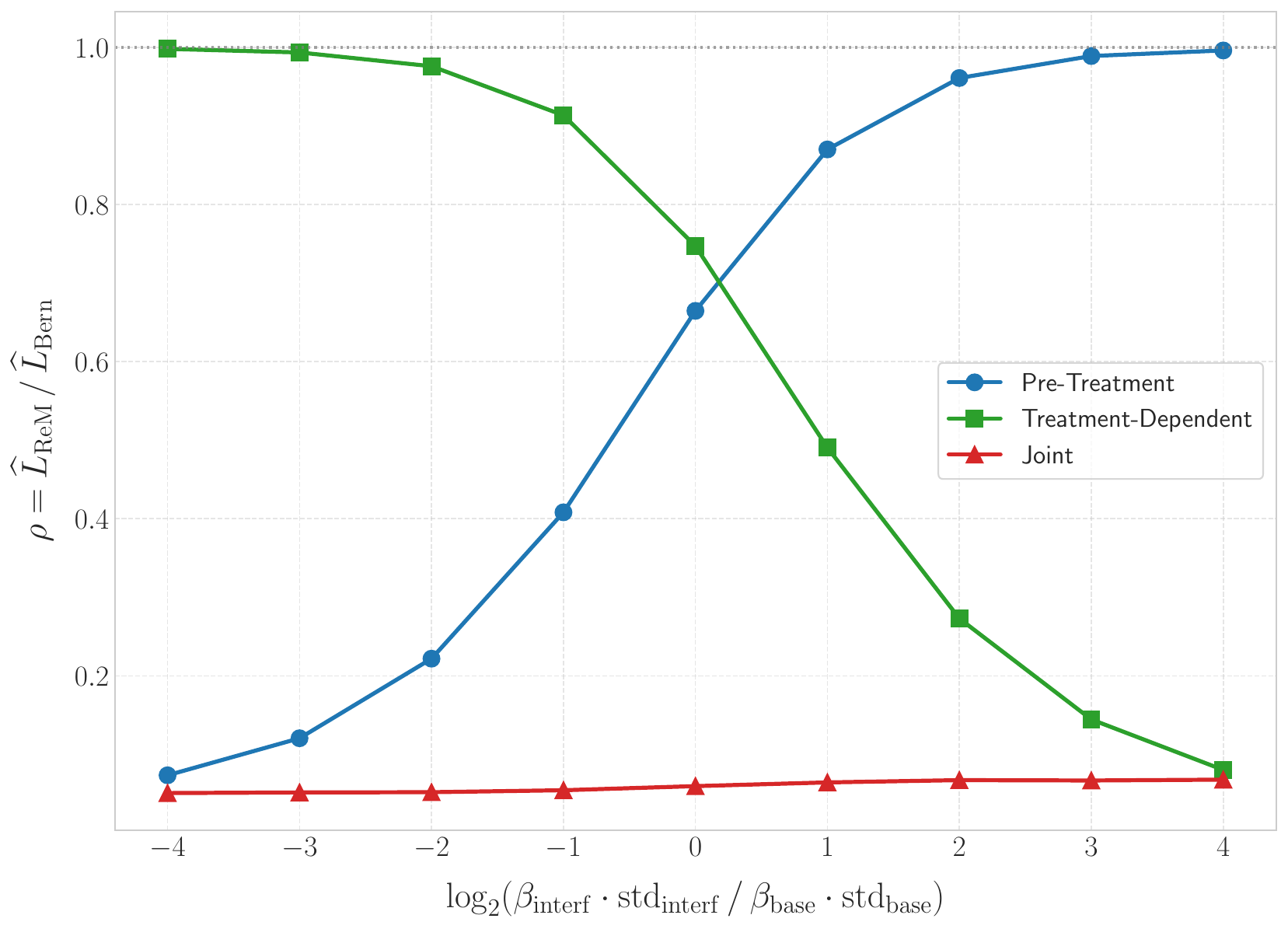"}
\caption{Linear, balancing set NP}
\label{fig:conserv_ratio_type1_np_main}
\end{subfigure}
\hfill
\begin{subfigure}[b]{0.48\textwidth}
\centering
\includegraphics[width=\linewidth]{"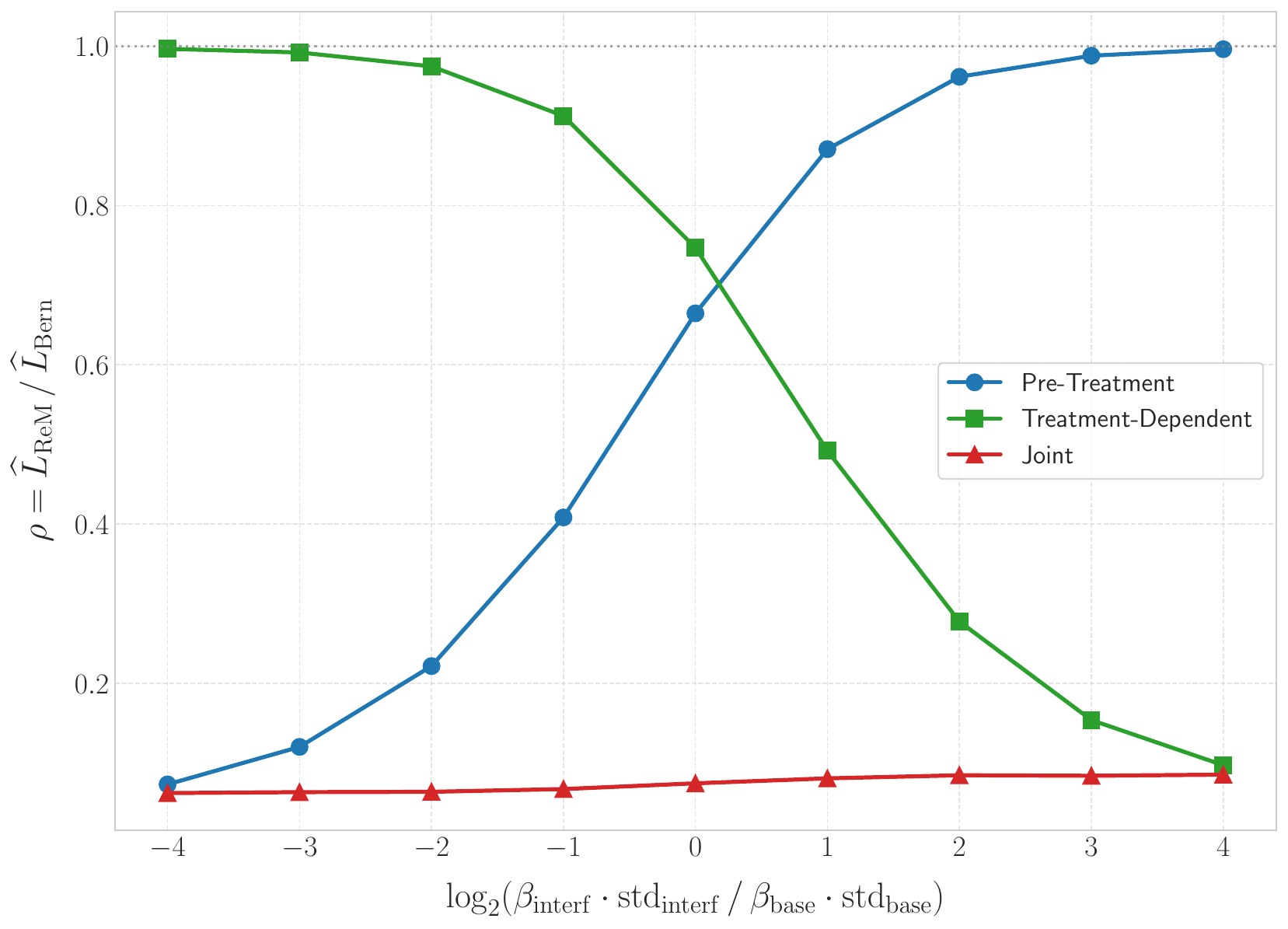"}
\caption{Linear, balancing set NP+NWX}
\label{fig:conserv_ratio_type1_npnwx_main}
\end{subfigure}

\vspace{0.5em}

\begin{subfigure}[b]{0.48\textwidth}
\centering
\includegraphics[width=\linewidth]{"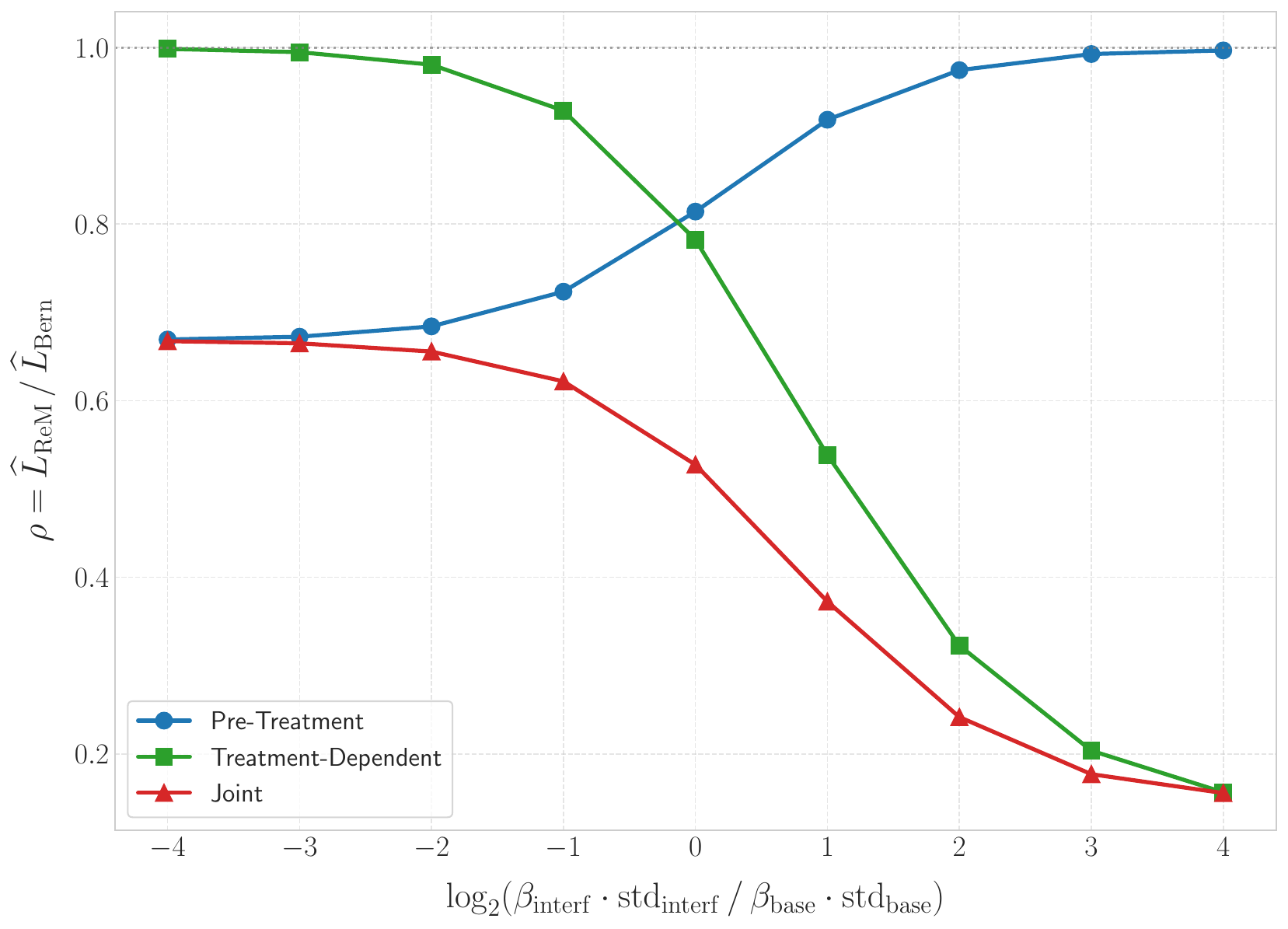"}
\caption{Nonlinear, balancing set NP}
\label{fig:conserv_ratio_type5_np_main}
\end{subfigure}
\hfill
\begin{subfigure}[b]{0.48\textwidth}
\centering
\includegraphics[width=\linewidth]{"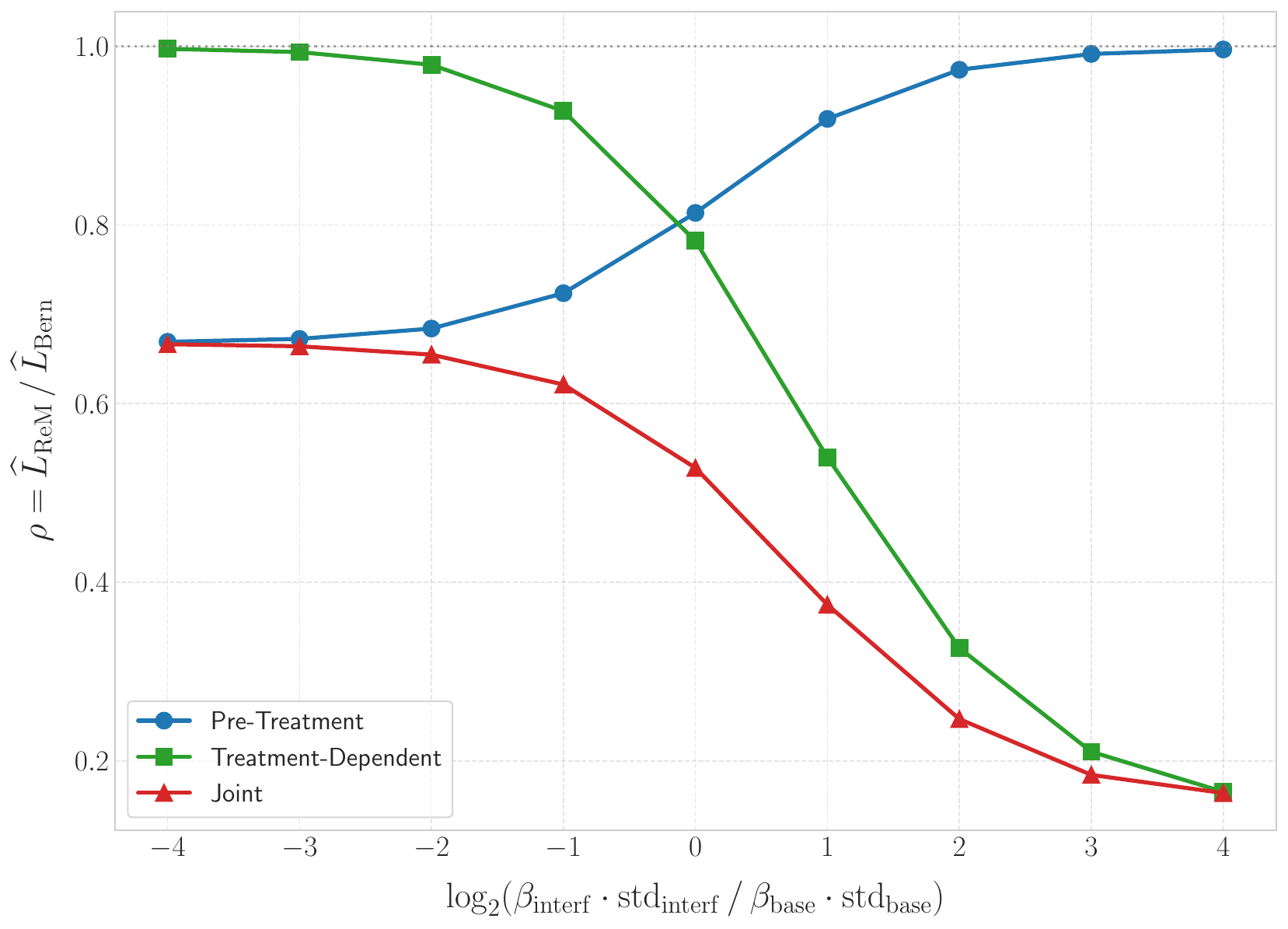"}
\caption{Nonlinear, balancing set NP+NWX}
\label{fig:conserv_ratio_type5_npnwx_main}
\end{subfigure}
\caption{Confidence interval length ratio $\rho$ in~\eqref{eq:ci_length_ratio} as a function of the log signal ratio $\log_2 \kappa$. Rows correspond to outcome settings (top: linear; bottom: nonlinear); columns correspond to balancing sets (left: NP; right: NP+NWX). Blue: pre-treatment strategy; green: treatment-dependent strategy; red: joint strategy. }
\label{fig:conserv_main}
\end{figure}

Figure~\ref{fig:conserv_main} shows the empirical ratio $\rho$ as a function of $\log_2\kappa$. Across all four panels and all three strategies, the ratio satisfies $\rho \le 1$. In other words, the proposed variance estimator successfully translates the design-stage variance reduction into a tighter inference: the confidence intervals constructed under rerandomization are at most as wide as those under the Bernoulli design and substantially shorter when the balancing covariates are predictive of the outcome.

The crossover pattern in Figure~\ref{fig:conserv_main} parallels that of Figure~\ref{fig:vr_main} for a transparent theoretical reason. Corollary~\ref{thm:variance_reduction} predicts that the conditional variance scales the unconditional variance by the factor $1-(1-v_p)R^2$, where $R^2$ is the population coefficient of determination between $\hat{\tau}$ and the balanced covariate imbalance vector $\tauxhat$. As $\log_2\kappa$ shifts from pre-treatment-dominated to interference-dominated, the dominant source of $R^2$ migrates from the pre-treatment covariates $\mathbf X$ to the treatment-dependent covariates $H_{i,\text{prop}}$ (and $\mathbf H_{i,\text{wX}}$ when included). The pre-treatment strategy thus delivers its largest improvement on the left, the treatment-dependent strategy on the right, and the joint strategy maintains a uniformly low $\rho$ because it retains explanatory power across both regimes. Comparing the NP and NP+NWX columns again shows that adding $\mathbf H_{i,\text{wX}}$ does not improve the ratio for these outcome models, in line with the practical caution noted in the previous subsection.

\subsection{Variance estimator: comparison with alternative estimators}
\label{subsection:sim_var_comp}
In this section, we examine alternative constructions of the variance estimator and compare them with the approach in Sections~\ref{sec:conservative_var} and~\ref{sec:opt}.

Specifically, we consider three different constructions of the unconditional joint covariance bound $\hat{\mathbf{U}}_n$: the estimator proposed in Definition~\ref{def:Un_estimator} and two alternatives. For each construction, we apply the scale-adaptive ridge below and plug the corrected bound into the ridge-adjusted optimization problem in~\eqref{eq:ridge-optimization-problem} to obtain the corresponding variance estimator $\hat v_{n,\delta}$ under rerandomization.

\begin{enumerate}
\item[(a)] \textbf{Proposed (group centering with per-unit local degrees).} This is the construction in Definition~\ref{def:Un_estimator}, repeated here for convenience:
\begin{equation}
\label{eq:Uhat_ours}
\hat{\mathbf{U}}_n^{\mathrm{ours}} = \frac{1}{n^2} \sum_{i=1}^n d_{n,i} \left[
\frac{Z_i\,\hat{\boldsymbol{\phi}}_i\hat{\boldsymbol{\phi}}_i^\top}{\pi^2} +
\frac{(1-Z_i)\,\hat{\boldsymbol{\phi}}_i\hat{\boldsymbol{\phi}}_i^\top}{(1-\pi)^2} \right],
\end{equation}
where $\hat{\boldsymbol{\phi}}_i$ is the within-treatment-group-centered version of $\boldsymbol{\phi}_i$ (i.e., the residual after subtracting the empirical mean over units sharing the treatment status of unit $i$), and $d_{n,i}$ is the local dependency degree introduced in Section~\ref{sec:conservative}.

\item[(b)] \textbf{No centering, per-unit local degrees.} This variant retains the per-unit weighting $d_{n,i}$ but uses the raw vectors $\boldsymbol{\phi}_i$ rather than the group-centered $\hat{\boldsymbol{\phi}}_i$:
\begin{equation}
\label{eq:Uhat_uncen}
\hat{\mathbf{U}}_n^{\mathrm{uncen}} = \frac{1}{n^2} \sum_{i=1}^n d_{n,i} \left[
\frac{Z_i\,\boldsymbol{\phi}_i\boldsymbol{\phi}_i^\top}{\pi^2} +
\frac{(1-Z_i)\,\boldsymbol{\phi}_i\boldsymbol{\phi}_i^\top}{(1-\pi)^2} \right].
\end{equation}
Comparing (a) and (b) isolates the contribution of the group-centering step.

\item[(c)] \textbf{Group centering, spectral-radius inflation.} This variant retains the centering step but replaces the per-unit weighting $d_{n,i}$ with a single global factor $\lambda_{\max}$, the spectral radius of the dependency adjacency matrix:
\begin{equation}
\label{eq:Uhat_spec}
\hat{\mathbf{U}}_n^{\mathrm{spec}} = \frac{\lambda_{\max}}{n^2} \sum_{i=1}^n \left[
\frac{Z_i\,\hat{\boldsymbol{\phi}}_i\hat{\boldsymbol{\phi}}_i^\top}{\pi^2} +
\frac{(1-Z_i)\,\hat{\boldsymbol{\phi}}_i\hat{\boldsymbol{\phi}}_i^\top}{(1-\pi)^2} \right].
\end{equation}
This mirrors the global spectral-radius inflation considered by \citet{savje2021average}, while retaining our within-arm centering step, which is not studied in \citet{savje2021average}. Comparing (a) and (c) isolates the contribution of replacing a single global graph summary with per-unit local degrees.
\end{enumerate}

For each construction $\mathrm{m} \in \{\mathrm{ours}, \mathrm{uncen}, \mathrm{spec}\}$, let
\[
\hat{\mathbf S}_n^{\mathrm m}
=\operatorname{diag}\!\left(
\hat U_{11}^{\mathrm m},\Sigma_{xx,11},\ldots,\Sigma_{xx,pp}
\right),
\qquad
\hat{\mathbf U}_{n,\delta}^{\mathrm m}
=\hat{\mathbf U}_n^{\mathrm m}+\delta\hat{\mathbf S}_n^{\mathrm m}.
\]
This is the variant-specific version of the scale-adaptive ridge in \eqref{eq:scale_adaptive_ridge}. We use the same prespecified $\delta$ for all three methods and apply the computation rule in Appendix~\ref{app:feasibility-refinement} to $\hat{\mathbf U}_{n,\delta}^{\mathrm m}$ to obtain the corresponding variance estimator $\hat v_{n,\delta}^{\mathrm m}$ under rerandomization.

To compare the three estimators, we report on the log scale the \emph{conservativeness ratio}
\begin{equation}
\label{eq:conserv_ratio_log}
\mathrm{CR}^{\mathrm{m}} \;=\; \log\!\left(
\frac{\E\!\left[\hat v_{n,\delta}^{\mathrm m} \,\middle|\, M_n
\le a\right]} {\Var\!\left(\hat{\tau} \,\middle|\, M_n \le a\right)} \right), \qquad \mathrm{m} \in
\{\mathrm{ours}, \mathrm{uncen}, \mathrm{spec}\}.
\end{equation}
Here, the numerator is the expected value of the variance estimator, and the denominator is the true variance of the H\'ajek estimator. A value $\mathrm{CR}^{\mathrm{m}} = 0$ corresponds to a perfectly calibrated variance estimator, $\mathrm{CR}^{\mathrm{m}} > 0$ indicates conservative behavior (overestimation of the conditional variance), and $\mathrm{CR}^{\mathrm{m}} < 0$ would indicate anti-conservative behavior and a violation of the inferential guarantee.

\begin{figure}
\centering
\begin{subfigure}[b]{0.48\textwidth}
\centering
\includegraphics[width=\linewidth]{"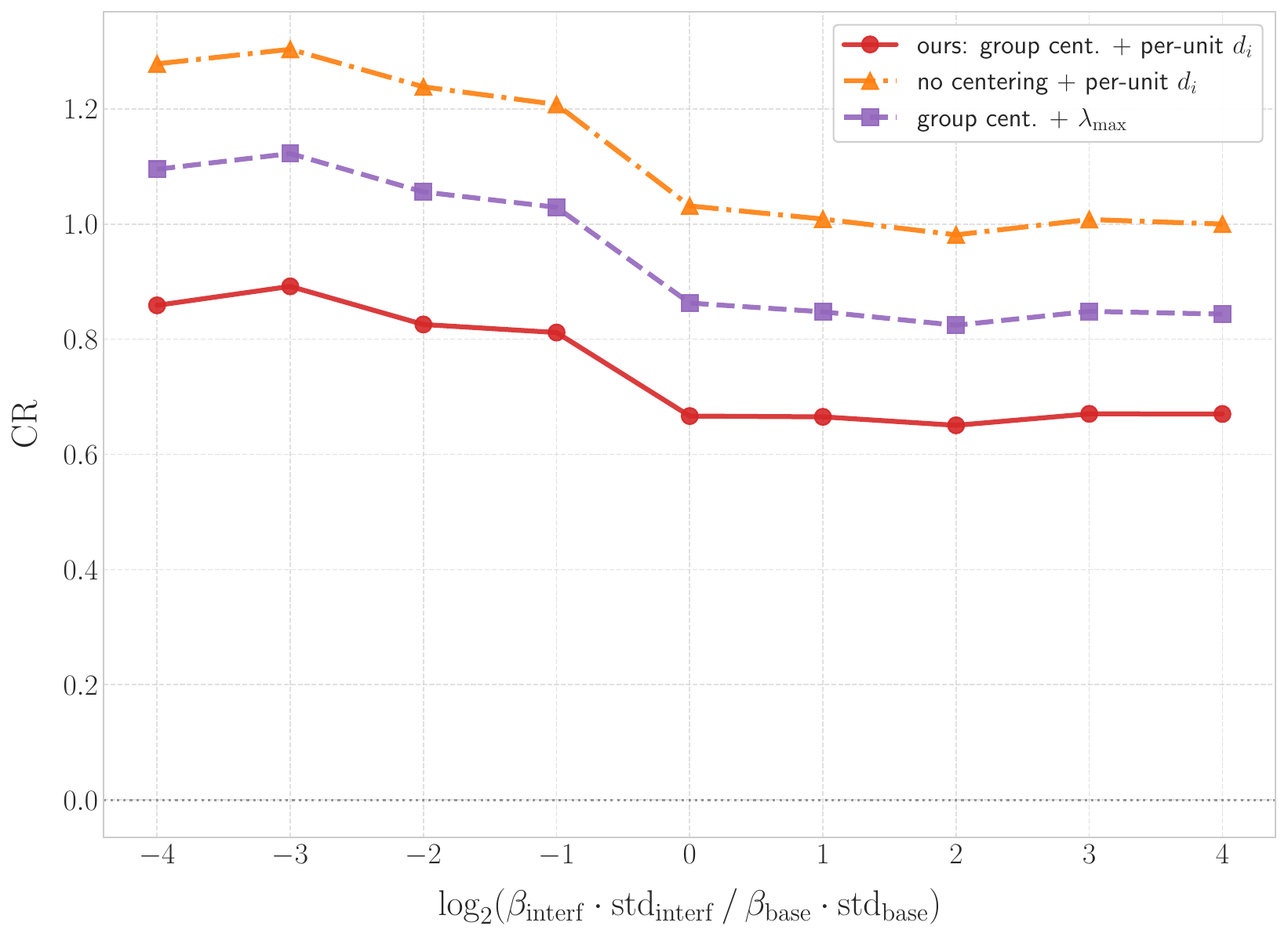"}
\caption{Linear, balancing set NP}
\label{fig:comparison_type1_np}
\end{subfigure}
\hfill
\begin{subfigure}[b]{0.48\textwidth}
\centering
\includegraphics[width=\linewidth]{"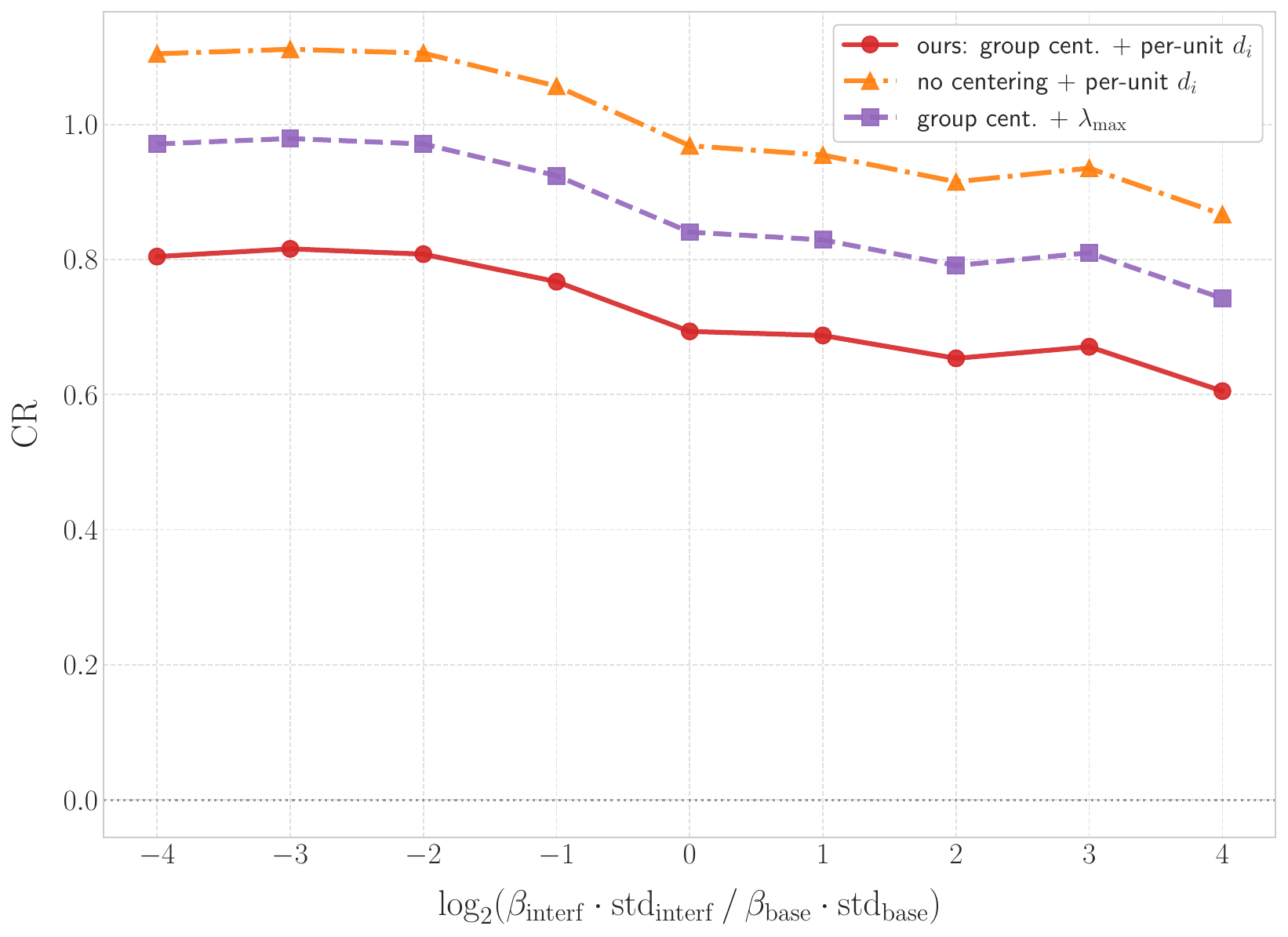"}
\caption{Linear, balancing set NP+NWX}
\label{fig:comparison_type1_npnwx}
\end{subfigure}

\vspace{0.5em}

\begin{subfigure}[b]{0.48\textwidth}
\centering
\includegraphics[width=\linewidth]{"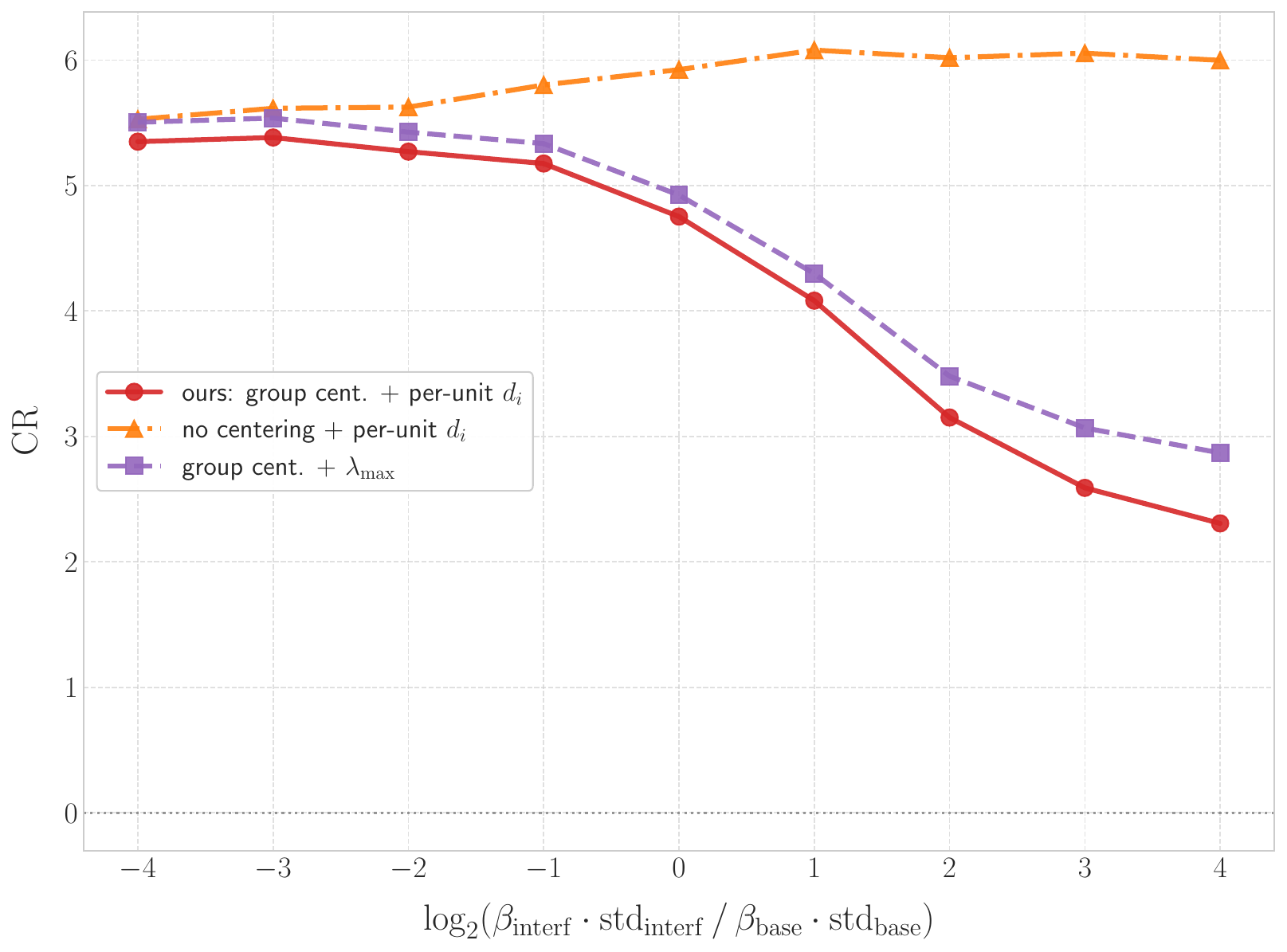"}
\caption{Nonlinear, balancing set NP}
\label{fig:comparison_type5_np}
\end{subfigure}
\hfill
\begin{subfigure}[b]{0.48\textwidth}
\centering
\includegraphics[width=\linewidth]{"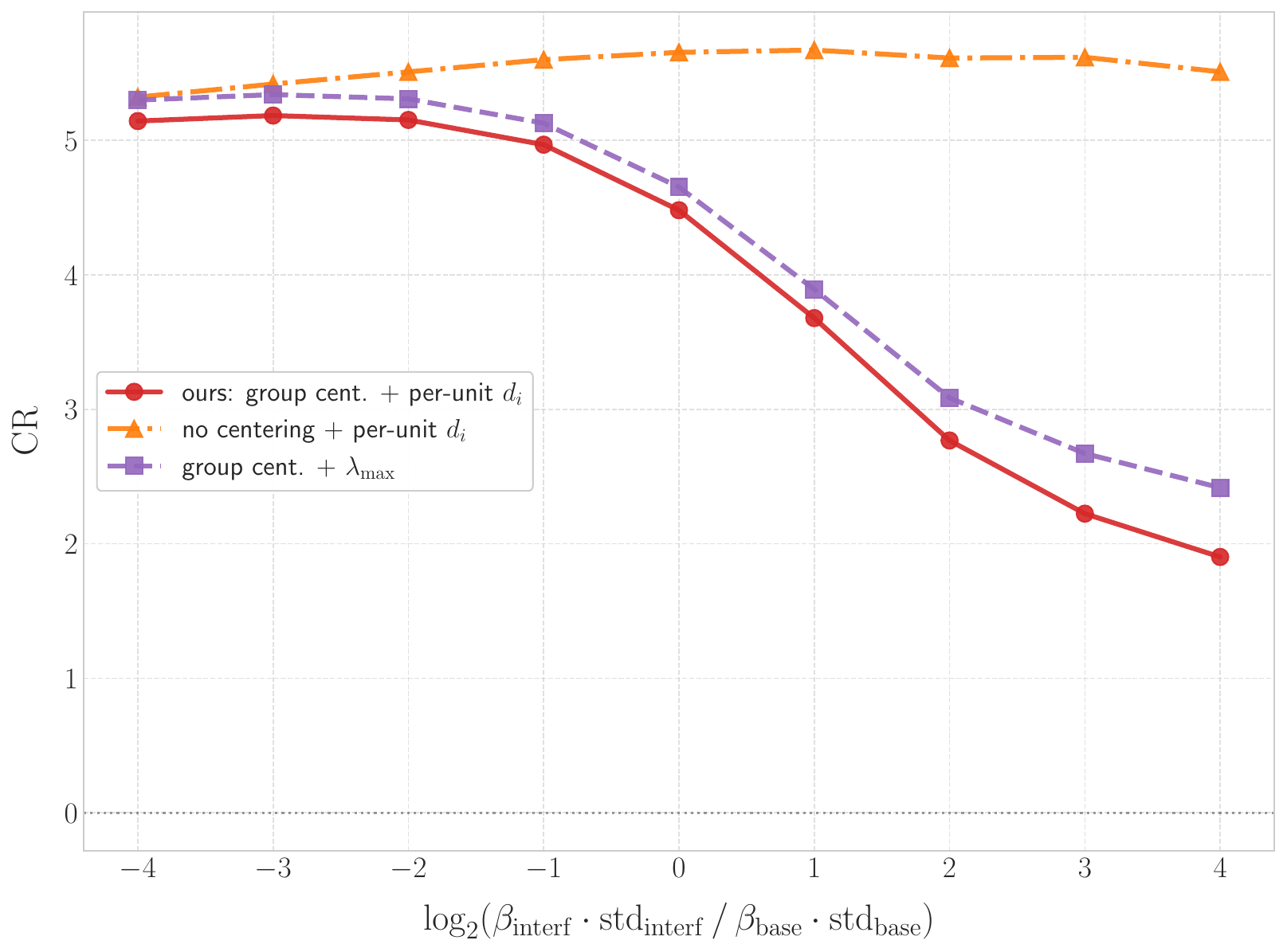"}
\caption{Nonlinear, balancing set NP+NWX}
\label{fig:comparison_type5_npnwx}
\end{subfigure}
\caption{Comparison of conservativeness ratios for three variance-estimator constructions, shown as a function of the log signal ratio $\log_2 \kappa$. Rows correspond to outcome settings (top: linear; bottom: nonlinear); columns correspond to balancing sets (left: NP; right: NP+NWX). Red: proposed group-centering with per-unit local degrees; orange: no-centering with per-unit local degrees; purple: group-centering with spectral-radius inflation.}
\label{fig:comparison_main}
\end{figure}

Figure~\ref{fig:comparison_main} plots the values of $\mathrm{CR}^{\mathrm{m}}$. Across all four panels and all three methods, $\mathrm{CR}^{\mathrm{m}} \ge 0$, confirming that all three variance estimators are conservative. The differences across constructions, however, are informative.

In the linear setting, the three estimators lie within a narrow band. The proposed estimator, $\hat v_{n,\delta}^{\mathrm{ours}}$, achieves the smallest CR, followed by the spectral-radius variant, $\hat v_{n,\delta}^{\mathrm{spec}}$, which incurs only a modest increase, while the no-centering variant, $\hat v_{n,\delta}^{\mathrm{uncen}}$, is slightly more conservative.

In the nonlinear setting, the contrast becomes more pronounced. The exponential outcome map can generate large within-group outcome means, helping explain why failing to center substantially inflates the variance estimator in these simulations: $\hat v_{n,\delta}^{\mathrm{uncen}}$ remains far above the other two across the entire signal range. The spectral-radius variant performs better and is closer to the proposed estimator than to the no-centering variant. However, it still trails $\hat v_{n,\delta}^{\mathrm{ours}}$ by a noticeable margin in the interference-dominated regime. This can occur because the degree distribution is highly heterogeneous, and a single global scaling factor can apply excessive inflation to units with smaller local degrees.

These results illustrate two features of the construction in Section~\ref{sec:conservative_var}. First, group centering removes group-level outcome components and substantially reduces the variance estimates in these simulations, with especially pronounced gains under the nonlinear specifications. Second, per-unit local weighting becomes more important as the dependency graph becomes more heterogeneous, since a global $\lambda_{\max}$-based bound applies the same worst-case inflation to every unit. The Barab\'{a}si--Albert network highlights this second effect, as its degree distribution is highly heterogeneous. In contrast, a more homogeneous design, such as Watts--Strogatz with fixed degree, would reduce this gap. We note that the ordering between the proposed estimator and the spectral-radius variant is not universal: when the outcome's magnitude scales with the unit's degree, hub residuals are amplified by both centering and the per-unit weight $d_{n,i}$, and the global $\lambda_{\max}$ alternative can be tighter. We document and explain this regime in Appendix~\ref{sec:sim_comparison_full} for the \texttt{exp-sumexp} and \texttt{exp-sumexp-full} models.

\section{Discussion}
\label{sec:discussion}

This paper develops a rerandomization framework for experiments with interference. In contrast to much of the interference literature, which requires modeling assumptions to incorporate covariates, we show that for the EATE and the H\'ajek estimator, rerandomization provides a model-agnostic way to achieve an asymptotic no-harm property: incorporating covariates at the design stage is guaranteed to be asymptotically no worse than not adjusting. Our results extend the classical rerandomization insights of \citet{morgan2012rerandomization} and \citet{li2018asymptotic} to settings where interference induces cross-unit dependence.

We conclude by outlining several directions for future work.

One direction is to move beyond the EATE to estimands that capture spillover or total effects. This is challenging because estimating spillover or total effects typically requires specifying an interference structure, such as exposure mappings or neighborhoods \citep{hudgens2008toward, leung2022causal}, making model-agnostic guarantees difficult. At the same time, existing estimators for spillover or total effects can be biased when the assumed structure is misspecified. This suggests that rerandomization may help not only with variance reduction, but also with mitigating finite-sample bias by balancing covariates and exposure summaries that drive spillover patterns.

Another direction is to move beyond Bernoulli experiments. Our analysis relies on viewing rerandomization as Bernoulli assignment conditional on a balance event, which is convenient for both the EATE and the asymptotic theory. In many applications, cluster randomization is used when interference is stronger within groups. In such settings, one could rerandomize at both the cluster and individual levels, balancing baseline cluster covariates as well as exposure summaries induced by within- and between-cluster spillovers.

It is also natural to consider combining rerandomization with regression adjustment. Under SUTVA, \citet{liding2020rerandomization} develop a unified randomization-inference theory for this combination, allowing the covariates used in the rerandomization criterion and those used in the analysis-stage adjustment to differ. Their results show that rerandomization can improve precision for a given estimator, while regression adjustment can further improve the estimated precision for a given design; they therefore recommend using rerandomization in the design stage together with regression adjustment and Huber--White robust standard errors in the analysis stage. Extending this conclusion to interference settings is nontrivial. Regression adjustment no longer has the same no-harm guarantee, because residualizing outcomes does not remove the cross-unit covariance induced by spillovers, and such no-harm properties may fail under interference \citep{gao2025causal}. A useful direction for future work is to develop regression-adjusted estimators tailored to rerandomized designs under interference, together with valid variance estimators and inference procedures.

An important practical question is how to choose covariates in interference-aware rerandomization. The theory highlights a trade-off: the benefit of rerandomization depends on the predictive power of the balance vector through $R^2$, while the threshold factor $v_p$ becomes less favorable as the number of covariates increases. Important pre-treatment variables and exposure summaries should be included, but adding many weak covariates can make the acceptance criterion harder to satisfy without improving precision. Our simulations illustrate this effect when incorporating weakly predictive neighbor-based summaries. Future work could study data-independent or sample-splitting approaches to covariate selection, for example using historical data or pilot studies to rank candidates before treatment assignment. In network settings, useful candidates may include baseline covariates, degree or centrality measures, and exposure summaries motivated by the application.

Finally, a practical limitation of the current inference procedure is that the variance estimator can be conservative, especially under nonlinear outcome models or highly heterogeneous networks. This conservativeness ensures coverage but may lead to confidence intervals that are wider than necessary. Developing sharper yet still design-valid variance estimators is an important direction for future work.

\bibliographystyle{plainnat}
\bibliography{ref}

@article{Andrews1984NonstrongMA,
  title={Non-strong mixing autoregressive processes},
  author={Donald W. K. Andrews},
  journal={Journal of Applied Probability},
  year={1984},
  volume={21},
  pages={930 - 934},
}

@article{Samii2012EstimatingAC,
  author  = {Peter M. Aronow and Cyrus Samii},
  title   = {Estimating average causal effects under general interference, with application to a social network experiment},
  journal = {The Annals of Applied Statistics},
  year    = {2017},
  volume  = {11},
  number  = {4},
  pages   = {1912--1947},
  doi     = {10.1214/16-AOAS1005}
}

@article{athey2018exact,
  title={Exact p-values for network interference},
  author={Athey, Susan and Eckles, Dean and Imbens, Guido W},
  journal={Journal of the American Statistical Association},
  volume={113},
  number={521},
  pages={230--240},
  year={2018},
  publisher={Taylor \& Francis}
}

@article{bailey1987restricted,
  title={Restricted randomization: A practical example},
  author={Bailey, Rosemary Anne},
  journal={Journal of the American Statistical Association},
  volume={82},
  number={399},
  pages={712--719},
  year={1987},
  publisher={Taylor \& Francis}
}

@article{albert1999emergence,
  title={Emergence of scaling in random networks},
  author={Barab{\'a}si, Albert-L{\'a}szl{\'o} and Albert, R{\'e}ka},
  journal={Science},
  volume={286},
  number={5439},
  pages={509--512},
  year={1999},
  publisher={American Association for the Advancement of Science}
}

@article{basse2018model,
  title={Model-assisted design of experiments in the presence of network-correlated outcomes},
  author={Basse, Guillaume W and Airoldi, Edoardo M},
  journal={Biometrika},
  volume={105},
  number={4},
  pages={849--858},
  year={2018},
  publisher={Oxford University Press}
}

@article{bloniarz2016lasso,
  title={Lasso adjustments of treatment effect estimates in randomized experiments},
  author={Bloniarz, Adam and Liu, Hanzhong and Zhang, Cun-Hui and Sekhon, Jasjeet S and Yu, Bin},
  journal={Proceedings of the National Academy of Sciences},
  volume={113},
  number={27},
  pages={7383--7390},
  year={2016},
  publisher={National Academy of Sciences}
}

@article{branson2021ridge,
  title={Ridge rerandomization: An experimental design strategy in the presence of covariate collinearity},
  author={Branson, Zach and Shao, Stephane},
  journal={Journal of Statistical Planning and Inference},
  volume={211},
  pages={287--314},
  year={2021},
  publisher={Elsevier}
}

@article{branson2024power,
  title={Power and sample size calculations for rerandomization},
  author={Branson, Zach and Li, Xinran and Ding, Peng},
  journal={Biometrika},
  volume={111},
  number={1},
  pages={355--363},
  year={2024},
  publisher={Oxford University Press}
}

@article{cai2015social,
  title={Social networks and the decision to insure},
  author={Cai, Jing and Janvry, Alain De and Sadoulet, Elisabeth},
  journal={American Economic Journal: Applied Economics},
  volume={7},
  number={2},
  pages={81--108},
  year={2015},
  publisher={American Economic Association 2014 Broadway, Suite 305, Nashville, TN 37203-2425}
}

@article{chang2024exact,
  title={Exact bias correction for linear adjustment of randomized controlled trials},
  author={Chang, Haoge and Middleton, Joel A and Aronow, PM},
  journal={Econometrica},
  volume={92},
  number={5},
  pages={1503--1519},
  year={2024},
  publisher={Wiley Online Library}
}

@article{cohen2022gaussian,
  title={Gaussian prepivoting for finite population causal inference},
  author={Cohen, Peter L and Fogarty, Colin B},
  journal={Journal of the Royal Statistical Society Series B: Statistical Methodology},
  volume={84},
  number={2},
  pages={295--320},
  year={2022},
  publisher={Oxford University Press}
}

@article{cohen2024no,
  title={No-harm calibration for generalized {Oaxaca--Blinder} estimators},
  author={Cohen, Peter L and Fogarty, Colin B},
  journal={Biometrika},
  volume={111},
  number={1},
  pages={331--338},
  year={2024},
  publisher={Oxford University Press}
}

@article{cortez2023exploiting,
  title={Exploiting neighborhood interference with low-order interactions under unit randomized design},
  author={Cortez-Rodriguez, Mayleen and Eichhorn, Matthew and Yu, Christina Lee},
  journal={Journal of Causal Inference},
  volume={11},
  number={1},
  pages={20220051},
  year={2023},
  publisher={De Gruyter}
}

@article{cox1982randomization,
  title={Randomization and concomitant variables in the design of experiments},
  author={Cox, DR},
  journal={Statistics and probability: Essays in honor of CR Rao},
  pages={197--202},
  year={1982},
  publisher={North-Holland Amsterdam}
}

@article{dedecker1998central,
  title={A central limit theorem for stationary random fields},
  author={Dedecker, J{\'e}r{\^o}me},
  journal={Probability Theory and Related Fields},
  volume={110},
  number={3},
  pages={397--426},
  year={1998},
  publisher={Springer}
}

@article{dedecker2001exponential,
  title={Exponential inequalities and functional central limit theorems for random fields},
  author={Dedecker, J{\'e}r{\^o}me},
  journal={ESAIM: Probability and Statistics},
  volume={5},
  pages={77--104},
  year={2001}
}

@incollection{dedecker2007weak,
  title={Weak dependence},
  author={Dedecker, J{\'e}r{\^o}me and Doukhan, Paul and Lang, Gabriel and Le{\'o}n R., Jos{\'e} Rafael and Louhichi, Sana and Prieur, Cl{\'e}mentine},
  booktitle={Weak dependence: With examples and applications},
  pages={9--20},
  year={2007},
  publisher={Springer}
}

@article{doukhan1999new,
  title={A new weak dependence condition and applications to moment inequalities},
  author={Doukhan, Paul and Louhichi, Sana},
  journal={Stochastic processes and their applications},
  volume={84},
  number={2},
  pages={313--342},
  year={1999},
  publisher={Elsevier}
}

@article{eckles2017design,
  title={Design and analysis of experiments in networks: Reducing bias from interference},
  author={Eckles, Dean and Karrer, Brian and Ugander, Johan},
  journal={Journal of Causal Inference},
  volume={5},
  number={1},
  pages={20150021},
  year={2017},
  publisher={De Gruyter}
}

@article{el2013central,
  title={A central limit theorem for stationary random fields},
  author={El Machkouri, Mohamed and Voln{\'y}, Dalibor and Wu, Wei Biao},
  journal={Stochastic Processes and their Applications},
  volume={123},
  number={1},
  pages={1--14},
  year={2013},
  publisher={Elsevier}
}

@article{fan2025causal,
  title={Causal Inference under Interference: Regression Adjustment and Optimality},
  author={Fan, Xinyuan and Leng, Chenlei and Wu, Weichi},
  journal={arXiv preprint arXiv:2502.06008},
  year={2025}
}

@article{fisher1992arrangement,
  title   = {The Arrangement of Field Experiments},
  author  = {Fisher, Ronald A.},
  journal = {Journal of the Ministry of Agriculture},
  volume  = {33},
  pages   = {503--513},
  year    = {1926}
}

@book{fisher1935design,
  author    = {Fisher, Ronald A.},
  title     = {The Design of Experiments},
  year      = {1935},
  publisher = {Oliver and Boyd},
  address   = {Edinburgh}
}

@article{fogarty2018regression,
  title={Regression-assisted inference for the average treatment effect in paired experiments},
  author={Fogarty, Colin B},
  journal={Biometrika},
  volume={105},
  number={4},
  pages={994--1000},
  year={2018},
  publisher={Oxford University Press}
}

@article{gao2025causal,
  title={Causal inference in network experiments: regression-based analysis and design-based properties},
  author={Gao, Mengsi and Ding, Peng},
  journal={Journal of Econometrics},
  volume={252},
  pages={106119},
  year={2025},
  publisher={Elsevier}
}

@article{guo2023generalized,
  title={The generalized {Oaxaca-Blinder} estimator},
  author={Guo, Kevin and Basse, Guillaume},
  journal={Journal of the American Statistical Association},
  volume={118},
  number={541},
  pages={524--536},
  year={2023},
  publisher={Taylor \& Francis}
}

@article{harshaw2026optimized,
  title={Optimized variance estimation under interference and complex experimental designs},
  author={Harshaw, Christopher and Middleton, Joel and S{\"a}vje, Fredrik},
  journal={Journal of the American Statistical Association},
  pages={1--14},
  year={2026},
  publisher={Taylor \& Francis}
}

@article{hu2022average,
  title={Average direct and indirect causal effects under interference},
  author={Hu, Yuchen and Li, Shuangning and Wager, Stefan},
  journal={Biometrika},
  volume={109},
  number={4},
  pages={1165--1172},
  year={2022},
  publisher={Oxford University Press}
}

@article{hudgens2008toward,
  title={Toward causal inference with interference},
  author={Hudgens, Michael G and Halloran, M Elizabeth},
  journal={Journal of the american statistical association},
  volume={103},
  number={482},
  pages={832--842},
  year={2008},
  publisher={Taylor \& Francis}
}

@article{imai2009essential,
  author  = {Imai, Kosuke and King, Gary and Nall, Clayton},
  title   = {The Essential Role of Pair Matching in Cluster-Randomized Experiments, with Application to the Mexican Universal Health Insurance Evaluation},
  journal = {Statistical Science},
  volume  = {24},
  number  = {1},
  pages   = {29--53},
  year    = {2009},
}

@article{imbens2011experimental,
  title={Experimental design for unit and cluster randomized trials},
  author={Imbens, Guido W},
  journal={International Initiative for Impact Evaluation Paper},
  year={2011}
}

@book{imbens2015causal,
  title={Causal inference for statistics, social, and biomedical sciences: An introduction},
  author={Imbens, Guido W and Rubin, Donald B},
  year={2015},
  publisher={Cambridge university press}
}

@article{johansson2022rerandomization,
  title={Rerandomization: A complement or substitute for stratification in randomized experiments?},
  author={Johansson, Per and Schultzberg, M{\aa}rten},
  journal={Journal of Statistical Planning and Inference},
  volume={218},
  pages={43--58},
  year={2022},
  publisher={Elsevier}
}

@article{junior2025does,
  title={Does Rerandomization Help Beyond Covariate Adjustment? A Review and Guide for Theory and Practice},
  author={Ribeiro Junior, Ant{\^o}nio Carlos Herling and Branson, Zach},
  journal={arXiv preprint arXiv:2512.05290},
  year={2025}
}

@article{kallus2018optimal,
  title={Optimal a priori balance in the design of controlled experiments},
  author={Kallus, Nathan},
  journal={Journal of the Royal Statistical Society Series B: Statistical Methodology},
  volume={80},
  number={1},
  pages={85--112},
  year={2018},
  publisher={Oxford University Press}
}

@article{kapelner2021harmonizing,
  title={Harmonizing optimized designs with classic randomization in experiments},
  author={Kapelner, Adam and Krieger, Abba M and Sklar, Michael and Shalit, Uri and Azriel, David},
  journal={The American Statistician},
  volume={75},
  number={2},
  pages={195--206},
  year={2021},
  publisher={Taylor \& Francis}
}

@article{kernan1999stratified,
  title={Stratified randomization for clinical trials},
  author={Kernan, Walter N and Viscoli, Catherine M and Makuch, Robert W and Brass, Lawrence M and Horwitz, Ralph I},
  journal={Journal of clinical epidemiology},
  volume={52},
  number={1},
  pages={19--26},
  year={1999},
  publisher={Elsevier}
}

@article{kojevnikov2021limit,
  title={Limit theorems for network dependent random variables},
  author={Kojevnikov, Denis and Marmer, Vadim and Song, Kyungchul},
  journal={Journal of Econometrics},
  volume={222},
  number={2},
  pages={882--908},
  year={2021},
  publisher={Elsevier}
}

@article{lee2019stable,
  title={Stable limit theorems for empirical processes under conditional neighborhood dependence},
  author={Lee, Ji Hyung and Song, Kyungchul},
  journal={Bernoulli},
  volume={25},
  number={2},
  pages={1189--1224},
  year={2019},
  publisher={JSTOR}
}

@article{leung2020treatment,
  title={Treatment and spillover effects under network interference},
  author={Leung, Michael P},
  journal={Review of Economics and Statistics},
  volume={102},
  number={2},
  pages={368--380},
  year={2020},
  publisher={MIT Press One Rogers Street, Cambridge, MA 02142-1209, USA journals-info~…}
}

@article{leung2022causal,
  title={Causal inference under approximate neighborhood interference},
  author={Leung, Michael P},
  journal={Econometrica},
  volume={90},
  number={1},
  pages={267--293},
  year={2022},
  publisher={Wiley Online Library}
}

@article{li2022random,
  title={Random graph asymptotics for treatment effect estimation under network interference},
  author={Li, Shuangning and Wager, Stefan},
  journal={The Annals of Statistics},
  volume={50},
  number={4},
  pages={2334--2358},
  year={2022},
  publisher={Institute of Mathematical Statistics}
}

@article{li2017general,
  title={General forms of finite population central limit theorems with applications to causal inference},
  author={Li, Xinran and Ding, Peng},
  journal={Journal of the American Statistical Association},
  volume={112},
  number={520},
  pages={1759--1769},
  year={2017},
  publisher={Taylor \& Francis}
}

@article{liding2020rerandomization,
  title={Rerandomization and regression adjustment},
  author={Li, Xinran and Ding, Peng},
  journal={Journal of the Royal Statistical Society Series B: Statistical Methodology},
  volume={82},
  number={1},
  pages={241--268},
  year={2020},
  publisher={Oxford University Press}
}

@article{li2018asymptotic,
  title={Asymptotic theory of rerandomization in treatment--control experiments},
  author={Li, Xinran and Ding, Peng and Rubin, Donald B},
  journal={Proceedings of the National Academy of Sciences},
  volume={115},
  number={37},
  pages={9157--9162},
  year={2018},
  publisher={National Academy of Sciences}
}

@article{lidingrubin2020rerandomization,
  title={Rerandomization in $2^K$ factorial experiments},
  author={Li, Xinran and Ding, Peng and Rubin, Donald B},
  journal={The Annals of Statistics},
  volume={48},
  number={1},
  pages={43--63},
  year={2020},
  publisher={JSTOR}
}

@article{lin2013agnostic,
  title={Agnostic notes on regression adjustments to experimental data: Reexamining {Freedman}'s critique},
  author={Lin, Winston},
  journal={The Annals of Applied Statistics},
  pages={295--318},
  year={2013},
  publisher={JSTOR}
}

@article{liu2025bayesian,
  title={A {Bayesian} criterion for rerandomization},
  author={Liu, Zhaoyang and Han, Tingxuan and Rubin, Donald B and Deng, Ke},
  journal={Journal of the American Statistical Association},
  volume={120},
  number={552},
  pages={2809--2821},
  year={2025},
  publisher={Taylor \& Francis}
}

@article{lu2023design,
  title={Design-based theory for cluster rerandomization},
  author={Lu, Xin and Liu, Tianle and Liu, Hanzhong and Ding, Peng},
  journal={Biometrika},
  volume={110},
  number={2},
  pages={467--483},
  year={2023},
  publisher={Oxford University Press}
}

@article{lu2024adjusting,
  title={Adjusting auxiliary variables under approximate neighborhood interference},
  author={Lu, Xin and Wang, Yuhao and Zhang, Zhiheng},
  journal={arXiv preprint arXiv:2411.19789},
  year={2024}
}

@article{miratrix2013adjusting,
  title={Adjusting treatment effect estimates by post-stratification in randomized experiments},
  author={Miratrix, Luke W and Sekhon, Jasjeet S and Yu, Bin},
  journal={Journal of the Royal Statistical Society Series B: Statistical Methodology},
  volume={75},
  number={2},
  pages={369--396},
  year={2013},
  publisher={Oxford University Press}
}

@article{morgan2012rerandomization,
  author  = {Morgan, Kari Lock and Rubin, Donald B.},
  title   = {Rerandomization to Improve Covariate Balance in Experiments},
  journal = {The Annals of Statistics},
  volume  = {40},
  number  = {2},
  pages   = {1263--1282},
  year    = {2012},
}

@article{morgan2015rerandomization,
  title={Rerandomization to balance tiers of covariates},
  author={Morgan, Kari Lock and Rubin, Donald B},
  journal={Journal of the American Statistical Association},
  volume={110},
  number={512},
  pages={1412--1421},
  year={2015},
  publisher={Taylor \& Francis}
}

@article{negi2021revisiting,
  title={Revisiting regression adjustment in experiments with heterogeneous treatment effects},
  author={Negi, Akanksha and Wooldridge, Jeffrey M},
  journal={Econometric Reviews},
  volume={40},
  number={5},
  pages={504--534},
  year={2021},
  publisher={Taylor \& Francis}
}

@article{neyman1923application,
  title={On the application of probability theory to agricultural experiments. Essay on principles},
  author={Neyman, Jerzy},
  journal={Ann. Agricultural Sciences},
  pages={1--51},
  year={1923}
}

@article{qu2025randomization,
  title={Randomization-based {Z}-estimation for evaluating average and individual treatment effects},
  author={Qu, Tianyi and Du, Jiangchuan and Li, Xinran},
  journal={Biometrika},
  volume={112},
  number={2},
  pages={asaf002},
  year={2025},
  publisher={Oxford University Press}
}

@article{rosenbaum2002covariance,
  title={Covariance adjustment in randomized experiments and observational studies},
  author={Rosenbaum, Paul R},
  journal={Statistical Science},
  volume={17},
  number={3},
  pages={286--327},
  year={2002},
  publisher={Institute of Mathematical Statistics}
}

@article{rosenblatt1956central,
  title={A central limit theorem and a strong mixing condition},
  author={Rosenblatt, Murray},
  journal={Proceedings of the national Academy of Sciences},
  volume={42},
  number={1},
  pages={43--47},
  year={1956}
}

@article{rubin1974estimating,
  title={Estimating causal effects of treatments in randomized and nonrandomized studies.},
  author={Rubin, Donald B},
  journal={Journal of educational Psychology},
  volume={66},
  number={5},
  pages={688},
  year={1974},
  publisher={American Psychological Association}
}

@article{rubindonald1980comment,
	author = {Rubin, D. B.},
	journal = {Journal of American Statistical Association},
	pages = {591--593},
	title = {Comment on ``{R}andomization analysis of experimental data: the {F}isher randomization test'' by {D}. {B}asu.},
	volume = {75},
	year = {1980}}

@article{sacerdote2001peer,
  title={Peer effects with random assignment: Results for {Dartmouth} roommates},
  author={Sacerdote, Bruce},
  journal={The Quarterly journal of economics},
  volume={116},
  number={2},
  pages={681--704},
  year={2001},
  publisher={MIT Press}
}

@article{savje2021average,
  title={Average treatment effects in the presence of unknown interference},
  author={S{\"a}vje, Fredrik and Aronow, Peter and Hudgens, Michael},
  journal={Annals of statistics},
  volume={49},
  number={2},
  pages={673},
  year={2021}
}

@article{simon1979restricted,
  title={Restricted randomization designs in clinical trials},
  author={Simon, Richard},
  journal={Biometrics},
  pages={503--512},
  year={1979},
  publisher={JSTOR}
}

@article{sobel2006randomized,
  title={What do randomized studies of housing mobility demonstrate? Causal inference in the face of interference},
  author={Sobel, Michael E},
  journal={Journal of the American Statistical Association},
  volume={101},
  number={476},
  pages={1398--1407},
  year={2006},
  publisher={Taylor \& Francis}
}

@article{su2021model,
  title={Model-assisted analyses of cluster-randomized experiments},
  author={Su, Fangzhou and Ding, Peng},
  journal={Journal of the Royal Statistical Society Series B: Statistical Methodology},
  volume={83},
  number={5},
  pages={994--1015},
  year={2021},
  publisher={Oxford University Press}
}

@article{tchetgen2012causal,
  title={On causal inference in the presence of interference},
  author={Tchetgen Tchetgen, Eric J. and VanderWeele, Tyler J},
  journal={Statistical methods in medical research},
  volume={21},
  number={1},
  pages={55--75},
  year={2012},
  publisher={SAGE Publications Sage UK: London, England}
}

@inproceedings{toulis2013estimation,
  title={Estimation of causal peer influence effects},
  author={Toulis, Panos and Kao, Edward},
  booktitle={International conference on machine learning},
  pages={1489--1497},
  year={2013},
  organization={PMLR}
}

@article{wang2023model,
  title={Model-robust inference for clinical trials that improve precision by stratified randomization and covariate adjustment},
  author={Wang, Bingkai and Susukida, Ryoko and Mojtabai, Ramin and Amin-Esmaeili, Masoumeh and Rosenblum, Michael},
  journal={Journal of the American Statistical Association},
  volume={118},
  number={542},
  pages={1152--1163},
  year={2023},
  publisher={Taylor \& Francis}
}

@article{wang2024model,
  title={Model-robust and efficient covariate adjustment for cluster-randomized experiments},
  author={Wang, Bingkai and Park, Chan and Small, Dylan S and Li, Fan},
  journal={Journal of the American Statistical Association},
  volume={119},
  number={548},
  pages={2959--2971},
  year={2024},
  publisher={Taylor \& Francis}
}

@article{wang2023rerandomization,
  title={Rerandomization in stratified randomized experiments},
  author={Wang, Xinhe and Wang, Tingyu and Liu, Hanzhong},
  journal={Journal of the American Statistical Association},
  volume={118},
  number={542},
  pages={1295--1304},
  year={2023},
  publisher={Taylor \& Francis}
}

@article{wang2025covariate,
  title={Covariate Adjustment Cannot Hurt: Treatment Effect Estimation under Interference with Low-Order Outcome Interactions},
  author={Wang, Xinyi and Li, Shuangning},
  journal={arXiv preprint arXiv:2509.03050},
  year={2025}
}

@article{wang2022rerandomization,
  title={Rerandomization with diminishing covariate imbalance and diverging number of covariates},
  author={Wang, Yuhao and Li, Xinran},
  journal={The Annals of Statistics},
  volume={50},
  number={6},
  pages={3439--3465},
  year={2022},
  publisher={JSTOR}
}

@article{yang2026design,
    author = {Haoyu Yang and Yichen Qin and Yan Yu and Yang Li},
    title = {Design Strategies for Networked Experiments via Interference Balancing},
    journal = {Journal of Business \& Economic Statistics},
    volume = {0},
    number = {ja},
    pages = {1--31},
    year = {2026},
    publisher = {Taylor \& Francis},
}

@article{yang2023rejective,
  title={Rejective sampling, rerandomization, and regression adjustment in survey experiments},
  author={Yang, Zihao and Qu, Tianyi and Li, Xinran},
  journal={Journal of the American Statistical Association},
  volume={118},
  number={542},
  pages={1207--1221},
  year={2023},
  publisher={Taylor \& Francis}
}

@article{zhang2024pca,
  title={{PCA} rerandomization},
  author={Zhang, Hengtao and Yin, Guosheng and Rubin, Donald B},
  journal={Canadian Journal of Statistics},
  volume={52},
  number={1},
  pages={5--25},
  year={2024},
  publisher={Wiley Online Library}
}

@article{zhao2024no,
  title={No star is good news: A unified look at rerandomization based on p-values from covariate balance tests},
  author={Zhao, Anqi and Ding, Peng},
  journal={Journal of Econometrics},
  volume={241},
  number={1},
  pages={105724},
  year={2024},
  publisher={Elsevier}
}

@article{zhao2022reconciling,
  title={Reconciling design-based and model-based causal inferences for split-plot experiments},
  author={Zhao, Anqi and Ding, Peng},
  journal={The Annals of Statistics},
  volume={50},
  number={2},
  pages={1170--1192},
  year={2022},
  publisher={Institute of Mathematical Statistics}
}

@article{zhao2024covariate,
  title={Covariate adjustment in randomized experiments with missing outcomes and covariates},
  author={Zhao, Anqi and Ding, Peng and Li, Fan},
  journal={Biometrika},
  volume={111},
  number={4},
  pages={1413--1420},
  year={2024},
  publisher={Oxford University Press}
}

@article{janson1988normal,
  title={Normal convergence by higher semiinvariants with applications to sums of dependent random variables and random graphs},
  author={Janson, Svante},
  journal={The Annals of Probability},
  pages={305--312},
  year={1988},
  publisher={JSTOR}
}

\clearpage
\appendix
\etocdepthtag.toc{appendix}

\begingroup
\etocsetnexttocdepth{subsubsection}
\etocsettagdepth{main}{none}
\etocsettagdepth{appendix}{subsubsection}
\etocsettocstyle{\section*{Appendix Contents}}{}
\begin{singlespace}
\small
\pdfbookmark[0]{Appendix Contents}{appendix-contents}
\tableofcontents
\end{singlespace}
\endgroup
\clearpage

\section{Proofs for Section~\ref{sec:asymptotic}}
\label{app:section3-proofs}

\subsection{Notation and H\'ajek linearization}
\label{app:common-linearization}

Throughout this appendix, the covariate dimension $p$ is fixed and $d_\phi=p+1$. Recall
\[
\boldsymbol\phi_i=\boldsymbol\phi_i(\mathbf Z)
=
\begin{pmatrix}
Y_i(\mathbf Z)\\
\mathbf X_i(\mathbf Z)
\end{pmatrix}
\in\mathbb R^{d_\phi}.
\]
For $z\in\{0,1\}$, define the unit-specific and population-average Bernoulli marginal means
\[
\boldsymbol\mu_{n,i,z}
=
\mathbb E\{\boldsymbol\phi_i(z,\mathbf Z_{-i})\},
\qquad
\boldsymbol\mu_{n,z}
=
\frac1n\sum_{i=1}^n\boldsymbol\mu_{n,i,z},
\]
and let
\[
\mathbf D_n
=
\operatorname{diag}\bigl(\sigma_{n,y}^{-1},\sigma_{n,x_1}^{-1},\ldots,
\sigma_{n,x_p}^{-1}\bigr).
\]
Define the influence vector
\begin{equation}
\boldsymbol\eta_{n,i}
=
\frac{Z_i}{\pi}(\boldsymbol\phi_i-\boldsymbol\mu_{n,1})
-
\frac{1-Z_i}{1-\pi}(\boldsymbol\phi_i-\boldsymbol\mu_{n,0}).
\label{eq:t1-influence}
\end{equation}
Writing
\[
\hat{\boldsymbol\tau}_n
=
\begin{pmatrix}
\hat\tau\\
\hat{\boldsymbol\tau}_x
\end{pmatrix},
\qquad
\boldsymbol\tau_n
=
\begin{pmatrix}
\tau\\
\boldsymbol\tau_x
\end{pmatrix},
\qquad
\boldsymbol\tau_x=\mathbb E(\hat{\boldsymbol\tau}_x),
\]
set
\begin{equation}
\mathbf T_n
=
\mathbf D_n(\hat{\boldsymbol\tau}_n-\boldsymbol\tau_n),
\qquad
\widetilde{\mathbf T}_n
=
\frac1n\mathbf D_n\sum_{i=1}^n\boldsymbol\eta_{n,i}.
\label{eq:t1-T-Ttilde}
\end{equation}
Because $Z_i$ is independent of $\mathbf Z_{-i}$,
\[
\mathbb E(\boldsymbol\eta_{n,i})
=
(\boldsymbol\mu_{n,i,1}-\boldsymbol\mu_{n,1})
-
(\boldsymbol\mu_{n,i,0}-\boldsymbol\mu_{n,0}),
\]
and therefore
\begin{equation}
\sum_{i=1}^n\mathbb E(\boldsymbol\eta_{n,i})=0.
\label{eq:t1-total-eta-mean-zero}
\end{equation}

For the random denominators in the H\'ajek estimator, define the standardized arm scores
\begin{equation}
\boldsymbol\zeta_{n,i}
=
\frac1{\sqrt n}
\begin{pmatrix}
\displaystyle \frac{Z_i}{\pi}\mathbf D_n(\boldsymbol\phi_i-\boldsymbol\mu_{n,1})\\[3mm]
\displaystyle \frac{1-Z_i}{1-\pi}\mathbf D_n(\boldsymbol\phi_i-\boldsymbol\mu_{n,0})
\end{pmatrix}
=
\begin{pmatrix}
\boldsymbol\zeta_{n,i}^{(1)}\\
\boldsymbol\zeta_{n,i}^{(0)}
\end{pmatrix}
\in\mathbb R^{2d_\phi},
\label{eq:t1-augmented-score}
\end{equation}
and let $\boldsymbol\zeta_{n,i}^{\circ}=\boldsymbol\zeta_{n,i}-\mathbb E(\boldsymbol\zeta_{n,i})$. The same calculation as above gives
\begin{equation}
\sum_{i=1}^n\mathbb E(\boldsymbol\zeta_{n,i}^{(z)})=0,
\qquad z=0,1.
\label{eq:t1-arm-score-total-mean-zero}
\end{equation}
For $z\in\{0,1\}$, put
\begin{equation}
\mathbf Q_{n,z}
=
\frac1n\sum_{i=1}^n\boldsymbol\zeta_{n,i}^{(z)}
=
\frac1n\sum_{i=1}^n\boldsymbol\zeta_{n,i}^{\circ,(z)},
\label{eq:t1-Qnz}
\end{equation}
where $\boldsymbol\zeta_{n,i}^{\circ,(z)}$ denotes the corresponding $d_\phi$-dimensional block of $\boldsymbol\zeta_{n,i}^{\circ}$. Under Assumptions~\ref{ass:boundedness} and \ref{ass:variance_lower_bound}, there is a constant $C<\infty$ such that
\begin{equation}
\sup_{n\ge1}\max_{i\in[n]}
\left\{\|\boldsymbol\zeta_{n,i}\|_2+\|\boldsymbol\zeta_{n,i}^{\circ}\|_2\right\}
\le C.
\label{eq:t1-zeta-uniform}
\end{equation}
We will also use the elementary inequality
\begin{equation}
|\operatorname{Cov}(U,W)|
\le
2\|W\|_\infty\,\mathbb E|U|
\label{eq:t1-elementary-cov-bound}
\end{equation}
for integrable $U$ and bounded $W$.

\begin{lemma}[$L^2$ H\'ajek linearization]
\label{lem:t1-hajek-linearization}
Under Assumptions~\ref{ass:variance_stability}--\ref{ass:variance_lower_bound}, suppose that, for $z\in\{0,1\}$,
\[
\mathbb E\|\mathbf Q_{n,z}\|_2^2\to0.
\]
Then
\begin{equation}
\mathbb E\|\mathbf T_n-\widetilde{\mathbf T}_n\|_2^2\to0. \label{eq:t1-L2-linearization}
\end{equation}
Consequently,
\begin{equation}
\mathbf T_n-\widetilde{\mathbf T}_n=o_p(1) \qquad\text{and}\qquad
\operatorname{Var}(\widetilde{\mathbf T}_n)\to \mathbf V.
\label{eq:t1-variance-equivalence}
\end{equation}
\end{lemma}

\begin{proof}
By~\eqref{eq:t1-zeta-uniform}, $\|\mathbf Q_{n,z}\|_2\le C$ almost surely. Hence the assumed second-moment convergence also gives
\begin{equation}
\mathbb E\|\mathbf Q_{n,z}\|_2^4
\le C^2\mathbb E\|\mathbf Q_{n,z}\|_2^2
\to0.
\label{eq:t1-Q-fourth-generic}
\end{equation}
Let
\[
\hat\pi_n=\frac1n\sum_{i=1}^nZ_i.
\]
To keep track of the convention on the two extreme assignments, define
\[
\hat{\boldsymbol\tau}_n^{\mathrm{arm}}
=
\frac{\sum_iZ_i\boldsymbol\phi_i}{\sum_iZ_i}
-
\frac{\sum_i(1-Z_i)\boldsymbol\phi_i}{\sum_i(1-Z_i)},
\]

where each ratio is set to zero when its denominator vanishes. The estimator $\hat{\boldsymbol\tau}_n$ in the main text additionally sets the entire contrast to zero whenever one arm is empty. Hence, with
\[
\mathbf H_n=\hat{\boldsymbol\tau}_n-\hat{\boldsymbol\tau}_n^{\mathrm{arm}},
\]
$\mathbf H_n$ is supported on $\{\hat\pi_n\in\{0,1\}\}$. Assumptions~\ref{ass:boundedness} and~\ref{ass:variance_lower_bound} therefore give
\begin{equation}
\mathbb E\|\mathbf D_n\mathbf H_n\|_2^2
\le Cn\{\pi^n+(1-\pi)^n\}
\to0.
\label{eq:t1-extreme-assignment-difference}
\end{equation}
Define the treated-arm Hájek mean, with its zero-denominator ratio set to zero, by
\[
\bar{\boldsymbol{\phi}}_{n,1} = \frac{\sum_iZ_i\boldsymbol\phi_i}{\sum_iZ_i}.
\]
Set
\begin{equation}
\mathbf A_{n,1} = \frac1n\sum_{i=1}^n \frac{Z_i}{\pi} (\boldsymbol\phi_i-\boldsymbol\mu_{n,1}).
\label{eq:t1-An1}
\end{equation}
On the event $\{\hat\pi_n>0\}$,
\begin{equation}
\bar{\boldsymbol{\phi}}_{n,1}-\boldsymbol\mu_{n,1} = \frac{\pi}{\hat\pi_n}\mathbf A_{n,1} =
\mathbf A_{n,1}+\mathbf R_{n,1}, \qquad \mathbf R_{n,1} = \left(\frac{\pi}{\hat\pi_n}-1\right)
\mathbf A_{n,1}. \label{eq:t1-treated-exact}
\end{equation}
By \eqref{eq:t1-augmented-score} and \eqref{eq:t1-Qnz},
\begin{equation}
\mathbf D_n\mathbf A_{n,1} = \sqrt n\,\mathbf Q_{n,1}. \label{eq:t1-DA-Q}
\end{equation}

Define the event
\[
\mathcal E_{n,1} = \left\{\hat\pi_n\ge\frac\pi2\right\}.
\]
On $\mathcal E_{n,1}$,
\[
\left|\frac\pi{\hat\pi_n}-1\right| = \frac{|\hat\pi_n-\pi|}{\hat\pi_n} \le
\frac2\pi|\hat\pi_n-\pi|.
\]
Using \eqref{eq:t1-DA-Q},
\begin{equation}
\|\mathbf D_n\mathbf R_{n,1}\|_2^2 \mathbf 1\{\mathcal E_{n,1}\} \le C n(\hat\pi_n-\pi)^2
\|\mathbf Q_{n,1}\|_2^2. \label{eq:t1-good-rem-bound}
\end{equation}
The fourth central moment of the binomial average $\hat \pi_n$ satisfies
\begin{equation}
\mathbb E(\hat\pi_n-\pi)^4=\mathcal O(n^{-2}), \qquad \sup_n\mathbb E\{n^2(\hat\pi_n-\pi)^4\}<\infty.
\label{eq:t1-binomial-fourth}
\end{equation}
Thus, by Cauchy--Schwarz, \eqref{eq:t1-good-rem-bound}, and \eqref{eq:t1-Q-fourth-generic},
\begin{align}
&\mathbb E\left[ \|\mathbf D_n\mathbf R_{n,1}\|_2^2 \mathbf 1\{\mathcal E_{n,1}\} \right] \notag\\
&\quad\le C \left[ \mathbb E\{n^2(\hat\pi_n-\pi)^4\} \right]^{1/2} \left[ \mathbb E\|\boldsymbol
Q_{n,1}\|_2^4 \right]^{1/2} \to0. \label{eq:t1-good-rem-expectation}
\end{align}

For the complement, a Chernoff bound gives
\begin{equation}
\mathbb P(\mathcal E_{n,1}^c)\le e^{-c_1n} \label{eq:t1-chernoff-treated}
\end{equation}
for some $c_1>0$.
\[
\mathbf R_{n,1} = \bar{\boldsymbol{\phi}}_{n,1} -\boldsymbol\mu_{n,1} -\mathbf A_{n,1}
\]
is uniformly bounded: each of the three displayed vectors is uniformly bounded under Assumption~\ref{ass:boundedness}. Assumption~\ref{ass:variance_lower_bound} gives $\|\mathbf D_n\|_{\mathrm{op}}\le C\sqrt n$. Therefore,
\begin{equation}
\mathbb E\left[ \|\mathbf D_n\mathbf R_{n,1}\|_2^2 \mathbf 1\{\mathcal E_{n,1}^c\} \right] \le Cn
e^{-c_1n} \to0. \label{eq:t1-bad-rem-expectation}
\end{equation}
Combining \eqref{eq:t1-good-rem-expectation} and \eqref{eq:t1-bad-rem-expectation},
\begin{equation}
\mathbb E\|\mathbf D_n\mathbf R_{n,1}\|_2^2\to0. \label{eq:t1-treated-rem-L2}
\end{equation}

For the control arm, define
\[
\bar{\boldsymbol{\phi}}_{n,0}
=
\frac{\sum_i(1-Z_i)\boldsymbol\phi_i}{\sum_i(1-Z_i)},
\qquad
\mathbf A_{n,0}
=
\frac1n\sum_{i=1}^n
\frac{1-Z_i}{1-\pi}
(\boldsymbol\phi_i-\boldsymbol\mu_{n,0}),
\]
and
\[
\mathbf R_{n,0}
=
\bar{\boldsymbol{\phi}}_{n,0}
-\boldsymbol\mu_{n,0}
-\mathbf A_{n,0}.
\]
The same argument, applied on $\{1-\hat\pi_n\ge(1-\pi)/2\}$, gives
\begin{equation}
\mathbb E\|\mathbf D_n\mathbf R_{n,0}\|_2^2\to0. \label{eq:t1-control-rem-L2}
\end{equation}

Subtracting the two arm expansions gives
\begin{equation}
\hat{\boldsymbol\tau}_n^{\mathrm{arm}} -(\boldsymbol\mu_{n,1}-\boldsymbol\mu_{n,0}) =
\frac1n\sum_{i=1}^n\boldsymbol\eta_{n,i} +\mathbf R_n^{\circ},
\qquad
\mathbf R_n^{\circ}=\mathbf R_{n,1}-\mathbf R_{n,0},
\label{eq:t1-pretarget-expansion}
\end{equation}
and
\begin{equation}
\mathbb E\|\mathbf D_n\mathbf R_n^{\circ}\|_2^2\to0.
\label{eq:t1-pretarget-rem-L2}
\end{equation}

The first coordinate of $\boldsymbol\mu_{n,1}-\boldsymbol\mu_{n,0}$ is precisely the EATE $\tau$. For the covariate coordinates, however, the target in the Mahalanobis criterion is $\boldsymbol\tau_x=\mathbb E(\hat{\boldsymbol\tau}_x)$. Recall from~\eqref{eq:t1-total-eta-mean-zero} that $\sum_{i=1}^n\mathbb E(\boldsymbol\eta_{n,i})=0$. Taking expectations in the covariate block of \eqref{eq:t1-pretarget-expansion}, and using $\hat{\boldsymbol\tau}_n=\hat{\boldsymbol\tau}_n^{\mathrm{arm}}+\mathbf H_n$, therefore yields
\begin{equation}
\boldsymbol\tau_x
-
(\boldsymbol\mu_{n,1,x}-\boldsymbol\mu_{n,0,x})
=
\mathbb E(\mathbf R_{n,x}^{\circ}+\mathbf H_{n,x}).
\label{eq:t1-covariate-center-shift}
\end{equation}

Define the final remainder
\begin{equation}
\mathbf R_n
=
\begin{pmatrix}
R_{n,y}^{\circ}+H_{n,y}\\
\mathbf R_{n,x}^{\circ}+\mathbf H_{n,x}
-\mathbb E(\mathbf R_{n,x}^{\circ}+\mathbf H_{n,x})
\end{pmatrix}.
\label{eq:t1-final-remainder}
\end{equation}
Equations~\eqref{eq:t1-pretarget-expansion} and \eqref{eq:t1-covariate-center-shift} give the exact expansion
\begin{equation}
\hat{\boldsymbol\tau}_n-\boldsymbol\tau_n = \frac1n\sum_{i=1}^n\boldsymbol\eta_{n,i} +\mathbf
R_n. \label{eq:t1-final-expansion}
\end{equation}
Centering the covariate block can only reduce its second moment, so
\begin{align}
\mathbb E\|\mathbf D_n\mathbf R_n\|_2^2
&\le
\mathbb E\|\mathbf D_n(\mathbf R_n^{\circ}+\mathbf H_n)\|_2^2 \notag\\
&\le
2\mathbb E\|\mathbf D_n\mathbf R_n^{\circ}\|_2^2
+2\mathbb E\|\mathbf D_n\mathbf H_n\|_2^2
\to0.
\label{eq:t1-final-rem-L2}
\end{align}
Equations~\eqref{eq:t1-T-Ttilde}, \eqref{eq:t1-final-expansion}, and \eqref{eq:t1-final-rem-L2} prove \eqref{eq:t1-L2-linearization}.

For variance equivalence, set $\boldsymbol\Delta_n=\mathbf T_n-\widetilde{\mathbf T}_n$. The $L^2$ convergence gives
\[
\|\operatorname{Var}(\boldsymbol\Delta_n)\|_{\mathrm{op}}\to0.
\]
Assumption~\ref{ass:variance_stability} gives $\operatorname{Var}(\mathbf T_n)\to\mathbf V$, and hence
\[
\|\operatorname{Cov}(\mathbf T_n,\boldsymbol\Delta_n)\|_{\mathrm{op}} \le
\|\operatorname{Var}(\mathbf T_n)\|_{\mathrm{op}}^{1/2}
\|\operatorname{Var}(\boldsymbol\Delta_n)\|_{\mathrm{op}}^{1/2} \to0.
\]
Since $\widetilde{\mathbf T}_n=\mathbf T_n-\boldsymbol\Delta_n$,
\[
\operatorname{Var}(\widetilde{\mathbf T}_n) = \operatorname{Var}(\mathbf T_n)+o(1) \to
\mathbf V.
\]
This proves \eqref{eq:t1-variance-equivalence}.
\end{proof}

\subsection{Proof of Theorem~\ref{thm:joint-normality}}
\label{app:proof-theorem-1}

Throughout this subsection, let $\Delta_n$ denote the maximum degree of $\mathcal G_n^{\mathrm{dep}}$; equivalently,
\[
\Delta_n
=
\max_{i\in[n]}d_{n,i}^{\mathrm{dep}}
=
\max_{i\in[n]}\sum_{j=1}^n A_{n,ij}^{\mathrm{dep}}
\,.
\]

\begin{lemma}[Arm-score law of large numbers under a dependency graph]
\label{lem:dg-arm-score-lln}
Under Assumptions~\ref{ass:boundedness}--\ref{ass:network_sparsity}, for $z\in\{0,1\}$,
\begin{equation}
\mathbb E\|\mathbf Q_{n,z}\|_2^2
\le
C\frac{\Delta_n}{n}
=o(1).
\label{eq:dg-Q-second}
\end{equation}
\end{lemma}

\begin{proof}
Since $\boldsymbol\zeta_{n,i}^{\circ,(z)}$ is a measurable function of $\mathbf W_{n,i}$, $\mathcal G_n^{\mathrm{dep}}$ is also a dependency graph for the collection $\{\boldsymbol\zeta_{n,i}^{\circ,(z)}:i\in[n]\}$. Fix a coordinate $\ell\in[d_\phi]$. By \eqref{eq:t1-Qnz},
\[
\mathbb E(Q_{n,z,\ell}^2)
=
\frac1{n^2}\sum_{i=1}^n\sum_{j=1}^n
\operatorname{Cov}\bigl(\zeta_{n,i,\ell}^{\circ,(z)},
\zeta_{n,j,\ell}^{\circ,(z)}\bigr).
\]
The covariance is zero whenever $A_{n,ij}^{\mathrm{dep}}=0$. By~\eqref{eq:t1-zeta-uniform}, every remaining covariance is bounded in absolute value by a constant. Since there are at most $\sum_i d_{n,i}^{\mathrm{dep}}\le n\Delta_n$ ordered pairs with $A_{n,ij}^{\mathrm{dep}}=1$,
\[
\mathbb E(Q_{n,z,\ell}^2)
\le
C\frac{\Delta_n}{n}.
\]
Summing over the fixed number $d_\phi$ of coordinates proves the displayed bound. Assumption ~\ref{ass:network_sparsity} implies $\Delta_n/n\to0$.
\end{proof}

We use the following dependency-graph central limit theorem.

\begin{lemma}[Janson's dependency-graph criterion]
\label{lem:janson-dependency-clt}
For each $n$, let $X_{n,1},\ldots,X_{n,N_n}$ have a dependency graph whose maximum degree is at most $D_n$. Suppose $|X_{n,i}|\le L_n$ almost surely, and write $S_n=\sum_{i=1}^{N_n}X_{n,i}$ and $s_n^2=\operatorname{Var}(S_n)$. If, for some fixed integer $r\ge3$,
\[
N_nD_n^{r-1}\left(\frac{L_n}{s_n}\right)^r\to0,
\]
then $(S_n-\mathbb E S_n)/s_n\xrightarrow{d}N(0,1)$.
\end{lemma}

Lemma~\ref{lem:janson-dependency-clt} is a result in \citet{janson1988normal}. We now verify it for every Cram\'er--Wold projection.

\begin{proof}[Proof of Theorem~\ref{thm:joint-normality}]
Lemma~\ref{lem:dg-arm-score-lln} verifies the condition of Lemma~\ref{lem:t1-hajek-linearization}. Hence
\begin{equation}
\mathbf T_n-\widetilde{\mathbf T}_n=o_p(1),
\qquad
\operatorname{Var}(\widetilde{\mathbf T}_n)\to\mathbf V.
\label{eq:dg-linearization}
\end{equation}

Fix a nonzero vector $\boldsymbol\lambda\in\mathbb R^{d_\phi}$ and define
\begin{equation}
X_{n,i}^{(\boldsymbol\lambda)}
=
\frac1{\sqrt n}\boldsymbol\lambda^\top\mathbf D_n
\left\{\boldsymbol\eta_{n,i}-\mathbb E(\boldsymbol\eta_{n,i})\right\},
\qquad
S_{n,\boldsymbol\lambda}
=
\sum_{i=1}^nX_{n,i}^{(\boldsymbol\lambda)}.
\label{eq:dg-summands}
\end{equation}
Because $X_{n,i}^{(\boldsymbol\lambda)}$ is a measurable function of $\mathbf W_{n,i}$, $\mathcal G_n^{\mathrm{dep}}$ is a dependency graph for this scalar array. Assumptions ~\ref{ass:boundedness} and~\ref{ass:variance_lower_bound} imply
\begin{equation}
\max_{i\in[n]}|X_{n,i}^{(\boldsymbol\lambda)}|
\le C_{\boldsymbol\lambda}
\label{eq:dg-summand-bound}
\end{equation}
for a constant independent of $n$: each coordinate of $\boldsymbol\eta_{n,i}$ is uniformly bounded, and $(\sqrt n\,\sigma_{n,y})^{-1}$ and $(\sqrt n\,\sigma_{n,x_j})^{-1}$ are uniformly bounded.

By~\eqref{eq:t1-total-eta-mean-zero},
\begin{equation}
S_{n,\boldsymbol\lambda}
=
\frac1{\sqrt n}\boldsymbol\lambda^\top\mathbf D_n\sum_{i=1}^n\boldsymbol\eta_{n,i}
=
\sqrt n\,\boldsymbol\lambda^\top\widetilde{\mathbf T}_n.
\label{eq:dg-S-Ttilde}
\end{equation}
Let $s_{n,\boldsymbol\lambda}^2=\operatorname{Var}(S_{n,\boldsymbol\lambda})$. From \eqref{eq:dg-linearization} and~\eqref{eq:dg-S-Ttilde},
\begin{equation}
\frac{s_{n,\boldsymbol\lambda}^2}{n}
=
\operatorname{Var}(\boldsymbol\lambda^\top\widetilde{\mathbf T}_n)
\to
\boldsymbol\lambda^\top\mathbf V\boldsymbol\lambda
=:
v_{\boldsymbol\lambda}>0.
\label{eq:dg-variance-rate}
\end{equation}
Thus $s_{n,\boldsymbol\lambda}\asymp\sqrt n$.

Let $\delta>0$ be the constant in Assumption~\ref{ass:network_sparsity}, and choose a fixed integer $r\ge3$ such that $(r-1)\delta>1/2$. After deleting the self-loops, the ordinary maximum degree is at most $\Delta_n$. Applying Lemma~\ref{lem:janson-dependency-clt} with $N_n=n$, $D_n=\Delta_n$, and $L_n=C_{\boldsymbol\lambda}$, the required quantity is bounded by
\begin{align*}
n\Delta_n^{r-1}
\left(\frac{C_{\boldsymbol\lambda}}{s_{n,\boldsymbol\lambda}}\right)^r
&\le
C n\,n^{(r-1)(1/2-\delta)}n^{-r/2}\\
&=
C n^{1/2-(r-1)\delta}
\to0.
\end{align*}
Therefore,
\[
\frac{S_{n,\boldsymbol\lambda}}{s_{n,\boldsymbol\lambda}}
\xrightarrow{d}N(0,1).
\]
Combining this with~\eqref{eq:dg-S-Ttilde} and~\eqref{eq:dg-variance-rate} yields
\[
\boldsymbol\lambda^\top\widetilde{\mathbf T}_n
=
\frac{s_{n,\boldsymbol\lambda}}{\sqrt n}
\frac{S_{n,\boldsymbol\lambda}}{s_{n,\boldsymbol\lambda}}
\xrightarrow{d}
N\bigl(0,\boldsymbol\lambda^\top\mathbf V\boldsymbol\lambda\bigr).
\]
By~\eqref{eq:dg-linearization} and Slutsky's theorem, the same limit holds with $\mathbf T_n$ in place of $\widetilde{\mathbf T}_n$. The Cram\'er--Wold theorem then gives
\[
\mathbf D_n(\hat{\boldsymbol\tau}_n-\boldsymbol\tau_n)
\xrightarrow{d}N(\boldsymbol 0,\mathbf V),
\]
which proves the result.
\end{proof}

\subsection{Proof of Theorem~\ref{thm:rerand_dist}}

By Theorem~\ref{thm:joint-normality}, the standardized joint vector
\begin{equation}\label{eq:Wn_def}
\mathbf W_n
:=
\begin{pmatrix}
W_{1n}\\
\widetilde{\mathbf X}_n
\end{pmatrix}
=
\begin{pmatrix}
\sigma_{n,y}^{-1}(\hat\tau-\tau)\\[2pt]
\mathbf D_{n,x}(\tauxhat-\boldsymbol\tau_x)
\end{pmatrix}
\xrightarrow{d}
\begin{pmatrix}
W_1^\ast\\
\widetilde{\mathbf X}^\ast
\end{pmatrix}
\sim
 N(\boldsymbol 0,\mathbf V),
\end{equation}
where
\[
\mathbf D_{n,x}
=
\operatorname{diag}
\bigl(\sigma_{n,x_1}^{-1},\ldots,\sigma_{n,x_p}^{-1}\bigr).
\]
In particular, $V_{11}=1$, and Assumption~\ref{ass:variance_stability} gives
\[
\mathbf B_n
:=
\mathbf D_{n,x}\boldsymbol\Sigma_{xx}\mathbf D_{n,x}
\to
\mathbf V_{22}.
\]
Since $\mathbf V_{22}\succ0$, $\mathbf B_n$ is positive definite for all sufficiently large $n$. The acceptance criterion can then be written as
\begin{equation}\label{eq:Mn_in_W}
M_n
=
\widetilde{\mathbf X}_n^\top
\mathbf B_n^{-1}
\widetilde{\mathbf X}_n.
\end{equation}
It follows from~\eqref{eq:Wn_def}, \eqref{eq:Mn_in_W}, Slutsky's theorem, and the continuous mapping theorem that
\[
(W_{1n},M_n)
\xrightarrow{d}
(W_1^\ast,M^\ast),
\qquad
M^\ast
:=
\widetilde{\mathbf X}^{\ast\top}
\mathbf V_{22}^{-1}
\widetilde{\mathbf X}^\ast
\sim\chi_p^2.
\]

We next justify the passage to the conditional law. Since $\Pr(M^\ast=a)=0$ and $\Pr(M^\ast\le a)=\Pr(\chi_p^2\le a)>0$, for every bounded continuous function $h$,
\[
\begin{aligned}
\E\{h(W_{1n})\mid M_n\le a\}
&=
\frac{\E[h(W_{1n})\mathbf 1\{M_n\le a\}]}
{\Pr(M_n\le a)}\\
&\to
\frac{\E[h(W_1^\ast)\mathbf 1\{M^\ast\le a\}]}
{\Pr(M^\ast\le a)}
=
\E\{h(W_1^\ast)\mid M^\ast\le a\}.
\end{aligned}
\]
Here the convergence of the numerator follows because the discontinuity set of $(w,m)\mapsto h(w)\mathbf 1\{m\le a\}$ has zero probability under the limiting law. Therefore,
\[
W_{1n}\mid\{M_n\le a\}
\xrightarrow{d}
W_1^\ast\mid\{M^\ast\le a\}.
\]

It remains to identify this conditional Gaussian law. Define
\[
\mathbf J^\ast
=
\mathbf V_{22}^{-1/2}\widetilde{\mathbf X}^\ast
\sim N(\boldsymbol 0,\mathbf I_p),
\qquad
\mathbf c
=
\mathbf V_{22}^{-1/2}\mathbf V_{21}.
\]
The Gaussian projection decomposition gives
\[
W_1^\ast
=
\mathbf c^\top\mathbf J^\ast
+
\sqrt{1-R^2}\,\varepsilon_0,
\qquad
\|\mathbf c\|_2^2
=
\mathbf V_{12}\mathbf V_{22}^{-1}\mathbf V_{21}
=
R^2,
\]
where $\varepsilon_0\sim N(0,1)$ is independent of $\mathbf J^\ast$. Moreover,
\[
\{M^\ast\le a\}
=
\{\|\mathbf J^\ast\|_2^2\le a\}.
\]
By rotational invariance of the standard Gaussian distribution,
\[
\mathbf c^\top\mathbf J^\ast\ \bigm|\,\{\|\mathbf J^\ast\|_2^2\le a\}
\ \overset d=\,\sqrt{R^2}\,L_{p,a}.
\]
Combining the preceding displays proves
\[
\sigma_{n,y}^{-1}(\hat\tau-\tau)\ \bigm|\,\{M_n\le a\}
\xrightarrow{d}
\sqrt{1-R^2}\,\varepsilon_0+\sqrt{R^2}\,L_{p,a}.
\]

\subsection{Proof of Corollary \ref{cor:unbiased}}

The limiting distribution in~\eqref{eqn:limit_unscaled} is symmetric about zero and has mean zero. The result therefore follows directly from the definition of $\mathbb E_{\operatorname{asymp}}$.

\subsection{Proof of Corollary~\ref{thm:variance_reduction}}
\label{proof:thm2}

The standard truncated-Gaussian identity gives
\[
\Var(L_{p,a})
=
v_p
=
\frac{\Pr(\chi_{p+2}^2\le a)}{\Pr(\chi_p^2\le a)}.
\]
Therefore, by Theorem~\ref{thm:rerand_dist} and the independence of $\varepsilon_0$ and $L_{p,a}$,
\[
\Var\left\{
\sqrt{1-R^2}\,\varepsilon_0+\sqrt{R^2}\,L_{p,a}
\right\}
=
1-R^2+v_pR^2
=
1-(1-v_p)R^2.
\]
Multiplying by $\sigma_{n,y}^2$ proves the stated variance formula.

For completeness, the asymptotic-correlation representation of $R^2$ follows by writing
\[
\Cov_{\operatorname{asymp}}(\hat\tau,\tauxhat)
=
\sigma_{n,y}\mathbf D_{n,x}^{-1}\mathbf V_{21},
\qquad
\Var_{\operatorname{asymp}}(\tauxhat)
=
\mathbf D_{n,x}^{-1}\mathbf V_{22}\mathbf D_{n,x}^{-1}.
\]
Substitution cancels the diagonal scaling matrices and gives
\[
R^2
=
\frac{\mathbf V_{12}\mathbf V_{22}^{-1}\mathbf V_{21}}{V_{11}}.
\]

\section{Discussion and Extensions to Section~\ref{sec:asymptotic}}
\label{app:asymptotic-extensions}
\subsection{Joint asymptotic normality under approximate locality}
\label{app:approximate-locality}

Assumption~\ref{ass:network_sparsity} requires exact independence whenever two sets of vertices are separated in the dependency graph. We now replace this requirement by a local approximation condition. A unit-level score may depend on every treatment assignment, but it must be well approximated using assignments in a growing graph neighborhood around that unit.

Let $\mathcal L_n=([n],E_n^{\mathrm{loc}})$ be a connected graph with graph distance $\operatorname{dist}_n$. Note that this graph need not be observed by the experimenter. We will later see that the assumption only requires the existence of such a graph with certain properties.

For a unit $i$ and an integer $r\ge0$, define the radius-$r$ ball, the distance-$r$ shell, and the assignments observed inside the ball by
\begin{gather*}
\mathcal N_n(i;r)=\{j\in[n]:\operatorname{dist}_n(i,j)\le r\},
\qquad
\partial\mathcal N_n(i;r)=\{j\in[n]:\operatorname{dist}_n(i,j)=r\},\\
\mathcal F_{n,i}(r)=\sigma\{Z_j:j\in\mathcal N_n(i;r)\}.
\end{gather*}
Thus, $\mathcal N_n(i;r)$ contains all vertices within $r$ graph steps of $i$, whereas $\partial\mathcal N_n(i;r)$ contains only those vertices exactly $r$ steps away.

For a random vector $\mathbf X$, write $\|\mathbf X\|_{L^2}=\{\mathbb E\|\mathbf X\|_2^2\}^{1/2}$. Define the radius-$r$ local projection of the centered score and its worst-case approximation error by
\begin{equation}
\boldsymbol\zeta_{n,i}^{[r]}
=
\mathbb E\bigl(\boldsymbol\zeta_{n,i}^{\circ}\mid\mathcal F_{n,i}(r)\bigr),
\qquad
\rho_n(r)
=
\max_{i\in[n]}
\left\|\boldsymbol\zeta_{n,i}^{\circ}-\boldsymbol\zeta_{n,i}^{[r]}\right\|_{L^2}.
\label{eq:approx-local-projection}
\end{equation}
Because conditional expectation is the best $L^2$ approximation based on $\mathcal F_{n,i}(r)$, the vector $\boldsymbol\zeta_{n,i}^{[r]}$ is the part of the score that can be recovered from assignments inside the radius-$r$ ball. Hence $\rho_n(r)$ measures the largest remaining contribution of assignments outside that ball. The sequence $r\mapsto\rho_n(r)$ is non-increasing. If every centered score is measurable with respect to $\mathcal F_{n,i}(R)$ for a fixed radius $R$, then $\rho_n(r)=0$ for all $r\ge R$; approximate locality permits $\rho_n(r)$ to remain positive at every finite radius, provided that it decays sufficiently fast.

For $s\ge1$, let
\begin{equation}
h_s=\left\lfloor\frac{s-1}{2}\right\rfloor,
\qquad
\theta_n(0)=1,
\qquad
\theta_n(s)=\min\{1,4\rho_n(h_s)\}.
\label{eq:approx-local-dependence-envelope}
\end{equation}
The radius $h_s$ has a simple geometric meaning. If two vertices are at least $s$ steps apart, their radius-$h_s$ balls are disjoint. Under i.i.d. Bernoulli assignment, the corresponding local projections are therefore independent. The quantity $\theta_n(s)$ bounds the dependence that remains after replacing the original scores by these independent local projections. It is thus a distance-$s$ dependence envelope, and it is non-increasing in $s$.

To describe network growth, for $t\ge1$ define
\begin{equation}
\mathcal B_{n,t}(r)
=
\left\{\frac1n\sum_{i=1}^n|\mathcal N_n(i;r)|^t\right\}^{1/t},
\qquad
\mathcal S_{n,t}(s)
=
\left\{\frac1n\sum_{i=1}^n|\partial\mathcal N_n(i;s)|^t\right\}^{1/t}.
\label{eq:approx-local-ball-shell-growth}
\end{equation}
Here $\mathcal B_{n,t}(r)$ and $\mathcal S_{n,t}(s)$ are the empirical $L^t$ norms of ball and shell sizes across the $n$ vertices. Thus $\mathcal B_{n,t}(r)$ measures a typical radius-$r$ ball size, whereas $\mathcal S_{n,t}(s)$ measures a typical distance-$s$ shell size. Taking $t>1$ gives more weight to unusually dense parts of an irregular graph, without imposing a maximum-degree bound.

\begin{figure}[H]
\centering
\begin{tikzpicture}[
    x=0.66cm,
    y=0.66cm,
    lattice edge/.style={draw=gray!45,line width=0.35pt},
    outside vertex/.style={circle,fill=gray!45,inner sep=1.25pt},
    ball vertex/.style={circle,fill=blue!35,draw=blue!65,inner sep=1.8pt},
    shell vertex/.style={circle,fill=orange!80,draw=orange!80!black,inner sep=2.0pt},
    center vertex/.style={circle,fill=black,draw=black,inner sep=2.5pt},
    annotation/.style={align=left,font=\small}
]
\fill[blue!8] (0,2.35)--(2.35,0)--(0,-2.35)--(-2.35,0)--cycle;

\foreach \x in {-3,...,2}{
    \foreach \y in {-3,...,3}{
        \draw[lattice edge] (\x,\y)--({\x+1},\y);
    }
}
\foreach \x in {-3,...,3}{
    \foreach \y in {-3,...,2}{
        \draw[lattice edge] (\x,\y)--(\x,{\y+1});
    }
}

\foreach \x in {-3,...,3}{
    \foreach \y in {-3,...,3}{
        \pgfmathtruncatemacro{\dxy}{abs(\x)+abs(\y)}
        \ifnum\dxy=0
            \node[center vertex] at (\x,\y) {};
        \else
            \ifnum\dxy=2
                \node[shell vertex] at (\x,\y) {};
            \else
                \ifnum\dxy<2
                    \node[ball vertex] at (\x,\y) {};
                \else
                    \node[outside vertex] at (\x,\y) {};
                \fi
            \fi
        \fi
    }
}

\node[annotation,anchor=east] at (-3.65,0.15) {center $i$};
\draw[-{Stealth[length=2mm]},line width=0.5pt] (-3.55,0.05)--(-0.18,0.02);

\node[annotation,anchor=west] at (3.65,1.45)
{$\mathcal N_n(i;2)$: radius-$2$ ball\\$1+4+8=13$ vertices};
\draw[-{Stealth[length=2mm]},line width=0.5pt] (3.55,1.25)--(0.75,0.75);

\node[annotation,anchor=west] at (3.65,-1.15)
{$\partial\mathcal N_n(i;2)$: distance-$2$ shell\\$8$ vertices};
\draw[-{Stealth[length=2mm]},line width=0.5pt] (3.55,-0.9)--(2.12,0);
\end{tikzpicture}
\caption{A local view of a two-dimensional square-lattice torus. The shaded diamond is $\mathcal N_n(i;2)$, and the highlighted outer ring is $\partial\mathcal N_n(i;2)$. Since these counts do not depend on the center, $\mathcal B_{n,t}(2)=13$ and $\mathcal S_{n,t}(2)=8$ for every $t\ge1$.}
\label{fig:approx-locality-lattice}
\end{figure}
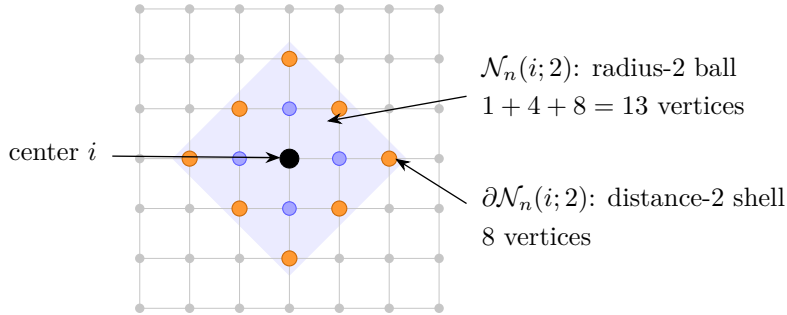

For example, consider a two-dimensional square-lattice torus with $n=L_n^2$ vertices. For $2r<L_n$ and $1\le s<L_n/2$, every vertex has
\begin{equation}
|\mathcal N_n(i;r)|=1+2r(r+1),
\qquad
|\partial\mathcal N_n(i;s)|=4s.
\label{eq:lattice-ball-shell-size}
\end{equation}
Consequently, for every $t\ge1$,
\[
\mathcal B_{n,t}(r)=1+2r(r+1)\asymp r^2,
\qquad
\mathcal S_{n,t}(s)=4s\asymp s.
\]
The same orders hold away from the outer boundary of an ordinary finite square lattice. More generally, on a fixed-dimensional $d$-dimensional lattice, $\mathcal B_{n,t}(r)\asymp r^d$ and $\mathcal S_{n,t}(s)\asymp (1+s)^{d-1}$ for ball radius $r\ge1$ and shell distance $s\ge0$ small relative to the side length.

\begin{assumption}[Approximate network locality and average network growth]
\label{assump:network-locality-growth}
There exist a sequence of connected graphs $\mathcal L_n=([n],E_n^{\mathrm{loc}})$, constants $q>4$ and $\nu>1$, and positive integers $m_n\to\infty$ such that, with $\nu'=\nu/(\nu-1)$,
\begin{enumerate}
\item[(i)] $n^{3/2}\theta_n(m_n)^{1-1/q}\to0$;
\item[(ii)]
\[
\frac{\mathcal B_{n,\nu}(m_n)}{\sqrt n}
\sum_{s\ge0}\mathcal S_{n,\nu'}(s)\theta_n(s)^{1-3/q}
\to0;
\]
\item[(iii)]
\[
\frac{\mathcal B_{n,2\nu}(m_n)^2}{n}
\sum_{s\ge0}\mathcal S_{n,\nu'}(s)\theta_n(s)^{1-4/q}
\to0.
\]
\end{enumerate}
\end{assumption}

Assumption~\ref{assump:network-locality-growth} asks for an intermediate radius $m_n$ at which long-range approximation error is already negligible but local graph neighborhoods are still small. Condition~(i) controls the part of each score that remains outside its radius-$m_n$ ball. In Conditions~(ii) and~(iii), the sums are dependence-weighted shell sizes: the number of vertices at distance $s$ is discounted by the dependence envelope $\theta_n(s)$. The factors $\mathcal B_{n,\nu}(m_n)$ and $\mathcal B_{n,2\nu}(m_n)^2$ control first- and second-order local configurations, and the normalizations $\sqrt n$ and $n$ match the corresponding central-limit scales. The constants $q$ and $\nu$ are technical moment and H\"older exponents that tune these sufficient rates. The substantive requirement is the existence of a radius sequence $m_n$ satisfying this approximation--growth tradeoff.

The square lattice gives a useful benchmark. If, uniformly in $n$, $\theta_n(s)\le C_0\exp(-c_0s)$, then both dependence-weighted shell sums in Assumption~\ref{assump:network-locality-growth} are bounded. Taking $m_n=\lceil C\log n\rceil$ with $C$ sufficiently large makes Condition~(i) hold, while $\mathcal B_{n,t}(m_n)\asymp(\log n)^2$ makes Conditions~(ii) and~(iii) hold as well.

\begin{theorem}[Joint CLT under approximate locality]
\label{thm:joint-normality-approx-locality}
Suppose that $Z_i\overset{\mathrm{i.i.d.}}{\sim}\operatorname{Bernoulli}(\pi)$. Under Assumptions~\ref{ass:variance_stability}--\ref{ass:variance_lower_bound} and Assumption~\ref{assump:network-locality-growth},
\[
\mathbf D_n
\begin{pmatrix}
\hat\tau-\tau\\
\hat{\boldsymbol\tau}_x-\boldsymbol\tau_x
\end{pmatrix}
\xrightarrow{d}N(\boldsymbol 0,\mathbf V).
\]
Consequently, Theorem~\ref{thm:rerand_dist} and Corollaries~\ref{cor:unbiased} and \ref{thm:variance_reduction} remain valid when Assumption~\ref{ass:network_sparsity} is replaced by Assumption~\ref{assump:network-locality-growth}.
\end{theorem}

For the remainder of this subsection, let $q$ and $\nu$ be as in Assumption~\ref{assump:network-locality-growth}, write $\nu'=\nu/(\nu-1)$, and use the shorthand
\begin{equation}
\Gamma_n^{(k)}
=
\sum_{s\ge0}\mathcal S_{n,\nu'}(s)\theta_n(s)^{1-(k+2)/q},
\qquad k=1,2.
\label{eq:approx-local-weighted-shell-growth}
\end{equation}
Thus $\Gamma_n^{(k)}$ is the total shell growth after discounting vertices at distance $s$ by the corresponding dependence strength.

\subsubsection{Arm-score law of large numbers}

\begin{lemma}[Arm-score law of large numbers under approximate locality]
\label{lem:t1-arm-score-lln}
Under Assumptions~\ref{ass:boundedness} and~\ref{ass:variance_lower_bound}, together with Assumption~\ref{assump:network-locality-growth}, for $z\in\{0,1\}$,
\begin{equation}
\mathbb E\|\mathbf Q_{n,z}\|_2^2
\le
C\left\{\frac{\Gamma_n^{(2)}}{n}+\theta_n(m_n)\right\}
=o(1),
\label{eq:t1-Q-second}
\end{equation}
and
\begin{equation}
\mathbb E\|\mathbf Q_{n,z}\|_2^4=o(1).
\label{eq:t1-Q-fourth}
\end{equation}
\end{lemma}

\begin{proof}
We now derive the coordinate-wise covariance bound used below. Fix $i,j$ and coordinates $\ell_1,\ell_2\in[2d_\phi]$. If $s=\operatorname{dist}_n(i,j)\ge1$, put $r=h_s$. Let
\[
X=\zeta_{n,i,\ell_1}^{\circ}, \quad X_r=\zeta_{n,i,\ell_1}^{[r]}, \quad
Y=\zeta_{n,j,\ell_2}^{\circ}, \quad Y_r=\zeta_{n,j,\ell_2}^{[r]}.
\]
Because $s>2r$, the neighborhoods $\mathcal N_n(i;r)$ and $\mathcal N_n(j;r)$ are disjoint. Thus $X_r$ and $Y_r$ are functions of two disjoint sets of independent Bernoulli coordinates and are independent. Therefore,
\[
\operatorname{Cov}(X,Y)
=
\operatorname{Cov}(X-X_r,Y)
+
\operatorname{Cov}(X_r,Y-Y_r),
\]
because $\operatorname{Cov}(X_r,Y_r)=0$. Using the elementary covariance bound~\eqref{eq:t1-elementary-cov-bound}, \eqref{eq:t1-zeta-uniform}, Jensen's inequality for the conditional expectations, and $L^1\le L^2$, we have $\mathbb E|X-X_r|\le\rho_n(r)$ and $\mathbb E|Y-Y_r|\le\rho_n(r)$, and hence
\[
|\operatorname{Cov}(X,Y)| \le C\rho_n(r).
\]
The uniform bound in~\eqref{eq:t1-zeta-uniform} also gives $|\operatorname{Cov}(X,Y)|\le C$. Combining the two bounds with $\theta_n(s)=\min\{1,4\rho_n(h_s)\}$ gives, after changing the constant $C$ if necessary,
\begin{equation}
\left| \operatorname{Cov} \bigl( \zeta_{n,i,\ell_1}^{\circ}, \zeta_{n,j,\ell_2}^{\circ} \bigr)
\right| \le C\theta_n\!\left(\operatorname{dist}_n(i,j)\right). \label{eq:t1-coordinate-cov}
\end{equation}
The same conclusion holds when $i=j$, because $\theta_n(0)=1$ and \eqref{eq:t1-zeta-uniform} applies.

Fix one coordinate $\ell$ of the $z$th block. By \eqref{eq:t1-Qnz} and \eqref{eq:t1-coordinate-cov},
\begin{align}
\mathbb E\{Q_{n,z,\ell}^2\} &= \frac1{n^2} \sum_{i=1}^n\sum_{j=1}^n \operatorname{Cov} \bigl(
\zeta_{n,i,\ell}^{\circ,(z)}, \zeta_{n,j,\ell}^{\circ,(z)} \bigr) \notag\\
&\le \frac{C}{n^2} \sum_{i=1}^n \sum_{s=0}^{m_n} |\partial\mathcal N_n(i;s)|\,\theta_n(s) +
\frac{C}{n^2} \sum_{\substack{i,j:\
\operatorname{dist}_n(i,j)>m_n}}\theta_n\!\left(\operatorname{dist}_n(i,j)\right). \label{eq:t1-Q-cov-split}
\end{align}
The sequence $s\mapsto\theta_n(s)$ is non-increasing hence the second term in \eqref{eq:t1-Q-cov-split} is at most $C\theta_n(m_n)$. The first term equals
\[
\frac{C}{n} \sum_{s=0}^{m_n} \mathcal S_{n,1}(s)\theta_n(s).
\]
Since $\nu'=\nu/(\nu-1)>1$, monotonicity of the empirical $L_p$ norm gives
\[
\mathcal S_{n,1}(s) \le \mathcal S_{n,\nu'}(s).
\]
Moreover, $q>4$ and $0\le\theta_n(s)\le1$ imply
\[
\theta_n(s) \le \theta_n(s)^{1-4/q}.
\]
It follows that
\begin{equation}
\mathbb E\{Q_{n,z,\ell}^2\} \le C\left\{ \frac1n \sum_{s\ge0} \mathcal S_{n,\nu'}(s)
\theta_n(s)^{1-4/q} + \theta_n(m_n) \right\} = C\left\{ \frac{\Gamma_n^{(2)}}n +
\theta_n(m_n) \right\}. \label{eq:t1-Q-coordinate-bound}
\end{equation}
Summing over the fixed number $d_\phi$ of coordinates proves the inequality in \eqref{eq:t1-Q-second}.

It remains to verify that the right-hand side vanishes. Every closed neighborhood contains its center, so $\mathcal B_{n,2\nu}(m_n)\ge1$. Assumption~\ref{assump:network-locality-growth}(iii) therefore gives
\[
0\le \frac{\Gamma_n^{(2)}}n \le \frac{\mathcal B_{n,2\nu}(m_n)^2}{n} \Gamma_n^{(2)}
\to0.
\]
Assumption~\ref{assump:network-locality-growth}(i) gives $n^{3/2}\theta_n(m_n)^{1-1/q}\to0$, and hence $\theta_n(m_n)\to0$. This proves \eqref{eq:t1-Q-second}.

Finally, \eqref{eq:t1-zeta-uniform} implies $\|\mathbf Q_{n,z}\|_2\le C$ almost surely. Therefore
\[
\|\mathbf Q_{n,z}\|_2^4 \le C^2\|\mathbf Q_{n,z}\|_2^2,
\]
and \eqref{eq:t1-Q-fourth} follows from \eqref{eq:t1-Q-second}.
\end{proof}

\subsubsection{Bridging lemma for KMS \texorpdfstring{\citep{kojevnikov2021limit}}{(Kojevnikov et al., 2021)} Assumption~2.1}
For a positive integer $a$, write
\[
\boldsymbol u=(\boldsymbol u_1,\ldots,\boldsymbol u_a),
\qquad
\boldsymbol v=(\boldsymbol v_1,\ldots,\boldsymbol v_a)
\in(\mathbb R^{2d_\phi})^a,
\]
and equip $(\mathbb R^{2d_\phi})^a$ with the product-sum metric
\[
d_a(\boldsymbol u,\boldsymbol v)
=
\sum_{\ell=1}^a
\|\boldsymbol u_\ell-\boldsymbol v_\ell\|_2.
\]
For a bounded function $f:(\mathbb R^{2d_\phi})^a\to\mathbb R$, define
\[
\operatorname{Lip}_a(f)
=
\sup_{\boldsymbol u\neq\boldsymbol v}
\frac{|f(\boldsymbol u)-f(\boldsymbol v)|}
{d_a(\boldsymbol u,\boldsymbol v)},
\qquad
\|f\|_{\mathrm{BL},a}
=
\|f\|_\infty+\operatorname{Lip}_a(f).
\]
The normalized bounded-Lipschitz class is
\[
\mathrm{BL}_{1,a}
=
\left\{
f:(\mathbb R^{2d_\phi})^a\to\mathbb R:
\|f\|_{\mathrm{BL},a}\leq1
\right\}.
\]

For node sets $A,B\subset[n]$, write
\[
\operatorname{dist}_n(A,B)=\min_{i\in A,\,j\in B}\operatorname{dist}_n(i,j),
\]
and let
\[
\boldsymbol\zeta_{n,A}^{\circ} =(\boldsymbol\zeta_{n,i}^{\circ}:i\in A), \qquad \boldsymbol
\zeta_{n,A}^{[r]} =(\boldsymbol\zeta_{n,i}^{[r]}:i\in A).
\]

\begin{lemma}[Approximate locality implies blockwise $\psi$-dependence]
\label{lem:t1-locality-to-psi}
Let $A,B\subset[n]$, $|A|=a$, $|B|=b$, and suppose $\operatorname{dist}_n(A,B)\ge s\ge1$. Then, for every $f\in\mathrm{BL}_{1,a}$ and $g\in\mathrm{BL}_{1,b}$,
\begin{equation}
\left| \operatorname{Cov} \left\{ f(\boldsymbol\zeta_{n,A}^{\circ}),
g(\boldsymbol\zeta_{n,B}^{\circ}) \right\} \right| \le ab\,\theta_n(s).
\label{eq:t1-block-covariance}
\end{equation}
Consequently, $\{\boldsymbol\zeta_{n,i}^{\circ}:i\in[n]\}$ satisfies Assumption~2.1 of \citet{kojevnikov2021limit} with dependence coefficients $\{\theta_n(s)\}_{s\ge0}$.
\end{lemma}

\begin{proof}
Let $r=h_s=\lfloor(s-1)/2\rfloor$. We first explain carefully why the two local blocks are independent. Define
\[
\mathcal N_n(A;r)=\bigcup_{i\in A}\mathcal N_n(i;r), \qquad \mathcal N_n(B;r)=\bigcup_{j\in
B}\mathcal N_n(j;r).
\]
These sets are disjoint: otherwise the triangle inequality would give $\operatorname{dist}_n(A,B)\le2r<s$. The vector $\boldsymbol\zeta_{n,A}^{[r]}$ is measurable with respect to $\sigma\{Z_k:k\in \mathcal N_n(A;r)\}$, and the analogous statement holds for $B$. Because the Bernoulli coordinates are independent and the two index sets are disjoint,
\begin{equation}
\boldsymbol\zeta_{n,A}^{[r]} \ \perp\! \!\!\perp\ \boldsymbol\zeta_{n,B}^{[r]}.
\label{eq:t1-local-block-independence}
\end{equation}

Write
\[
F=f(\boldsymbol\zeta_{n,A}^{\circ}), \quad F_r=f(\boldsymbol\zeta_{n,A}^{[r]}), \quad
G=g(\boldsymbol\zeta_{n,B}^{\circ}), \quad G_r=g(\boldsymbol\zeta_{n,B}^{[r]}).
\]
By \eqref{eq:t1-local-block-independence}, $\operatorname{Cov}(F_r,G_r)=0$, and hence
\begin{equation}
\operatorname{Cov}(F,G) = \operatorname{Cov}(F-F_r,G) + \operatorname{Cov}(F_r,G-G_r).
\label{eq:t1-covariance-decomposition}
\end{equation}
Using \eqref{eq:t1-elementary-cov-bound},
\begin{align*}
|\operatorname{Cov}(F,G)| &\le 2\|g\|_\infty\,\mathbb E|F-F_r| +2\|f\|_\infty\,\mathbb E|G-G_r|.
\end{align*}
The product-sum metric and the Lipschitz property give
\begin{align*}
\mathbb E|F-F_r| &\le \operatorname{Lip}_a(f) \sum_{i\in A} \mathbb
E\|\boldsymbol\zeta_{n,i}^{\circ} -\boldsymbol\zeta_{n,i}^{[r]}\|_2 \le
a\operatorname{Lip}_a(f)\rho_n(r), \\
\mathbb E|G-G_r| &\le b\operatorname{Lip}_b(g)\rho_n(r).
\end{align*}
Therefore,
\begin{equation}
|\operatorname{Cov}(F,G)| \le 2\rho_n(r) \left\{ a\|g\|_\infty\operatorname{Lip}_a(f)
+b\|f\|_\infty\operatorname{Lip}_b(g) \right\}. \label{eq:t1-locality-cov-bound}
\end{equation}
Because $f$ and $g$ are in the normalized bounded-Lipschitz classes, all four quantities $\|f\|_\infty$, $\operatorname{Lip}_a(f)$, $\|g\|_\infty$, and $\operatorname{Lip}_b(g)$ are at most one. Since $a,b\ge1$,
\[
a\|g\|_\infty\operatorname{Lip}_a(f) +b\|f\|_\infty\operatorname{Lip}_b(g) \le a+b\le2ab.
\]
Thus \eqref{eq:t1-locality-cov-bound} is at most $4ab\rho_n(r)$. In addition,
\[
|\operatorname{Cov}(F,G)|
\le
\sqrt{\operatorname{Var}(F)\operatorname{Var}(G)}
\le1\le ab.
\]
Combining the two bounds gives
\[
|\operatorname{Cov}(F,G)| \le ab\min\{1,4\rho_n(r)\} = ab\theta_n(s),
\]
which proves \eqref{eq:t1-block-covariance}.

For arbitrary bounded Lipschitz $f$ and $g$, apply the preceding bound to
\[
f^\circ= \frac{f}{\|f\|_\infty+\operatorname{Lip}_a(f)}, \qquad g^\circ=
\frac{g}{\|g\|_\infty+\operatorname{Lip}_b(g)},
\]
with the zero-norm cases handled trivially. By bilinearity of covariance,
\[
\left| \operatorname{Cov} \{f(\boldsymbol\zeta_{n,A}^{\circ}), g(\boldsymbol\zeta_{n,B}^{\circ})\}
\right| \le ab \|f\|_{\mathrm{BL},a} \|g\|_{\mathrm{BL},b} \theta_n(s).
\]
This is Definition~2.2 and Assumption~2.1(a) of \citet{kojevnikov2021limit} with
\[
\psi_{a,b}(f,g)=ab\|f\|_{\mathrm{BL},a}\|g\|_{\mathrm{BL},b}.
\]
Assumption~2.1(b) follows immediately from $0\le\theta_n(s)\le1$.
\end{proof}

\subsubsection{Bridging lemma II for KMS \texorpdfstring{\citep{kojevnikov2021limit}}{(Kojevnikov et al., 2021)} Assumption~3.4}

For completeness, introduce the network quantities used by \citet{kojevnikov2021limit}. For $t\ge1$, let
\[
\delta_n^{\partial}(s;t) = \frac1n\sum_{i=1}^n|\partial\mathcal N_n(i;s)|^t,
\]
\[
\Delta_n(s,m;t) = \frac1n\sum_{i=1}^n \max_{j\in\partial\mathcal N_n(i;s)} |\mathcal
N_n(i;m)\setminus \mathcal N_n(j;s-1)|^t,
\]
where $\mathcal N_n(j;-1)=\varnothing$ and a maximum over an empty shell is interpreted as zero. Their network-denseness quantity is
\begin{equation}
c_n(s,m;k) = \inf_{\alpha>1} \left\{\Delta_n(s,m;k\alpha)\right\}^{1/\alpha} \left\{
\delta_n^{\partial} \left(s;\frac{\alpha}{\alpha-1}\right) \right\}^{1-1/\alpha}.
\label{eq:t1-kms-cn}
\end{equation}

\begin{lemma}[Average ball/shell growth bounds the KMS network quantity]
\label{lem:t1-kms-growth}
For every fixed $\nu>1$ and $k\in\{1,2\}$,
\begin{equation}
c_n(s,m;k) \le \mathcal B_{n,k\nu}(m)^k \mathcal S_{n,\nu'}(s), \qquad
\nu'=\frac\nu{\nu-1}. \label{eq:t1-cn-bound}
\end{equation}
Consequently, Assumption~\ref{assump:network-locality-growth} implies, for $k=1,2$,
\begin{equation}
\frac1{n^{k/2}} \sum_{s\ge0} c_n(s,m_n;k) \theta_n(s)^{1-(k+2)/q} \to 0.
\label{eq:t1-kms-local-rates}
\end{equation}

\end{lemma}

\begin{proof}
In the infimum in \eqref{eq:t1-kms-cn}, choose $\alpha=\nu$. For every $i$ and $j$,
\[
|\mathcal N_n(i;m)\setminus \mathcal N_n(j;s-1)| \le|\mathcal N_n(i;m)|.
\]
Therefore,
\begin{align*}
\left\{\Delta_n(s,m;k\nu)\right\}^{1/\nu} &\le \left\{ \frac1n\sum_{i=1}^n|\mathcal N_n(i;m)|^{k\nu}
\right\}^{1/\nu} \\
&= \left[ \left\{ \frac1n\sum_{i=1}^n|\mathcal N_n(i;m)|^{k\nu} \right\}^{1/(k\nu)} \right]^k \\
&= \mathcal B_{n,k\nu}(m)^k.
\end{align*}
For the shell factor, since $1-1/\nu=1/\nu'$,
\[
\left\{ \delta_n^{\partial}(s;\nu') \right\}^{1-1/\nu} = \left\{
\frac1n\sum_{i=1}^n|\partial\mathcal N_n(i;s)|^{\nu'} \right\}^{1/\nu'} = \mathcal
S_{n,\nu'}(s).
\]
Multiplying these two bounds proves \eqref{eq:t1-cn-bound}. Hence, for $k=1,2$,
\begin{align*}
&\frac1{n^{k/2}} \sum_{s\ge0} c_n(s,m_n;k) \theta_n(s)^{1-(k+2)/q} \\
&\quad\le \frac{\mathcal B_{n,k\nu}(m_n)^k}{n^{k/2}} \sum_{s\ge0} \mathcal S_{n,\nu'}(s)
\theta_n(s)^{1-(k+2)/q} \\
&\quad= \frac{\mathcal B_{n,k\nu}(m_n)^k}{n^{k/2}} \Gamma_n^{(k)} \to0
\end{align*}
by Assumption~\ref{assump:network-locality-growth}(ii)--(iii).
\end{proof}

\subsubsection{Proof of Theorem~\ref{thm:joint-normality-approx-locality}}

Lemma~\ref{lem:t1-arm-score-lln} verifies the condition of Lemma~\ref{lem:t1-hajek-linearization}. Therefore,
\begin{equation}
\mathbf T_n-\widetilde{\mathbf T}_n=o_p(1), \qquad \operatorname{Var}(\widetilde{\mathbf
T}_n)\to \mathbf V.
\label{eq:t1-start-clt}
\end{equation}
So by Slutsky's Theorem, it remains to establish a joint central limit theorem for $\widetilde{\mathbf T}_n$.

Fix an arbitrary nonzero vector $\boldsymbol\lambda\in\mathbb R^{d_\phi}$. Define the scalar triangular array
\begin{equation}
Y_{n,i}^{(\boldsymbol\lambda)} =
\frac1{\sqrt n} \boldsymbol\lambda^\top \mathbf D_n \left\{ \boldsymbol\eta_{n,i} -\mathbb
E(\boldsymbol\eta_{n,i}) \right\}. \label{eq:t1-scalar-summand}
\end{equation}
Let
\begin{equation}
S_{n,\boldsymbol\lambda} = \sum_{i=1}^nY_{n,i}^{(\boldsymbol\lambda)}, \qquad
\varsigma_{n,\boldsymbol\lambda}^2 = \operatorname{Var}(S_{n,\boldsymbol\lambda}).
\label{eq:t1-S-varsigma}
\end{equation}
By construction, $\mathbb E Y_{n,i}^{(\boldsymbol\lambda)}=0$ for every $n,i$. Moreover, \eqref{eq:t1-total-eta-mean-zero} gives
\begin{equation}
S_{n,\boldsymbol\lambda} = \frac1{\sqrt n} \boldsymbol\lambda^\top \mathbf D_n
\sum_{i=1}^n\boldsymbol\eta_{n,i} = \sqrt n\,\boldsymbol\lambda^\top\widetilde{\mathbf T}_n.
\label{eq:t1-S-Ttilde}
\end{equation}

We verify Assumptions~2.1, 3.3, and 3.4 of \cite{kojevnikov2021limit} to apply their central limit theorem.

\paragraph{KMS Assumption 2.1.}
By Lemma~\ref{lem:t1-locality-to-psi}, and because a fixed linear transformation preserves $\psi$-dependence up to a multiplicative constant depending only on $\boldsymbol\lambda$, the scalar array $\{Y_{n,i}^{(\boldsymbol\lambda)}\}$ is conditionally $\psi$-dependent (with the conditioning sigma-field taken to be trivial in the present design-based setting). The multiplicative constant is absorbed into the $\psi$-functional, while the dependence coefficients remain $\{\theta_n(s)\}_{s\ge0}$.

\paragraph{KMS Assumption 3.3.}
By \eqref{eq:t1-zeta-uniform},
\[
\sup_{n\ge1}\max_{i\in[n]} |Y_{n,i}^{(\boldsymbol\lambda)}| \le C_{\boldsymbol\lambda}<\infty.
\]
Therefore, for the $q>4$ appearing in Assumption~\ref{assump:network-locality-growth},
\begin{equation}
\sup_{n\ge1}\max_{i\in[n]} \mathbb E |Y_{n,i}^{(\boldsymbol\lambda)}|^q <\infty,
\label{eq:t1-kms-moment}
\end{equation}
which is KMS Assumption~3.3.

\paragraph{Limit variance.}
By \eqref{eq:t1-S-Ttilde} and \eqref{eq:t1-start-clt},
\begin{equation}
\frac{\varsigma_{n,\boldsymbol\lambda}^2}{n} =
\operatorname{Var}(\boldsymbol\lambda^\top\widetilde{\mathbf T}_n) \to
\boldsymbol\lambda^\top \mathbf V\boldsymbol\lambda =:v_{\boldsymbol\lambda}.
\label{eq:t1-long-run-variance}
\end{equation}
Because $\mathbf V$ is positive definite and $\boldsymbol\lambda\ne0$, $v_{\boldsymbol\lambda}>0$. Hence
\begin{equation}
\frac{\varsigma_{n,\boldsymbol\lambda}}{\sqrt n} \to \sqrt{v_{\boldsymbol\lambda}}>0.
\label{eq:t1-varsigma-rate}
\end{equation}

\paragraph{KMS Assumption 3.4.}
For $k=1,2$, write
\begin{align}
&\frac{n}{\varsigma_{n,\boldsymbol\lambda}^{2+k}} \sum_{s\ge0} c_n(s,m_n;k)
\theta_n(s)^{1-(k+2)/q} \notag\\
&\quad= \frac{n^{1+k/2}} {\varsigma_{n,\boldsymbol\lambda}^{2+k}} \left[ \frac1{n^{k/2}}
\sum_{s\ge0} c_n(s,m_n;k) \theta_n(s)^{1-(k+2)/q} \right]. \label{eq:t1-kms-34-local}
\end{align}
By \eqref{eq:t1-varsigma-rate}, the first factor in \eqref{eq:t1-kms-34-local} converges to $v_{\boldsymbol\lambda}^{-(2+k)/2}$, while the bracketed factor tends to zero by Lemma~\ref{lem:t1-kms-growth}. Thus the first part of KMS Assumption~3.4 holds. For the remote term,
\begin{align}
\frac{n^2 \theta_n(m_n)^{1-1/q}} {\varsigma_{n,\boldsymbol\lambda}} &= \frac{\sqrt
n}{\varsigma_{n,\boldsymbol\lambda}} \left[n^{3/2}\theta_n(m_n)^{1-1/q}\right] \to0.
\label{eq:t1-kms-34-remote}
\end{align}
This verifies KMS Assumption~3.4.

All conditions of Theorem 3.2 of \cite{kojevnikov2021limit} are now satisfied. Therefore,
\begin{equation}
\frac{S_{n,\boldsymbol\lambda}} {\varsigma_{n,\boldsymbol\lambda}} \xrightarrow{d}N(0,1).
\label{eq:t1-kms-clt}
\end{equation}
Combining \eqref{eq:t1-S-Ttilde}, \eqref{eq:t1-varsigma-rate}, and \eqref{eq:t1-kms-clt} gives
\begin{align*}
\boldsymbol\lambda^\top\widetilde{\mathbf T}_n &= \frac{S_{n,\boldsymbol\lambda}}{\sqrt n} =
\frac{\varsigma_{n,\boldsymbol\lambda}}{\sqrt n} \frac{S_{n,\boldsymbol\lambda}}
{\varsigma_{n,\boldsymbol\lambda}} \\
&\xrightarrow{d} N\left(0,\boldsymbol\lambda^\top \mathbf V\boldsymbol\lambda\right).
\end{align*}
Finally, Lemma~\ref{lem:t1-hajek-linearization} and Slutsky's theorem yield
\[
\boldsymbol\lambda^\top \mathbf T_n \xrightarrow{d} N\left(0,\boldsymbol\lambda^\top \mathbf
V\boldsymbol\lambda\right).
\]
Since this holds for every fixed $\boldsymbol\lambda\in\mathbb R^{d_\phi}$, the Cram\'er--Wold theorem gives
\[
\mathbf T_n = \mathbf D_n(\hat{\boldsymbol\tau}_n-\boldsymbol\tau_n) \xrightarrow{d} N(0,\mathbf
V).
\]

\subsection{Derivations for the \texorpdfstring{$R^2$}{R-squared} examples}
\label{app:r2-examples}

All calculations in this subsection are under the Bernoulli assignment mechanism used in Examples~\ref{ex:r2-sutva} and~\ref{ex:r2-pair}. Let
\[
\psi_i = \frac{Z_i}{\pi}-\frac{1-Z_i}{1-\pi} = \frac{Z_i-\pi}{\pi(1-\pi)}.
\]
For a scalar quantity $W_i$, define its H\'ajek treated-minus-control contrast by
\[
\hat\tau_W = \frac{\sum_{i=1}^n Z_i W_i}{\sum_{i=1}^n Z_i} - \frac{\sum_{i=1}^n
(1-Z_i)W_i}{\sum_{i=1}^n (1-Z_i)}.
\]
We use the following linearization only under the two specific structures considered below: (i) $\{W_i\}_{i=1}^n$ is an i.i.d.\ collection with a finite second moment and is independent of $\mathbf Z$; or (ii) $W_i=H_i$ under the disjoint-pair interference structure in Example~\ref{ex:r2-pair}. In both cases, let
\[
\mu_W
=
\E(W_i\mid Z_i=1)
=
\E(W_i\mid Z_i=0),
\]
and define
\[
A_{1n}
=
\frac{1}{n}\sum_{i=1}^n Z_i(W_i-\mu_W),
\qquad
A_{0n}
=
\frac{1}{n}\sum_{i=1}^n (1-Z_i)(W_i-\mu_W).
\]
Under structure~(i), independence across units and the finite-second-moment condition imply
\[
A_{1n}=\mathcal O_p(n^{-1/2}),
\qquad
A_{0n}=\mathcal O_p(n^{-1/2}).
\]
The same orders hold under structure~(ii), because the corresponding sums can be grouped into independent and bounded pair-level contributions.

Let
\[
\hat\pi_n=\frac{1}{n}\sum_{i=1}^n Z_i.
\]
On the event that both treatment arms are nonempty,
\[
\hat\tau_W
=
\frac{A_{1n}}{\hat\pi_n}
-
\frac{A_{0n}}{1-\hat\pi_n}.
\]
Because $\hat\pi_n-\pi=\mathcal O_p(n^{-1/2})$, expanding the random denominators yields
\begin{align}
\hat\tau_W
&=
\frac{A_{1n}}{\pi}
-
\frac{A_{0n}}{1-\pi}
+
\mathcal O_p(n^{-1})
\nonumber\\
&=
\frac{1}{n}\sum_{i=1}^n
\psi_i(W_i-\mu_W)
+
o_p(n^{-1/2}).
\label{eq:model-specific-imbalance-expansion}
\end{align}
The convention for zero denominators used throughout the paper does not affect this expansion, since those events have exponentially small probability under Bernoulli assignment.

Under structure~(i), if $\Var(W_i)=\sigma_W^2$, this expansion implies
\[
\Var_{\operatorname{asymp}}(\hat\tau_W)
=
\frac{\sigma_W^2}{n\pi(1-\pi)}.
\]

\noindent\textit{Example~\ref{ex:r2-sutva}.} In the SUTVA model, the observed outcome is $Y_i=\alpha+\tau Z_i+\beta X_i+\varepsilon_i$. Thus, up to the negligible zero-denominator events, $\hat\tau-\tau = \beta\hat\tau_X+\hat\tau_\varepsilon$, where $\hat\tau_X$ and $\hat\tau_\varepsilon$ denote the H\'ajek contrasts formed from $X_i$ and $\varepsilon_i$, respectively. Since the collections $\{X_i\}_{i=1}^n$, $\{\varepsilon_i\}_{i=1}^n$, and $\mathbf Z$ are mutually independent, and since both $X_i$ and $\varepsilon_i$ are i.i.d.\ with mean zero and finite second moments, structure~(i) of the preceding linearization applies separately to $W_i=X_i$ and $W_i=\varepsilon_i$. Therefore,
\[
\Var_{\operatorname{asymp}}(\hat\tau_X) = \frac{\sigma_X^2}{n\pi(1-\pi)}, \qquad
\Var_{\operatorname{asymp}}(\hat\tau_\varepsilon) = \frac{\sigma_\varepsilon^2}{n\pi(1-\pi)},
\]
and
\[
\Cov_{\operatorname{asymp}}(\hat\tau_X,\hat\tau_\varepsilon)=0.
\]
Therefore,
\[
R^2 = \frac{ \Cov_{\operatorname{asymp}}(\hat\tau,\hat\tau_X)^2 }{
\Var_{\operatorname{asymp}}(\hat\tau_X) \Var_{\operatorname{asymp}}(\hat\tau) } = \frac{
\beta^2\Var_{\operatorname{asymp}}(\hat\tau_X) }{ \beta^2\Var_{\operatorname{asymp}}(\hat\tau_X) +
\Var_{\operatorname{asymp}}(\hat\tau_\varepsilon) } = \frac{\beta^2\sigma_X^2}
{\beta^2\sigma_X^2+\sigma_\varepsilon^2}.
\]
This is the same as the marginal unit-level predictive quantity in this SUTVA model, since
\[
\frac{\Var\{\E(Y_i(\mathbf{z})\mid X_i)\}}{\Var\{Y_i(\mathbf{z})\}} = \frac{\beta^2\sigma_X^2}
{\beta^2\sigma_X^2+\sigma_\varepsilon^2}.
\]

\noindent\textit{Example~\ref{ex:r2-pair}.} Write the two units in pair $k$ as $k1$ and $k2$, so that $H_{k1}=Z_{k2}$, $H_{k2}=Z_{k1}$. Although $H_i$ is treatment-dependent, it is independent of unit $i$'s own treatment under independent Bernoulli assignment. In particular,
\[
\E(H_i\mid Z_i=1)
=
\E(H_i\mid Z_i=0)
=
\pi.
\]
Moreover, the centered arm-specific sums can be grouped into independent and bounded pair-level contributions. Consequently, $A_{1n}=\mathcal O_p(n^{-1/2})$ and $A_{0n}=\mathcal O_p(n^{-1/2})$, so structure~(ii) of the preceding linearization applies with $W_i=H_i$ and $\mu_W=\pi$:
\begin{align}
\hat\tau_H &= \frac{1}{n}\sum_{k=1}^{n/2}\sum_{i=1}^{2} \left\{ \frac{Z_{ki}}{\pi} -
\frac{1-Z_{ki}}{1-\pi} \right\} (H_{ki}-\pi) +o_p(n^{-1/2}) \nonumber\\
&= \frac{1}{n}\sum_{k=1}^{n/2} \frac{2(Z_{k1}-\pi)(Z_{k2}-\pi)}{\pi(1-\pi)} +o_p(n^{-1/2}).
\label{eq:H-expansion}
\end{align}
The second equality uses the fact that the two centered contributions within the same pair are identical:
\[
\left\{ \frac{Z_{k1}}{\pi} - \frac{1-Z_{k1}}{1-\pi} \right\}(Z_{k2}-\pi) = \left\{
\frac{Z_{k2}}{\pi} - \frac{1-Z_{k2}}{1-\pi} \right\}(Z_{k1}-\pi).
\]
The pair-level summands in \eqref{eq:H-expansion} are independent across pairs and have mean zero. Moreover, $\E\{(Z_{k1}-\pi)^2(Z_{k2}-\pi)^2\} = \pi^2(1-\pi)^2$, so each pair-level summand in \eqref{eq:H-expansion} has variance $4$. Therefore, $\Var_{\operatorname{asymp}}(\hat\tau_H) = \frac{1}{n^2}\cdot \frac{n}{2}\cdot 4 = \frac{2}{n}$.

The baseline covariate and residual contrasts satisfy $\Var_{\operatorname{asymp}}(\hat\tau_X) = \frac{\sigma_X^2}{n\pi(1-\pi)}$, $\Var_{\operatorname{asymp}}(\hat\tau_\varepsilon) = \frac{\sigma_\varepsilon^2}{n\pi(1-\pi)}$. The three first-order imbalance terms $\hat\tau_X$, $\hat\tau_H$, and $\hat\tau_\varepsilon$ are asymptotically uncorrelated, because the collections $\{X_i\}_{i=1}^n$, $\{\varepsilon_i\}_{i=1}^n$, and $\mathbf Z$ are mutually independent, and $X_i$ and $\varepsilon_i$ are mean zero.

The observed outcome satisfies $Y_i = \alpha+\tau Z_i+\beta X_i+\lambda H_i+\varepsilon_i$, and therefore $\hat\tau-\tau = \beta\hat\tau_X+\lambda\hat\tau_H+\hat\tau_\varepsilon$. With $\tauxhat=(\hat\tau_X,\hat\tau_H)^\top$, the asymptotic covariance matrix of $\tauxhat$ is diagonal:
\[
\Var_{\operatorname{asymp}}(\tauxhat) =
\begin{pmatrix}
\sigma_X^2/\{n\pi(1-\pi)\} & 0\\
0 & 2/n \end{pmatrix}.
\]
Moreover,
\[
\Cov_{\operatorname{asymp}}(\hat\tau,\tauxhat) =
\begin{pmatrix}
\beta\sigma_X^2/\{n\pi(1-\pi)\}\\
2\lambda/n \end{pmatrix}.
\]
Thus,
\[
\Cov_{\operatorname{asymp}}(\hat\tau,\tauxhat)^\top \Var_{\operatorname{asymp}}(\tauxhat)^{-1}
\Cov_{\operatorname{asymp}}(\hat\tau,\tauxhat) = \frac{\beta^2\sigma_X^2}{n\pi(1-\pi)} +
\frac{2\lambda^2}{n}.
\]
Since
\[
\Var_{\operatorname{asymp}}(\hat\tau) = \frac{\beta^2\sigma_X^2}{n\pi(1-\pi)} + \frac{2\lambda^2}{n}
+ \frac{\sigma_\varepsilon^2}{n\pi(1-\pi)},
\]
the $R^2$ in Corollary~\ref{thm:variance_reduction} is
\begin{equation}\label{eq:appendix-r2-pair-design} R^2 = \frac{
\beta^2\sigma_X^2/\{n\pi(1-\pi)\}+2\lambda^2/n }{ \beta^2\sigma_X^2/\{n\pi(1-\pi)\}+2\lambda^2/n
+\sigma_\varepsilon^2/\{n\pi(1-\pi)\} } = \frac{\beta^2\sigma_X^2+2\pi(1-\pi)\lambda^2}
{\beta^2\sigma_X^2+2\pi(1-\pi)\lambda^2+\sigma_\varepsilon^2}.
\end{equation}

For comparison, the marginal unit-level predictive quantity focuses on the prognostic component of a single unit's outcome,
\[
Y_i(\mathbf Z)-\tau Z_i = \alpha+\beta X_i+\lambda H_i+\varepsilon_i.
\]
Because $X_i$ is independent of $H_i$, and $\Var(H_i)=\pi(1-\pi)$, this unit-level quantity is
\begin{equation}\label{eq:appendix-r2-pair-unit} R^2_{\operatorname{unit}} = \frac{
\Var\{\E(Y_i(\mathbf Z)-\tau Z_i\mid X_i,H_i)\} }{ \Var\{Y_i(\mathbf Z)-\tau Z_i\} } =
\frac{\beta^2\sigma_X^2+\pi(1-\pi)\lambda^2}
{\beta^2\sigma_X^2+\pi(1-\pi)\lambda^2+\sigma_\varepsilon^2}.
\end{equation}
Comparing \eqref{eq:appendix-r2-pair-design} and \eqref{eq:appendix-r2-pair-unit}, the two expressions differ only in the exposure term: the $R^2$ contains $2\pi(1-\pi)\lambda^2$, whereas the marginal unit-level predictive quantity contains $\pi(1-\pi)\lambda^2$. The factor of two comes from the covariance structure of the H\'ajek exposure imbalance. To first order, the two units in the same pair contribute the same centered term to $\hat\tau_H$, so the asymptotic variance of the exposure imbalance is $2/n$. By contrast, an independent baseline covariate with the same marginal variance $\pi(1-\pi)$ would have asymptotic imbalance variance
\[
\frac{\pi(1-\pi)}{n\pi(1-\pi)}=\frac1n.
\]
This within-pair alignment is the source of the factor $2\pi(1-\pi)\lambda^2$ in the design-based expression.

\section{Proofs for Section~\ref{sec:conservative}}
\label{app:conservative-proofs}

\subsection{Preliminaries}

Let
\[
\psi_i = \frac{Z_i}{\pi}-\frac{1-Z_i}{1-\pi}, \qquad \boldsymbol\phi_{n,i}^{\mathrm{pop}} =
\boldsymbol\phi_i -Z_i\boldsymbol\mu_{n,1} -(1-Z_i)\boldsymbol\mu_{n,0},
\]
and recall the within-arm centered vector
\[
\hat{\boldsymbol\phi}_i = \boldsymbol\phi_i -Z_i\bar{\boldsymbol{\phi}}_1
-(1-Z_i)\bar{\boldsymbol{\phi}}_0.
\]
Define the arm-mean errors
\[
\boldsymbol\delta_{n,1} = \bar{\boldsymbol{\phi}}_1-\boldsymbol\mu_{n,1}, \qquad
\boldsymbol\delta_{n,0} = \bar{\boldsymbol{\phi}}_0-\boldsymbol\mu_{n,0}.
\]

Recall that $\mathcal G_n=([n],E_n)$ is the dependency graph used to construct the estimator in Definition~\ref{def:Un_estimator}, with adjacency matrix $\mathbf A_n$ and degrees $d_{n,i}=\sum_{j=1}^n A_{n,ij}$. Write
\[
\Delta_n=\max_{i\in[n]}d_{n,i}.
\]
For the proof, introduce the population-centered analogue of the estimator in Definition~\ref{def:Un_estimator}:
\begin{equation}
\label{eq:t3-pop-estimator}
\mathbf U_n^{\mathrm{pop}}
=
\frac1{n^2}\sum_{i=1}^n
 d_{n,i}\psi_i^2
\boldsymbol\phi_{n,i}^{\mathrm{pop}}
{\boldsymbol\phi_{n,i}^{\mathrm{pop}}}^{\top}.
\end{equation}
The estimator itself can be written as
\[
\hat{\mathbf U}_n
=
\frac1{n^2}\sum_{i=1}^n
 d_{n,i}\psi_i^2
\hat{\boldsymbol\phi}_i\hat{\boldsymbol\phi}_i^{\top}.
\]

Finally, let
\[
\mathbf U_n = \operatorname{Var}(\hat{\boldsymbol\eta}), \qquad \mathbf V_n = \mathbf D_n\mathbf
U_n\mathbf D_n.
\]
By Assumption~\ref{ass:variance_stability}, $\mathbf V_n\to\mathbf V$ entrywise. Since $d_\phi=p+1$ is fixed,
\begin{equation}
\label{eq:t3-Vn-op}
\|\mathbf V_n-\mathbf V\|_{\mathrm{op}}\to0.
\end{equation}

\begin{lemma}[Concentration and dominance of the population-centered estimator]
\label{lm:conservative_infeasible}
Suppose that $Z_i\stackrel{\mathrm{i.i.d.}}{\sim}\operatorname{Bernoulli}(\pi)$. Under Assumptions~\ref{ass:variance_stability}--\ref{ass:variance_lower_bound}, suppose that $\mathcal G_n$ is a dependency graph and that, for some constant $\delta>0$, $\Delta_n=\mathcal O(n^{1/2-\delta})$. Then there exists a deterministic sequence $b_n\to0$ such that, with
\[
\boldsymbol\Lambda_n
=
\mathbf D_n\,\mathbb E(\mathbf U_n^{\mathrm{pop}})\,\mathbf D_n,
\]
we have
\begin{equation}
\label{eq:t3-Sn-dominates-Vn}
\boldsymbol\Lambda_n \succeq \mathbf V_n-b_n\mathbf I_{d_\phi}.
\end{equation}
Moreover, for every $\varepsilon>0$,
\begin{equation}
\label{eq:t3-pop-centered-conservative}
\mathbb P\left\{
\mathbf D_n\mathbf U_n^{\mathrm{pop}}\mathbf D_n
\succeq
\mathbf V_n-\varepsilon\mathbf I_{d_\phi}
\right\}
\to1.
\end{equation}
\end{lemma}

\begin{proof}
\medskip\noindent\textbf{Step 1: Dominance in expectation.} Lemma~\ref{lem:dg-arm-score-lln}, applied with $\mathcal G_n^{\mathrm{dep}}=\mathcal G_n$, gives, for $z\in\{0,1\}$,
\[
\mathbb E\|\mathbf Q_{n,z}\|_2^2
\le
C\frac{\Delta_n}{n}
=
o(1).
\]
Thus Lemma~\ref{lem:t1-hajek-linearization} applies. Since $\boldsymbol\eta_{n,i}=\psi_i\boldsymbol\phi_{n,i}^{\mathrm{pop}}$, it yields
\[
\mathbf D_n\operatorname{Var}\left(\frac1n\sum_{i=1}^n\boldsymbol\eta_{n,i}\right)\mathbf D_n
=
\operatorname{Var}(\widetilde{\mathbf T}_n)
\to
\mathbf V.
\]
Together with~\eqref{eq:t3-Vn-op}, this implies that the deterministic sequence
\begin{equation}
\label{eq:t3-linearized-var-close}
b_n
:=
\left\|
\mathbf V_n-
\mathbf D_n\operatorname{Var}\left(\frac1n\sum_{i=1}^n\boldsymbol\eta_{n,i}\right)\mathbf D_n
\right\|_{\mathrm{op}}
\to0.
\end{equation}

Put $\boldsymbol\eta_{n,i}^{\circ} =\boldsymbol\eta_{n,i}-\mathbb E(\boldsymbol\eta_{n,i})$. Because $\boldsymbol\eta_{n,i}$ is a measurable function of $\mathbf W_{n,i}=(Z_i,\boldsymbol\phi_i^{\top})^{\top}$, the dependency-graph property gives
\[
A_{n,ij}=0
\quad\to\quad
\boldsymbol\eta_{n,i}^{\circ}\indep\boldsymbol\eta_{n,j}^{\circ}.
\]
Thus all covariance terms corresponding to nonedges vanish. Grouping the remaining terms by unordered edges and using
\[
\mathbf a\mathbf b^{\top}+\mathbf b\mathbf a^{\top}
\preceq
\mathbf a\mathbf a^{\top}+\mathbf b\mathbf b^{\top},
\]
which follows from $(\mathbf a-\mathbf b)(\mathbf a-\mathbf b)^{\top}\succeq\mathbf0$, yields
\begin{align}
\operatorname{Var}\left(\frac1n\sum_{i=1}^n\boldsymbol\eta_{n,i}\right)
&\preceq
\frac1{n^2}\sum_{i=1}^n d_{n,i}
\mathbb E\left(
\boldsymbol\eta_{n,i}^{\circ}\boldsymbol\eta_{n,i}^{\circ\top}
\right) \notag\\
&\preceq
\frac1{n^2}\sum_{i=1}^n d_{n,i}
\mathbb E\left(
\boldsymbol\eta_{n,i}\boldsymbol\eta_{n,i}^{\top}
\right)
=
\mathbb E(\mathbf U_n^{\mathrm{pop}}).
\label{eq:t3-expectation-domination}
\end{align}
After premultiplication and postmultiplication by $\mathbf D_n$, \eqref{eq:t3-linearized-var-close} and~\eqref{eq:t3-expectation-domination} prove \eqref{eq:t3-Sn-dominates-Vn}.

Since $\mathbf V_n\to\mathbf V\succ\mathbf0$ and $b_n\to0$, \eqref{eq:t3-Sn-dominates-Vn} implies that, for all sufficiently large $n$,
\begin{equation}
\label{eq:t3-Sn-positive}
\lambda_{\min}(\boldsymbol\Lambda_n)
\ge
c_0
:=
\frac12\lambda_{\min}(\mathbf V)
>
0.
\end{equation}

\medskip\noindent\textbf{Step 2: Concentration around the expectation.} For all sufficiently large $n$, define the relatively standardized estimator
\begin{equation}
\label{eq:t3-Kn-def}
\mathbf K_n
=
\boldsymbol\Lambda_n^{-1/2}
\mathbf D_n\mathbf U_n^{\mathrm{pop}}\mathbf D_n
\boldsymbol\Lambda_n^{-1/2}.
\end{equation}
Then $\mathbf K_n\succeq\mathbf0$ and
\begin{equation}
\label{eq:t3-EKn-I}
\mathbb E(\mathbf K_n)=\mathbf I_{d_\phi}.
\end{equation}

Fix a unit vector $\mathbf u\in\mathbb R^{d_\phi}$ and define
\begin{equation}
\label{eq:t3-qni-def}
q_{n,i}(\mathbf u)
=
\psi_i^2
\left(
\frac{
\mathbf u^\top\boldsymbol\Lambda_n^{-1/2}\mathbf D_n
\boldsymbol\phi_{n,i}^{\mathrm{pop}}
}{\sqrt n}
\right)^2.
\end{equation}
Assumptions~\ref{ass:boundedness} and~\ref{ass:variance_lower_bound}, together with fixed $d_\phi$, imply
\[
\sup_{n,i}
\frac{\|\mathbf D_n\boldsymbol\phi_{n,i}^{\mathrm{pop}}\|_2}{\sqrt n}
<
\infty.
\]
By~\eqref{eq:t3-Sn-positive} and $|\psi_i|\le\{\min(\pi,1-\pi)\}^{-1}$, there is a constant $C<\infty$ such that $0\le q_{n,i}(\mathbf u)\le C$, uniformly in $n,i,\mathbf u$. Moreover, \eqref{eq:t3-Kn-def} and~\eqref{eq:t3-EKn-I} give
\begin{equation}
\label{eq:t3-directional-identities}
\mathbf u^\top\mathbf K_n\mathbf u
=
\frac1n\sum_{i=1}^n d_{n,i}q_{n,i}(\mathbf u),
\qquad
\sum_{i=1}^n d_{n,i}\mathbb E\{q_{n,i}(\mathbf u)\}=n.
\end{equation}

Because $q_{n,i}(\mathbf u)$ is a measurable function of $\mathbf W_{n,i}$, it is independent of $q_{n,j}(\mathbf u)$ whenever $A_{n,ij}=0$. Its uniform boundedness implies
\[
\operatorname{Var}\{q_{n,i}(\mathbf u)\}
\le
C\mathbb E\{q_{n,i}(\mathbf u)\},
\]
and, whenever $A_{n,ij}=1$,
\[
\left|
\operatorname{Cov}\{q_{n,i}(\mathbf u),q_{n,j}(\mathbf u)\}
\right|
\le
C\left[
\mathbb E\{q_{n,i}(\mathbf u)\}
+
\mathbb E\{q_{n,j}(\mathbf u)\}
\right].
\]
Using~\eqref{eq:t3-directional-identities}, $d_{n,i}\le\Delta_n$, and $\sum_{j:A_{n,ij}=1}d_{n,j}\le d_{n,i}\Delta_n$, we obtain
\[
\operatorname{Var}\{\mathbf u^\top\mathbf K_n\mathbf u\}
\le
\frac{C}{n^2}\left(n\Delta_n+2n\Delta_n^2\right)
\le
C\frac{\Delta_n^2}{n}
\to
0.
\]
Thus $\mathbf u^\top\mathbf K_n\mathbf u\to_p1$ for every fixed unit vector $\mathbf u$.

Because $d_\phi$ is fixed, no sphere-covering argument is needed. Let
\[
\mathcal U_{d_\phi}
=
\left\{\mathbf e_j:1\le j\le d_\phi\right\}
\cup
\left\{
\frac{\mathbf e_j+\mathbf e_k}{\sqrt2}:1\le j<k\le d_\phi
\right\}.
\]
This is a finite set. Applying the preceding convergence to every $\mathbf u\in\mathcal U_{d_\phi}$ and using a union bound imply
\[
\max_{\mathbf u\in\mathcal U_{d_\phi}}
\left|\mathbf u^\top\mathbf K_n\mathbf u-1\right|
\xrightarrow{p}
0.
\]
The directions $\mathbf e_j$ determine the diagonal entries of $\mathbf K_n-\mathbf I_{d_\phi}$, while $(\mathbf e_j+\mathbf e_k)/\sqrt2$ determines the $(j,k)$ entry through
\[
(\mathbf K_n)_{jk}
=
\frac{
(\mathbf e_j+\mathbf e_k)^\top
\mathbf K_n
(\mathbf e_j+\mathbf e_k)
}{2}
-
\frac{(\mathbf K_n)_{jj}+(\mathbf K_n)_{kk}}{2}.
\]
It follows, since $d_\phi$ is fixed, that
\begin{equation}
\label{eq:t3-Kn-op-convergence}
\|\mathbf K_n-\mathbf I_{d_\phi}\|_{\mathrm{op}}
\xrightarrow{p}
0.
\end{equation}
Thus, for every fixed $\alpha\in(0,1)$,
\begin{equation}
\label{eq:t3-Kn-lower}
\mathbb P\left\{
\mathbf K_n\succeq(1-\alpha)\mathbf I_{d_\phi}
\right\}
\to1.
\end{equation}

\medskip\noindent\textbf{Step 3: Conclusion.} To finish, fix $\varepsilon>0$. Since $\mathbf V_n\to\mathbf V$ in operator norm, there exists $C_V<\infty$ such that $\mathbf V_n\preceq C_V\mathbf I_{d_\phi}$ for all sufficiently large $n$. Choose $\alpha>0$ so that $\alpha C_V\le\varepsilon/2$, and then choose $n$ sufficiently large that $b_n\le\varepsilon/2$. On the event in~\eqref{eq:t3-Kn-lower},
\begin{align*}
\mathbf D_n\mathbf U_n^{\mathrm{pop}}\mathbf D_n
&=
\boldsymbol\Lambda_n^{1/2}\mathbf K_n\boldsymbol\Lambda_n^{1/2}\\
&\succeq
(1-\alpha)\boldsymbol\Lambda_n\\
&\succeq
(1-\alpha)(\mathbf V_n-b_n\mathbf I_{d_\phi})\\
&\succeq
\mathbf V_n-\{\alpha C_V+b_n\}\mathbf I_{d_\phi}\\
&\succeq
\mathbf V_n-\varepsilon\mathbf I_{d_\phi}.
\end{align*}
Together with~\eqref{eq:t3-Kn-lower}, this proves \eqref{eq:t3-pop-centered-conservative}.
\end{proof}

\begin{lemma}[One-sided group-centering bound]
\label{lemma:group_centering}
Suppose that $Z_i\stackrel{\mathrm{i.i.d.}}{\sim}\operatorname{Bernoulli}(\pi)$. Under Assumptions~\ref{ass:boundedness} and~\ref{ass:variance_lower_bound}, suppose that $\mathcal G_n$ is a dependency graph and that, for some constant $\delta>0$, $\Delta_n=\mathcal O(n^{1/2-\delta})$. Then, for every fixed $\rho\in(0,1)$, there exists a random matrix $\mathbf R_{n,\rho}\succeq\mathbf0$ such that
\begin{equation}
\label{eq:t3-group-centering-conclusion}
\mathbf D_n\hat{\mathbf U}_n\mathbf D_n
\succeq
(1-\rho)\mathbf D_n\mathbf U_n^{\mathrm{pop}}\mathbf D_n
-
\mathbf R_{n,\rho},
\qquad
\|\mathbf R_{n,\rho}\|_{\mathrm{op}}=o_p(1).
\end{equation}
\end{lemma}

\begin{proof}
Let
\[
\boldsymbol\gamma_{n,i}
=
\hat{\boldsymbol\phi}_i-\boldsymbol\phi_{n,i}^{\mathrm{pop}}
=
-Z_i\boldsymbol\delta_{n,1}-(1-Z_i)\boldsymbol\delta_{n,0},
\]
and set
\[
L_n
=
\max_{z\in\{0,1\}}
\|\mathbf D_n\boldsymbol\delta_{n,z}\|_2.
\]
We first derive a rate for $L_n$. Define
\[
p_1=\pi,
\qquad
p_0=1-\pi,
\qquad
\hat p_{n,1}=\frac1n\sum_{i=1}^n Z_i,
\qquad
\hat p_{n,0}=\frac1n\sum_{i=1}^n(1-Z_i),
\]
and let
\[
\mathcal E_n
=
\bigcap_{z\in\{0,1\}}
\left\{\hat p_{n,z}\ge\frac{p_z}{2}\right\}.
\]
The law of large numbers gives $\mathbb P(\mathcal E_n)\to1$. On $\mathcal E_n$, the definitions of the arm means and $\mathbf Q_{n,z}$ imply, for $z\in\{0,1\}$,
\[
\mathbf D_n\boldsymbol\delta_{n,z}
=
\frac{p_z}{\hat p_{n,z}}\sqrt n\,\mathbf Q_{n,z}.
\]
Lemma~\ref{lem:dg-arm-score-lln}, applied with $\mathcal G_n^{\mathrm{dep}}=\mathcal G_n$, gives
\[
\mathbb E\|\sqrt n\,\mathbf Q_{n,z}\|_2^2
\le
C\Delta_n.
\]
Markov's inequality therefore yields $\sqrt n\,\mathbf Q_{n,z}=\mathcal O_p(\sqrt{\Delta_n})$ for each $z$. Since $p_z/\hat p_{n,z}\le2$ on $\mathcal E_n$ and $\mathbb P(\mathcal E_n)\to1$, the preceding identity implies
\[
L_n
=
\mathcal O_p(\sqrt{\Delta_n}).
\]

For any $\mathbf a,\mathbf b\in\mathbb R^{d_\phi}$ and any $\rho\in(0,1)$,
\[
(\mathbf a+\mathbf b)(\mathbf a+\mathbf b)^\top
\succeq
(1-\rho)\mathbf a\mathbf a^\top
-
\frac{1-\rho}{\rho}\mathbf b\mathbf b^\top.
\]
Indeed, after moving the right-hand side to the left, the resulting matrix equals
\[
\left(\sqrt\rho\,\mathbf a+\rho^{-1/2}\mathbf b\right)
\left(\sqrt\rho\,\mathbf a+\rho^{-1/2}\mathbf b\right)^\top
\succeq
\mathbf0.
\]
Apply the preceding inequality with $\mathbf a=\mathbf D_n\boldsymbol\phi_{n,i}^{\mathrm{pop}}$ and $\mathbf b=\mathbf D_n\boldsymbol\gamma_{n,i}$, multiply by the nonnegative weight $n^{-2}d_{n,i}\psi_i^2$, and sum over $i$. This gives
\[
\mathbf D_n\hat{\mathbf U}_n\mathbf D_n
\succeq
(1-\rho)\mathbf D_n\mathbf U_n^{\mathrm{pop}}\mathbf D_n
-
\mathbf R_{n,\rho},
\]
where
\[
\mathbf R_{n,\rho}
=
\frac{1-\rho}{\rho n^2}
\sum_{i=1}^n
 d_{n,i}\psi_i^2
(\mathbf D_n\boldsymbol\gamma_{n,i})
(\mathbf D_n\boldsymbol\gamma_{n,i})^\top
\succeq
\mathbf0.
\]
Because $\|\mathbf D_n\boldsymbol\gamma_{n,i}\|_2\le L_n$, $\sup_i|\psi_i|\le\{\min(\pi,1-\pi)\}^{-1}$, and $\sum_i d_{n,i}\le n\Delta_n$,
\begin{align*}
\|\mathbf R_{n,\rho}\|_{\mathrm{op}}
&\le
C_\rho\frac{\Delta_n}{n}L_n^2\\
&=
\mathcal O_p\!\left(\frac{\Delta_n^2}{n}\right)
=
o_p(1),
\end{align*}
where the last equality follows from $\Delta_n=\mathcal O(n^{1/2-\delta})$. This proves~\eqref{eq:t3-group-centering-conclusion}.
\end{proof}

\subsection{Proof of Theorem~\ref{thm:conservative_matrix}}
\label{sec:thm3_proof}

Fix $\varepsilon>0$. Because $\mathbf V_n$ is positive semidefinite and $\|\mathbf V_n-\mathbf V\|_{\mathrm{op}}\to0$, there is a constant $C_V<\infty$ such that
\[
\mathbf0\preceq\mathbf V_n\preceq C_V\mathbf I_{d_\phi}
\]
for all sufficiently large $n$. Choose a fixed $\rho\in(0,1)$ such that $\rho C_V\le\varepsilon/4$.

By Lemma~\ref{lm:conservative_infeasible},
\begin{equation}
\label{eq:t3-final-a}
\mathbb P\left\{
\mathbf D_n\mathbf U_n^{\mathrm{pop}}\mathbf D_n
\succeq
\mathbf V_n-\frac{\varepsilon}{4}\mathbf I_{d_\phi}
\right\}
\to1.
\end{equation}
By Lemma~\ref{lemma:group_centering}, there is a random matrix $\mathbf R_{n,\rho}\succeq\mathbf0$ satisfying \eqref{eq:t3-group-centering-conclusion} and $\|\mathbf R_{n,\rho}\|_{\mathrm{op}}=o_p(1)$. Hence
\begin{equation}
\label{eq:t3-final-b}
\mathbb P\left\{
\|\mathbf R_{n,\rho}\|_{\mathrm{op}}
\le
\frac{\varepsilon}{4}
\right\}
\to1.
\end{equation}

On the event in~\eqref{eq:t3-final-b}, the facts that $\mathbf R_{n,\rho}\succeq\mathbf0$ and $\|\mathbf R_{n,\rho}\|_{\mathrm{op}}\le\varepsilon/4$ imply $\mathbf R_{n,\rho}\preceq(\varepsilon/4)\mathbf I_{d_\phi}$. Therefore, \eqref{eq:t3-group-centering-conclusion} gives
\begin{equation}
\label{eq:t3-final-one-sided}
\mathbf D_n\hat{\mathbf U}_n\mathbf D_n
\succeq
(1-\rho)\mathbf D_n\mathbf U_n^{\mathrm{pop}}\mathbf D_n
-
\frac{\varepsilon}{4}\mathbf I_{d_\phi}
\end{equation}
on that event. On the intersection of the events in~\eqref{eq:t3-final-a} and~\eqref{eq:t3-final-b}, combining~\eqref{eq:t3-final-a} and \eqref{eq:t3-final-one-sided} yields
\begin{align*}
\mathbf D_n\hat{\mathbf U}_n\mathbf D_n
&\succeq
(1-\rho)\left(
\mathbf V_n-\frac{\varepsilon}{4}\mathbf I_{d_\phi}
\right)
-
\frac{\varepsilon}{4}\mathbf I_{d_\phi}\\
&\succeq
\mathbf V_n
-
\left\{
\rho C_V+\frac{(1-\rho)\varepsilon}{4}+\frac{\varepsilon}{4}
\right\}\mathbf I_{d_\phi}\\
&\succeq
\mathbf V_n-\frac{3\varepsilon}{4}\mathbf I_{d_\phi}.
\end{align*}
Finally, by~\eqref{eq:t3-Vn-op}, for all sufficiently large $n$,
\begin{equation}
\label{eq:t3-final-Vn-V}
\mathbf V_n
\succeq
\mathbf V-\frac{\varepsilon}{4}\mathbf I_{d_\phi}.
\end{equation}
The intersection of the two events above has probability tending to one. On that intersection, for all sufficiently large $n$, the preceding bound and~\eqref{eq:t3-final-Vn-V} give
\[
\mathbf D_n\hat{\mathbf U}_n\mathbf D_n
\succeq
\mathbf V-\varepsilon\mathbf I_{d_\phi}.
\]
Equivalently,
\[
\mathbb P\left\{
\mathbf D_n\hat{\mathbf U}_n\mathbf D_n
-
\mathbf V
+
\varepsilon\mathbf I_{d_\phi}
\succeq
\mathbf0
\right\}
\to1,
\]
which proves Theorem~\ref{thm:conservative_matrix}.

\subsection{Proof of Corollary \ref{cor:conditional_psd_conservative}}
\label{sec:cor2_proof}

This follows immediately from Theorem~\ref{thm:conservative_matrix}. Indeed, let
\[
A_n(\varepsilon)
=\left\{
\mathbf D_n\hat{\mathbf U}_n\mathbf D_n-\mathbf V
+\varepsilon\mathbf I_{d_\phi}\succeq\mathbf0
\right\}.
\]
Then
\[
\mathbb P\!\left(A_n(\varepsilon)^c\mid M_n\le a\right)
\le
\frac{\mathbb P\!\left(A_n(\varepsilon)^c\right)}
{\mathbb P(M_n\le a)}.
\]
By Theorem~\ref{thm:conservative_matrix}, the numerator converges to zero. On the other hand, by Theorem~\ref{thm:joint-normality} and the continuous mapping theorem, $M_n \xrightarrow{d} \chi^2_p$, so $\mathbb{P}(M_n \le a) \to \mathbb{P}(\chi^2_p \le a) > 0$. Hence the conditional probability above converges to zero.

\subsection{Proof of Theorem \ref{thm:opt_solution}}

We first establish the feasibility criterion. The constraint $\boldsymbol\Omega\preceq\hat{\mathbf U}_{n}$ together with $\boldsymbol\Omega_{22}=\boldsymbol\Sigma_{xx}$ necessarily implies $\hat{\mathbf U}_{22}\succeq\boldsymbol\Sigma_{xx}$. Conversely, suppose this block inequality holds and set $\mathbf C=\hat{\mathbf U}_{22}\succ\mathbf0$ and
\[
\mathbf L=
\begin{pmatrix}
\hat{\mathbf U}_{12}\mathbf C^{-1}\\
\mathbf I_p
\end{pmatrix}.
\]
The Schur complement of $\hat{\mathbf U}_{n}\succeq\mathbf0$ gives the decomposition
\[
\hat{\mathbf U}_{n}
=
\begin{pmatrix}
\hat U_{11}-\hat{\mathbf U}_{12}\mathbf C^{-1}
\hat{\mathbf U}_{21}&\mathbf0\\
\mathbf0&\mathbf0
\end{pmatrix}
+\mathbf L\mathbf C\mathbf L^\top.
\]
Thus $\boldsymbol\Omega_0=\mathbf L\boldsymbol\Sigma_{xx}\mathbf L^\top$ is positive semidefinite, has lower-right block $\boldsymbol\Sigma_{xx}$, and is dominated by $\hat{\mathbf U}_{n}$ because
\[
\hat{\mathbf U}_{n}-\boldsymbol\Omega_0
=
\begin{pmatrix}
\hat U_{11}-\hat{\mathbf U}_{12}\mathbf C^{-1}
\hat{\mathbf U}_{21}&\mathbf0\\
\mathbf0&\mathbf0
\end{pmatrix}
+\mathbf L(\mathbf C-\boldsymbol\Sigma_{xx})\mathbf L^\top
\succeq\mathbf0.
\]
This proves sufficiency.

For ease of presentation we additionally assume $\hat{\mathbf U}_{22} - \boldsymbol\Sigma_{xx} \succ 0$ in the derivation below; the positive-semidefinite boundary case is justified at the end of the proof.

\medskip\noindent\textbf{Step 1: Two-sided bounds on $\Omega_{11}$.} Apply the Schur complement to each of the two PSD constraints. The lower constraint $\boldsymbol\Omega \succeq 0$, with lower-right block $\boldsymbol\Sigma_{xx} \succ 0$, gives
\begin{equation}\label{eq:opt_Omega11_lower} \Omega_{11} \;\ge\;
\boldsymbol\Omega_{12}\,\boldsymbol\Sigma_{xx}^{-1}\,\boldsymbol\Omega_{21}.
\end{equation}
The upper constraint $\hat{\mathbf U}_{n} - \boldsymbol\Omega \succeq 0$, with lower-right block $\hat{\mathbf U}_{22} - \boldsymbol\Sigma_{xx} \succ 0$, gives
\begin{equation}\label{eq:opt_Omega11_upper} \Omega_{11} \;\le\; \hat U_{11} -
(\hat{\mathbf U}_{12} - \boldsymbol\Omega_{12})(\hat{\mathbf U}_{22} -
\boldsymbol\Sigma_{xx})^{-1}(\hat{\mathbf U}_{21} - \boldsymbol\Omega_{21}).
\end{equation}
Hence, for any choice of $(\boldsymbol\Omega_{12}, \boldsymbol\Omega_{21})$, the variable $\Omega_{11}$ ranges within the interval defined by~\eqref{eq:opt_Omega11_lower}--\eqref{eq:opt_Omega11_upper}, and feasibility of the triple requires this interval to be nonempty (the right-hand side of~\eqref{eq:opt_Omega11_upper} must dominate the right-hand side of~\eqref{eq:opt_Omega11_lower}). Within any feasible interval, the objective $\Omega_{11} - (1-v_p)\,\boldsymbol\Omega_{12}\boldsymbol\Sigma_{xx}^{-1}\boldsymbol\Omega_{21}$ is strictly increasing in $\Omega_{11}$, so the maximum places $\Omega_{11}$ at its upper endpoint:
\begin{equation}\label{eq:opt_Omega11_in_Omega21} \Omega_{11} \;=\; \hat U_{11} -
(\hat{\mathbf U}_{12} - \boldsymbol\Omega_{12})(\hat{\mathbf U}_{22} -
\boldsymbol\Sigma_{xx})^{-1}(\hat{\mathbf U}_{21} - \boldsymbol\Omega_{21}).
\end{equation}
This saturates the upper PSD constraint but does \emph{not} automatically guarantee the lower PSD constraint: substituting~\eqref{eq:opt_Omega11_in_Omega21} into~\eqref{eq:opt_Omega11_lower} produces a single quadratic constraint that must be tracked alongside the reduced objective:
\begin{equation}\label{eq:opt_lower_quad} (\hat{\mathbf U}_{12} -
\boldsymbol\Omega_{12})(\hat{\mathbf U}_{22} -
\boldsymbol\Sigma_{xx})^{-1}(\hat{\mathbf U}_{21} - \boldsymbol\Omega_{21}) +
\boldsymbol\Omega_{12}\boldsymbol\Sigma_{xx}^{-1}\boldsymbol\Omega_{21}
\;\le\; \hat U_{11}.
\end{equation}
Equivalently, \eqref{eq:opt_lower_quad} is precisely the condition that the upper bound~\eqref{eq:opt_Omega11_upper} dominates the lower bound~\eqref{eq:opt_Omega11_lower} for the chosen $(\boldsymbol\Omega_{12}, \boldsymbol\Omega_{21})$. The reduced problem is
\begin{equation}\label{eq:opt_reduced} \max_{\boldsymbol\Omega_{21}} \;\; \hat U_{11} -
(\hat{\mathbf U}_{12} - \boldsymbol\Omega_{12})(\hat{\mathbf U}_{22} -
\boldsymbol\Sigma_{xx})^{-1}(\hat{\mathbf U}_{21} - \boldsymbol\Omega_{21}) -
(1-v_p)\,\boldsymbol\Omega_{12}\boldsymbol\Sigma_{xx}^{-1}\boldsymbol\Omega_{21}
\quad\text{s.t.~\eqref{eq:opt_lower_quad}}.
\end{equation}

\medskip\noindent\textbf{Step 2: Parametric form of the optimizer.} The objective in~\eqref{eq:opt_reduced} is strictly concave in $\boldsymbol\Omega_{21}$ and~\eqref{eq:opt_lower_quad} defines a compact convex ellipsoid, so an optimum exists. When this ellipsoid has nonempty interior, the KKT conditions apply. Introducing a Lagrange multiplier $\mu\ge0$ for~\eqref{eq:opt_lower_quad} and equating the gradient in $\boldsymbol\Omega_{21}$ to zero yields
\[
(1+\mu)(\hat{\mathbf U}_{22} - \boldsymbol\Sigma_{xx})^{-1}
(\hat{\mathbf U}_{21} -
\boldsymbol\Omega_{21}) \;=\; (1-v_p+\mu)\,\boldsymbol\Sigma_{xx}^{-1}\boldsymbol\Omega_{21},
\]
or equivalently $(\hat{\mathbf U}_{22} - \boldsymbol\Sigma_{xx})^{-1} \hat{\mathbf U}_{21} = \bigl[(\hat{\mathbf U}_{22} - \boldsymbol\Sigma_{xx})^{-1} + t\,\boldsymbol\Sigma_{xx}^{-1}\bigr]\boldsymbol\Omega_{21}$, where
\[
t \;:=\; \frac{1-v_p+\mu}{1+\mu} \in [1-v_p,\,1]
\]
(so $t = 1-v_p$ at $\mu = 0$ and $t \to 1$ as $\mu \to \infty$). If the ellipsoid has empty interior, strict convexity of its defining quadratic makes it a singleton at the quadratic's minimizer, which is exactly the same path at $t=1$. Thus both cases are covered by $t\in[1-v_p,1]$. Using the algebraic identity
\[
(\hat{\mathbf U}_{22} - \boldsymbol\Sigma_{xx})^{-1}
+t\,\boldsymbol\Sigma_{xx}^{-1}
\;=\;
(\hat{\mathbf U}_{22} - \boldsymbol\Sigma_{xx})^{-1}
\bigl[\boldsymbol\Sigma_{xx}
+t(\hat{\mathbf U}_{22} - \boldsymbol\Sigma_{xx})\bigr]
\boldsymbol\Sigma_{xx}^{-1}
\;=\;
(\hat{\mathbf U}_{22} - \boldsymbol\Sigma_{xx})^{-1}
\mathbf M(t)\boldsymbol\Sigma_{xx}^{-1},
\]
where $\mathbf M(t)$ is as in the theorem statement and $\mathbf M(t)\succ0$ on $[1-v_p,1]$ (since $\boldsymbol\Sigma_{xx}\succ0$, $\hat{\mathbf U}_{22}-\boldsymbol\Sigma_{xx}\succeq0$, and $t\ge0$), we solve for $\boldsymbol\Omega_{21}$:
\begin{equation}\label{eq:opt_Omega21_path}
\boldsymbol\Omega_{21}(t) \;=\; \boldsymbol\Sigma_{xx}\,\mathbf
M(t)^{-1}\,\hat{\mathbf U}_{21}, \qquad
\boldsymbol\Omega_{12}(t) \;=\;
\hat{\mathbf U}_{12}\,\mathbf M(t)^{-1}\,\boldsymbol\Sigma_{xx}.
\end{equation}
A direct computation gives
\begin{align*}
\hat{\mathbf U}_{21}-\boldsymbol\Omega_{21}(t)
&=\{\mathbf M(t)-\boldsymbol\Sigma_{xx}\}
\mathbf M(t)^{-1}\hat{\mathbf U}_{21}\\
&=t(\hat{\mathbf U}_{22}-\boldsymbol\Sigma_{xx})
\mathbf M(t)^{-1}\hat{\mathbf U}_{21}.
\end{align*}
Substituting into~\eqref{eq:opt_Omega11_in_Omega21} gives
\begin{equation}\label{eq:opt_Omega11_path}
\Omega_{11}(t) \;=\; \hat U_{11}
-t^2\hat{\mathbf U}_{12}\mathbf M(t)^{-1}
(\hat{\mathbf U}_{22}-\boldsymbol\Sigma_{xx})
\mathbf M(t)^{-1}\hat{\mathbf U}_{21}.
\end{equation}

\medskip\noindent\textbf{Step 3: Constraints on the path.} Substituting~\eqref{eq:opt_Omega21_path} into the left-hand side of~\eqref{eq:opt_lower_quad} and using
\[
\hat{\mathbf U}_{21}-\boldsymbol\Omega_{21}(t)
=t(\hat{\mathbf U}_{22}-\boldsymbol\Sigma_{xx})
\mathbf M(t)^{-1}\hat{\mathbf U}_{21}
\]
from Step~2 yields, by direct expansion,
\[
\begin{aligned}
&\boldsymbol\Omega_{12}(t)\boldsymbol\Sigma_{xx}^{-1}\boldsymbol\Omega_{21}(t)\\
&\quad+
(\hat{\mathbf U}_{12}-\boldsymbol\Omega_{12}(t))
(\hat{\mathbf U}_{22}-\boldsymbol\Sigma_{xx})^{-1}
(\hat{\mathbf U}_{21}-\boldsymbol\Omega_{21}(t))
=G(t).
\end{aligned}
\]
Thus the constraint~\eqref{eq:opt_lower_quad} on the path is equivalent to $G(t)\le\hat U_{11}$.

\smallskip \noindent\emph{Monotonicity of $G(t)$.} Let
\[
\mathbf L
:=
\boldsymbol\Sigma_{xx}^{-1/2}
(\hat{\mathbf U}_{22}-\boldsymbol\Sigma_{xx})
\boldsymbol\Sigma_{xx}^{-1/2}
\succeq0,
\]
with eigendecomposition $\mathbf L=\mathbf P\operatorname{diag}(\lambda_1,\ldots,\lambda_p)\mathbf P^\top$, $\lambda_j\ge0$. Set $\mathbf r=\mathbf P^\top\boldsymbol\Sigma_{xx}^{-1/2}\hat{\mathbf U}_{21}$. Then
\[
\mathbf M(t)^{-1}
=
\boldsymbol\Sigma_{xx}^{-1/2}\mathbf P
\operatorname{diag}\bigl((1+t\lambda_j)^{-1}\bigr)
\mathbf P^\top\boldsymbol\Sigma_{xx}^{-1/2},
\]
and the theorem's definition of $G(t)$ becomes
\begin{equation}
\label{eq:opt_G_diagonal}
G(t)
=
\sum_{j=1}^p r_j^2\frac{1+t^2\lambda_j}{(1+t\lambda_j)^2},
\qquad
G'(t)
=
2(t-1)\sum_{j=1}^p\frac{\lambda_jr_j^2}{(1+t\lambda_j)^3}.
\end{equation}
For $t\in[0,1]$, $G'(t)\le0$, so $G$ is non-increasing on $[1-v_p,1]$. At $t=1$, $\mathbf M(1)=\hat{\mathbf U}_{22}$ and
\[
G(1)
=
\hat{\mathbf U}_{12}
(\hat{\mathbf U}_{22})^{-1}
\hat{\mathbf U}_{21}
\le\hat U_{11}.
\]
This is the Schur complement of $\hat{\mathbf U}_{n}\succeq0$.

\smallskip

\textbf{Case A:} If $G(1-v_p)\le\hat U_{11}$, then at $\mu=0$ ($t=1-v_p$) the unconstrained maximizer of~\eqref{eq:opt_reduced} is feasible. By concavity and KKT this is the global optimum, so $t^\star=1-v_p$.

\textbf{Case B:} If $G(1-v_p)>\hat U_{11}$, complementary slackness forces~\eqref{eq:opt_lower_quad} to bind. Combining $G(1-v_p)>\hat U_{11}$, $G(1)\le\hat U_{11}$, monotonicity, and continuity yields $t^\star\in(1-v_p,1]$ with $G(t^\star)=\hat U_{11}$.

If $G$ is constant at $\hat U_{11}$ on a sub-interval, every point of that sub-interval satisfies the KKT conditions and yields the same optimal value $\hat v_{n}$. We pick the smallest such $t^\star$ for definiteness.

\medskip\noindent\textbf{Step 4: Optimal value.} With $t^\star$ from Cases A and B, $\boldsymbol\Omega^\star$ has blocks
\[
\boldsymbol\Omega^\star_{22} = \boldsymbol\Sigma_{xx}, \qquad \boldsymbol\Omega^\star_{12} =
\hat{\mathbf U}_{12}\,\mathbf M(t^\star)^{-1}\boldsymbol\Sigma_{xx}, \qquad
\boldsymbol\Omega^\star_{21} = (\boldsymbol\Omega^\star_{12})^\top,
\]
and $\Omega^\star_{11} = \Omega_{11}(t^\star)$ from~\eqref{eq:opt_Omega11_path}. The optimal value is
\begin{align*}
\hat v_{n}
&=\Omega^\star_{11}
-(1-v_p)\,\boldsymbol\Omega^\star_{12}
\boldsymbol\Sigma_{xx}^{-1}\boldsymbol\Omega^\star_{21}\\
&=\hat U_{11}
-\hat{\mathbf U}_{12}\mathbf M(t^\star)^{-1}
\bigl[(1-v_p)\boldsymbol\Sigma_{xx}
+(t^\star)^2(\hat{\mathbf U}_{22}-\boldsymbol\Sigma_{xx})\bigr]
\mathbf M(t^\star)^{-1}\hat{\mathbf U}_{21}.
\end{align*}
In Case A ($t^\star = 1-v_p$), the bracketed matrix simplifies to
\[
(1-v_p)\boldsymbol\Sigma_{xx}
+(1-v_p)^2(\hat{\mathbf U}_{22}-\boldsymbol\Sigma_{xx})
\;=\;
(1-v_p)\mathbf M(1-v_p),
\]
so one $\mathbf M(1-v_p)^{-1}$ cancels and
\[
\hat v_{n}
\;=\;
\hat U_{11}
-(1-v_p)\hat{\mathbf U}_{12}
\{v_p\boldsymbol\Sigma_{xx}
+(1-v_p)\hat{\mathbf U}_{22}\}^{-1}
\hat{\mathbf U}_{21},
\]
where the last equality uses $\mathbf M(1-v_p) =v_p\boldsymbol\Sigma_{xx}+(1-v_p)\hat{\mathbf U}_{22}$.

It remains to justify the boundary case $\hat{\mathbf U}_{22}-\boldsymbol\Sigma_{xx}\succeq\mathbf0$. For $\xi>0$, replace $\hat{\mathbf U}_{n}$ by
\[
\hat{\mathbf U}_{n,\xi}
=\hat{\mathbf U}_{n}+
\begin{pmatrix}
0&\mathbf0\\
\mathbf0&\xi\mathbf I_p
\end{pmatrix}.
\]
The preceding argument applies to each $\xi$. As $\xi\downarrow0$, the corresponding feasible sets are nested and contained in a common compact set. Hence every sequence of optimizers has a convergent subsequence whose limit is feasible for the original problem, while every originally feasible matrix remains feasible for all $\xi>0$. Continuity of the objective therefore implies convergence of the optimal values. Moreover, $\mathbf M_\xi(t)^{-1}$ and $G_\xi(t)$ converge uniformly on $[1-v_p,1]$ to their $\xi=0$ counterparts because $\mathbf M_\xi(t)\succeq\boldsymbol\Sigma_{xx}\succ\mathbf0$. Any limit point of the corresponding $t_\xi^\star$ consequently satisfies the limiting KKT conditions; if the limiting solution is not unique, all points in the resulting level interval give the same objective value. Taking the smallest such point gives the formulas stated in the theorem.

\subsection{Stability lemma}
\label{app:variance-optimization-stability}
Feasibility of the optimization does not ensure that its feasible set contains the true covariance matrix $\mathbf U_n$. By Theorem~\ref{thm:opt_solution}, the program is feasible if and only if $\hat{\mathbf U}_{22}\succeq\boldsymbol\Sigma_{xx}$, whereas feasibility of $\mathbf U_n$ itself requires the full matrix inequality $\hat{\mathbf U}_n\succeq\mathbf U_n$. Thus, a feasible program may exclude $\mathbf U_n$ and return a value below the objective evaluated at $\mathbf U_n$, without triggering the fallback in Definition~\ref{def:rerand-variance-estimator}.

Corollary~\ref{cor:conditional_psd_conservative} controls the covariance bound only up to a vanishing eigenvalue error after standardization. The following lemma shows that this failure of exact dominance can lower the optimized value, including its infeasibility fallback, by at most a constant times that error. This provides the deterministic bound needed to prove Theorem~\ref{thm:asymp_conservativeness} and will also be used for the ridge corrections in Appendix~\ref{app:feasibility-refinement}.

For a positive definite matrix $\mathbf B$ and a positive semidefinite matrix $\mathbf A$, partitioned with a scalar first block, write
\[
q_{\mathbf B}(\boldsymbol\Omega)
=\Omega_{11}-(1-v_p)\boldsymbol\Omega_{12}
\mathbf B^{-1}\boldsymbol\Omega_{21},
\]
and let $\Phi_{\mathbf B}(\mathbf A)$ be the maximum of $q_{\mathbf B}$ over $\mathbf0\preceq\boldsymbol\Omega\preceq\mathbf A$ with $\boldsymbol\Omega_{22}=\mathbf B$, using $A_{11}$ as the fallback when this set is empty. Let $\widetilde\Phi_{\mathbf B}(\mathbf A)$ denote the same optimized value with the fallback
\[
F(\mathbf A)=A_{11}-(1-v_p)\mathbf A_{12}\mathbf A_{22}^{-1}\mathbf A_{21}
\]
when $\mathbf A_{22}\succ\mathbf0$, and $A_{11}$ otherwise, using the same singular-block convention as Definition~\ref{def:rerand-variance-estimator}.

\begin{samepage}
\begin{lemma}[Stability under near dominance]
\label{lem:near-dominance-optimization}
Let $\mathbf W\succ\mathbf0$ satisfy $\mathbf W_{22}=\mathbf B\succ\mathbf0$, let $\mathbf A\succeq\mathbf0$, and put
\[
r=\bigl[-\lambda_{\min}(\mathbf A-\mathbf W)\bigr]_+.
\]
If the eigenvalues of $\mathbf W$ and $\mathbf B$ are bounded away from zero and their operator norms are bounded above, then, for all sufficiently small $r$,
\[
\min\{\Phi_{\mathbf B}(\mathbf A),\widetilde\Phi_{\mathbf B}(\mathbf A)\}
\ge q_{\mathbf B}(\mathbf W)-Cr,
\]
where $C<\infty$ depends only on these uniform bounds and on $v_p$.
\end{lemma}
\end{samepage}

\begin{proof}
Let $c=1-v_p$ and $\mathbf H=\mathbf A+r\mathbf I-\mathbf W\succeq\mathbf0$. If the program is infeasible, the fallback and the $(1,1)$ entry of $\mathbf A+r\mathbf I\succeq\mathbf W$ give
\[
\Phi_{\mathbf B}(\mathbf A)=A_{11}
\ge W_{11}-r
\ge q_{\mathbf B}(\mathbf W)-r.
\]
For the plug-in fallback, take $r$ small enough that $\mathbf A\succeq\mathbf W-r\mathbf I\succ\mathbf0$. The map $F$ is monotone in the positive semidefinite order: it is the sum of $v_p A_{11}$ and $c$ times the Schur complement of $\mathbf A_{22}$. Consequently,
\begin{align*}
\widetilde\Phi_{\mathbf B}(\mathbf A)
=F(\mathbf A)
&\ge F(\mathbf W-r\mathbf I)\\
&=q_{\mathbf B}(\mathbf W)-r
-c\mathbf W_{12}\bigl\{(\mathbf B-r\mathbf I)^{-1}-\mathbf B^{-1}\bigr\}\mathbf W_{21}\\
&\ge q_{\mathbf B}(\mathbf W)-Cr.
\end{align*}
The last bound is uniform for small $r$, since $(\mathbf B-r\mathbf I)^{-1}-\mathbf B^{-1} =r(\mathbf B-r\mathbf I)^{-1}\mathbf B^{-1}$ and the assumed spectral bounds apply.

Suppose instead that the program is feasible, so the two optimized values coincide. Feasibility implies $\mathbf A_{22}\succeq\mathbf B$, and hence $\mathbf H_{22}\succeq r\mathbf I$. The conclusion is immediate when $r=0$, because $\mathbf W$ is then feasible. For $r>0$, the sufficiency argument in the proof of Theorem~\ref{thm:opt_solution}, applied with upper bound $\mathbf H$ and prescribed lower-right block $r\mathbf I$, gives a matrix $\mathbf0\preceq\mathbf K\preceq\mathbf H$ with $\mathbf K_{22}=r\mathbf I$. When $r$ is smaller than $\lambda_{\min}(\mathbf W)$,
\[
\boldsymbol\Omega=\mathbf W-r\mathbf I+\mathbf K
\]
is positive semidefinite, is dominated by $\mathbf A$, and has lower-right block $\mathbf B$; therefore it is feasible. Write $\mathbf k=\mathbf K_{21}$. The Schur complement of $\mathbf K\succeq\mathbf0$ gives $K_{11}\ge\|\mathbf k\|_2^2/r$. Consequently,
\begin{align*}
q_{\mathbf B}(\boldsymbol\Omega)-q_{\mathbf B}(\mathbf W)
&=-r+K_{11}
-2c\mathbf W_{12}\mathbf B^{-1}\mathbf k
-c\mathbf k^\top\mathbf B^{-1}\mathbf k\\
&\ge-r+
\left(r^{-1}-c\|\mathbf B^{-1}\|_{\mathrm{op}}\right)\|\mathbf k\|_2^2
-2c\|\mathbf W_{12}\mathbf B^{-1}\|_2\|\mathbf k\|_2\\
&\ge-Cr,
\end{align*}
where the last inequality follows by completing the square, uniformly for small $r$.
\end{proof}

\subsection{Proof of Theorem~\ref{thm:asymp_conservativeness}}
\label{app:proof-asymp-conservativeness}

Write
\[
\mathbf A_n=\mathbf D_n\hat{\mathbf U}_n\mathbf D_n,
\qquad
\mathbf V_n=\mathbf D_n\mathbf U_n\mathbf D_n,
\qquad
\mathbf B_n=(\mathbf V_n)_{22}
=\mathbf D_{n,x}\boldsymbol\Sigma_{xx}\mathbf D_{n,x},
\]
where $\mathbf D_{n,x}$ is the covariate block of $\mathbf D_n$, and define the oracle standardized deficiency
\[
r_n=\bigl[-\lambda_{\min}(\mathbf A_n-\mathbf V_n)\bigr]_+.
\]
Corollary~\ref{cor:conditional_psd_conservative} and $\|\mathbf V_n-\mathbf V\|_{\mathrm{op}}\to0$ imply $r_n=o_p(1)$ conditionally on $M_n\le a$. Since $\mathbf V_n\to\mathbf V\succ\mathbf0$, $\mathbf A_n\succeq\mathbf V_n-r_n\mathbf I\succ\mathbf0$ with conditional probability tending to one, so the singular-block convention above has no effect on the asymptotic conclusion. Applying Lemma~\ref{lem:near-dominance-optimization} to the plug-in fallback gives
\[
\frac{\hat v_n}{\sigma_{n,y}^2}
=\widetilde\Phi_{\mathbf B_n}(\mathbf A_n)
\ge
q_{\mathbf B_n}(\mathbf V_n)-Cr_n
\]
with conditional probability tending to one. Assumption~\ref{ass:variance_stability} gives $\mathbf V_n\to\mathbf V\succ\mathbf0$, and hence
\[
q_{\mathbf B_n}(\mathbf V_n)
\to
V_{11}-(1-v_p)\mathbf V_{12}\mathbf V_{22}^{-1}\mathbf V_{21}
=\frac{v_n}{\sigma_{n,y}^2}.
\]
The claimed probability statement follows.

\subsection{Proof of Theorem~\ref{thm:ci_coverage}}
\label{app:proof-ci-coverage}

Fix $\alpha \in (0,1)$. Recall from~\eqref{eq:rerand_critical_value} that, with $Q(R^2;p,a) = \sqrt{1-R^2}\,\varepsilon_0 + \sqrt{R^2}\,L_{p,a}$ and $\omega_{1-\alpha/2}(R^2;p,a)$ its $(1-\alpha/2)$-quantile,
\[
q_{1-\alpha/2}(R^2;p,a) = \frac{\omega_{1-\alpha/2}(R^2;p,a)}{\sqrt{1-(1-v_p)R^2}}, \qquad \bar
q_{1-\alpha/2}(p,a) = \sup_{R^2\in[0,1]} q_{1-\alpha/2}(R^2;p,a).
\]
Write $\bar q := \bar q_{1-\alpha/2}(p,a)$, a deterministic constant depending only on $(p,a,\alpha)$, and let $D := 1-(1-v_p)R^2$ denote the true (unknown) design parameter, so that $\operatorname{Var}\bigl(Q(R^2;p,a)\bigr) = D \ge v_p > 0$ and $v_n = \sigma_{n,y}^2\,D$. The constant $\bar q$ is finite: Chebyshev's inequality gives $\omega_{1-\alpha/2}(R^2;p,a)\le\sqrt{2/\alpha}$ uniformly in $R^2$, whereas $\sqrt D\ge\sqrt{v_p}$. Coverage is equivalent to
\[
\tau \in \mathcal I_{1-\alpha}
\;\Longleftrightarrow\;
|\hat\tau-\tau|\le\bar q\,(\hat v_n)^{1/2}.
\]

Fix $\eta\in(0,v_p)$. By Theorem~\ref{thm:asymp_conservativeness}, the event
\[
\mathcal B_n(\eta)
:=\bigl\{\hat v_n\ge v_n-\eta\sigma_{n,y}^2\bigr\}
\]
satisfies $\mathbb P\bigl(\mathcal B_n(\eta)\mid M_n\le a\bigr)\to1$. Since $v_n=\sigma_{n,y}^2D$ and $D-\eta\ge v_p-\eta>0$, on $\mathcal B_n(\eta)$ we have
\[
\hat v_n\ge\sigma_{n,y}^2(D-\eta),
\qquad\text{hence}\qquad
\bar q\,(\hat v_n)^{1/2}
\ge\bar q\,\sigma_{n,y}\sqrt{D-\eta}.
\]
Thus on $\mathcal B_n(\eta)$ the interval $\mathcal I_{1-\alpha}$ contains the deterministic-half-width interval $\bigl[\hat\tau \pm \sigma_{n,y}\,\bar q\sqrt{D-\eta}\bigr]$.

Let $A_n := \bigl\{|\hat\tau-\tau|/\sigma_{n,y} \le \bar q\sqrt{D-\eta}\bigr\}$. By the preceding inclusion, $A_n\cap\mathcal B_n(\eta)\subseteq\{\tau\in\mathcal I_{1-\alpha}\}$, so by the elementary inequality $\mathbb P(A\cap B)\ge\mathbb P(A)-\mathbb P(B^c)$,
\[
\mathbb P\bigl(\tau\in\mathcal I_{1-\alpha}\mid M_n\le a\bigr)
\ge
\mathbb P\bigl(A_n\mid M_n\le a\bigr)
-\mathbb P\bigl(\mathcal B_n(\eta)^c\mid M_n\le a\bigr),
\]
where the last term tends to $0$. By~\eqref{eqn:limit_unscaled}, conditional on $M_n\le a$, $(\hat\tau-\tau)/\sigma_{n,y}\xrightarrow{d} Q(R^2;p,a)$, whose distribution is continuous for every $R^2\in[0,1]$ (by Gaussian convolution when $R^2<1$, and because $L_{p,a}$ has a continuous density when $R^2=1$). Hence $\{x:|x|\le \bar q\sqrt{D-\eta}\}$ is a continuity set of the limit, and by the Portmanteau theorem
\[
\liminf_{n\to\infty}
\mathbb P\bigl(\tau\in\mathcal I_{1-\alpha}\mid M_n\le a\bigr)
\ge
\mathbb P\bigl(|Q(R^2;p,a)|\le\bar q\sqrt{D-\eta}\bigr).
\]
Letting $\eta\downarrow0$ and using continuity of the CDF of $|Q(R^2;p,a)|$,
\begin{equation}\label{eq:proof_after_eta} \liminf_{n\to\infty}\mathbb P\bigl(\tau\in\mathcal
I_{1-\alpha}\mid M_n\le a\bigr)
\ge\mathbb P\bigl(|Q(R^2;p,a)|\le\bar q\sqrt D\bigr).
\end{equation}

Let $T(R^2;p,a) := Q(R^2;p,a)/\sqrt D$ be the studentized limit (unit variance). By construction~\eqref{eq:rerand_critical_value},
\[
q_{1-\alpha/2}(R^2;p,a)
=\frac{\omega_{1-\alpha/2}(R^2;p,a)}{\sqrt D}
\]
is exactly the $(1-\alpha/2)$-quantile of $T(R^2;p,a)$. Since $Q$, and hence $T$, is symmetric about zero,
\[
\mathbb P\bigl(|T(R^2;p,a)| \le q_{1-\alpha/2}(R^2;p,a)\bigr) = 1-\alpha.
\]
Because $\bar q = \sup_{r\in[0,1]} q_{1-\alpha/2}(r;p,a) \ge q_{1-\alpha/2}(R^2;p,a)$ at the true $R^2$, monotonicity of the CDF of $|T(R^2;p,a)|$ gives
\[
\mathbb P\bigl(|Q(R^2;p,a)| \le \bar q\sqrt D\bigr) = \mathbb P\bigl(|T(R^2;p,a)| \le \bar q\bigr)
\ge \mathbb P\bigl(|T(R^2;p,a)| \le q_{1-\alpha/2}(R^2;p,a)\bigr) = 1-\alpha.
\]
Combining this with~\eqref{eq:proof_after_eta} yields
\[
\liminf_{n\to\infty}
\mathbb P\bigl(\tau\in\mathcal I_{1-\alpha}\mid M_n\le a\bigr)
\ge1-\alpha.
\]

\section{ Discussion and Extensions to Section~\ref{sec:conservative}}
\label{sec:ex_conservative}

\subsection{Alternative approaches to conservative inference}
\label{app:alternative_inference}

In this section, we briefly discuss several alternative conservative procedures, including simpler variance bounds, alternative confidence-interval constructions, and a Gaussian-prepivoting approach inspired by \citet{cohen2022gaussian}. Although these approaches are generally less efficient than our proposed method or require stronger assumptions, they provide useful intuition.

\subsubsection{Alternative variance estimators}
\label{sec:alt_var_estimator}

Recall from Corollary~\ref{thm:variance_reduction} that the asymptotic variance under rerandomization can be written as
\[
\Var_{\operatorname{asymp}}(\hat{\tau} \mid M_n \le a)= \sigma_{n,y}^2 V_{11} \Bigl\{ (1-R^2)+v_p
R^2 \Bigr\} = \sigma_{n,y}^2 V_{11} (1-R^2) + v_p \sigma_{n,y}^2 V_{11} R^2 ,
\]
where
\[
R^2 = \frac{ \mathbf V_{12} \mathbf V_{22}^{-1} \mathbf V_{21} }{ V_{11} }.
\]

Our proposed estimator is obtained by solving the optimization problem~(\ref{eqn:optimization_problem}) and is specifically designed to provide a relatively sharp conservative upper bound. Nevertheless, two simpler alternatives are available.

\paragraph{Bernoulli estimator}
The first alternative is the variance estimator inherited from the Bernoulli design:
\[
\hat U_{11}.
\]
Since $\hat{\mathbf U}$ is conservative for the covariance matrix under Bernoulli randomization, this estimator remains conservative under rerandomization. However, it completely ignores the variance reduction induced by rerandomization and is therefore typically overly conservative.

\paragraph{Small-\texorpdfstring{$a$}{a} plug-in estimator}
A second alternative becomes particularly natural when the rerandomization threshold is small. When $a \approx 0$, we have $v_p \approx 0$. Moreover, the solution $\hat v_{n,\delta}$ of the optimization problem~\eqref{eq:ridge-optimization-problem} satisfies
\[
\hat v_{n,\delta}
\approx
\hat U_{11,\delta}
-\hat{\mathbf U}_{12,\delta}
(\hat{\mathbf U}_{22,\delta})^{-1}
\hat{\mathbf U}_{21,\delta}.
\]
Combining this observation with Corollary~\ref{thm:variance_reduction}, we have the approximation
\[
\Var_{\operatorname{asymp}}(\hat{\tau} \mid M_n \le a) \approx \sigma_{n,y}^2 V_{11}(1-R^2) =
U_{11}-\mathbf U_{12} \mathbf U_{22}^{-1} \mathbf U_{21},
\]
which motivates a plug-in estimator for $\Var_{\operatorname{asymp}}(\hat{\tau} \mid M_n \le a)$ when $a \approx 0$. More generally, $\Var_{\operatorname{asymp}}(\hat{\tau} \mid M_n \le a) =\sigma_{n,y}^2 V_{11} (1-R^2)+v_p\sigma_{n,y}^2V_{11}R^2$, while $\sigma_{n,y}^2V_{11}R^2\le\sigma_{n,y}^2V_{11}=U_{11}$ can be conservatively bounded by $\hat U_{11,\delta}$. Combining the two pieces yields the variance estimator
\[
\hat U_{11,\delta}(1+v_p)
-\hat{\mathbf U}_{12,\delta}
(\hat{\mathbf U}_{22,\delta})^{-1}
\hat{\mathbf U}_{21,\delta}.
\]

When $v_p$ is close to 0, it is often substantially less conservative than $\hat U_{11,\delta}$ and may provide a useful approximation to our optimization-based estimator, because the plug-in estimate used here is motivated by approximating the solution $\hat v_{n,\delta}$ in the optimization problem \eqref{eq:ridge-optimization-problem}, with reduced computational complexity. On the other hand, when $v_p$ is large, this estimator may become more conservative than the naive estimator $\hat U_{11,\delta}$.

Although this estimator is heuristic and does not currently enjoy the same theoretical guarantees as our proposed optimization-based procedure, it provides useful intuition and may serve as a computationally simple approximation in settings where $a$ is small. In practice, $a$ is typically chosen to be small. In our simulations, we set the acceptance probability to $p_a = 0.05$, with $a$ equal to the corresponding lower $p_a$-quantile of the $\chi_p^2$ distribution.

\paragraph{Relaxing the covariate covariance constraint}
\label{sec:alt_var_estimator_opt}
The rerandomization conditional variance takes the Schur-complement form
\[
U_{11}-(1-v_p)\,\mathbf U_{12}\mathbf U_{22}^{-1}\mathbf U_{21},
\]
in which the cross term $\mathbf U_{12}$ and the covariate block $\mathbf U_{22}$ are tightly coupled through the positive semidefinite constraint. A naive plug-in that replaces $\hat{\mathbf U}_{22,\delta}$ in $\hat{\mathbf U}_{n,\delta}$ with the known $\boldsymbol\Sigma_{xx}$ while leaving $\hat{\mathbf U}_{12,\delta}$ unchanged can produce a matrix that is no longer positive semidefinite, potentially yielding negative variance estimates.

Another possible alternative is to remove the constraint
\[
\boldsymbol\Omega_{22}=\boldsymbol\Sigma_{xx}
\]
and optimize over the entire covariance matrix, with the objective replaced by the Schur complement form
\[
\Omega_{11} - (1-v_p) \boldsymbol\Omega_{12}\boldsymbol\Omega_{22}^{-1}\boldsymbol\Omega_{21}.
\]
Here $\boldsymbol\Omega_{22}\succ\mathbf0$ is required for the inverse to exist. The relaxed problem is feasible whenever $\hat{\mathbf U}_{22,\delta}\succ\mathbf0$, since $\boldsymbol\Omega=\hat{\mathbf U}_{n,\delta}$ is then admissible. However, it no longer exploits the known covariance structure of the rerandomization criterion, where the covariate covariance is fixed at $\boldsymbol\Sigma_{xx}$. Moreover, allowing $\boldsymbol\Omega_{22}$ to vary may lead to substantially more conservative solutions, since enlarging $\boldsymbol\Omega_{22}$ can reduce the penalty term in the Schur complement. Therefore, we retain the constraint $\boldsymbol\Omega_{22}=\boldsymbol\Sigma_{xx}$ and address the resulting exact-feasibility issue through the ridge correction in Appendix~\ref{app:feasibility-refinement}; the simulations use the fixed scale-adaptive ridge in~\eqref{eq:scale_adaptive_ridge}.

A further alternative is to remove the constraint $\boldsymbol\Omega_{22}=\boldsymbol\Sigma_{xx}$ while retaining the known covariate covariance in the objective, leading to
\[
\Omega_{11} - (1-v_p) \boldsymbol\Omega_{12} \boldsymbol\Sigma_{xx}^{-1} \boldsymbol\Omega_{21}.
\]
This formulation has two potential advantages. First, it continues to exploit the known information in $\boldsymbol\Sigma_{xx}$ while guaranteeing feasibility. Second, the objective is continuous in $\boldsymbol\Omega$. The argument underlying Theorem~\ref{thm:conservative_matrix} shows that any failure of $\hat{\mathbf U}_n$ to dominate $\mathbf U_n$ is asymptotically negligible in the relevant standardized positive-semidefinite order. Consequently, a small finite-sample violation of conservativeness may not materially affect the optimized value under this continuous formulation. The main drawback, however, is that removing the constraint $\boldsymbol\Omega_{22}=\boldsymbol\Sigma_{xx}$ substantially enlarges the feasible set and allows the lower-right block to differ from the known covariate covariance. Maximization over this enlarged set may therefore produce an excessively conservative variance bound. For this reason, we ultimately adopt the constrained formulation presented in the main text.

\subsubsection{Alternative confidence intervals}

Under rerandomization,
\[
\hat\tau-\tau \;\Big|\; M_n \le a \;\dot\sim\; \bigl( U_{11} - \mathbf U_{12}
\boldsymbol\Sigma_{xx}^{-1} \mathbf U_{21} \bigr)^{1/2} \varepsilon_0 + \bigl( \mathbf U_{12}
\boldsymbol\Sigma_{xx}^{-1} \mathbf U_{21} \bigr)^{1/2} L_{p,a}.
\]
where
\[
\varepsilon_0 \sim N(0,1), \qquad L_{p,a} = J_1 \mid \mathbf J^\top \mathbf J \le a.
\]

This representation suggests several possible conservative confidence interval constructions.

\paragraph{Bernoulli Wald interval}
A first and particularly simple approach is to ignore the rerandomization structure entirely and use the conservative Bernoulli approximation
\[
\hat U_{11}^{1/2}\varepsilon_0.
\]
This leads to the standard Wald-type interval
\[
\hat\tau \pm z_{1-\alpha/2} \hat U_{11}^{1/2}.
\]

\paragraph{Direct quantile optimization}
A second approach is to directly optimize the quantile of the asymptotic rerandomization distribution,
\[
\bigl( U_{11} - \mathbf U_{12} \boldsymbol\Sigma_{xx}^{-1} \mathbf U_{21} \bigr)^{1/2} \varepsilon_0
+ \bigl( \mathbf U_{12} \boldsymbol\Sigma_{xx}^{-1} \mathbf U_{21} \bigr)^{1/2} L_{p,a},
\]
subject to the constraints in \eqref{eqn:optimization_problem}. In principle, this procedure may yield a sharper conservative interval than the alternatives considered above. However, the resulting optimization problem appears computationally challenging.

\paragraph{Plug-in rerandomization interval}
A third approach parallels the plug-in variance estimator above. Specifically, one may conservatively estimate the two variance components $ \sigma_{n,y}^2 V_{11}(1-R^2) \quad\text{and}\quad v_p \sigma_{n,y}^2 V_{11}R^2 $ separately as we mentioned above, leading to the conservative limiting distribution
\[
\bigl( \hat U_{11} - \hat{\mathbf U}_{12} \hat{\mathbf U}_{22}^{-1} \hat{\mathbf U}_{21}
\bigr)^{1/2} \varepsilon_0 + \hat U_{11}^{1/2} L_{p,a}.
\]
Confidence intervals can then be constructed using quantiles of this distribution.

Theoretically, this yields an asymptotically conservative confidence interval. Empirically, however, it can be substantially more conservative especially when $a$ is large.

Compared with the direct quantile optimization approach, this method is computationally much simpler. Moreover, when $v_p$ is small, it is often close to the direct quantile optimization approach because the explained component contributes little to the rerandomization variance. Consequently, the asymptotic distribution is dominated by the unexplained component, $\left( U_{11} - \mathbf U_{12} \boldsymbol\Sigma_{xx}^{-1} \mathbf U_{21} \right)^{1/2} \varepsilon_0,$ which is estimated directly by this approximation.

\subsubsection{Gaussian prepivoting}

Another alternative, closer in spirit to randomization inference, is the Gaussian prepivoting approach of \citet{cohen2022gaussian}. In rerandomized experiments without interference, their method first transforms a base statistic by the Gaussian approximation to its conditional distribution given the balance event, and then uses the randomization distribution of the prepivoted statistic under Fisher's sharp null. For example, for a two-sided test based on the treatment-effect estimator, a natural Gaussian-prepivoted statistic has the form
\[
G(\mathbf z;\tau_0) = \frac{ \Pr_{\hat{\boldsymbol\Omega}}\left\{ |A| \le |\hat{\tau}(\mathbf
z)-\tau_0|,\quad \mathbf B^\top \boldsymbol\Sigma_{xx}^{-1} \mathbf B \le a \right\} }{
\Pr_{\hat{\boldsymbol\Omega}_{22}}\left\{ \mathbf B^\top \boldsymbol\Sigma_{xx}^{-1} \mathbf B
\le a \right\} },
\]
where $(A,\mathbf B^\top)^\top$ is a mean-zero Gaussian vector with covariance $\hat{\boldsymbol\Omega}$, intended to approximate the joint distribution of the treatment-effect estimator and the covariate imbalance. The randomization distribution of $G(\mathbf w;\tau_0)$ over allowable assignments $\mathbf w$ is then used for testing. In this way, Gaussian prepivoting avoids estimating a scalar $R^2$ directly; instead, it uses an estimate of the full joint Gaussian law of the estimator and the imbalance vector.

This approach is closely related to the conditional Gaussian approximation used in our analysis. Indeed, if the true joint covariance of $(\hat{\tau},\tauxhat)$ were known, then the Gaussian-prepivoting reference distribution would be based on the same conditional law that appears in Theorem~\ref{thm:rerand_dist}. In that oracle case, a prepivoted test and the corresponding large-sample conditional-Gaussian test would have the same rejection rule. This observation helps explain why Gaussian prepivoting is a natural candidate for avoiding an explicit plug-in estimator of $R^2$.

However, Gaussian prepivoting is not directly applicable under our assumptions. First, its validity theory is developed for finite-population SUTVA settings with fixed covariates, using a sharp-null randomization distribution while obtaining asymptotic validity for weak nulls. A null hypothesis such as $\tau=\tau_0$ does not impute all missing potential outcomes under interference, so the exact Fisher-randomization component of their method does not transfer directly. Second, applying the large-sample Gaussian component of their method would require a covariance estimator that can be validly inserted into the conditional Gaussian law of $(\hat{\tau},\tauxhat)$ given the rerandomization event. Their theory relies on a blockwise covariance-estimation condition, whereas our Theorem~\ref{thm:conservative_matrix} provides a joint PSD conservative bound.

A joint PSD conservative covariance bound does yield a conservative plug-in variance bound, as discussed in Section~\ref{sec:opt}, but this alone does not establish the validity of the full Gaussian-prepivoting procedure. The optimization in Definition~\ref{def:rerand-variance-estimator} uses the known covariate-imbalance covariance $\boldsymbol\Sigma_{xx}$ to sharpen the variance bound: it takes the worst-case conditional variance over all covariance matrices compatible with both this known block and the joint conservative bound. Our confidence intervals then use the worst-case critical value from Section~\ref{sec:CI}.

\subsection{Ridge correction for feasibility}
\label{app:feasibility-refinement}

Theorem~\ref{thm:conservative_matrix} establishes covariance dominance only up to a vanishing eigenvalue error after standardization, rather than exact finite-sample domination $\hat{\mathbf U}_n\succeq\mathbf U_n$. This distinction creates two difficulties for the optimization in~\eqref{eqn:optimization_problem}. If $\hat{\mathbf U}_{22}\not\succeq\boldsymbol\Sigma_{xx}$, the known covariate block cannot satisfy the upper-bound constraint, so the feasible set is empty. Even when the program is feasible, it may exclude the true covariance matrix $\mathbf U_n$, as discussed in Appendix~\ref{app:variance-optimization-stability}.

Theorem~\ref{thm:asymp_conservativeness} establishes asymptotic conservativeness despite these possibilities, using the fallback and Lemma~\ref{lem:near-dominance-optimization}. Here we study the scale-adaptive ridge correction introduced in Section~\ref{sec:opt} to improve feasibility while preserving this guarantee. We first consider a fixed positive correction, which restores full covariance dominance with probability tending to one, and then a minimal data-dependent correction that guarantees feasibility for every realization.

\begin{theorem}[Fixed ridge correction]
\label{thm:fixed-ridge-feasibility}
Suppose that the conditions of Theorem~\ref{thm:conservative_matrix} hold and that $\boldsymbol\Sigma_{xx}\succ\mathbf0$. For every fixed $\delta>0$,
\[
\mathbb P\left\{
\hat{\mathbf U}_{n,\delta}\succeq\mathbf U_n
\,\middle|\,M_n\le a
\right\}\to1.
\]
Consequently,
\[
\mathbb P\left\{
\hat{\mathbf U}_{22,\delta}\succeq\boldsymbol\Sigma_{xx}
\,\middle|\,M_n\le a
\right\}\to1,
\]
so the ridge-corrected optimization problem is feasible with probability tending to one. The corresponding unconditional statements also hold.
\end{theorem}

A fixed $\delta$ is simple, but it introduces a nonvanishing relative correction. We next define a data-adaptive choice that adds only the amount needed for feasibility. Let
\[
\mathbf S_{x,n}
=\operatorname{diag}(\Sigma_{xx,11},\ldots,\Sigma_{xx,pp})
\]
denote the lower-right block of $\hat{\mathbf S}_n$.

\begin{definition}[Minimal feasibility correction]
\label{def:minimal-feasibility-correction}
Assume that $\boldsymbol\Sigma_{xx}\succ\mathbf0$, and define
\begin{equation}
\label{eq:minimal-feasibility-correction}
\hat\delta_n^{\mathrm{feas}}
=
\left[
-\lambda_{\min}\left\{
\mathbf S_{x,n}^{-1/2}
(\hat{\mathbf U}_{22}-\boldsymbol\Sigma_{xx})
\mathbf S_{x,n}^{-1/2}
\right\}
\right]_+,
\end{equation}
where $[x]_+=\max\{x,0\}$. Define
\[
\hat{\mathbf U}_n^{\mathrm{feas}}
=\hat{\mathbf U}_{n,\hat\delta_n^{\mathrm{feas}}},
\qquad
\hat v_n^{\mathrm{feas}}
=\hat v_{n,\hat\delta_n^{\mathrm{feas}}}.
\]
\end{definition}

The correction in~\eqref{eq:minimal-feasibility-correction} is the smallest nonnegative scalar within the scale-adaptive family~\eqref{eq:scale_adaptive_ridge} that makes the optimization feasible. Indeed, for every $\delta\ge0$,
\[
\hat{\mathbf U}_{22,\delta}\succeq\boldsymbol\Sigma_{xx}
\quad\Longleftrightarrow\quad
\delta\ge\hat\delta_n^{\mathrm{feas}}.
\]
Thus the program based on $\hat{\mathbf U}_n^{\mathrm{feas}}$ is feasible for every realization.

The next result shows that the minimal correction coefficient converges to zero in probability.

\begin{proposition}[Vanishing minimal correction]
\label{prop:minimal-feasibility-correction-vanishes}
Suppose that the conditions of Theorem~\ref{thm:conservative_matrix} hold and that $\boldsymbol\Sigma_{xx}\succ\mathbf0$. Then
\[
\hat\delta_n^{\mathrm{feas}}=o_p(1)
\]
both under the Bernoulli assignment law and under the conditional rerandomization law given $M_n\le a$. Equivalently, for every $\eta>0$,
\[
\mathbb P\{\hat\delta_n^{\mathrm{feas}}>\eta\}\to0,
\qquad
\mathbb P\{\hat\delta_n^{\mathrm{feas}}>\eta\mid M_n\le a\}\to0.
\]
\end{proposition}

The minimal correction guarantees feasibility for every realization, but need not make the corrected covariance bound dominate $\mathbf U_n$. Since this correction is data-dependent and may vanish, Theorem~\ref{thm:fixed-ridge-feasibility} does not apply directly. The following result establishes asymptotic conservativeness for this choice and, more generally, for any nonnegative ridge correction. If the correction leaves the program infeasible, the estimator retains the fallback specified after~\eqref{eq:ridge-optimization-problem}.

\begin{theorem}[Conservativeness under nonnegative corrections]
\label{thm:adaptive-ridge-conservativeness}
\label{thm:feasibility-refinements}
Suppose that the conditions of Theorem~\ref{thm:conservative_matrix} hold and that $\boldsymbol\Sigma_{xx}\succ\mathbf0$. Let $\tilde\delta_n$ be any almost surely finite, nonnegative, possibly data-dependent correction level. Then, for every $\eta>0$,
\[
\mathbb P\left\{
\hat v_{n,\tilde\delta_n}
\ge v_n-\eta\sigma_{n,y}^2
\,\middle|\,M_n\le a
\right\}\to1.
\]
The corresponding unconditional statement also holds. In particular, the conclusion applies to $\hat v_n^{\mathrm{feas}}$. The confidence-interval conclusion of Theorem~\ref{thm:ci_coverage} remains valid after replacing $\hat v_n$ by any of these corrected bounds.
\end{theorem}

\paragraph{Computation after the ridge correction.}
The derivation in the proof of Theorem~\ref{thm:opt_solution} is deterministic and therefore applies to fixed or data-dependent correction levels. Whenever the corrected program~\eqref{eq:ridge-optimization-problem} is feasible, its optimal value is given by that theorem after replacing each block of $\hat{\mathbf U}_n$ by the corresponding block of $\hat{\mathbf U}_{n,\delta}$ and defining $\mathbf M(t)$, $G(t)$, and $t^\star$ with those corrected blocks. The known matrix $\boldsymbol\Sigma_{xx}$ is unchanged. In particular, this formula applies for every realization when $\delta\ge\hat\delta_n^{\mathrm{feas}}$.

\subsection{Proofs for Section~\ref{app:feasibility-refinement}}
\label{sec:prf_feasible_refinement}

\subsubsection{Proof of Theorem~\ref{thm:fixed-ridge-feasibility}}

\begin{proof}
All statements below are proved under both the unconditional Bernoulli law and the conditional law given $M_n\le a$. Write
\[
\mathbf V_n=\mathbf D_n\mathbf U_n\mathbf D_n,
\qquad
\mathbf A_n=\mathbf D_n\hat{\mathbf U}_n\mathbf D_n,
\qquad
r_n=
\bigl[-\lambda_{\min}(\mathbf A_n-\mathbf V_n)\bigr]_+.
\]
Theorem~\ref{thm:conservative_matrix} and $\|\mathbf V_n-\mathbf V\|_{\mathrm{op}}\to0$ imply $r_n=o_p(1)$ under the unconditional law. Corollary~\ref{cor:conditional_psd_conservative} gives the same conclusion conditionally on $M_n\le a$.

Let
\[
\mathbf T_n=\mathbf D_n\hat{\mathbf S}_n\mathbf D_n
=\operatorname{diag}\!\left(
\frac{\hat U_{11}}{\sigma_{n,y}^2},
\frac{\Sigma_{xx,11}}{\sigma_{n,x_1}^2},\ldots,
\frac{\Sigma_{xx,pp}}{\sigma_{n,x_p}^2}
\right).
\]
Assumption~\ref{ass:variance_stability} gives $\Sigma_{xx,jj}/\sigma_{n,x_j}^2\to1$ for every $j$. Moreover,
\[
\frac{\hat U_{11}}{\sigma_{n,y}^2}
=(\mathbf A_n)_{11}
\ge (\mathbf V_n)_{11}-r_n
=1+o_p(1).
\]
It follows that
\[
\mathbb P\left\{\mathbf T_n\succeq\frac12\mathbf I_{d_\phi}\right\}\to1
\]
under both laws. For a fixed $\delta>0$, on the event $\{\mathbf T_n\succeq\mathbf I_{d_\phi}/2,\ r_n\le\delta/2\}$,
\begin{align*}
\mathbf D_n(\hat{\mathbf U}_{n,\delta}-\mathbf U_n)\mathbf D_n
&=\mathbf A_n-\mathbf V_n+\delta\mathbf T_n\\
&\succeq(-r_n+\delta/2)\mathbf I_{d_\phi}
\succeq\mathbf0.
\end{align*}
The probability of this event tends to one. Congruence by $\mathbf D_n^{-1}$ proves $\hat{\mathbf U}_{n,\delta}\succeq\mathbf U_n$ with probability tending to one. Since $(\mathbf U_n)_{22}=\boldsymbol\Sigma_{xx}$, the lower-right block implication gives $\hat{\mathbf U}_{22,\delta}\succeq\boldsymbol\Sigma_{xx}$, and Theorem~\ref{thm:opt_solution} then gives feasibility.
\end{proof}

\subsubsection{Proof of Proposition~\ref{prop:minimal-feasibility-correction-vanishes}}

\begin{proof}
Use the notation from the proof of Theorem~\ref{thm:fixed-ridge-feasibility}, and let
\[
\mathbf E_n
=\mathbf D_{n,x}
(\hat{\mathbf U}_{22}-\boldsymbol\Sigma_{xx})
\mathbf D_{n,x}.
\]
Because $\mathbf E_n$ is the lower-right principal block of $\mathbf A_n-\mathbf V_n$,
\[
\mathbf E_n\succeq-r_n\mathbf I_p.
\]
Let
\[
\mathbf H_n
=\mathbf D_{n,x}\mathbf S_{x,n}^{1/2}
=\operatorname{diag}\!\left(
\frac{\sqrt{\Sigma_{xx,11}}}{\sigma_{n,x_1}},\ldots,
\frac{\sqrt{\Sigma_{xx,pp}}}{\sigma_{n,x_p}}
\right).
\]
Assumption~\ref{ass:variance_stability} implies $\mathbf H_n\to\mathbf I_p$. Moreover,
\begin{align*}
&\mathbf S_{x,n}^{-1/2}
(\hat{\mathbf U}_{22}-\boldsymbol\Sigma_{xx})
\mathbf S_{x,n}^{-1/2}\\
&\qquad=
\mathbf H_n^{-1}\mathbf E_n\mathbf H_n^{-1}
\succeq-r_n\mathbf H_n^{-2}.
\end{align*}
Consequently,
\[
0\le\hat\delta_n^{\mathrm{feas}}
\le r_n\|\mathbf H_n^{-2}\|_{\mathrm{op}}.
\]
Since $r_n=o_p(1)$ and $\|\mathbf H_n^{-2}\|_{\mathrm{op}}\to1$, the desired conclusion holds both unconditionally and conditionally on $M_n\le a$.
\end{proof}

\subsubsection{Proof of Theorem~\ref{thm:adaptive-ridge-conservativeness}}

\begin{proof}
Write
\[
\mathbf A_n=\mathbf D_n\hat{\mathbf U}_n\mathbf D_n,
\qquad
\widetilde{\mathbf A}_n
=\mathbf D_n\hat{\mathbf U}_{n,\tilde\delta_n}\mathbf D_n,
\qquad
\mathbf V_n=\mathbf D_n\mathbf U_n\mathbf D_n,
\]
and let
\[
\mathbf B_n=(\mathbf V_n)_{22}
=\mathbf D_{n,x}\boldsymbol\Sigma_{xx}\mathbf D_{n,x}.
\]
Because $\tilde\delta_n\ge0$ and $\hat{\mathbf S}_n\succeq\mathbf0$, $\widetilde{\mathbf A}_n\succeq\mathbf A_n$. Therefore, with
\[
r_n=
\bigl[-\lambda_{\min}(\mathbf A_n-\mathbf V_n)\bigr]_+,
\qquad
\widetilde r_n=
\bigl[-\lambda_{\min}(\widetilde{\mathbf A}_n-\mathbf V_n)\bigr]_+,
\]
we have $0\le\widetilde r_n\le r_n$ for every realization. As in the proof of Theorem~\ref{thm:asymp_conservativeness}, $r_n=o_p(1)$ both unconditionally and conditionally on $M_n\le a$.

Congruence by $\mathbf D_n$ transforms the corrected optimization problem and its plug-in fallback into the standardized construction in Lemma~\ref{lem:near-dominance-optimization}. Hence
\[
\frac{\hat v_{n,\tilde\delta_n}}{\sigma_{n,y}^2}
=\widetilde\Phi_{\mathbf B_n}(\widetilde{\mathbf A}_n)
\ge q_{\mathbf B_n}(\mathbf V_n)-C\widetilde r_n
\]
with probability tending to one under either law. Assumption~\ref{ass:variance_stability} gives $\mathbf V_n\to\mathbf V\succ\mathbf0$, so
\[
q_{\mathbf B_n}(\mathbf V_n)
\to
V_{11}-(1-v_p)\mathbf V_{12}\mathbf V_{22}^{-1}\mathbf V_{21}
=\frac{v_n}{\sigma_{n,y}^2}.
\]
This proves the stated conservativeness result.

The conclusion for $\hat v_n^{\mathrm{feas}}$ follows by taking $\tilde\delta_n=\hat\delta_n^{\mathrm{feas}}$. The confidence-interval conclusion follows from the proof of Theorem~\ref{thm:ci_coverage}, which uses only the asymptotic lower bound on the variance estimator.
\end{proof}

\section{Simulation Details}
\label{appendx:simulation}

This appendix provides full implementation details for the simulation study summarized in Section~\ref{appendx:simulation_main}. We describe the complete data-generating process, all parameter specifications, the full simulation protocol, and detailed regime-specific analysis of results. All simulations are implemented in Python.

\subsection{Data generating process}
\label{sec:sim_dgp}

\paragraph{Network Topology.}
We generate an undirected physical network $\mathcal G_n^{\mathrm{BA}}=([n],E_n^{\mathrm{BA}})$ with $n = 2000$ units using the Barab\'{a}si--Albert preferential-attachment model \citep{albert1999emergence}. Starting from a complete graph on $m + 1 = 6$ seed units, we sequentially add new units one at a time; each new unit attaches $m = 5$ edges to existing units, with the probability of attaching to unit $j$ proportional to its current degree (``rich-get-richer''). This procedure produces a network whose degree distribution obeys a power law, $P(d) \propto d^{-3}$, so that a small number of high-degree hub units coexist with a long tail of low-degree peripheral units. This heterogeneous structure is characteristic of real social, information, and biological networks, and contrasts sharply with random-graph or small-world models in which the degree distribution is tightly concentrated around its mean. Let $\mathbf A_n^{\mathrm{BA}}\in\{0,1\}^{n\times n}$ denote the resulting adjacency matrix, with $A_{n,ij}^{\mathrm{BA}}=1$ if and only if units $i$ and $j$ are connected or $i=j$, and 0 otherwise. The graph $\mathcal G_n^{\mathrm{BA}}$ is held fixed across all simulation replications. Under this parameterization, the expected mean number of non-self neighbors is approximately $2m = 10$, while the maximum degree grows as $\mathcal O(m\sqrt{n})$, producing the degree heterogeneity that most stresses the localized variance-bounding machinery developed in Section~\ref{sec:conservative}. For variance estimation, let $\mathcal G_n^{\mathrm{dep}}$ denote the dependency graph induced by the data-generating process. In the simulations, the practitioner graph is taken to be exact, so $\mathcal G_n=\mathcal G_n^{\mathrm{dep}}$.

\paragraph{Treatment Assignment.}
Treatment is assigned via independent Bernoulli randomization,
\[
Z_i\overset{\mathrm{i.i.d.}}{\sim}\operatorname{Bern}(\pi),
\qquad \pi=0.5,
\]
subject to the rerandomization acceptance criterion described below.

\paragraph{Covariates.}
For each unit $i$, we generate the following covariates:
\begin{itemize}
\item \textbf{Observed Pre-Treatment Covariates} ($\mathbf X_i \in \mathbb{R}^{p_s}$, $p_s = 6$): Drawn from $\mathbf X_i \sim \mathcal{N}(\boldsymbol 0, \sigma_x^2 \mathbf{I}_{p_s})$ with $\sigma_x = 0.316$, making it comparable with treatment-dependent covariates, giving each component a variance of approximately $0.1$. These covariates are observable pre-treatment and available for rerandomization.

\item \textbf{Hidden Covariate} ($X_{i, \text{hidden}} \in \mathbb{R}$): Drawn independently from the same distribution as the pre-treatment covariates. This variable influences the outcome but is not used in the rerandomization criterion, thereby introducing realistic unobserved heterogeneity and testing the method's robustness to incomplete covariate information.
\end{itemize}

\paragraph{Treatment-Dependent Covariates.}
To capture the network interference structure, we construct treatment-dependent covariates that depend jointly on the treatment vector $\mathbf Z$, the adjacency matrix $\mathbf A_n^{\mathrm{BA}}$, and the pre-treatment covariates $\mathbf X$. Let $\mathcal N_i = \{j \neq i : A_{n,ij}^{\mathrm{BA}} = 1\}$ denote the set of direct neighbors of unit $i$ (excluding $i$ itself). We define the following interference summaries:
\begin{align}
\text{Neighbor Treatment Proportion:} \quad & H_{i, \text{prop}} = \frac{1}{|\mathcal{N}_i|} \sum_{j
\in \mathcal{N}_i} Z_j, \label{eq:sim_h_prop} \\[4pt]
\text{Neighbor-Weighted Covariates:} \quad & \mathbf H_{i, \text{wX}} = \frac{1}{|\mathcal{N}_i|}
\sum_{j \in \mathcal{N}_i} Z_j \mathbf X_j, \label{eq:sim_h_wx}
\end{align}

The covariates and outcomes used in our simulations do not strictly satisfy the uniform boundedness condition in Assumption~\ref{ass:boundedness}. Moreover, the simulated dependency graph is quite dense: at $n=2{,}000$, its largest degree is approximately $1{,}673$. We therefore do not design this finite-sample stress test to satisfy the polynomial sparsity requirement in Assumption~\ref{ass:network_sparsity}. These choices are made primarily for implementation and computational convenience and to provide a demanding assessment of the proposed method. Strong performance in this setting provides reassuring evidence---although not a formal theoretical implication---for performance in more regular settings that fully satisfy our assumptions.

\subsection{Outcome models}
\label{sec:sim_outcome}

We evaluate performance under seven outcome specifications that vary along two axes: the functional form of the pre-treatment block and the structure of the interference block. The general outcome for unit~$i$ takes the form
\begin{equation}
\label{eq:outcome_general}
\begin{aligned}
Y_i
={}&\beta_0+\tau Z_i
+\beta_{\text{base}}\,f(\mathbf X_i,\boldsymbol\beta_{\text{xbase}})
+X_{i,\text{hidden}}\beta_{\text{hidden}}\\
&+\beta_{\text{interf}}\,
 g_i(\mathbf Z,\mathbf A_n^{\mathrm{BA}},\mathbf X;
 \boldsymbol\beta_{\text{xinterf}})
+\varepsilon_i.
\end{aligned}
\end{equation}
We fix $\beta_0=0.5$ and $\tau=2.0$, where $\tau$ is the estimand of interest. The coefficient $\beta_{\text{hidden}}$ is drawn from $\mathrm{Uniform}(0.5,2.0)$ with a random sign and held fixed across replications. The noise variables satisfy $\varepsilon_i\overset{\mathrm{i.i.d.}}{\sim}\mathcal N(0,1)$. The functions $f(\cdot)$ and $g_i(\cdot)$ specify the pre-treatment and interference blocks; their detailed forms are given below.

\paragraph{Signal composition.}
The scalars $\beta_{\text{base}}$ and $\beta_{\text{interf}}$ are outer multipliers that control the relative contribution of the pre-treatment and treatment-dependent components. Let $\text{std}_{\text{base}} = \text{SD}\bigl(f(\mathbf X_i, \boldsymbol\beta_{\text{xbase}})\bigr)$ and $\text{std}_{\text{interf}} = \text{SD}\bigl(g_i(\mathbf Z, \mathbf A_n^{\mathrm{BA}}, \mathbf X; \boldsymbol\beta_{\text{xinterf}})\bigr)$ denote the standard deviations of the two blocks across units, where the latter is additionally averaged over the \emph{unconstrained} Bernoulli distribution of $\mathbf Z$ (so that the signal ratio characterizes the data-generating process and does not depend on the choice of rerandomization criterion). The quantity $\text{std}_{\text{base}}$ is computed exactly from the realized $\mathbf X$ and the fixed $\boldsymbol\beta_{\text{xbase}}$; $\text{std}_{\text{interf}}$ is computed by Monte Carlo with $N_{\text{std}} = 2{,}000$ Bernoulli draws. We define the \emph{signal ratio}
\begin{equation}
\label{eq:signal_ratio}
\kappa = \frac{\beta_{\text{interf}} \cdot \text{std}_{\text{interf}}}{\beta_{\text{base}} \cdot
\text{std}_{\text{base}}},
\end{equation}
and, given a target $\log_2\kappa$, choose $(\beta_{\text{base}}, \beta_{\text{interf}})$ to satisfy~\eqref{eq:signal_ratio} together with the constant total-signal constraint $(\beta_{\text{base}} \cdot \text{std}_{\text{base}})^2 + (\beta_{\text{interf}} \cdot \text{std}_{\text{interf}})^2 = C^2$ with $C = 30$. As $\log_2\kappa$ varies from $-4$ to $4$ in unit steps, the dominant source of outcome variation smoothly transitions from the pre-treatment covariates to the interference effects, while the total signal magnitude remains fixed.

\paragraph{Pre-treatment blocks.}
We use two functional forms for the pre-treatment block~$f$, indexed by a label $\mathsf{L}$ (linear) or $\mathsf{E}$ (exponential):
\begin{align*}
\text{(linear, }\mathsf{L}\text{):} \quad & f(\mathbf X_i) = \mathbf X_i^\top
\boldsymbol\beta_{\text{xbase}}, \\
\text{(exponential, }\mathsf{E}\text{):} \quad & f(\mathbf X_i) = \exp\!\big(\mathbf X_i^\top
\boldsymbol\beta_{\text{xbase}}\big).
\end{align*}
The vector $\boldsymbol\beta_{\text{xbase}} \in \mathbb{R}^{p_s}$ is drawn once from $\mathrm{Uniform}(0.5, 2.0)^{p_s}$ with each entry given a random sign from $\mathrm{Uniform}\{-1,+1\}$, and is held fixed across all signal ratios and replications.

\paragraph{Interference blocks.}
We consider three interference structures, denoted by descriptive identifiers that name the interference summaries appearing in $g_i$:
\begin{align}
\mathsf{prop}: \quad & g_i = \beta_{\text{prop}} \, H_{i,\text{prop}}, \label{eq:gprop} \\[0pt]
\mathsf{prop\text{+}NWX}: \quad & g_i = \beta_{\text{prop}} \, H_{i,\text{prop}} + \mathbf
H_{i,\text{wX}}^\top \boldsymbol\beta_{\text{xinterf}}, \label{eq:gpropnwx} \\[0pt]
\mathsf{NWX}: \quad & g_i = \mathbf H_{i,\text{wX}}^\top \boldsymbol\beta_{\text{xinterf}},
\label{eq:gnwx}
\end{align}
plus four additional structures that pair the exponential pre-treatment block with exponential interference:
\begin{align}
\mathsf{exp\text{-}prop}: \quad & g_i = \beta_{\text{prop}} \, \exp\!\big(H_{i,\text{prop}}\big),
\label{eq:gexpprop}\\[0pt]
\mathsf{sumexp}: \quad & g_i = \beta_{\text{prop}} \sum_{j \in \mathcal{N}_i}\! \exp\!\big(Z_j /
|\mathcal{N}_i|\big), \label{eq:gsumexp}\\[0pt]
\mathsf{exp\text{-}prop\text{+}NWX}: \quad & g_i = \beta_{\text{prop}}
\exp\!\big(H_{i,\text{prop}}\big) + \exp\!\big(\mathbf H_{i,\text{wX}}\big)^\top
\boldsymbol\beta_{\text{xinterf}}, \label{eq:gexpnwx} \\[0pt]
\mathsf{sumexp\text{-}full}: \quad & g_i = \beta_{\text{prop}}\!\sum_{j\in\mathcal{N}_i}\!
\exp\!\big(Z_j/|\mathcal{N}_i|\big) + \!\sum_{j\in\mathcal{N}_i}\!\exp\!\big(Z_j \mathbf
X_j/|\mathcal{N}_i|\big)^{\!\top}\!\boldsymbol\beta_{\text{xinterf}}. \label{eq:gsumexpfull}
\end{align}
In all cases, $\beta_{\text{prop}}$ is fixed at $5.0$, $\boldsymbol\beta_{\text{xinterf}}$ (when present) is drawn once from $\mathrm{Uniform}(0.5, 2.0)^{p_s}$ with random signs and held fixed, and the exponential acting on a vector is interpreted entry-wise.

\paragraph{The seven outcome models.}
Pairing each pre-treatment block with each interference block yields the following seven specifications, summarized in Table~\ref{tab:outcome_models}. The three ``linear'' models combine $f(\mathbf X_i) = \mathbf X_i^\top \boldsymbol\beta_{\text{xbase}}$ with the three interference structures~\eqref{eq:gprop}--\eqref{eq:gnwx}; the four ``exponential'' models combine $f(\mathbf X_i) = \exp(\mathbf X_i^\top \boldsymbol\beta_{\text{xbase}})$ with~\eqref{eq:gexpprop}--\eqref{eq:gsumexpfull}.

\begin{table}[ht]
\centering\small
\begin{tabular}{l l l}
\toprule
Model & Pre-treatment block $f$ & Interference block $g_i$ \\
\midrule
\texttt{lin-prop}              & $\mathbf X_i^\top\boldsymbol\beta_{\text{xbase}}$    & $\beta_{\text{prop}} H_{i,\text{prop}}$ \\
\texttt{lin-prop+nwx}          & $\mathbf X_i^\top\boldsymbol\beta_{\text{xbase}}$    & $\beta_{\text{prop}} H_{i,\text{prop}} + \mathbf H_{i,\text{wX}}^\top\boldsymbol\beta_{\text{xinterf}}$ \\
\texttt{lin-nwx}               & $\mathbf X_i^\top\boldsymbol\beta_{\text{xbase}}$    & $\mathbf H_{i,\text{wX}}^\top\boldsymbol\beta_{\text{xinterf}}$ \\
\texttt{exp-prop}              & $\exp(\mathbf X_i^\top\boldsymbol\beta_{\text{xbase}})$ & $\beta_{\text{prop}} \exp(H_{i,\text{prop}})$ \\
\texttt{exp-sumexp}            & $\exp(\mathbf X_i^\top\boldsymbol\beta_{\text{xbase}})$ & $\beta_{\text{prop}}\sum_{j\in\mathcal N_i} \exp(Z_j/|\mathcal{N}_i|)$ \\
\texttt{exp-prop+nwx}          & $\exp(\mathbf X_i^\top\boldsymbol\beta_{\text{xbase}})$ & $\beta_{\text{prop}}\exp(H_{i,\text{prop}}) + \exp(\mathbf H_{i,\text{wX}})^\top\boldsymbol\beta_{\text{xinterf}}$ \\
\texttt{exp-sumexp-full}       & $\exp(\mathbf X_i^\top\boldsymbol\beta_{\text{xbase}})$ & $\beta_{\text{prop}}\sum_{j\in\mathcal N_i} \exp(Z_j/|\mathcal{N}_i|) + \sum_{j\in\mathcal N_i} \exp(Z_j\mathbf X_j/|\mathcal{N}_i|)^\top\boldsymbol\beta_{\text{xinterf}}$ \\
\bottomrule
\end{tabular}
\caption{The seven outcome specifications. ``Linear'' models combine the linear pre-treatment block with one of three interference structures; ``exponential'' models pair the exponential pre-treatment block with the corresponding exponential interference structures. We refer to each model by its descriptive identifier (e.g.\ \texttt{lin-prop}) throughout the appendix.}
\label{tab:outcome_models}
\end{table}

The body of the paper presents \texttt{lin-prop} as the linear setting~\eqref{eq:outcome_linear} and \texttt{exp-prop} as the nonlinear setting~\eqref{eq:outcome_nonlinear}; the other five models are reported in this appendix.

\subsection{Simulation protocol}
\label{sec:sim_protocol}

The following protocol underlies both simulation studies. The variance-reduction study is described in Section~\ref{sec:sim_vr}, and the conservative-variance study is described in Section~\ref{sec:sim_conserv}. For each outcome model and each target value of $\log_2\kappa$ on the grid
\[
\{-4,-3,\ldots,3,4\},
\]
we perform the steps below. Quantities marked \emph{exact} are computed deterministically once $\mathbf X$ and the network $\mathcal G_n^{\mathrm{BA}}$ are fixed; quantities marked \emph{Monte Carlo} are estimated from explicit treatment-assignment draws as described.

\paragraph{Step 1: covariate covariance and acceptance threshold.}
We estimate the population covariance matrix of the H\'{a}jek covariate imbalance vector,
\begin{equation*}
\boldsymbol\Sigma_{xx} := \Var_{\text{Bern}}(\tauxhat),
\end{equation*}
by Monte Carlo with $N_{\text{cov}} = 5{,}000$ independent Bernoulli draws of $\mathbf Z$ under the unconstrained design (used as $\hat{\boldsymbol\Sigma}_{xx}$ throughout). The acceptance threshold for each strategy is set \emph{exactly} to $a = F^{-1}_{\chi^2_p}(p_a)$ with $p_a = 0.05$, where $p$ is the dimension of the covariate vector entering the Mahalanobis criterion. The variance-correction factor $v_p = F_{\chi^2_{p+2}}(a) / F_{\chi^2_p}(a)$ is also computed \emph{exactly} from the $\chi^2$ distribution. For all reported variance calculations, we prespecify the scale-adaptive ridge at $\delta=0.01$, corresponding to a one-percent inflation on each marginal-variance scale.

\paragraph{Step 2: rerandomization criteria.}
We compare three acceptance criteria, each calibrated to acceptance rate $p_a = 0.05$:
\begin{itemize}
\item \textbf{Pre-treatment strategy:} accepts $\mathbf Z$ if the Mahalanobis distance based on the $p_s=6$ pre-treatment covariates satisfies $M_n^{(\text{pre})}\le F^{-1}_{\chi^2_{p_s}}(p_a)$.
\item \textbf{Treatment-dependent strategy:} accepts if $M_n^{(\text{trt})}\le F^{-1}_{\chi^2_{p_d}}(p_a)$. The treatment-dependent block is \texttt{NP} ($H_{i,\text{prop}}$ alone, $p_d=1$), \texttt{NWX} ($\mathbf H_{i,\text{wX}}$ alone, $p_d=p_s=6$), or \texttt{NP+NWX} ($p_d=7$).
\item \textbf{Joint strategy:} accepts if $M_n^{(\text{joint})}\le F^{-1}_{\chi^2_{p_s+p_d}}(p_a)$, using the concatenation of pre-treatment and treatment-dependent covariates.
\end{itemize}

\paragraph{Step 3: rerandomization sampling.}
For each strategy and each signal ratio, we sample treatment assignments by acceptance--rejection: propose $\mathbf Z \sim \mathrm{Bern}(\pi)^{\otimes n}$ and accept if the corresponding Mahalanobis distance satisfies the criterion. We collect $M = 2{,}000$ accepted draws (\emph{Monte Carlo}). The crude (unconstrained) Bernoulli benchmark $\Var(\hat\tau)$ is estimated separately from $N_{\text{Bern}} = 10{,}000$ proposed draws (\emph{Monte Carlo}), regardless of acceptance.

\paragraph{Step 4: performance measures.}
We report four quantities, all estimated by Monte Carlo from the $M$ rerandomization-accepted draws and (where applicable) the $N_{\text{Bern}}$ crude Bernoulli draws produced in Step~3. Below, $\Var(\cdot)$, $\E[\cdot]$, and $\Pr\{\cdot\}$ denote the population quantities under the relevant assignment distribution; each is replaced by its Monte Carlo sample analogue when actually evaluated.
\begin{itemize}
\item \textbf{Variance ratio} (Section~\ref{sec:sim_vr}):
\[
\mathrm{VR} \;=\; \frac{\Var(\hat\tau \mid M_n \le a)}{\Var(\hat\tau)}.
\]
The numerator is estimated by the sample variance of $\hat\tau$ over the $M$ rerandomization-accepted draws; the denominator, by the sample variance over the $N_{\text{Bern}}$ Bernoulli draws.

\item \textbf{Confidence interval length ratio} $\rho$ (Section~\ref{sec:sim_conserv}, body Eq.~\eqref{eq:ci_length_ratio}):
\[
\rho \;=\; \frac{M^{-1}\sum_{m=1}^{M}\hat L_{\mathrm{ReM}}^{(m)}}
{N_{\text{Bern}}^{-1}\sum_{b=1}^{N_{\text{Bern}}}\hat L_{\mathrm{Bern}}^{(b)}}.
\]
Here $\hat L_{\mathrm{ReM}} =z_{1-\alpha/2}\,(\hat v_{n,\delta})^{1/2}$ is the rerandomization confidence interval half-length from~\eqref{eq:ci_half_length}, evaluated on each of the $M$ rerandomization-accepted draws; on each draw $\hat v_{n,\delta}$ is computed from the per-draw ridge-adjusted bound $\hat{\mathbf U}_{n,\delta}$ based on Definition~\ref{def:Un_estimator}, $\hat{\boldsymbol\Sigma}_{xx}$, and $v_p$ via the computation rule in Appendix~\ref{app:feasibility-refinement}, with the infeasibility fallback specified after~\eqref{eq:ridge-optimization-problem}. The denominator averages the Bernoulli half-length $\hat L_{\mathrm{Bern}} = z_{1-\alpha/2}\,\hat V_{\mathrm{Bern}}^{1/2}$ over the $N_{\text{Bern}}$ unconstrained Bernoulli draws, where $\hat V_{\mathrm{Bern}}$ is the $(1,1)$ entry of $\hat{\mathbf U}_n$. Following the discussion in Section~\ref{sec:CI}, we use the Gaussian critical value $z_{1-\alpha/2}$ as a numerical approximation to $\bar q_{1-\alpha/2}(p,a)$; over the values of $p$ and $a$ used in these experiments, the two differ by at most $0.2\%$.

\item \textbf{Coverage probability} (Section~\ref{sec:sim_conserv}):
\[
\mathrm{Cov} \;=\; \Pr\!\Bigl\{\bigl|\hat\tau - \tau\bigr| \le z_{1-\alpha/2}\sqrt{\hat
v_{n,\delta}}\,\Bigm|\, M_n \le a\Bigr\},
\]
estimated by the Monte Carlo frequency of CI-coverage events over the $M$ rerandomization-accepted draws, using the true treatment effect $\tau = 2.0$ from~\eqref{eq:outcome_general} and $\alpha = 0.05$.

\item \textbf{Conservativeness ratio on the log scale} $\mathrm{CR}^{\mathrm m}$ (body Eq.~\eqref{eq:conserv_ratio_log}; reported in Section~\ref{sec:sim_comparison_full} for $m \in \{\mathrm{ours},\mathrm{uncen},\mathrm{spec}\}$, joint strategy only):
\[
\mathrm{CR}^{\mathrm m} \;=\; \log\!\left(
\frac{\E\!\left[\hat v_{n,\delta}^{\mathrm m} \,\bigm|\, M_n \le a\right]}
{\Var(\hat\tau \mid M_n \le a)} \right).
\]
The numerator is estimated by the Monte Carlo average of $\hat v_{n,\delta}^{\mathrm m}$ over the $M$ rerandomization-accepted draws; the denominator by the sample variance of $\hat\tau$ over the same $M$ draws. Each variant uses its own ridge-adjusted joint covariance bound: $\hat{\mathbf U}_{n,\delta}^{\mathrm{ours}}$ (based on Definition~\ref{def:Un_estimator}), $\hat{\mathbf U}_{n,\delta}^{\mathrm{uncen}}$ (based on Eq.~\eqref{eq:Uhat_uncen}, without group centering), and $\hat{\mathbf U}_{n,\delta}^{\mathrm{spec}}$ (based on Eq.~\eqref{eq:Uhat_spec}, replacing the per-unit $d_{n,i}$ by the global $\lambda_{\max}$); each is plugged into the same closed form to obtain its $\hat v_{n,\delta}^{\mathrm m}$.
\end{itemize}

\paragraph{Single shared experiment.}
The network $\mathcal G_n^{\mathrm{BA}}$, the pre-treatment covariates $\mathbf X$, the hidden covariate $X_{i,\text{hidden}}$, and the coefficient vectors $\boldsymbol\beta_{\text{xbase}}, \boldsymbol\beta_{\text{xinterf}}, \beta_{\text{hidden}}$ are drawn \emph{once} and held fixed across all signal ratios, balancing sets, and outcome models, so that variation in the reported quantities reflects only differences in the design and the outcome functional form, not differences in the underlying realization.

\subsection{Variance ratio}
\label{sec:sim_vr}

We report the variance ratio $\mathrm{VR}$ as a function of $\log_2 \kappa$ for all seven outcome models in Table~\ref{tab:outcome_models}. The body of the paper showed \texttt{lin-prop} and \texttt{exp-prop}; the remaining five models are reported here. Smaller values indicate larger variance reduction.

Figure~\ref{fig:vr_lin_appendix} reports the three models with linear $f$. The crossover pattern is sharp for all three: the pre-treatment strategy reduces the variance ratio from above $0.9$ in the interference-dominated regime to roughly $0.20$--$0.25$ in the pre-treatment-dominated regime, while the treatment-dependent strategy moves in the opposite direction. The joint strategy under \texttt{NP} balancing provides precision gains across signal regimes. Under \texttt{NP+NWX} balancing, the variance ratio of the joint strategy is somewhat larger (roughly $0.36$--$0.40$ uniformly across $\log_2\kappa$); this is a consequence of the increased dimension $p$ of the criterion, which inflates $v_p$ without commensurate gains in $R^2$ when the outcome does not load on $\mathbf H_{i,\text{wX}}$. The exception is \texttt{lin-nwx}, whose interference block does load on $\mathbf H_{i,\text{wX}}$: there, the \texttt{NP+NWX} balancing set captures the relevant variation and outperforms \texttt{NP}-only balancing.

\begin{figure}[p]
\centering \rule{0pt}{1ex}\texttt{lin-prop}\\[1pt]
\includegraphics[width=0.36\textwidth]{"figures/variance_reduction/vr_covs-np_type1.pdf"}\hfill
\includegraphics[width=0.36\textwidth]{"figures/variance_reduction/vr_covs-np+nwx_type1.pdf"}\\[0pt]
\rule{0pt}{1ex}\texttt{lin-prop+nwx}\\[1pt]
\includegraphics[width=0.36\textwidth]{"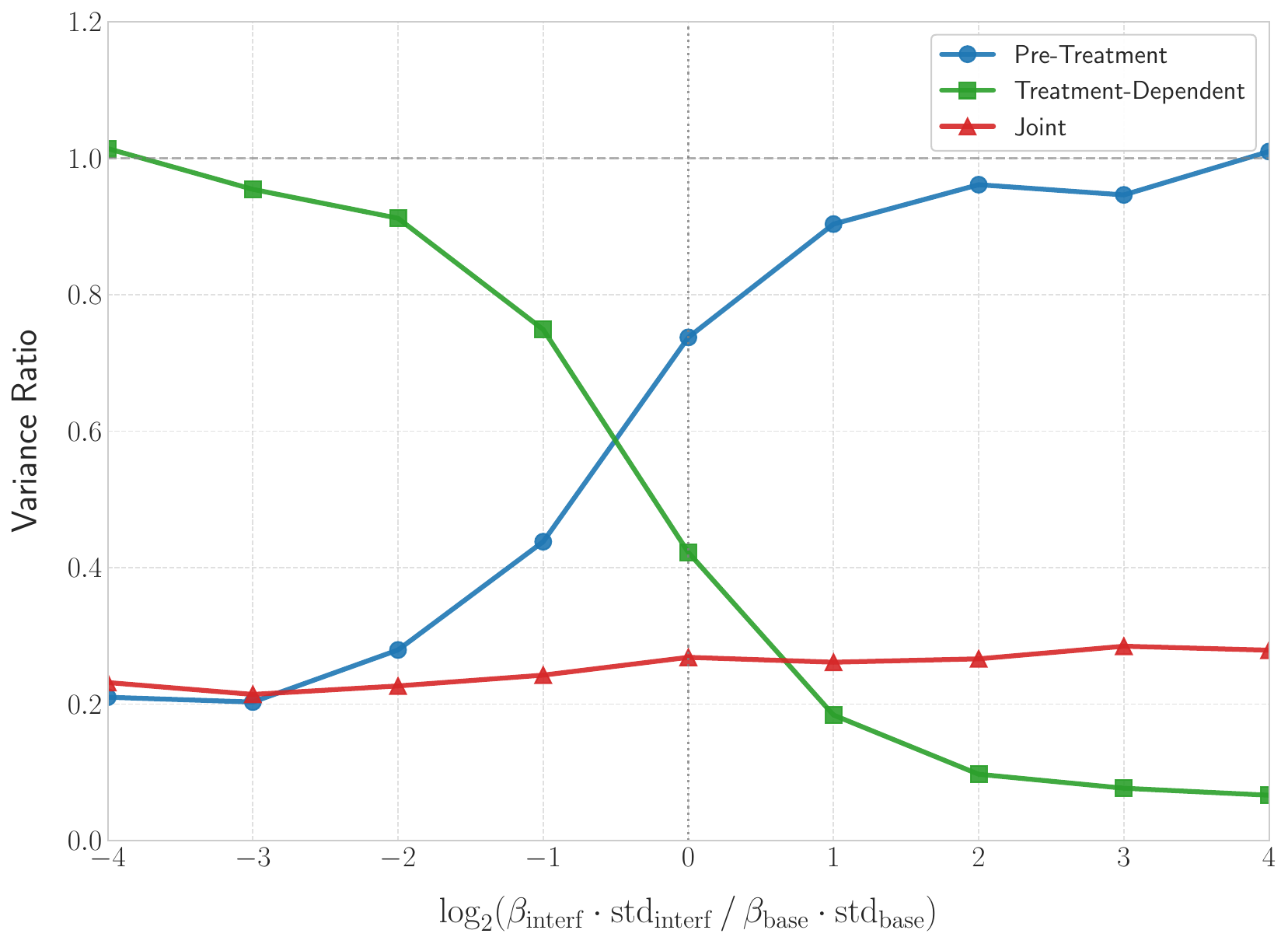"}\hfill
\includegraphics[width=0.36\textwidth]{"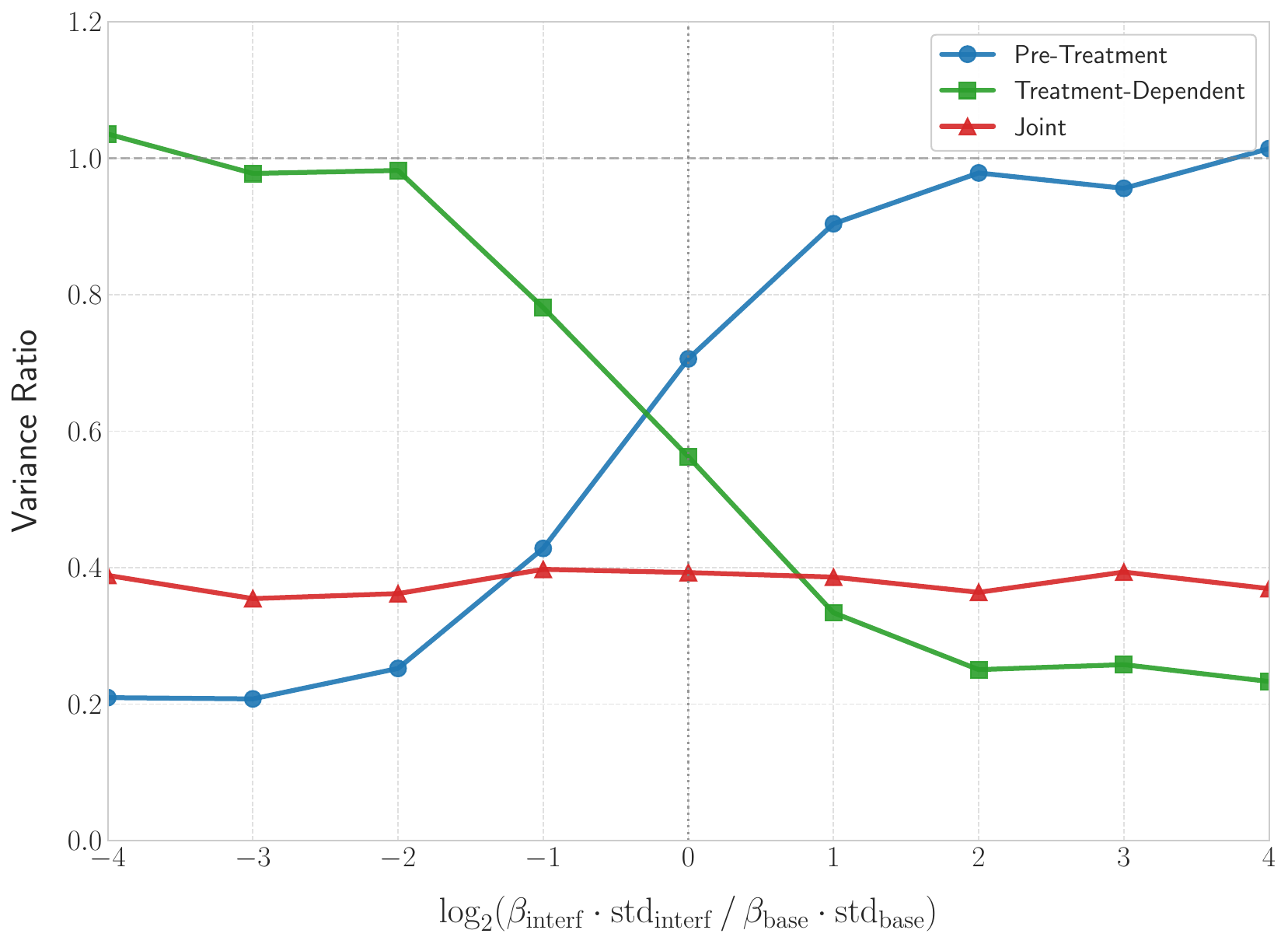"}\\[0pt]
\rule{0pt}{1ex}\texttt{lin-nwx}\\[1pt]
\includegraphics[width=0.36\textwidth]{"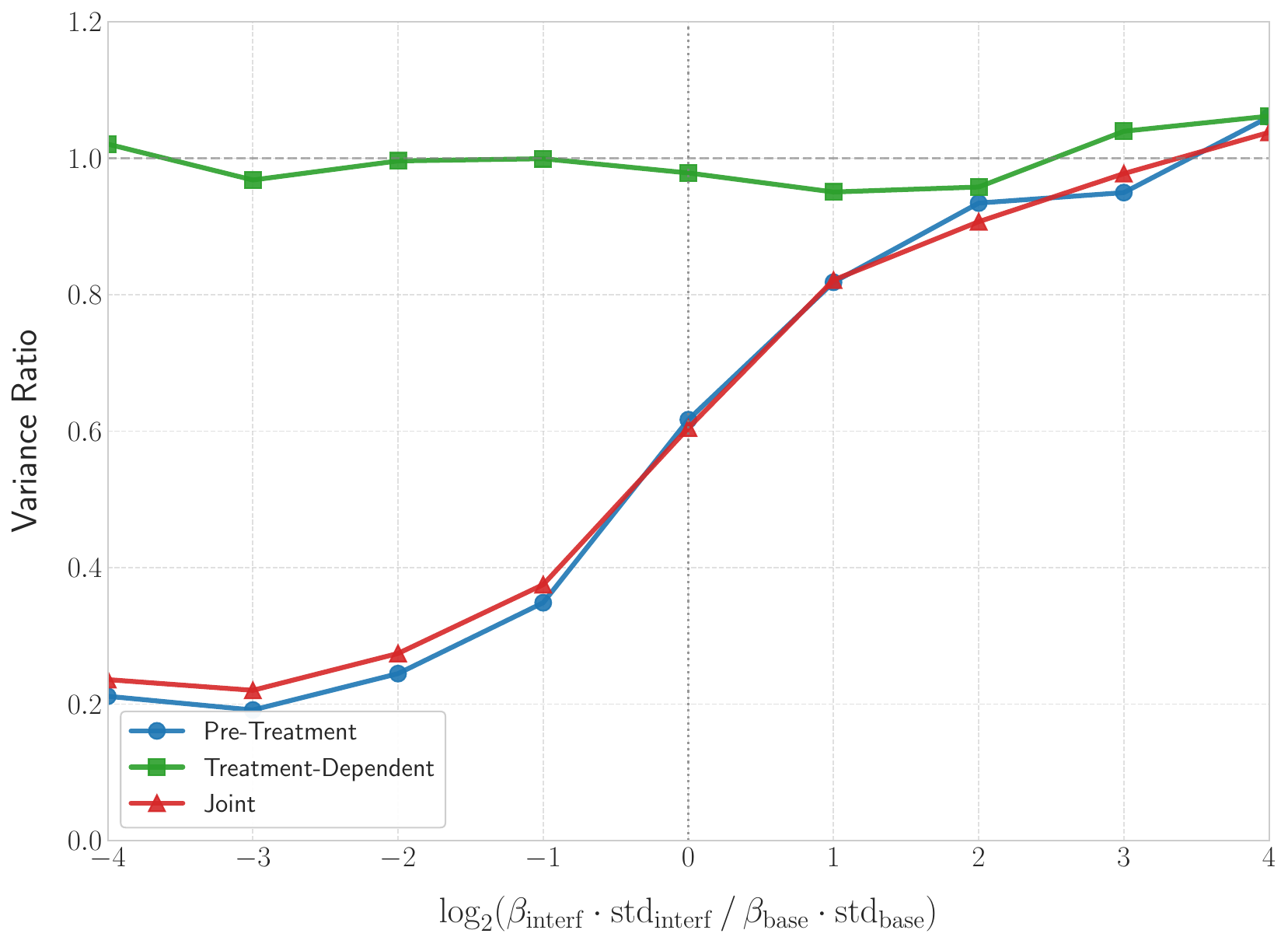"}\hfill
\includegraphics[width=0.36\textwidth]{"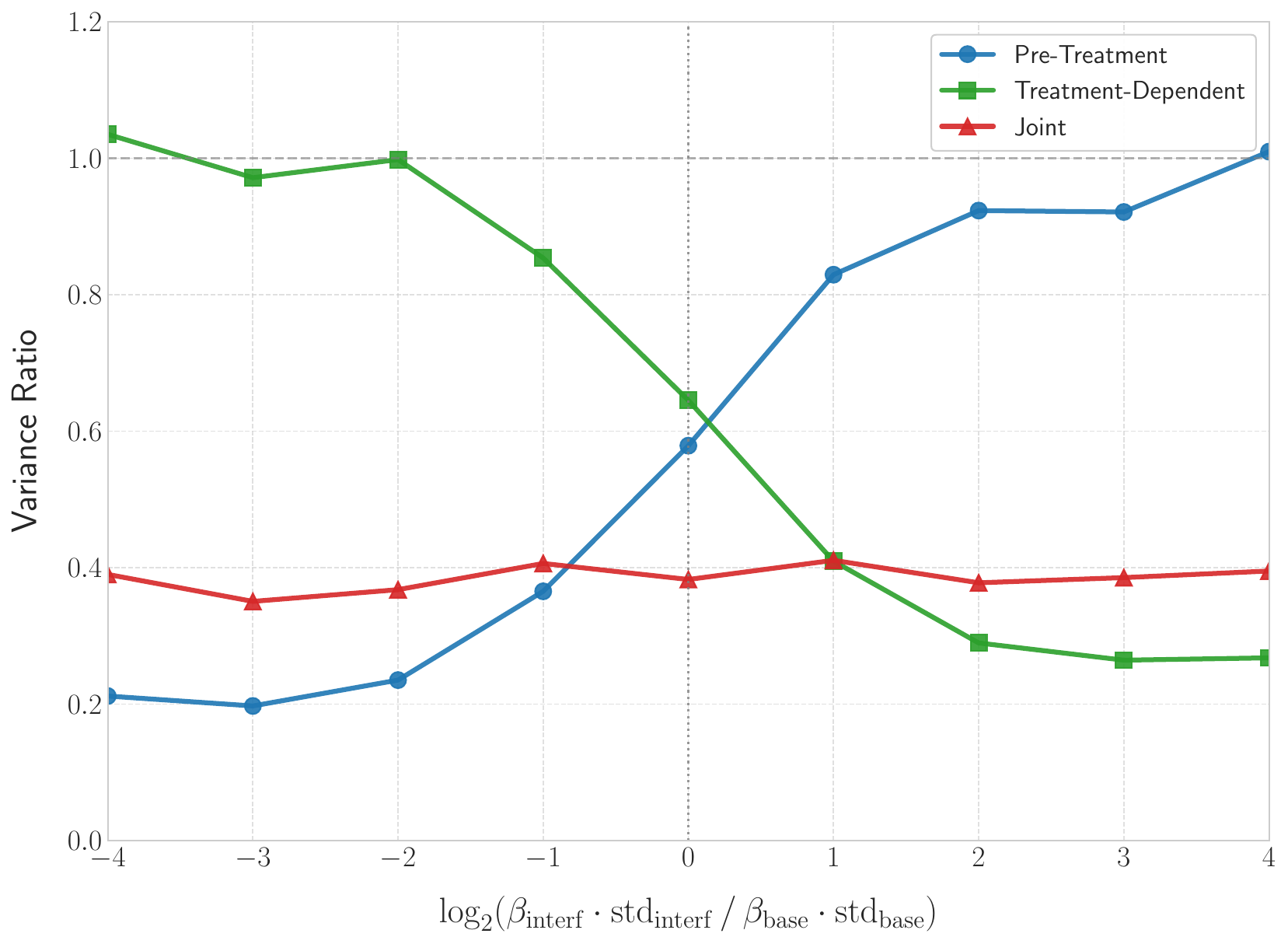"}\\[0pt]
\caption{Variance ratio for the three linear outcome models. Within each row, the left panel uses balancing set \texttt{NP} and the right panel uses balancing set \texttt{NP+NWX}. Blue: pre-treatment strategy; green: treatment-dependent strategy; red: joint strategy.}
\label{fig:vr_lin_appendix}
\end{figure}

Figure~\ref{fig:vr_exp_appendix} reports the four models with exponential $f$. The qualitative crossover is preserved, but the magnitudes are attenuated relative to the linear case: the pre-treatment strategy now bottoms out near $0.50$ rather than $0.20$, consistent with the weaker association between the treatment-effect estimator and the balanced first-moment contrasts under the nonlinear outcome model. The treatment-dependent strategy remains highly effective in the interference-dominated regime, often pushing the ratio to $0.05$ or below for \texttt{exp-prop} and \texttt{exp-prop+nwx}. The joint strategy provides variance reductions across these settings, with smaller gains under \texttt{NP+NWX} when $\mathbf H_{i,\text{wX}}$ does not enter the outcome.

\begin{figure}[p]
\centering \rule{0pt}{1ex}\texttt{exp-prop}\\[1pt]
\includegraphics[width=0.36\textwidth]{"figures/variance_reduction/vr_covs-np_type5.pdf"}\hfill
\includegraphics[width=0.36\textwidth]{"figures/variance_reduction/vr_covs-np+nwx_type5.pdf"}\\[0pt]
\rule{0pt}{1ex}\texttt{exp-sumexp}\\[1pt]
\includegraphics[width=0.36\textwidth]{"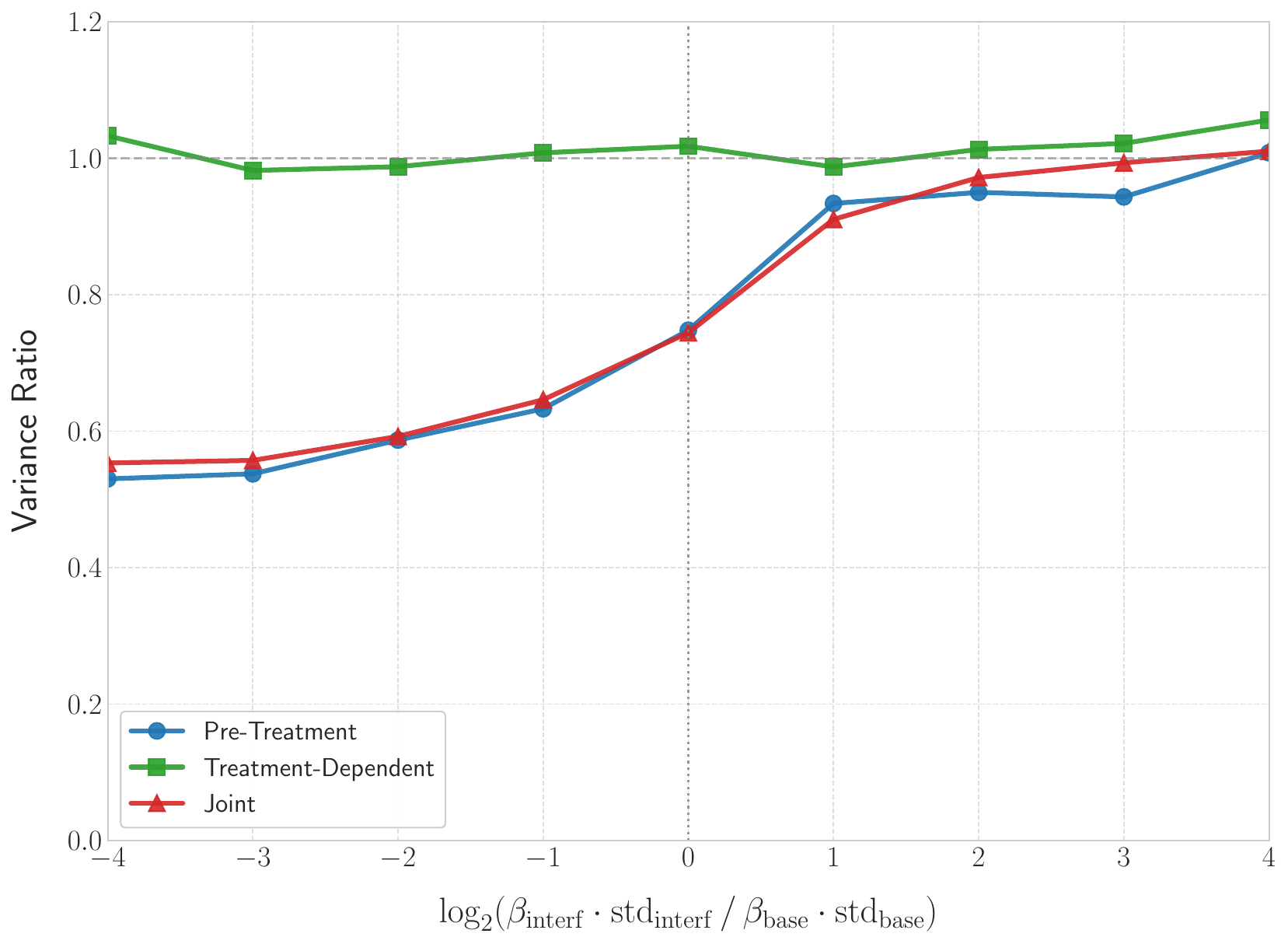"}\hfill
\includegraphics[width=0.36\textwidth]{"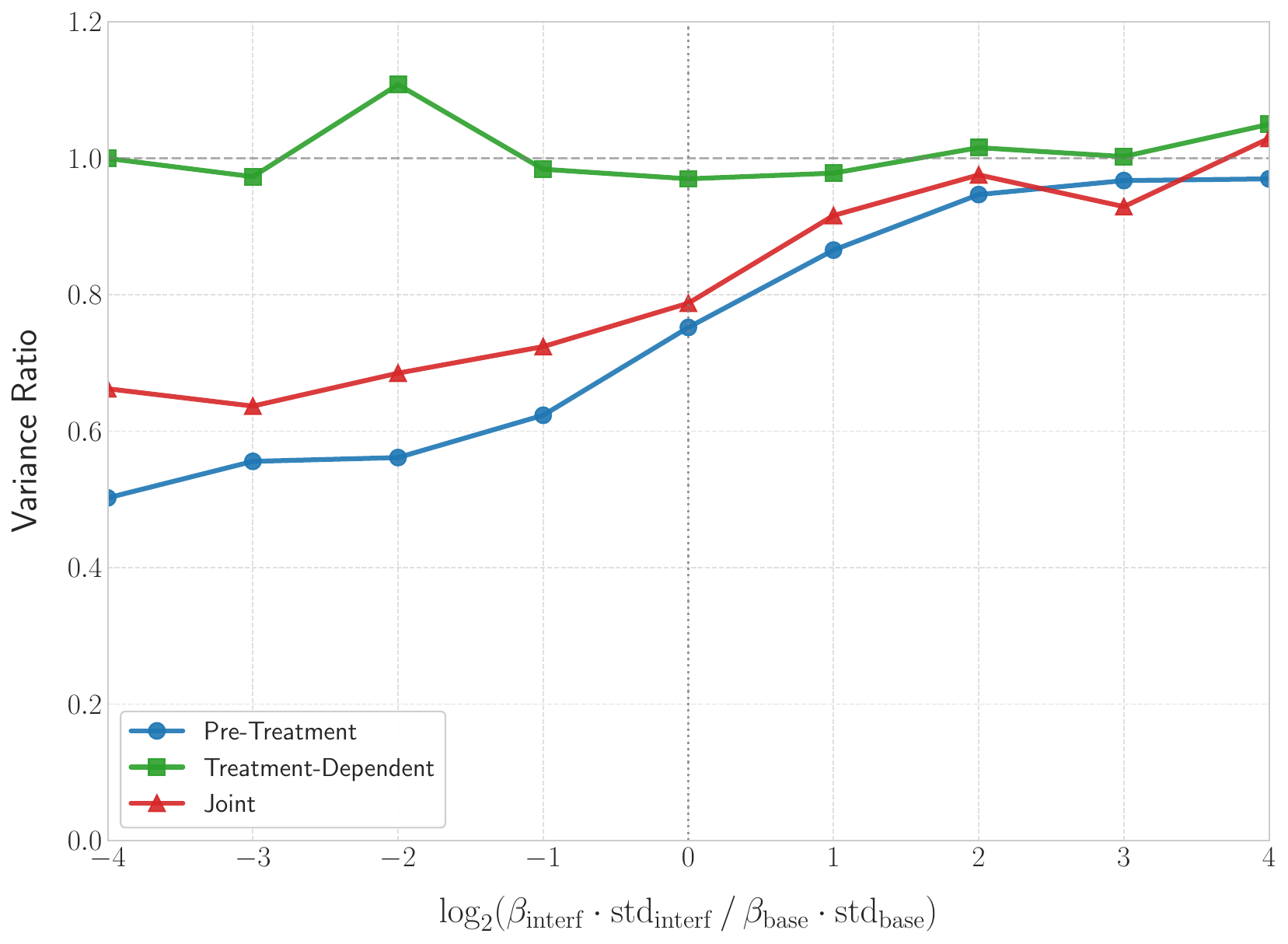"}\\[0pt]
\rule{0pt}{1ex}\texttt{exp-prop+nwx}\\[1pt]
\includegraphics[width=0.36\textwidth]{"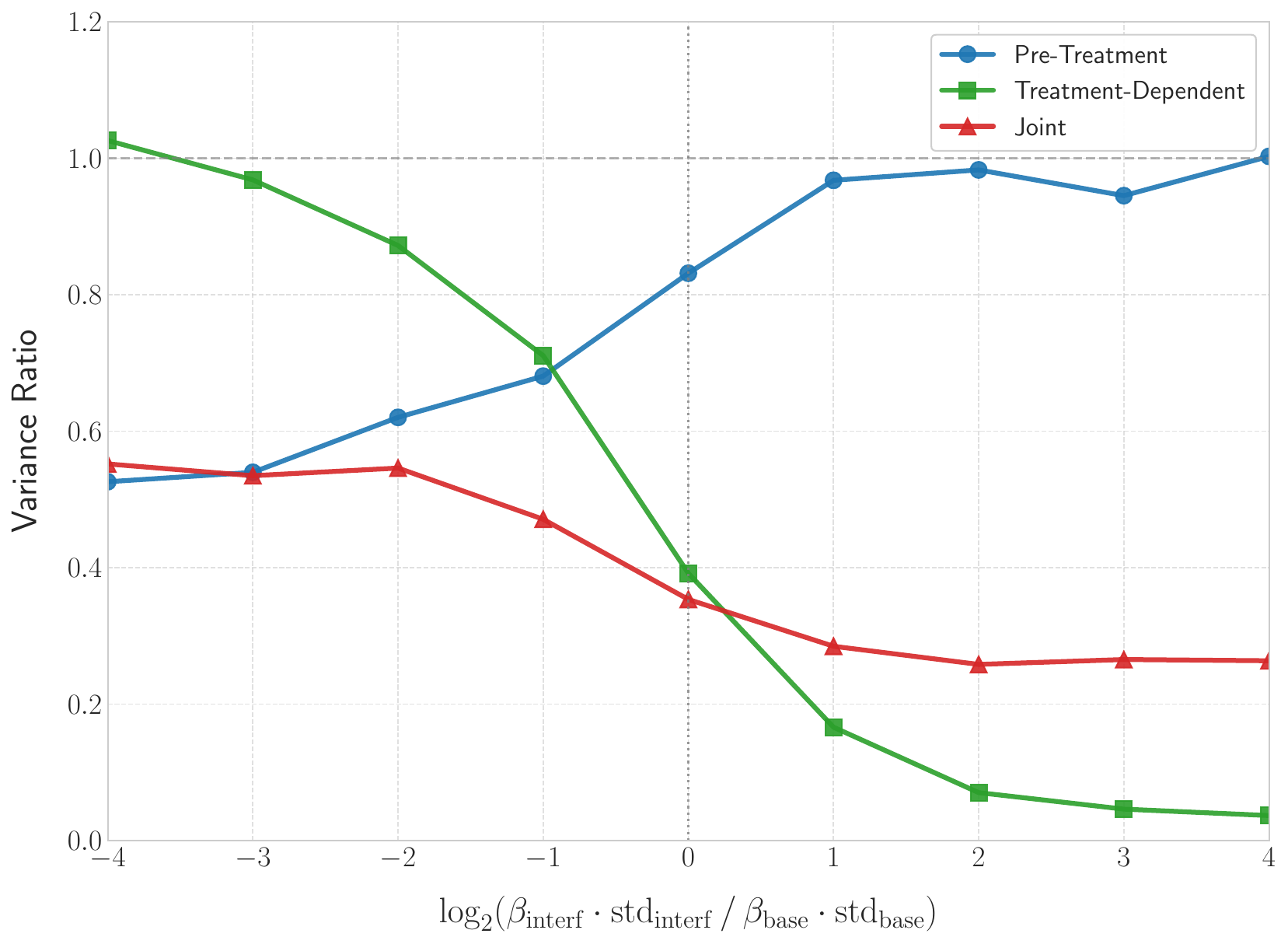"}\hfill
\includegraphics[width=0.36\textwidth]{"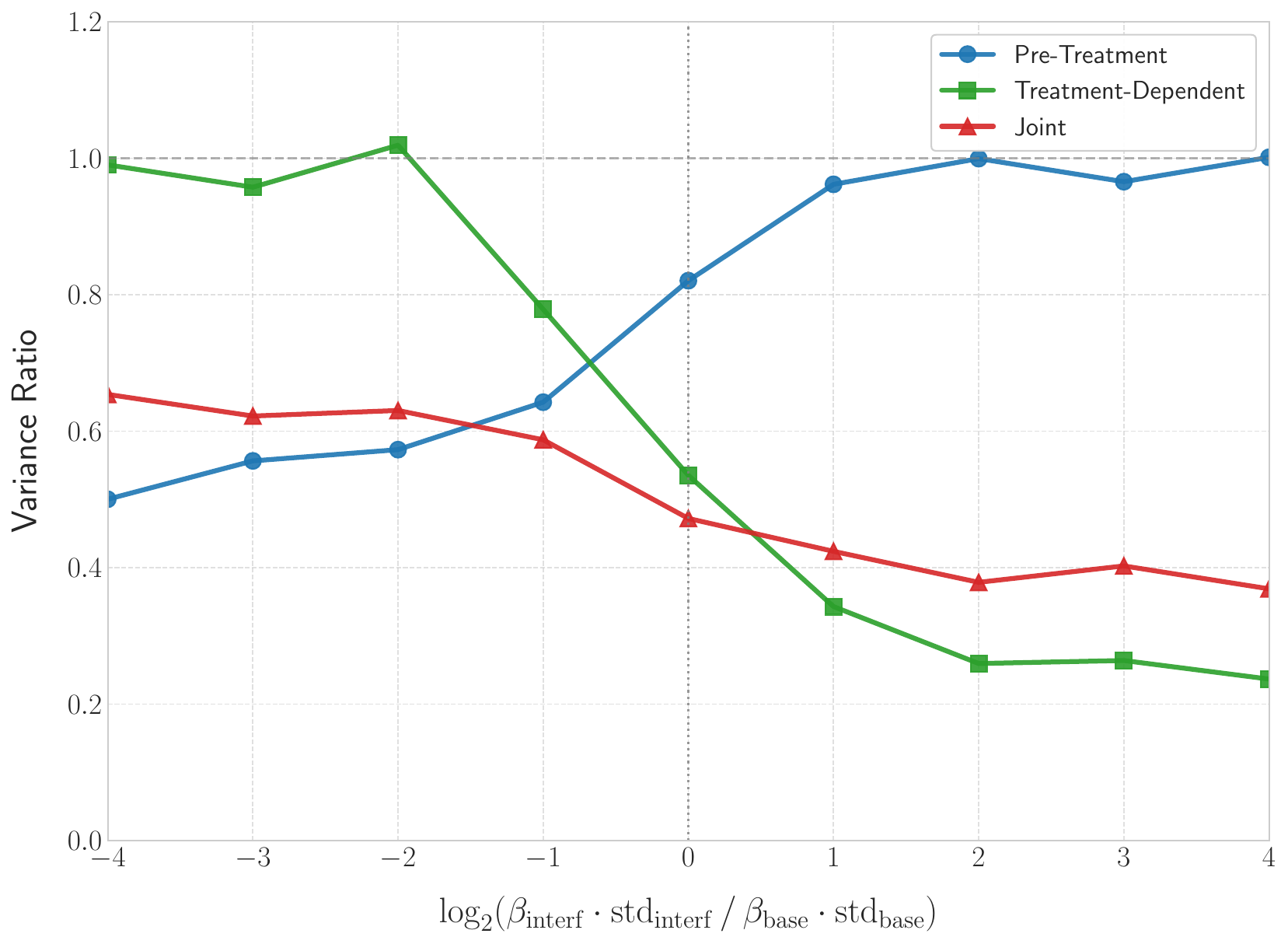"}\\[0pt]
\rule{0pt}{1ex}\texttt{exp-sumexp-full}\\[1pt]
\includegraphics[width=0.36\textwidth]{"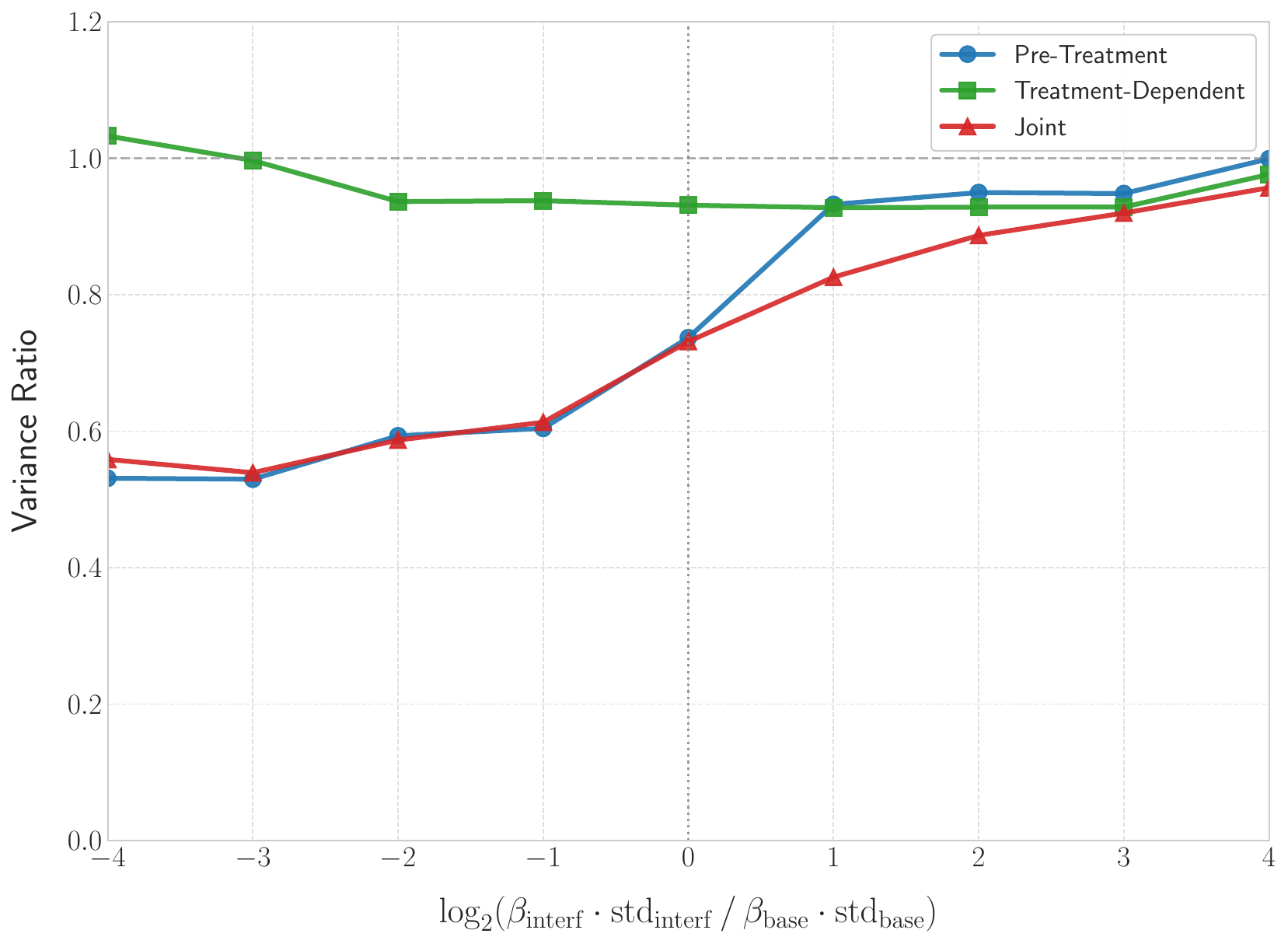"}\hfill
\includegraphics[width=0.36\textwidth]{"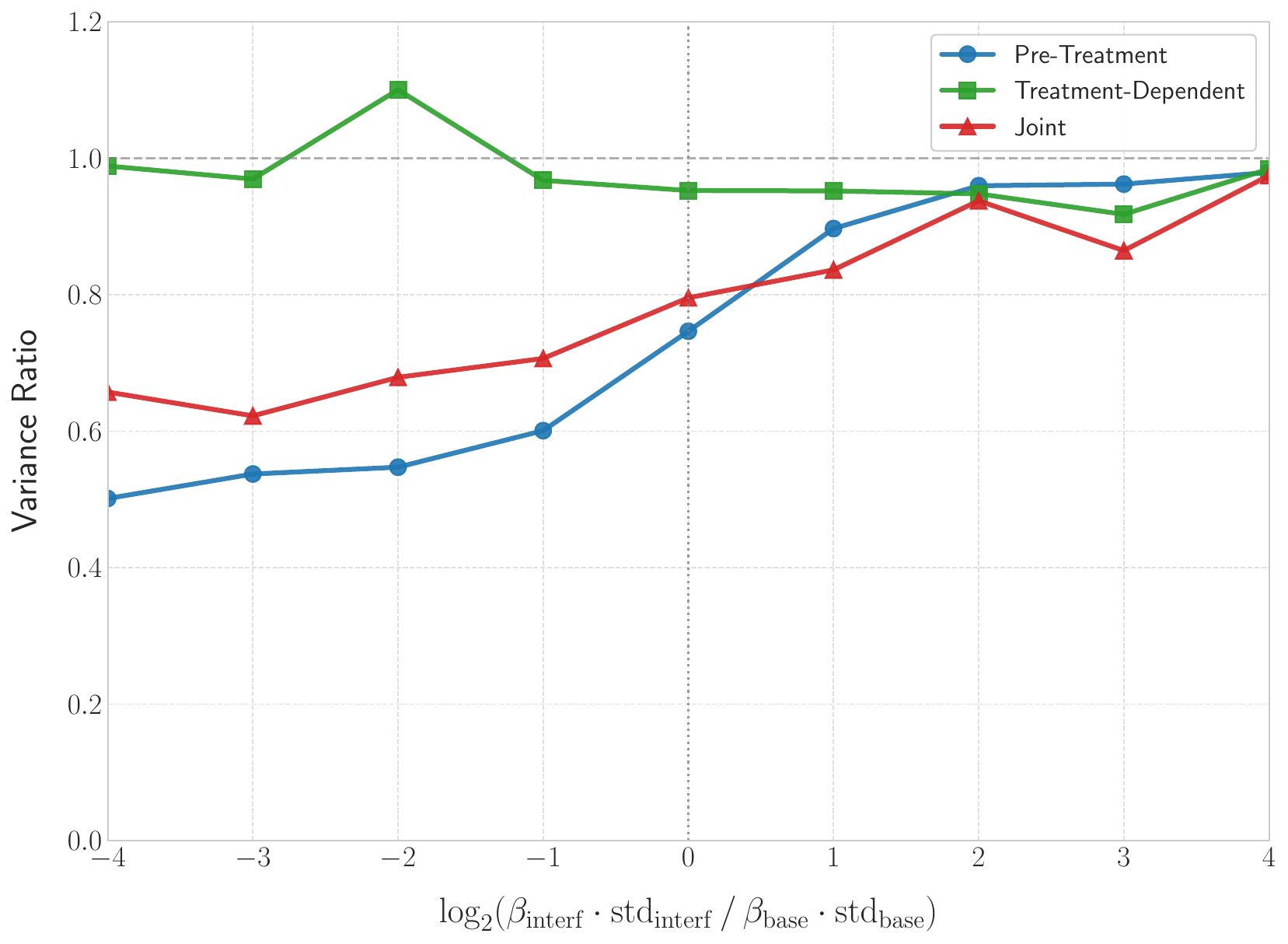"}\\[0pt]
\caption{Variance ratio for the four exponential outcome models. Layout, line types, and colors follow Figure~\ref{fig:vr_lin_appendix}.}
\label{fig:vr_exp_appendix}
\end{figure}

Two messages emerge from the variance-reduction sweep across all seven models. First, when the rerandomization criterion includes covariates whose corresponding directions are loaded by the outcome (e.g., the pre-treatment $\mathbf X$ in the pre-treatment-dominated regime; $H_{i,\text{prop}}$ for \texttt{prop}-type interference; $\mathbf H_{i,\text{wX}}$ for \texttt{nwx}-type interference), the variance reduction is dramatic. Second, the joint strategy is uniformly competitive but not uniformly optimal: when the criterion includes covariates that do not enter the outcome, the dimensional inflation of $v_p$ produces a measurable but modest loss. This argues for choosing the balancing set deliberately based on what is suspected to drive the outcome, rather than aggregating all available covariates.

\subsection{Conservative variance estimator}
\label{sec:sim_cv}

\subsubsection{Conservativeness ratio and confidence interval}
\label{sec:sim_conserv}

We now report values of the confidence interval length ratio $\rho = \hat L_{\mathrm{ReM}} / \hat L_{\mathrm{Bern}}$, and the empirical $95\%$ CI coverage, all defined and computed as in Section~\ref{sec:sim_protocol}, Step~4. Figures~\ref{fig:ratio_lin_appendix} and~\ref{fig:ratio_exp_appendix} report $\rho$ from~\eqref{eq:ci_length_ratio} for the linear and exponential model groups respectively. In almost every panel and for every strategy, $\rho \le 1$, confirming that the design-stage gain captured by the optimization procedure in Section~\ref{sec:opt} translates into a strictly shorter confidence interval under rerandomization. Across all seven models, the joint strategy is the only one that maintains $\rho$ uniformly below $1$. Coverage is empirically close to 1 in our Monte Carlo runs across all settings, so we display it for only two representative models, \texttt{lin-prop} and \texttt{exp-prop}, in Figure~\ref{fig:cov_appendix}. Coverage is essentially $1$ for all strategies and signal ratios, an immediate consequence of the conservativeness margin documented above. Because every other model exhibits the same near-unit coverage, we omit those panels.

\begin{figure}[p]
\centering \rule{0pt}{1ex}\texttt{lin-prop}\\[1pt]
\includegraphics[width=0.36\textwidth]{"figures/conservative_variance/cv_covs-np_type1_ratio.pdf"}\hfill
\includegraphics[width=0.36\textwidth]{"figures/conservative_variance/cv_covs-np+nwx_type1_ratio.pdf"}\\[0pt]
\rule{0pt}{1ex}\texttt{lin-prop+nwx}\\[1pt]
\includegraphics[width=0.36\textwidth]{"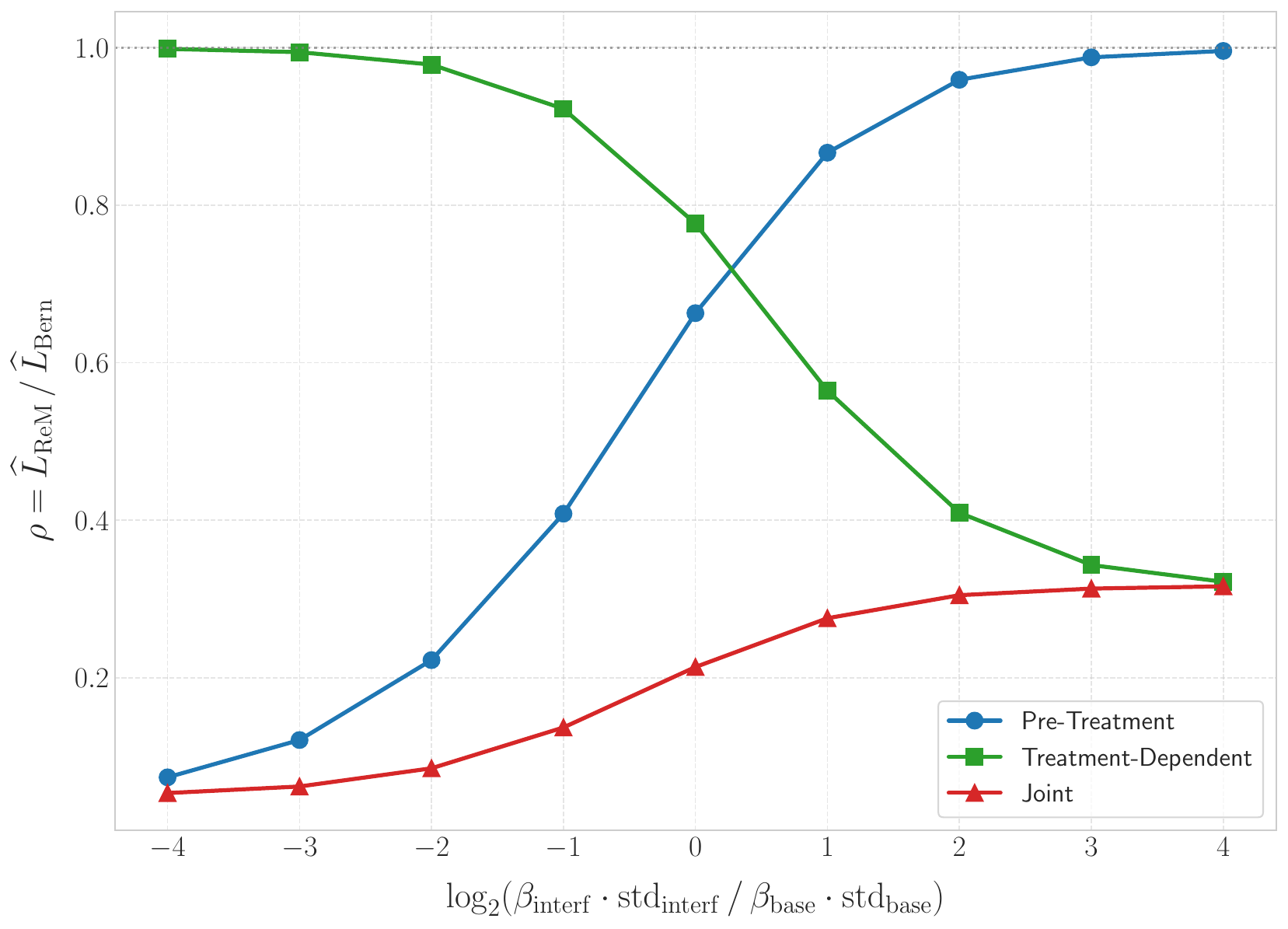"}\hfill
\includegraphics[width=0.36\textwidth]{"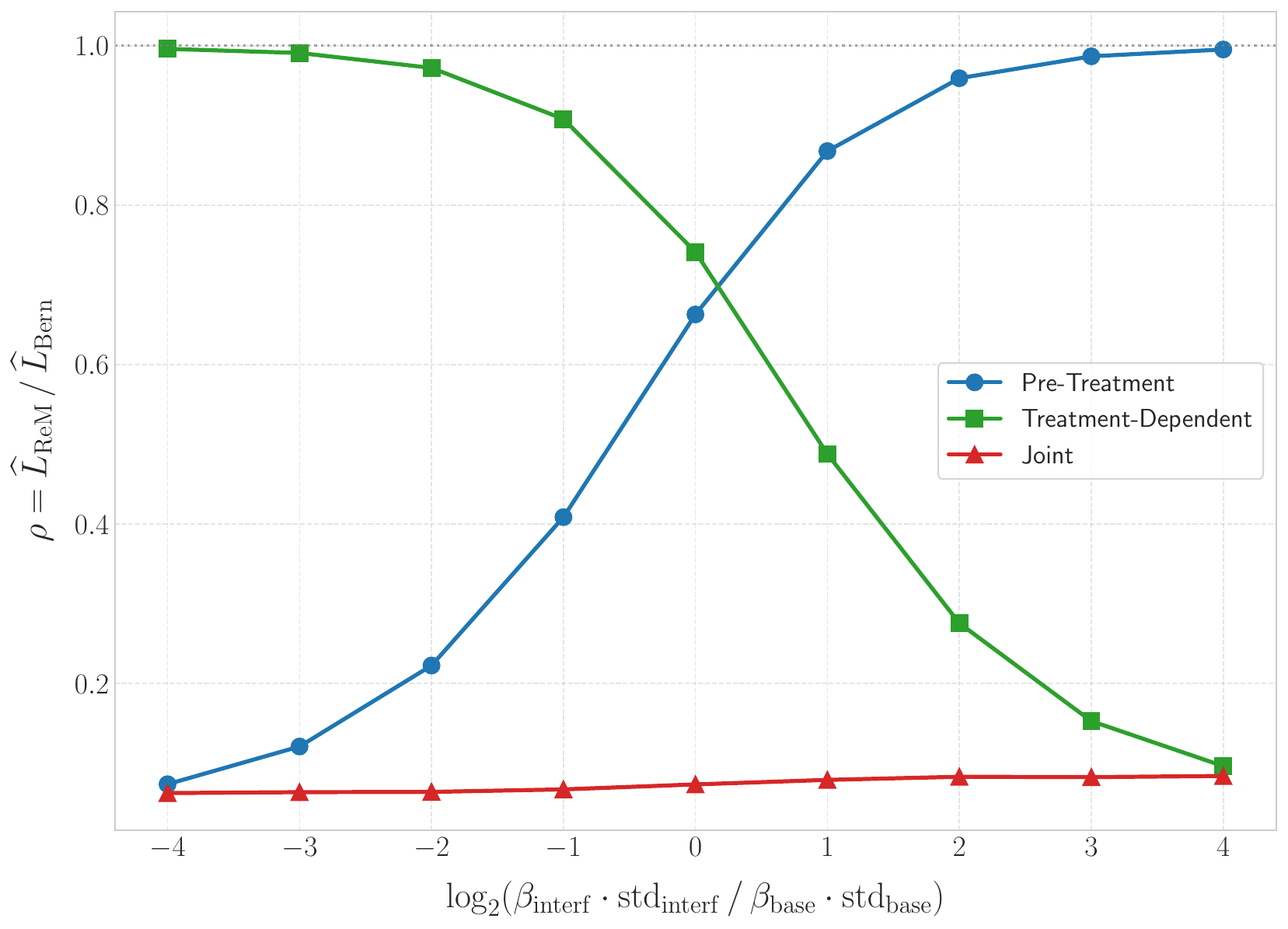"}\\[0pt]
\rule{0pt}{1ex}\texttt{lin-nwx}\\[1pt]
\includegraphics[width=0.36\textwidth]{"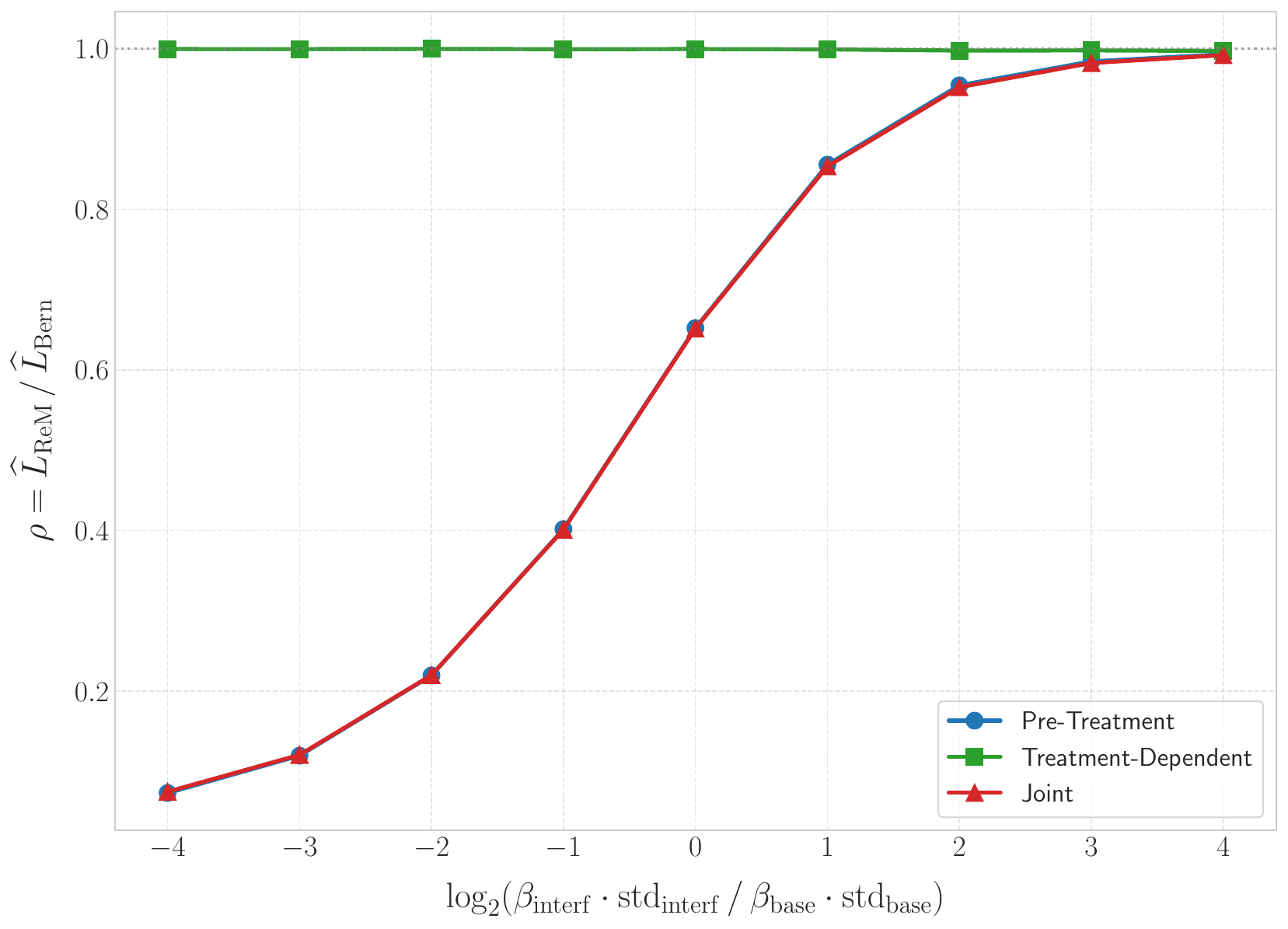"}\hfill
\includegraphics[width=0.36\textwidth]{"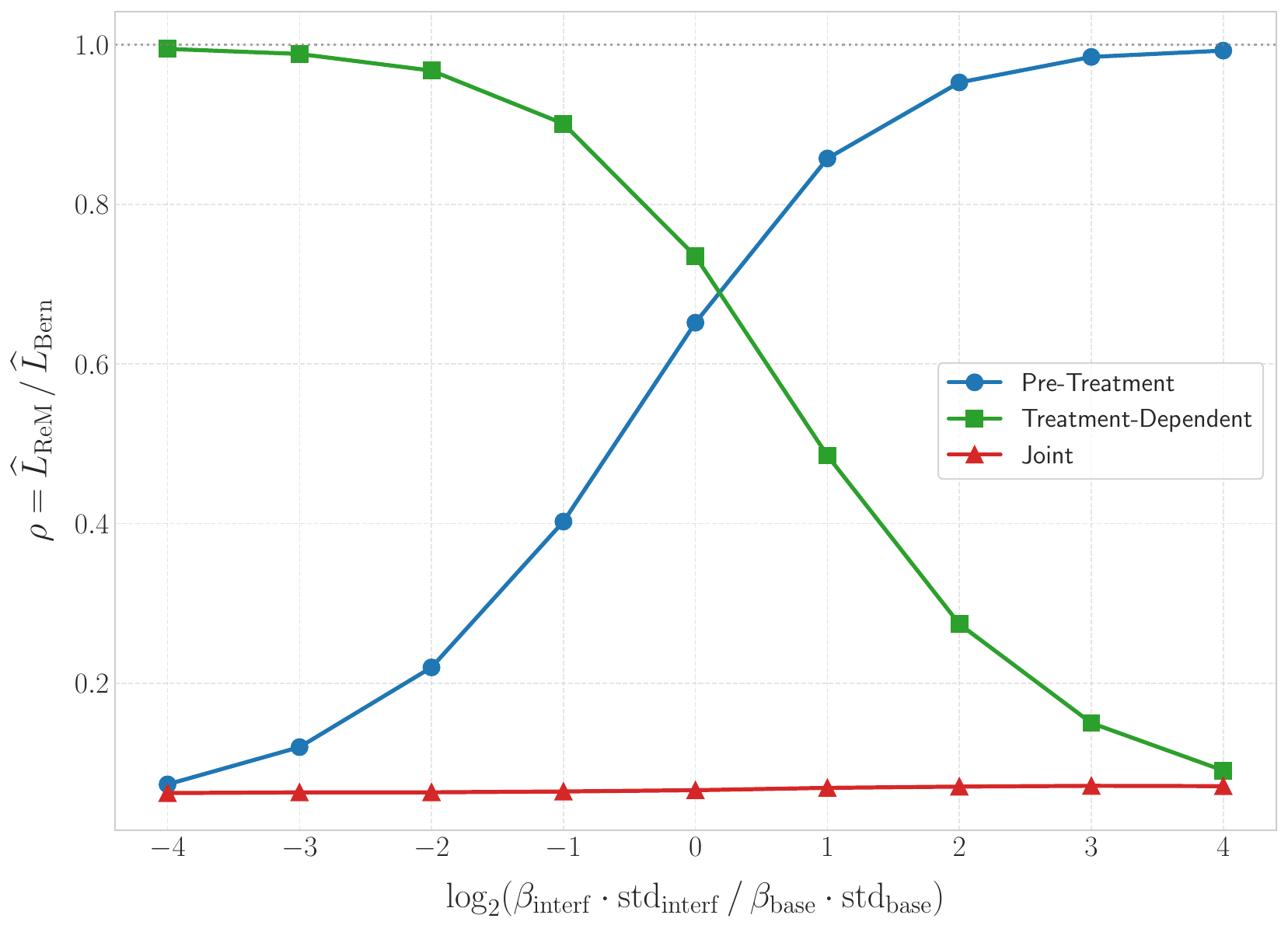"}\\[0pt]
\caption{Confidence interval length ratio $\rho$ from~\eqref{eq:ci_length_ratio} for the three linear outcome models. Within each row, left: balancing \texttt{NP}; right: \texttt{NP+NWX}.}
\label{fig:ratio_lin_appendix}
\end{figure}

\begin{figure}[p]
\centering \rule{0pt}{1ex}\texttt{exp-prop}\\[1pt]
\includegraphics[width=0.36\textwidth]{"figures/conservative_variance/cv_covs-np_type5_ratio.pdf"}\hfill
\includegraphics[width=0.36\textwidth]{"figures/conservative_variance/cv_covs-np+nwx_type5_ratio.pdf"}\\[0pt]
\rule{0pt}{1ex}\texttt{exp-sumexp}\\[1pt]
\includegraphics[width=0.36\textwidth]{"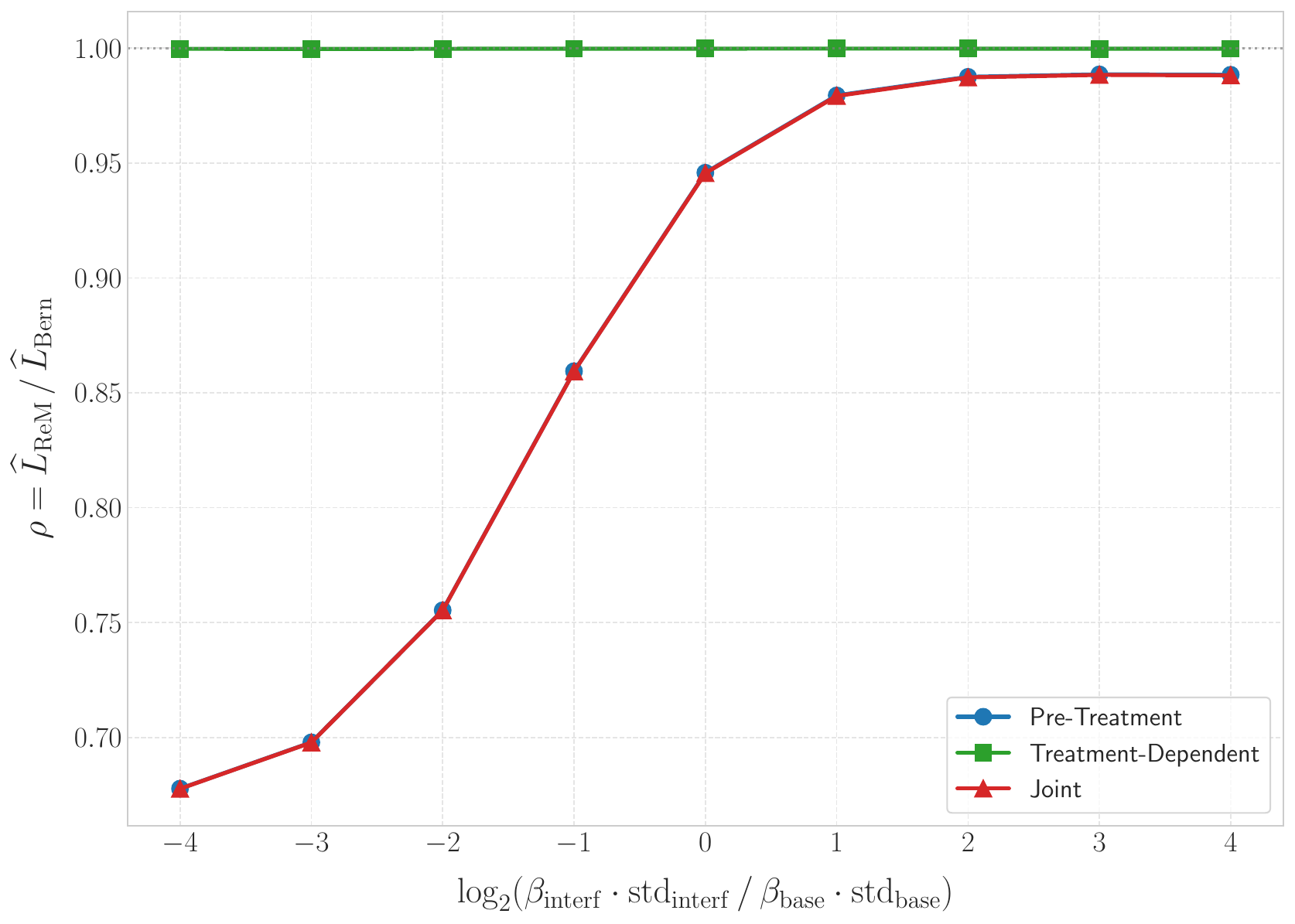"}\hfill
\includegraphics[width=0.36\textwidth]{"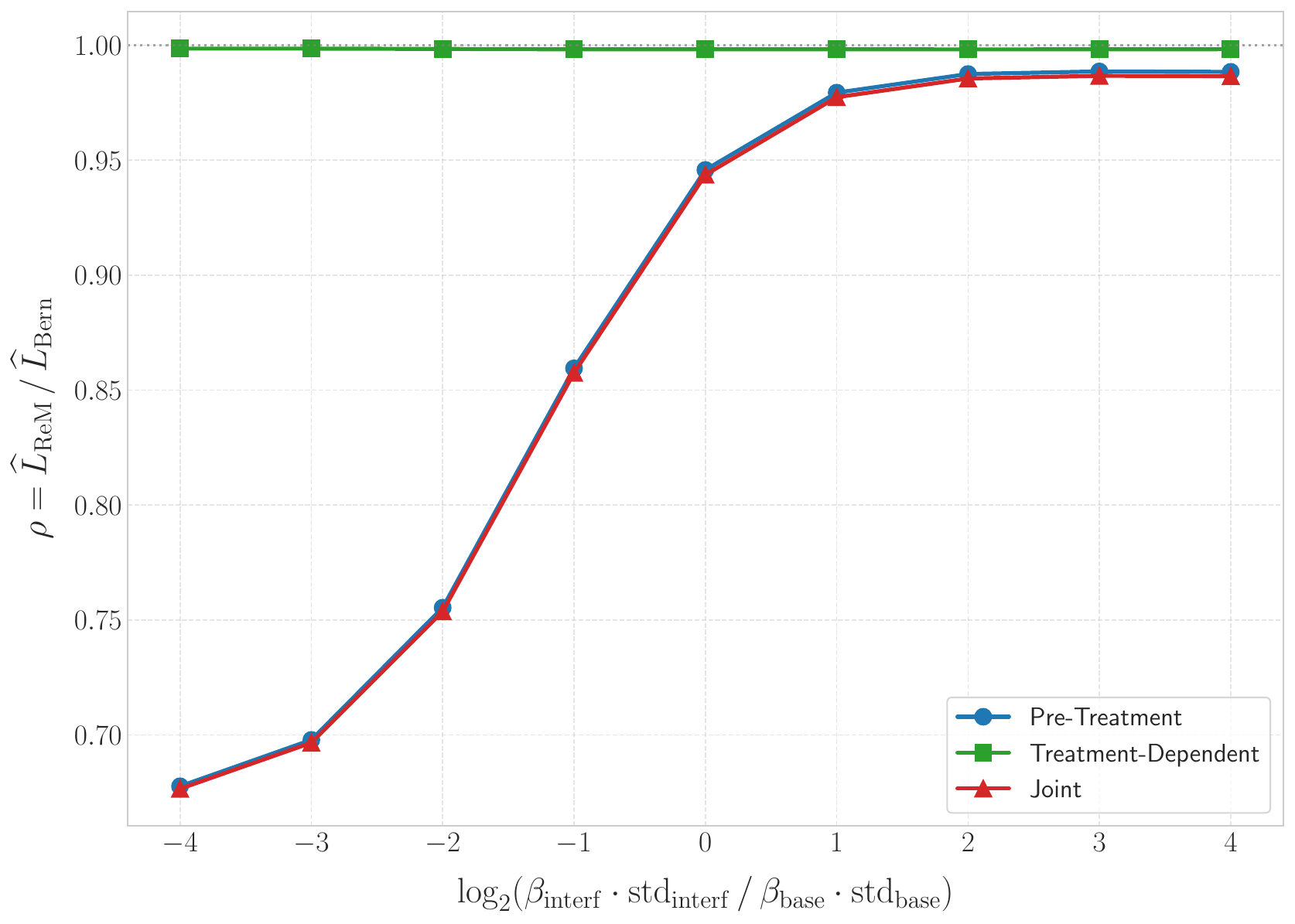"}\\[0pt]
\rule{0pt}{1ex}\texttt{exp-prop+nwx}\\[1pt]
\includegraphics[width=0.36\textwidth]{"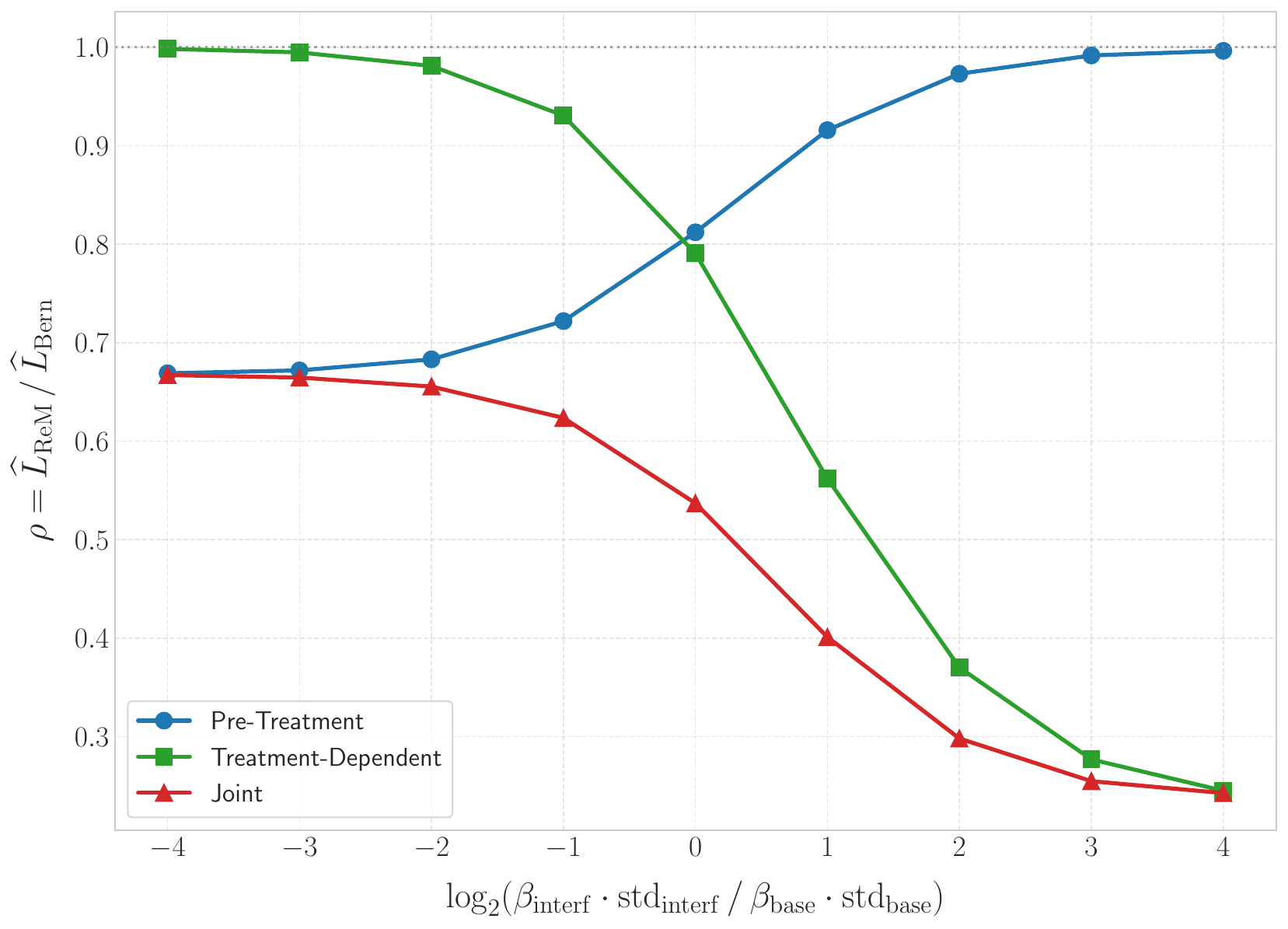"}\hfill
\includegraphics[width=0.36\textwidth]{"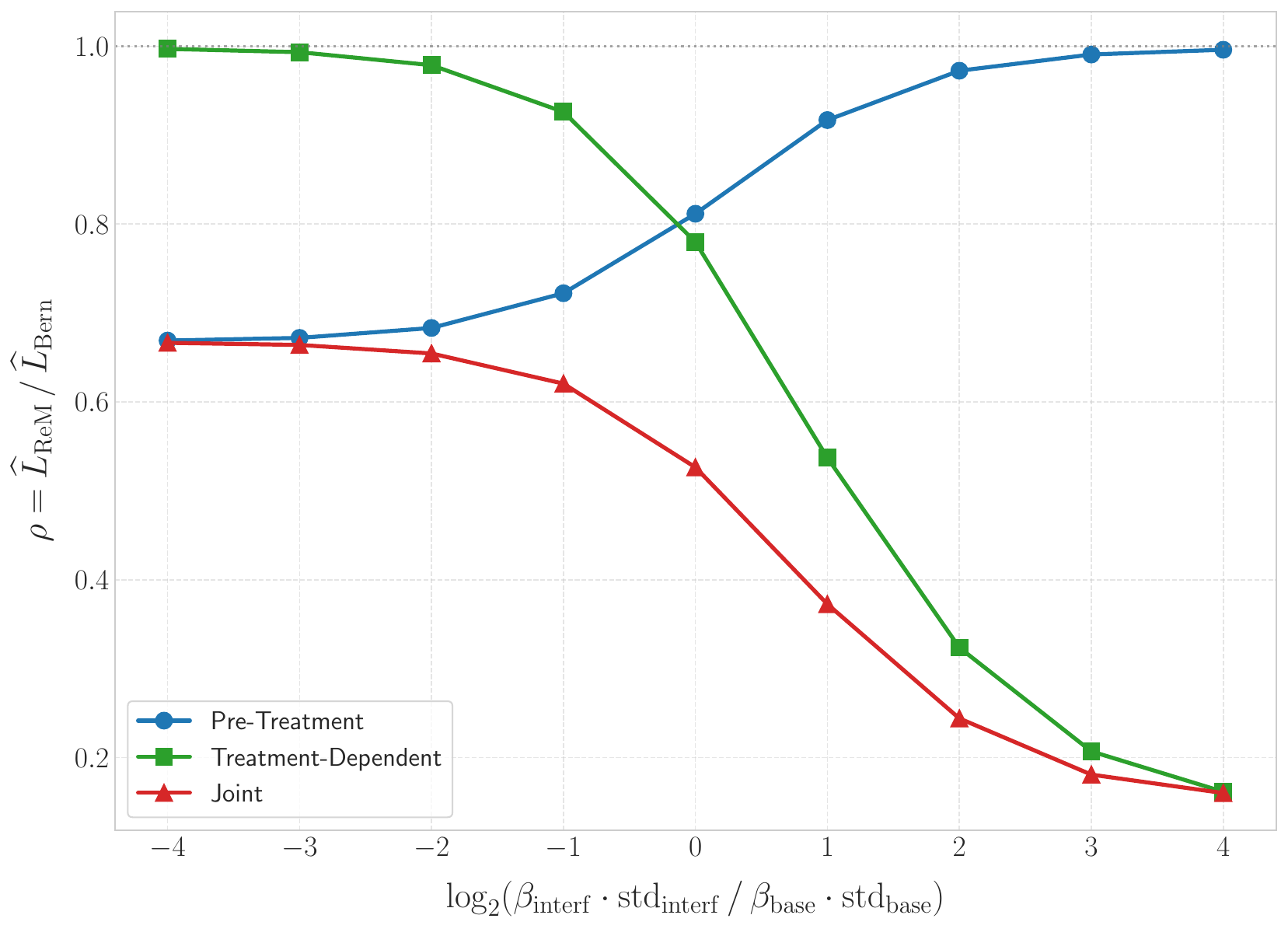"}\\[0pt]
\rule{0pt}{1ex}\texttt{exp-sumexp-full}\\[1pt]
\includegraphics[width=0.36\textwidth]{"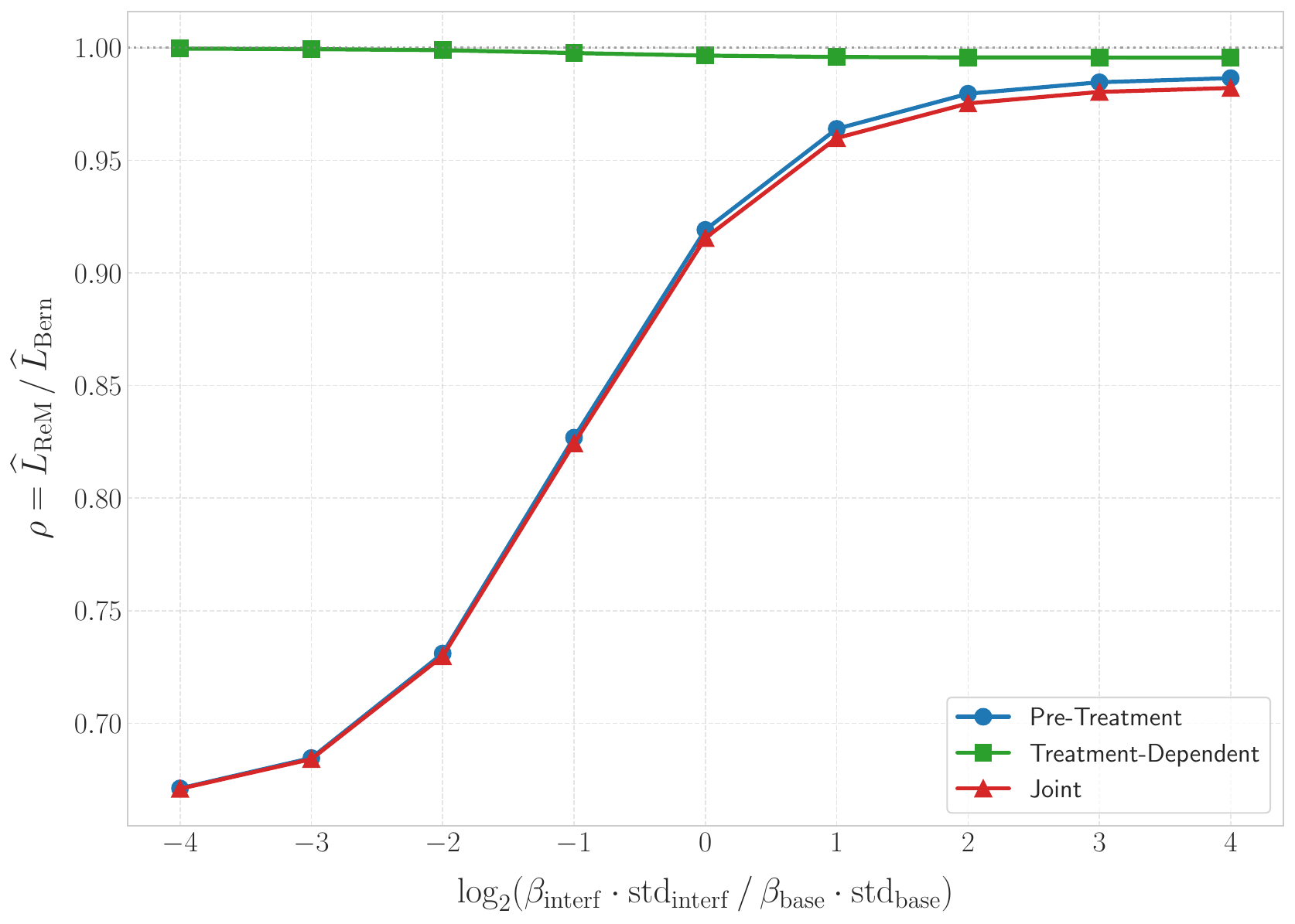"}\hfill
\includegraphics[width=0.36\textwidth]{"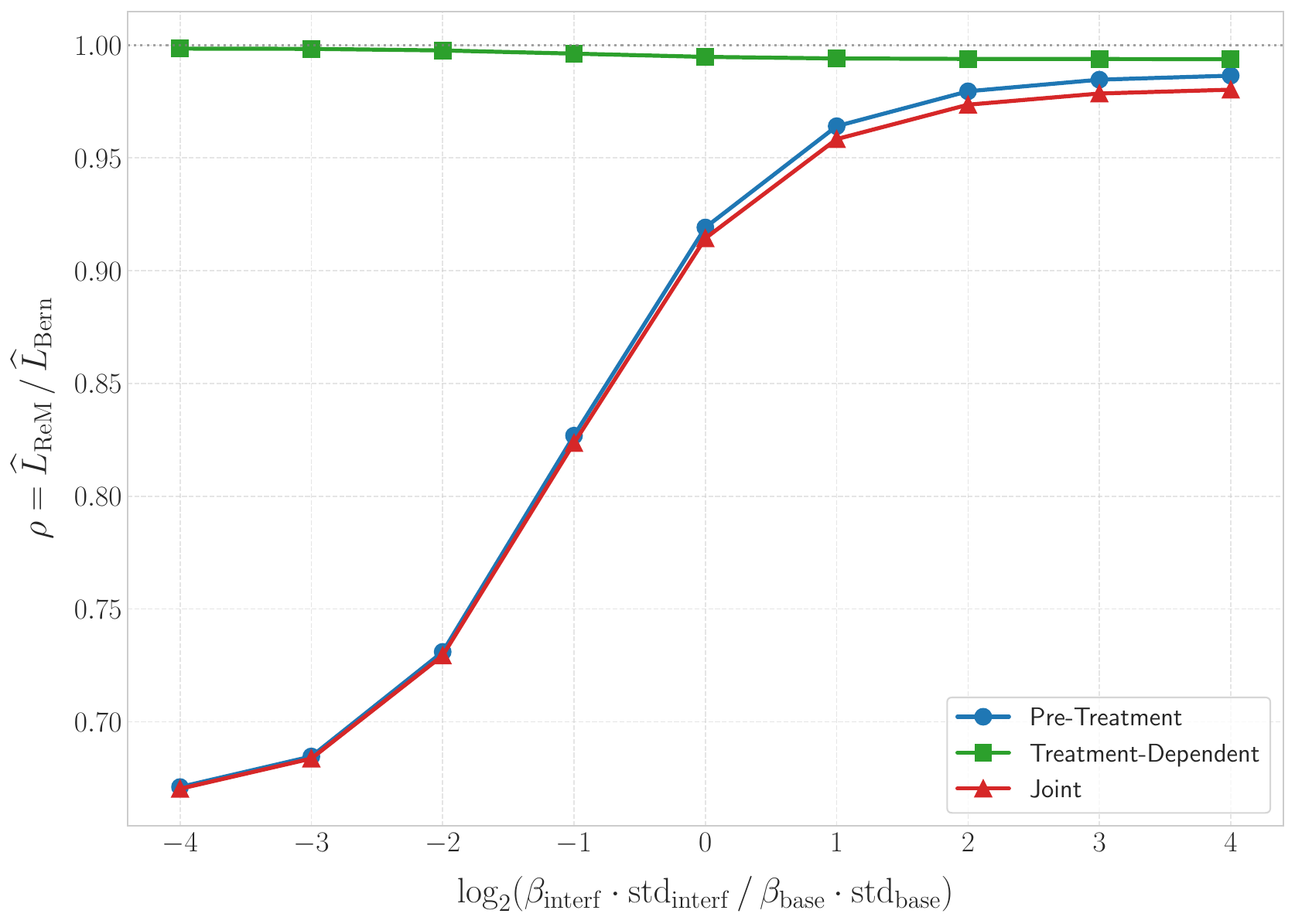"}\\[0pt]
\caption{Confidence interval length ratio $\rho$ from~\eqref{eq:ci_length_ratio} for the four exponential outcome models.}
\label{fig:ratio_exp_appendix}
\end{figure}

\begin{figure}[ht]
\centering
\begin{subfigure}[b]{0.85\textwidth}
\centering
\includegraphics[width=\linewidth]{"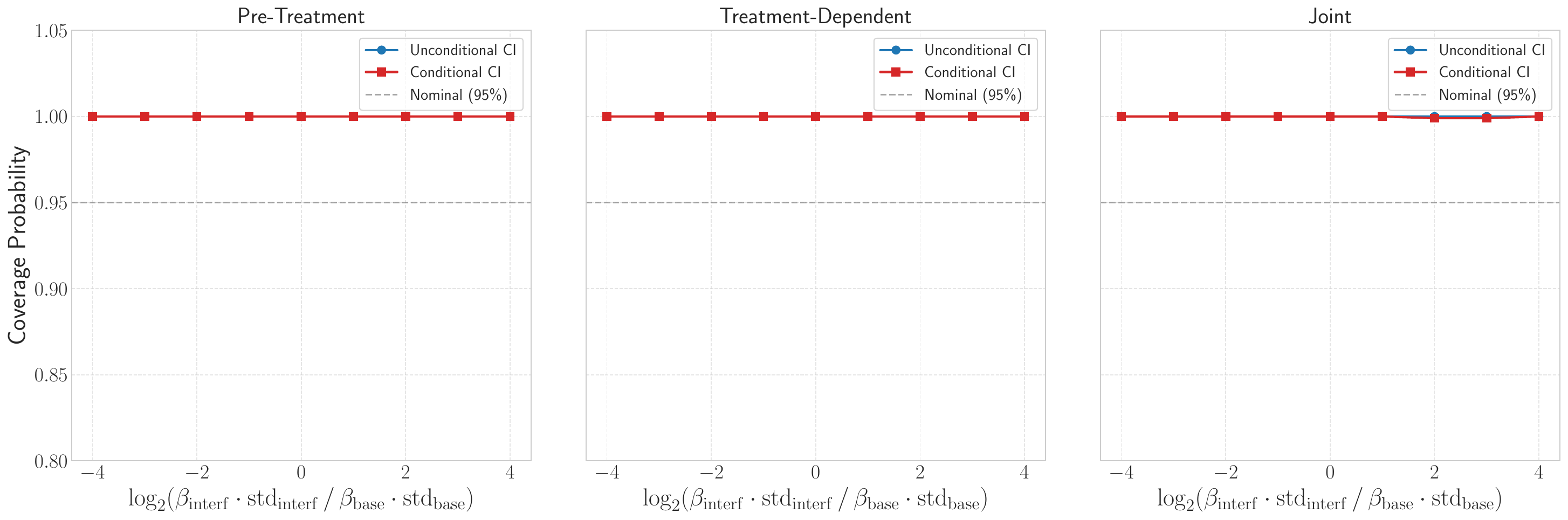"}
\caption{\texttt{lin-prop}, balancing \texttt{NP}}
\end{subfigure}

\vspace{0.6em}

\begin{subfigure}[b]{0.85\textwidth}
\centering
\includegraphics[width=\linewidth]{"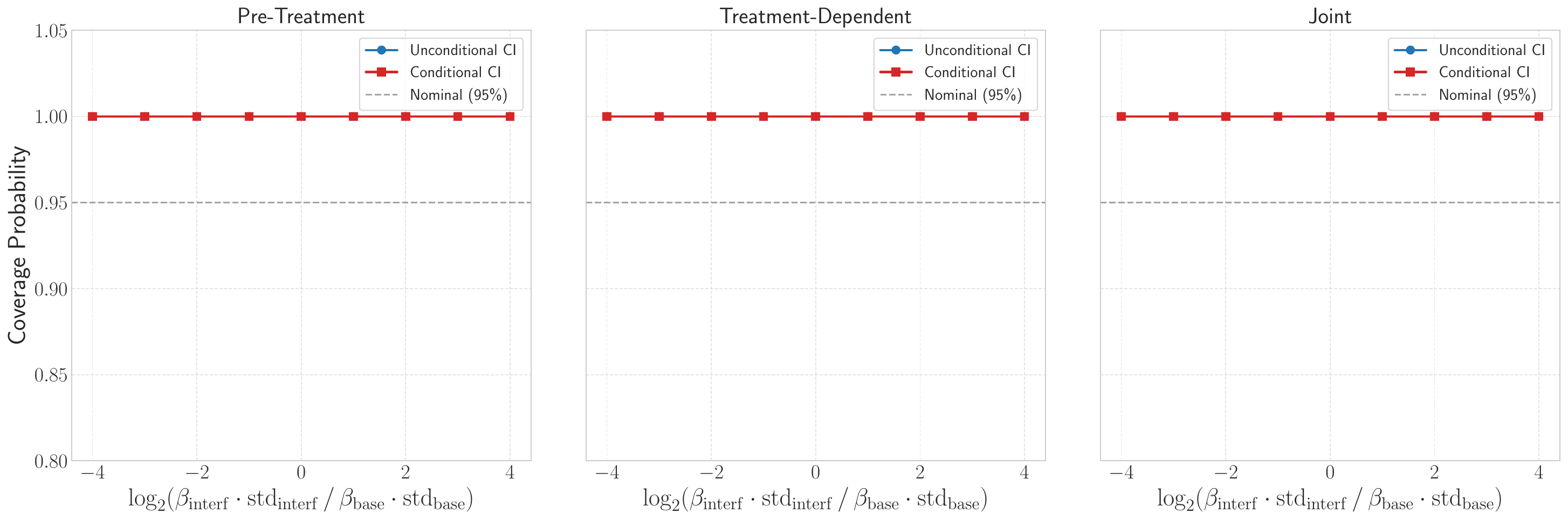"}
\caption{\texttt{exp-prop}, balancing \texttt{NP}}
\end{subfigure}
\caption{Empirical coverage probability of the 95\% confidence interval $\hat\tau\pm z_{1-\alpha/2}\sqrt{\hat v_{n,\delta}}$ as a function of $\log_2\kappa$, for two representative outcome models. The other five models exhibit qualitatively identical behavior and are omitted.}
\label{fig:cov_appendix}
\end{figure}

\subsubsection{Comparison with alternative constructions}
\label{sec:sim_comparison_full}

The body of the paper (Section~\ref{subsection:sim_var_comp}) reports the conservativeness-ratio comparison between the proposed estimator $\hat v_{n,\delta}^{\mathrm{ours}}$, the no-centering variant $\hat v_{n,\delta}^{\mathrm{uncen}}$, and the spectral-radius variant $\hat v_{n,\delta}^{\mathrm{spec}}$ for the two representative models \texttt{lin-prop} and \texttt{exp-prop}. We complete the picture here by reporting the same comparison across all seven outcome models. Each panel plots $\mathrm{CR}^{\mathrm{m}} =\log\bigl(\mathbb{E}[\hat v_{n,\delta}^{\mathrm{m}}\mid M_n\le a] / \mathrm{Var}(\hat\tau \mid M_n \le a)\bigr)$ for $\mathrm{m} \in \{\mathrm{ours}, \mathrm{uncen}, \mathrm{spec}\}$ as a function of $\log_2 \kappa$, with $\mathrm{CR}^{\mathrm{m}} \ge 0$ indicating asymptotic conservativeness; smaller values are tighter bounds.

The qualitative ordering observed in the body,
\[
\mathrm{CR}^{\mathrm{ours}}
\le\mathrm{CR}^{\mathrm{spec}}
\le\mathrm{CR}^{\mathrm{uncen}},
\]
holds across all linear models and the two exponential models with bounded interference blocks (\texttt{exp-prop} and \texttt{exp-prop+nwx}). The gap between the proposed estimator and the no-centering variant widens substantially under the exponential outcome models, consistent with the uncentered construction's sensitivity to large within-group outcome means before centering. The gap between the proposed estimator and the spectral-radius variant is most visible in the interference-dominated regime, where the heterogeneous BA degree distribution penalizes any approach that compresses the per-unit weights $d_i$ into a single global factor $\lambda_{\max}$.

The two remaining exponential models, \texttt{exp-sumexp} and \texttt{exp-sumexp-full}, show a different pattern: in the interference-dominated regime ($\log_2\kappa \ge -2$ for \texttt{exp-sumexp} and the entire sweep for \texttt{exp-sumexp-full}), $\mathrm{CR}^{\mathrm{ours}}$ overtakes $\mathrm{CR}^{\mathrm{spec}}$. Both estimators stay conservative ($\mathrm{CR}^{\mathrm{m}} \ge 0$ throughout), so the proposed bound retains its theoretical validity, but the comparison ceases to favor per-unit weighting.

The mechanism behind this reversal is confusing at first but actually straightforward. In the expansion of $g_i$, the degree term is $\beta_{\text{deg}}|\mathcal N_i|$, where
\[
\beta_{\text{deg}}=
\begin{cases}
\beta_{\text{prop}}, & \text{for \texttt{exp-sumexp}},\\
\beta_{\text{prop}}+\mathbf 1^\top\boldsymbol\beta_{\text{xinterf}},
& \text{for \texttt{exp-sumexp-full}}.
\end{cases}
\]
The second sum of exponentials in \texttt{exp-sumexp-full} contributes the additional constant $\mathbf 1^\top\boldsymbol\beta_{\text{xinterf}}$ per neighbor. Conditional on the fixed coefficients, this degree term depends only on the network, not on $\mathbf Z$ or $\mathbf X$. Write $\bar\nu_n=n^{-1}\sum_{i=1}^n|\mathcal N_i|$ for the average number of direct neighbors. When the degree component dominates the remaining outcome variation, substituting $Y_i \approx \alpha + \tau Z_i + \beta_{\text{interf}}\,\beta_{\text{deg}}\,|\mathcal{N}_i| + (\text{small fluctuation in }\mathbf Z)$ into the H\'ajek difference and linearizing as in Lemma~\ref{lem:t1-hajek-linearization} yields
\[
\hat\tau - \tau \;\approx\; \frac{\beta_{\text{interf}}\,\beta_{\text{deg}}}{n}\sum_{i=1}^n
\psi_i\bigl(|\mathcal{N}_i| - \bar\nu_n\bigr), \qquad \psi_i \;=\; \frac{Z_i}{\pi}
\;-\; \frac{1-Z_i}{1-\pi},
\]
which exhibits $\hat\tau - \tau$ as a linear-in-$\mathbf Z$ statistic with weights $|\mathcal{N}_i| - \bar\nu_n$. Since the $\psi_i$ are mean-zero and independent across $i$ with $\Var(\psi_i) = 1/[\pi(1-\pi)]$,
\[
\Var(\hat\tau) \;\approx\;
\frac{(\beta_{\text{interf}}\,\beta_{\text{deg}})^2}{n^2\,\pi(1-\pi)}\sum_{i=1}^n
\bigl(|\mathcal{N}_i| - \bar\nu_n\bigr)^2.
\]
On the Barab\'asi--Albert graph used in the simulation ($\bar\nu_n\approx 10$, $\max_i|\mathcal{N}_i|=242$), the squared deviations $(|\mathcal{N}_i| - \bar\nu_n)^2$ are sharply concentrated on the high-degree tail---a hub with $|\mathcal{N}_i|=100$ contributes $(90)^2 = 8{,}100$, while a peripheral unit with $|\mathcal{N}_i|=5$ contributes only $(5)^2=25$---so $\Var(\hat\tau)$ is dominated by the few dozen hub units. The balance set used in our rerandomization criterion---pre-treatment $\mathbf X$, $H_{i,\mathrm{prop}}$, and/or $\mathbf H_{i,\mathrm{wX}}$, defined in~\eqref{eq:sim_h_prop}--\eqref{eq:sim_h_wx}---does not include $|\mathcal{N}_i|$ or any quantity strongly correlated with it: $H_{i,\mathrm{prop}}$ takes values in $[0,1]$ and carries no $|\mathcal{N}_i|$-magnitude information, and $\mathbf H_{i,\mathrm{wX}}$ is similarly normalized by $|\mathcal{N}_i|$. Conditioning on $M_n \le a$ therefore barely reduces the $|\mathcal{N}_i|$-imbalance contribution to $\Var(\hat\tau)$, so the squared multiple correlation $R^2$ between $\hat\tau$ and $\tauxhat$ is small, the cross-term $\hat{\mathbf U}_{12,\delta}^{\mathrm m}$ contributes only a small correction in the closed form described in Appendix~\ref{app:feasibility-refinement}, and the optimization output effectively collapses to $\hat v_{n,\delta}^{\mathrm m}\approx\hat{\mathbf U}_{n,\delta}^{\mathrm m}(1,1)$ for both $\mathrm m\in\{\mathrm{ours},\mathrm{spec}\}$, where $\hat{\mathbf U}_{n,\delta}^{\mathrm m}(1,1)$ denotes the $(1,1)$ entry (the $Y$--$Y$ block) of the $d_\phi\times d_\phi$ matrix $\hat{\mathbf U}_{n,\delta}^{\mathrm m}$.

Because the same fixed marginal-scale ridge is applied to both constructions, their leading comparison remains one between the raw blocks $\hat{\mathbf U}_n^{\mathrm{ours}}(1,1)$ and $\hat{\mathbf U}_n^{\mathrm{spec}}(1,1)$. Writing $\hat\phi_{i,Y}:=Y_i-\bar Y_{Z_i}$ for the within-arm-centered $Y$ residual at unit $i$ (where $\bar Y_{Z_i}$ is the sample mean of $Y$ over units sharing $i$'s treatment status), the two scalar bounds take the explicit form
\[
\hat{\mathbf U}_n^{\mathrm{ours}}(1,1) \;=\; \frac{1}{n^2}\sum_{i=1}^n
d_{n,i}\,\psi_i^2\,\hat\phi_{i,Y}^2, \qquad \hat{\mathbf U}_n^{\mathrm{spec}}(1,1) \;=\;
\frac{\lambda_{\max}}{n^2}\sum_{i=1}^n \psi_i^2\,\hat\phi_{i,Y}^2,
\]
where $\psi_i^2 = Z_i/\pi^2 + (1-Z_i)/(1-\pi)^2$. Both bounds are positive linear combinations of the same residuals $\hat\phi_{i,Y}^2$; the only difference is the weight assigned to each summand. Two heavy-tail effects then compound to make ours the looser bound. First, by the linearization above, $\hat\phi_{i,Y}\approx \beta_{\text{interf}}\,\beta_{\text{deg}}(|\mathcal N_i| - \bar\nu_n)$, so $\hat\phi_{i,Y}^2$ inherits the heavy tail of the degree distribution and the few dozen hubs dominate $\sum_i \hat\phi_{i,Y}^2$ (this is the same calculation that gave the variance dominance above). Second, $d_{n,i}$ itself is heavy-tailed in $|\mathcal N_i|$, since a hub with many direct neighbors also reaches many units within two hops via those neighbors. Concretely, on our BA graph $d_{n,i}$ has mean $\bar d_n \approx 202.5$ and maximum $d_{n,\max}\approx 1{,}673$, against $\lambda_{\max}\approx 320.9$, so a hub residual receives a per-unit amplification roughly $d_{n,\max}/\lambda_{\max}\approx 5\times$ that of the spectral-radius bound. Because hubs simultaneously dominate $\sum_i \hat\phi_{i,Y}^2$ and receive this $\sim 5\times$ amplification under per-unit weighting, $\hat{\mathbf U}_n^{\mathrm{ours}}(1,1)$ ends up above $\hat{\mathbf U}_n^{\mathrm{spec}}(1,1)$, and the per-unit bound is therefore looser in this regime.

\begin{figure}[p]
\centering \rule{0pt}{1ex}\texttt{lin-prop}\\[1pt]
\includegraphics[width=0.36\textwidth]{"figures/conservative_variance/cv_covs-np_type1_comparison.pdf"}\hfill
\includegraphics[width=0.36\textwidth]{"figures/conservative_variance/cv_covs-np+nwx_type1_comparison.pdf"}\\[0pt]
\rule{0pt}{1ex}\texttt{lin-prop+nwx}\\[1pt]
\includegraphics[width=0.36\textwidth]{"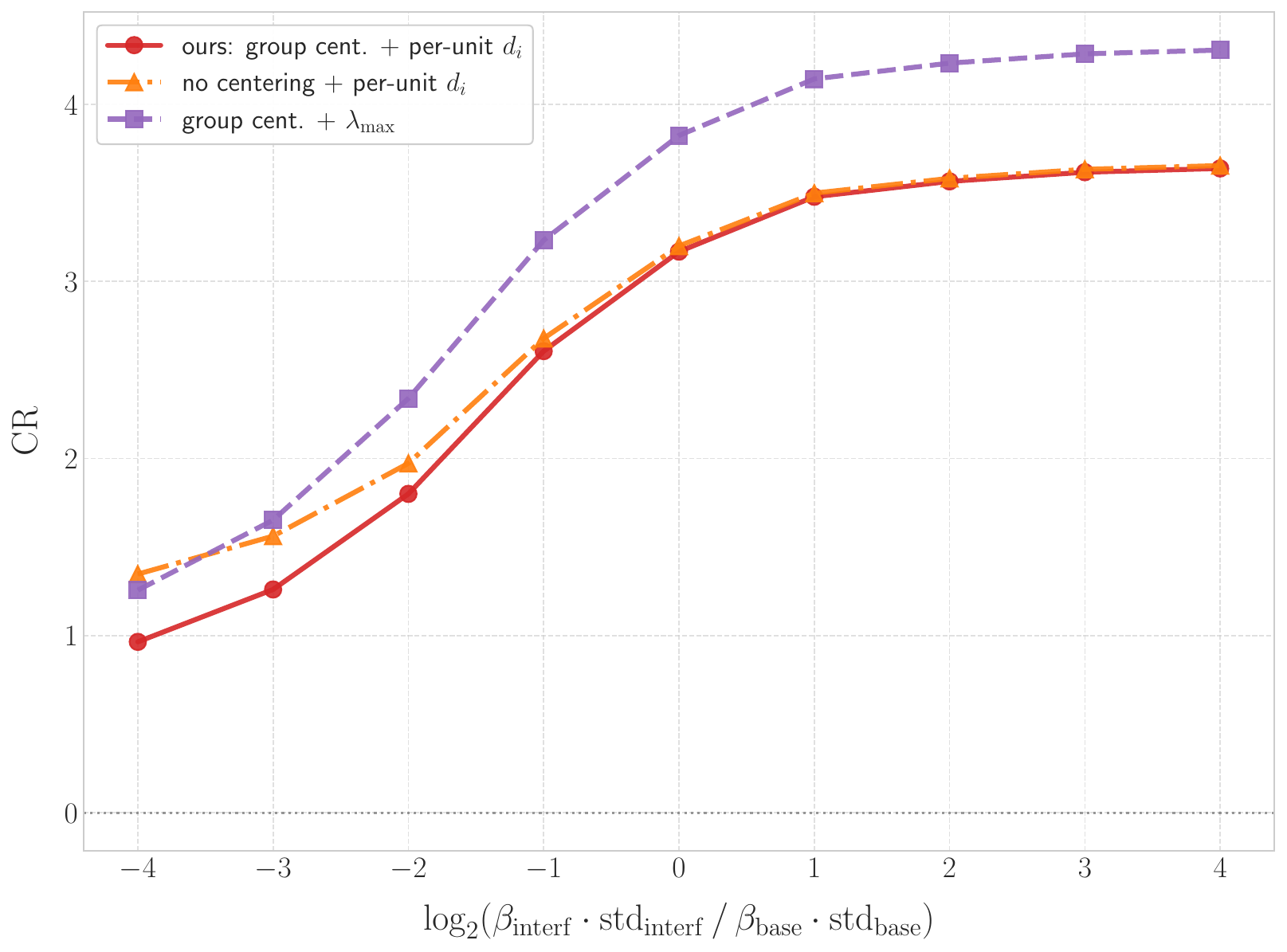"}\hfill
\includegraphics[width=0.36\textwidth]{"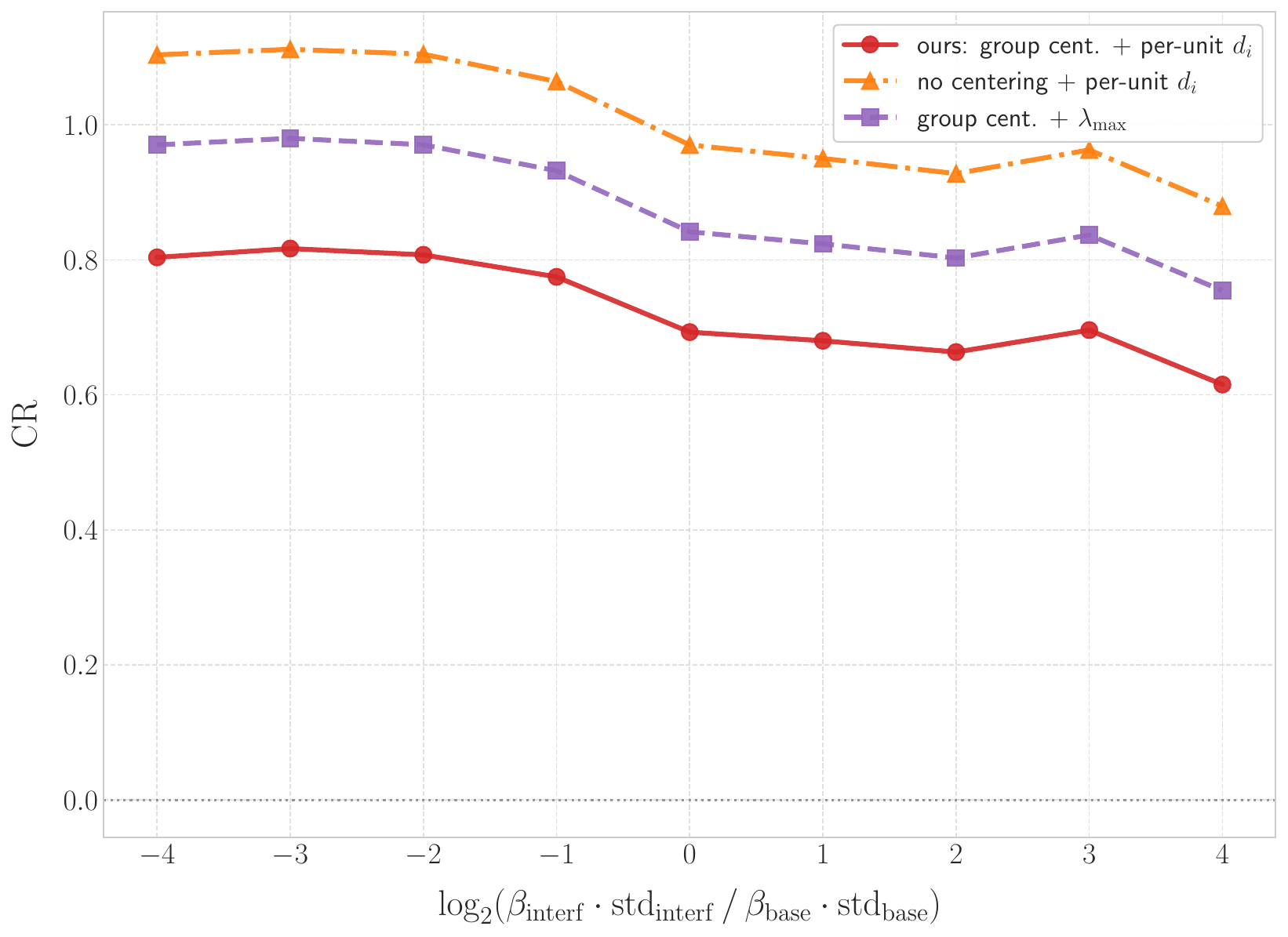"}\\[0pt]
\rule{0pt}{1ex}\texttt{lin-nwx}\\[1pt]
\includegraphics[width=0.36\textwidth]{"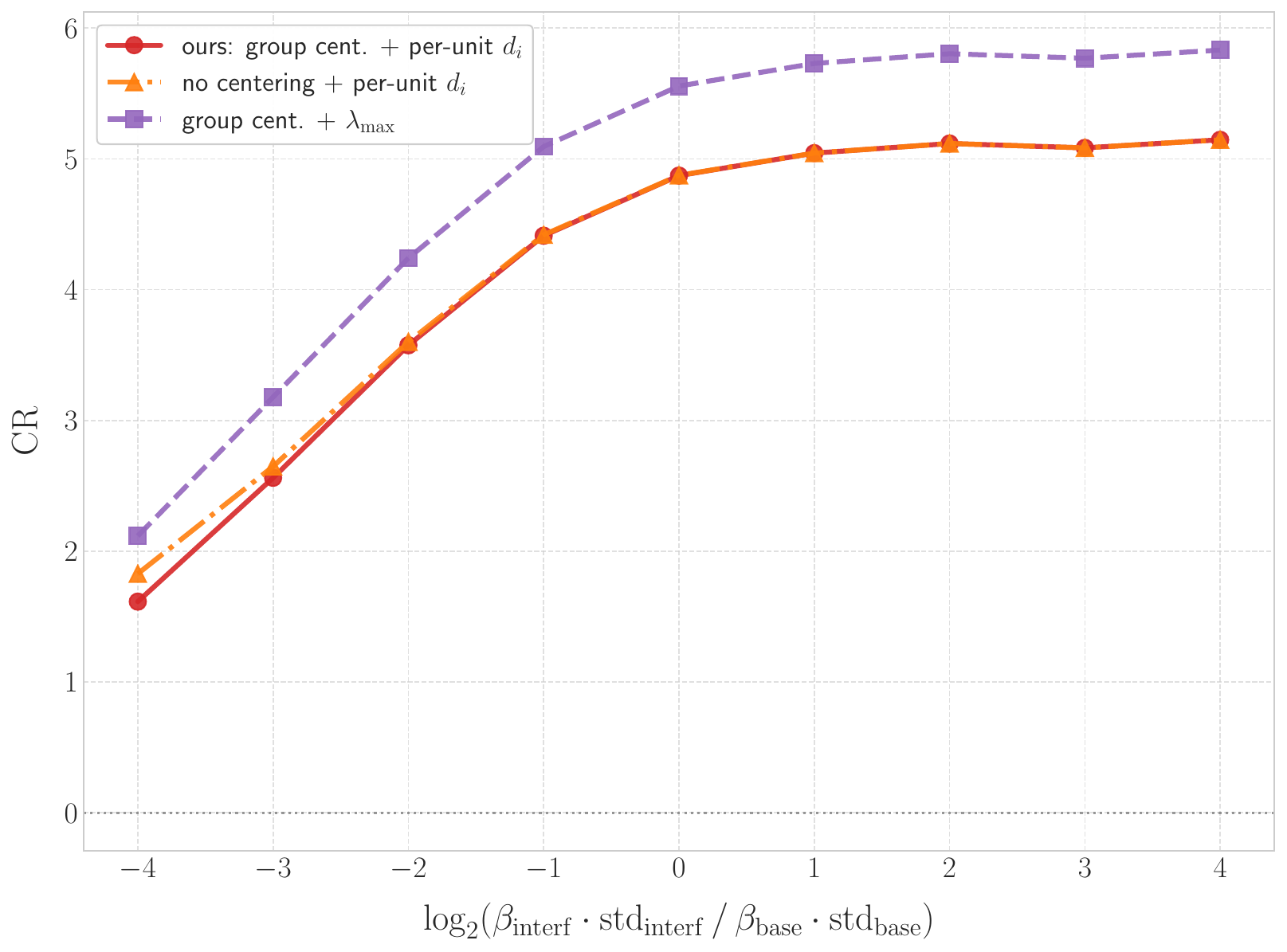"}\hfill
\includegraphics[width=0.36\textwidth]{"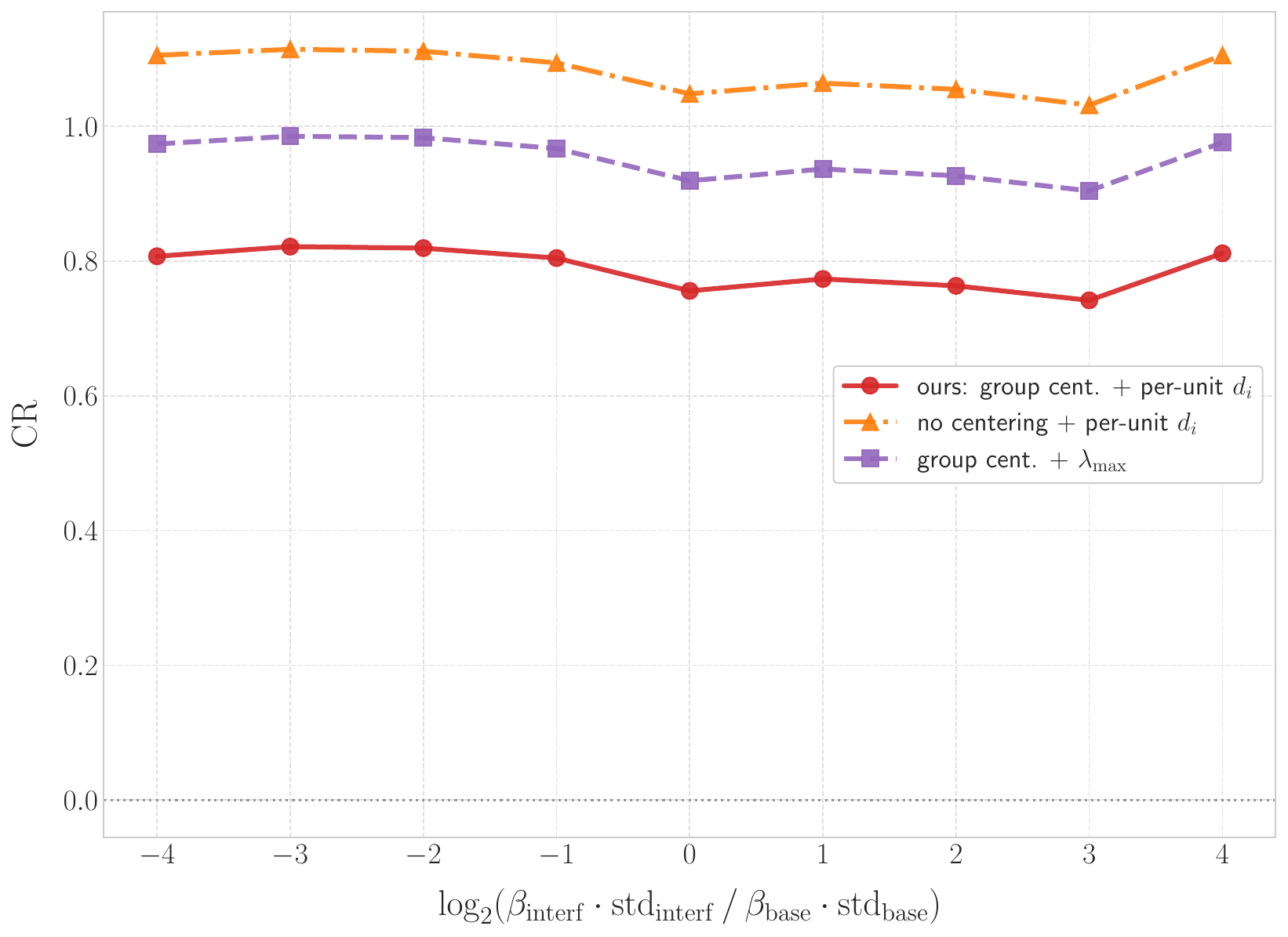"}\\[0pt]
\caption{Conservativeness-ratio comparison ($\mathrm{CR}^{\mathrm{m}}$ on the log scale) across the three linear outcome models. Within each row, left: balancing \texttt{NP}; right: balancing \texttt{NP+NWX}. Red: ours (group centering + per-unit $d_{n,i}$). Orange: no centering + per-unit $d_{n,i}$. Purple: group centering + spectral radius $\lambda_{\max}$.}
\label{fig:comparison_lin_appendix}
\end{figure}

\begin{figure}[p]
\centering \rule{0pt}{1ex}\texttt{exp-prop}\\[1pt]
\includegraphics[width=0.36\textwidth]{"figures/conservative_variance/cv_covs-np_type5_comparison.pdf"}\hfill
\includegraphics[width=0.36\textwidth]{"figures/conservative_variance/cv_covs-np+nwx_type5_comparison.pdf"}\\[0pt]
\rule{0pt}{1ex}\texttt{exp-sumexp}\\[1pt]
\includegraphics[width=0.36\textwidth]{"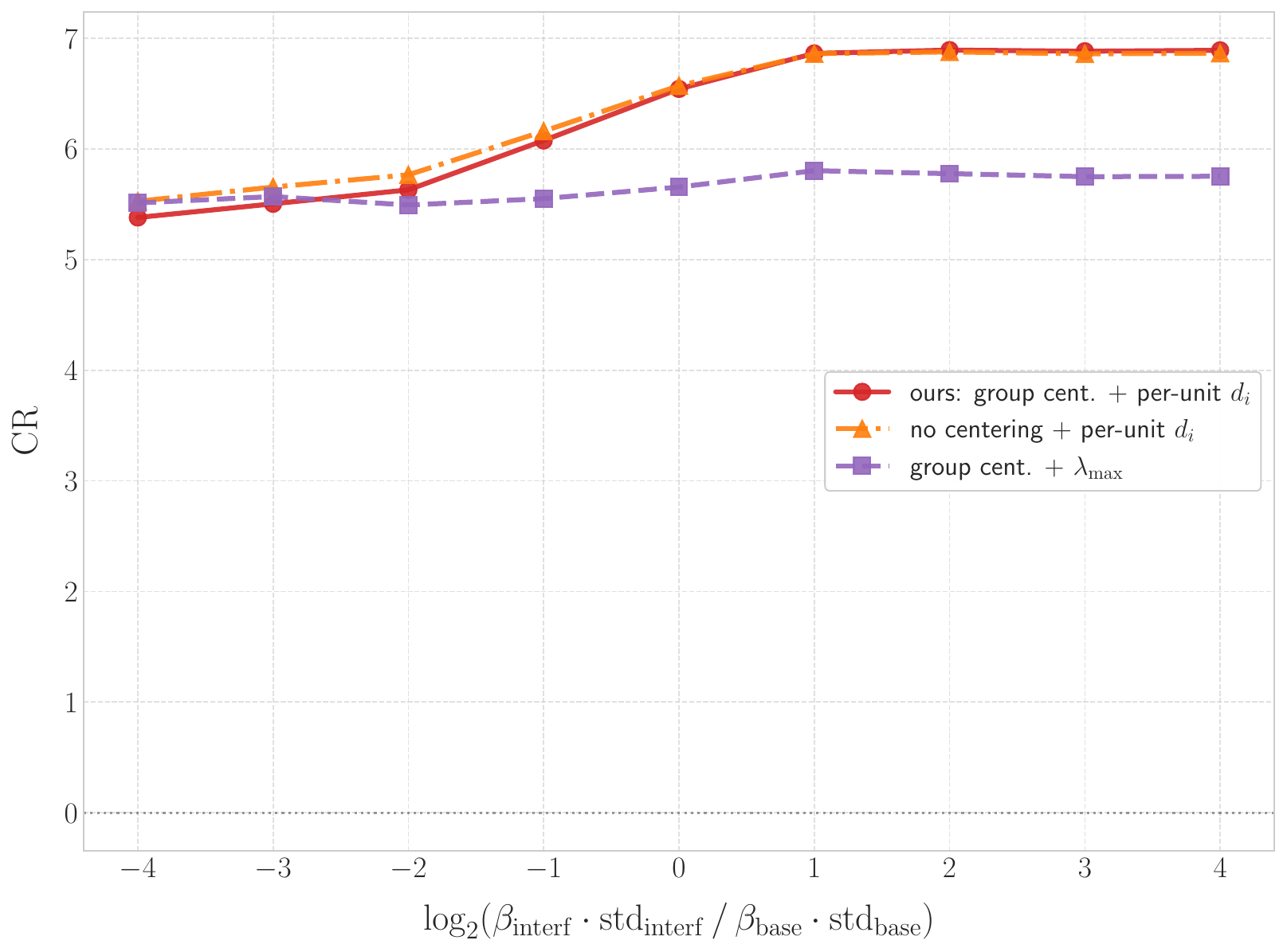"}\hfill
\includegraphics[width=0.36\textwidth]{"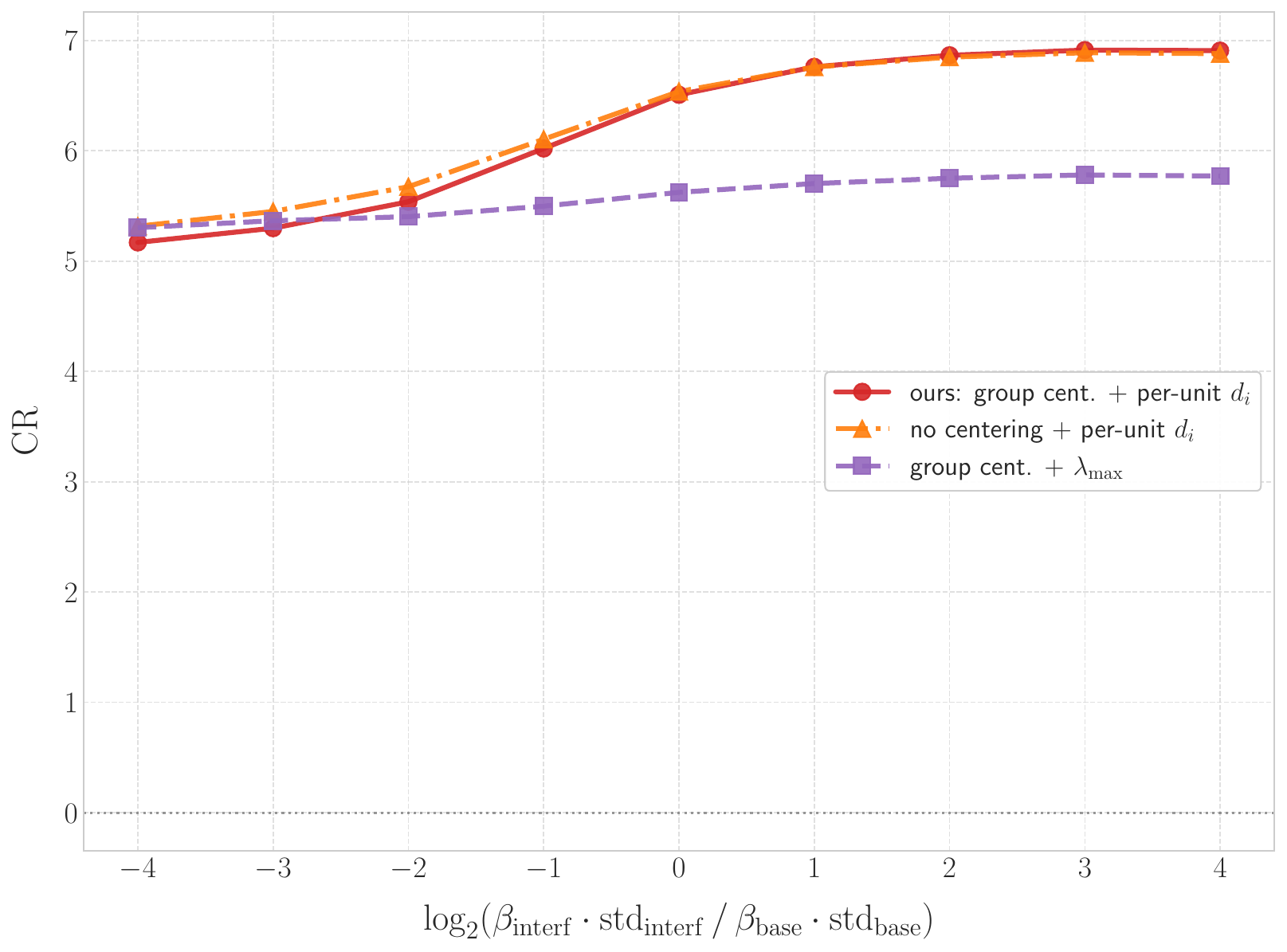"}\\[0pt]
\rule{0pt}{1ex}\texttt{exp-prop+nwx}\\[1pt]
\includegraphics[width=0.36\textwidth]{"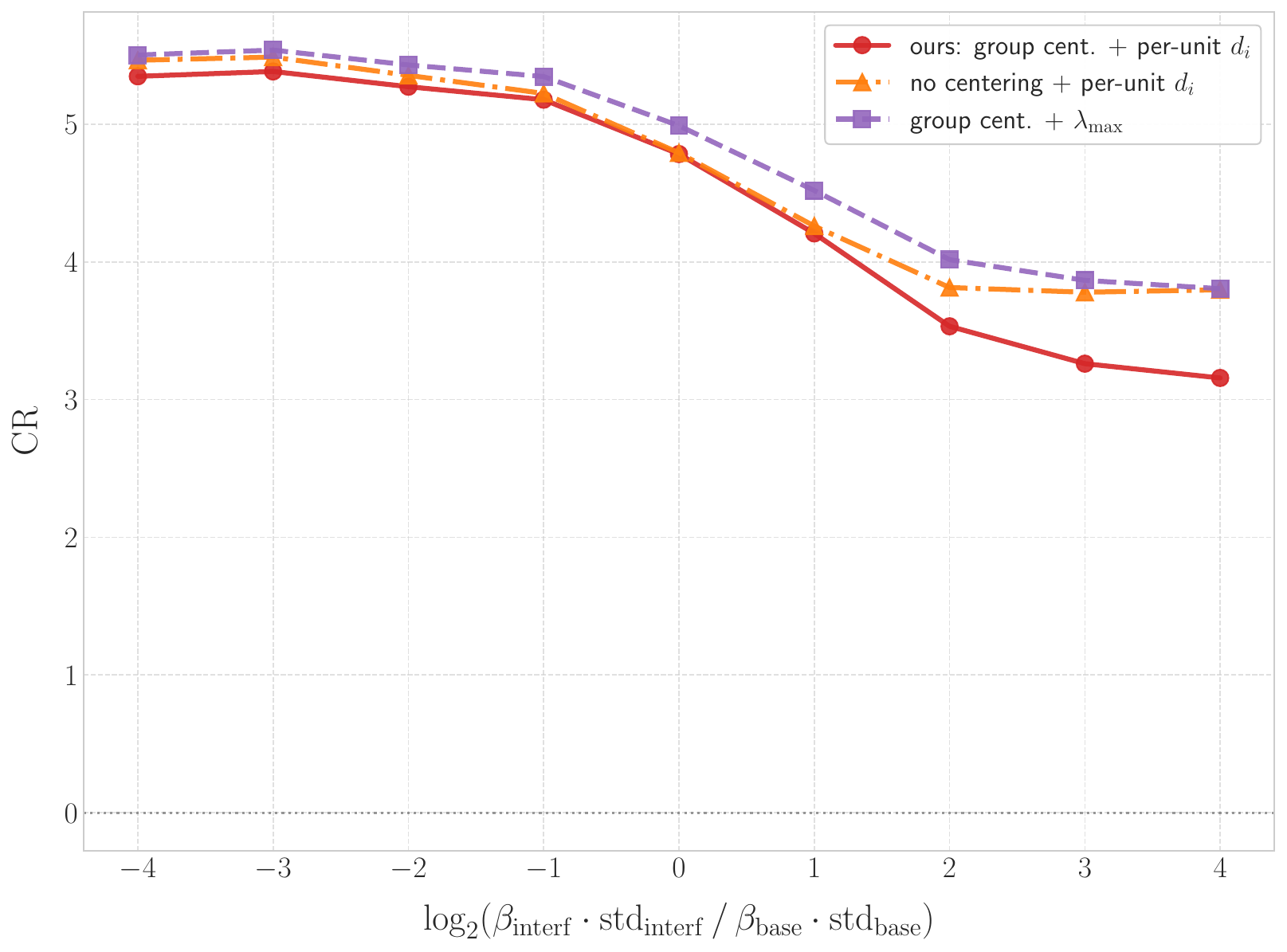"}\hfill
\includegraphics[width=0.36\textwidth]{"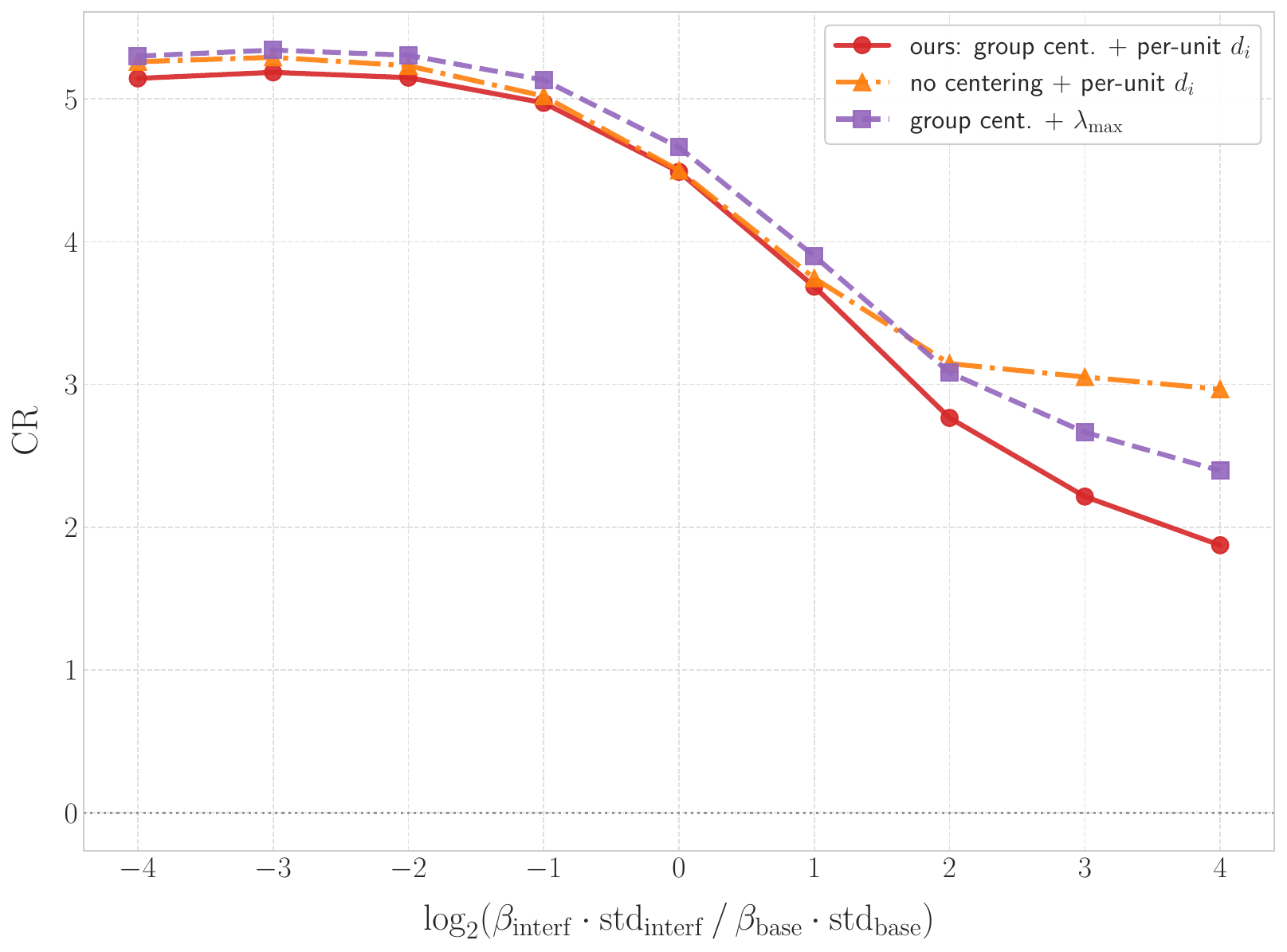"}\\[0pt]
\rule{0pt}{1ex}\texttt{exp-sumexp-full}\\[1pt]
\includegraphics[width=0.36\textwidth]{"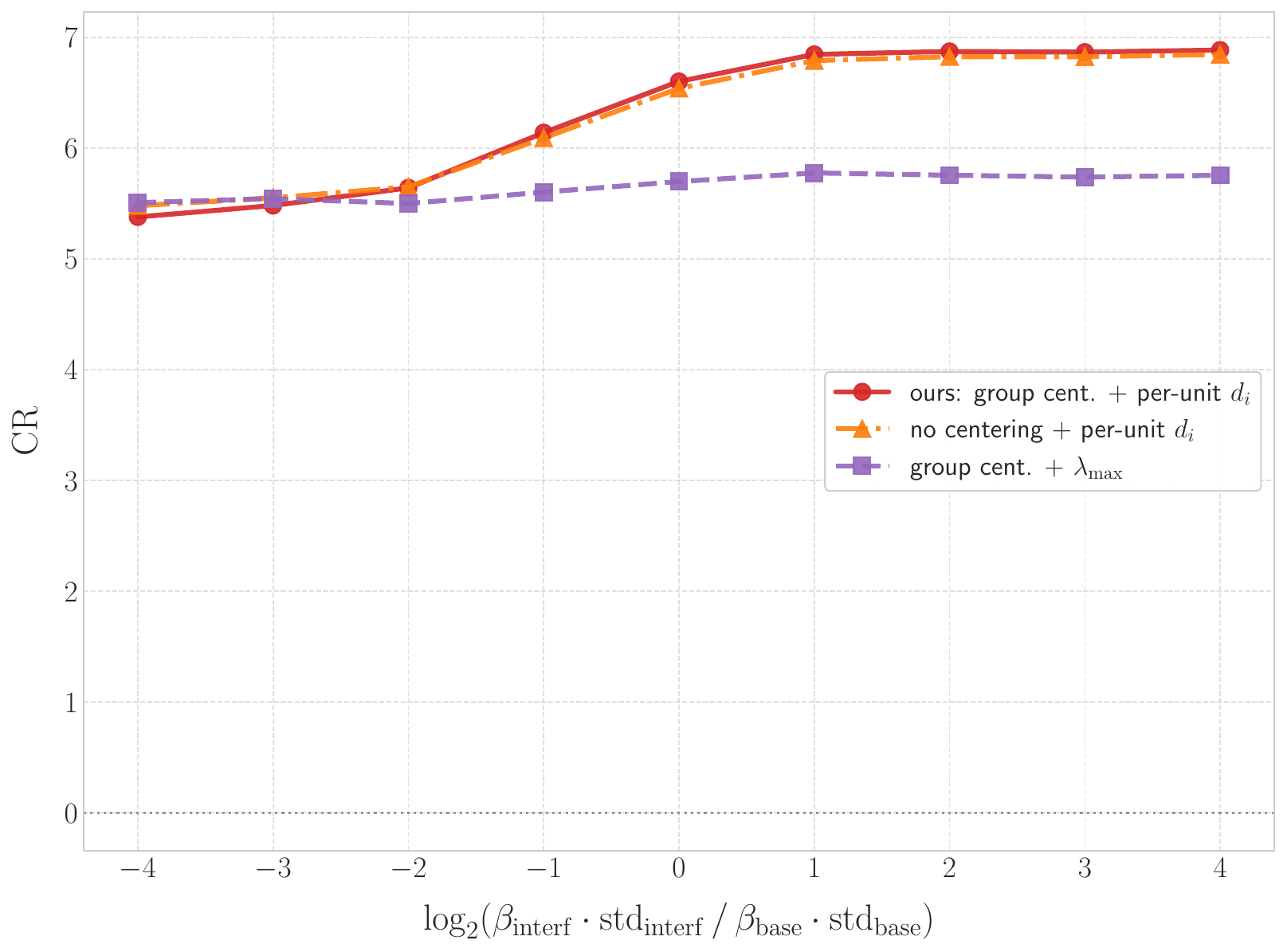"}\hfill
\includegraphics[width=0.36\textwidth]{"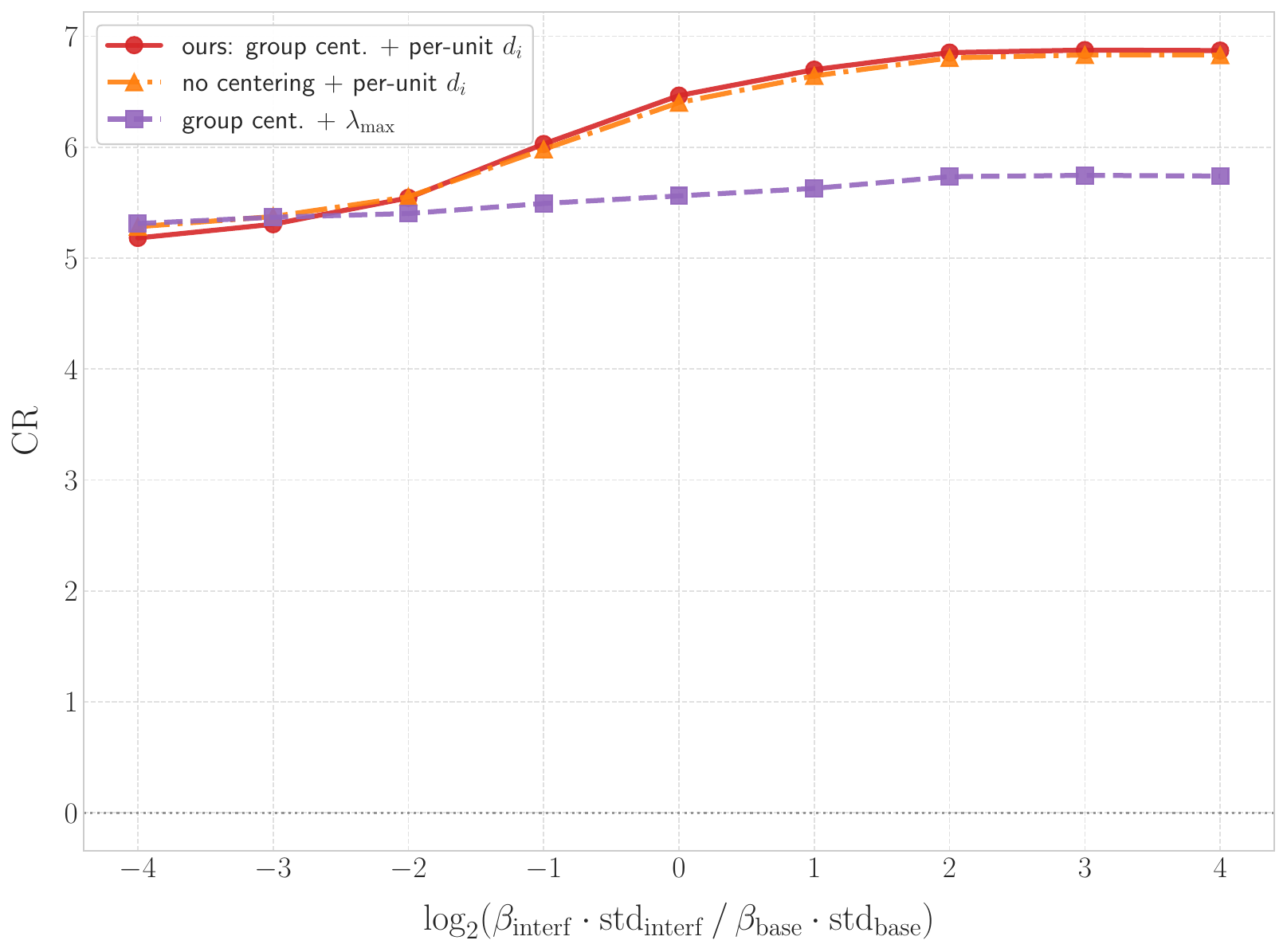"}\\[0pt]
\caption{Conservativeness-ratio comparison across the four exponential outcome models. Layout, line types, and colors follow Figure~\ref{fig:comparison_lin_appendix}. The gap between ours and the no-centering variant widens substantially under the exponential outcome map.}
\label{fig:comparison_exp_appendix}
\end{figure}

\subsection{Comparison of alternative inference procedures}
\label{sec:sim_alternative_inference}

In Appendix~\ref{app:alternative_inference} we introduced several simpler alternatives to our optimization-based inference procedure. Here we compare them empirically. On the variance side, we consider three estimators of the conditional variance $\Var_{\operatorname{asymp}}(\hat\tau \mid M_n \le a)$: our optimization-based estimator $\hat v_{n,\delta}$; the \emph{Bernoulli} estimator $\hat U_{11}$, which ignores the variance reduction from rerandomization; and the \emph{small-$a$ plug-in} estimator $\hat U_{11,\delta}(1+v_p) -\hat{\mathbf U}_{12,\delta}\hat{\mathbf U}_{22,\delta}^{-1} \hat{\mathbf U}_{21,\delta}$. On the inference side, we compare four confidence intervals: our interval $\hat\tau\pm\bar q_{1-\alpha/2}(p,a)\,(\hat v_{n,\delta})^{1/2}$; the \emph{Bernoulli Wald} interval $\hat\tau \pm z_{1-\alpha/2}\,\hat U_{11}^{1/2}$; the \emph{direct quantile optimization} interval; and the \emph{plug-in rerandomization} interval. The direct quantile optimization problem is computationally challenging in general (Appendix~\ref{app:alternative_inference}); we approximate it by searching the one-parameter family that defines our variance optimum.

We use the \texttt{lin-prop} and \texttt{exp-prop} outcome models with $n=500$ units, the joint balancing strategy, and the same signal-ratio sweep as in the previous subsections; all other settings follow Appendix~\ref{appendx:simulation}.

\begin{figure}[ht]
\centering
\begin{subfigure}[t]{0.48\textwidth}
\includegraphics[width=\linewidth]{"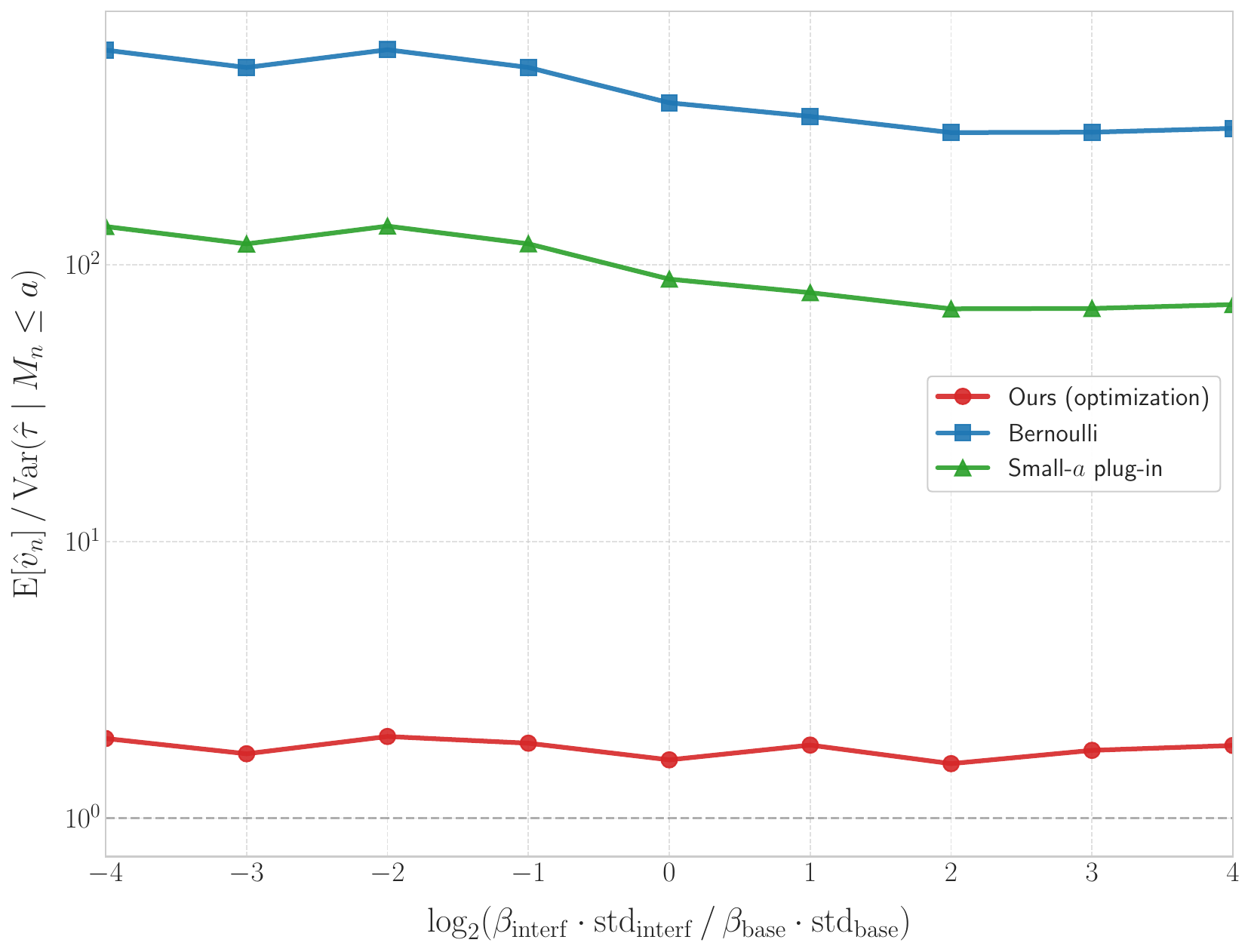"}
\caption{Linear (\texttt{lin-prop})}
\end{subfigure}
\hfill
\begin{subfigure}[t]{0.48\textwidth}
\includegraphics[width=\linewidth]{"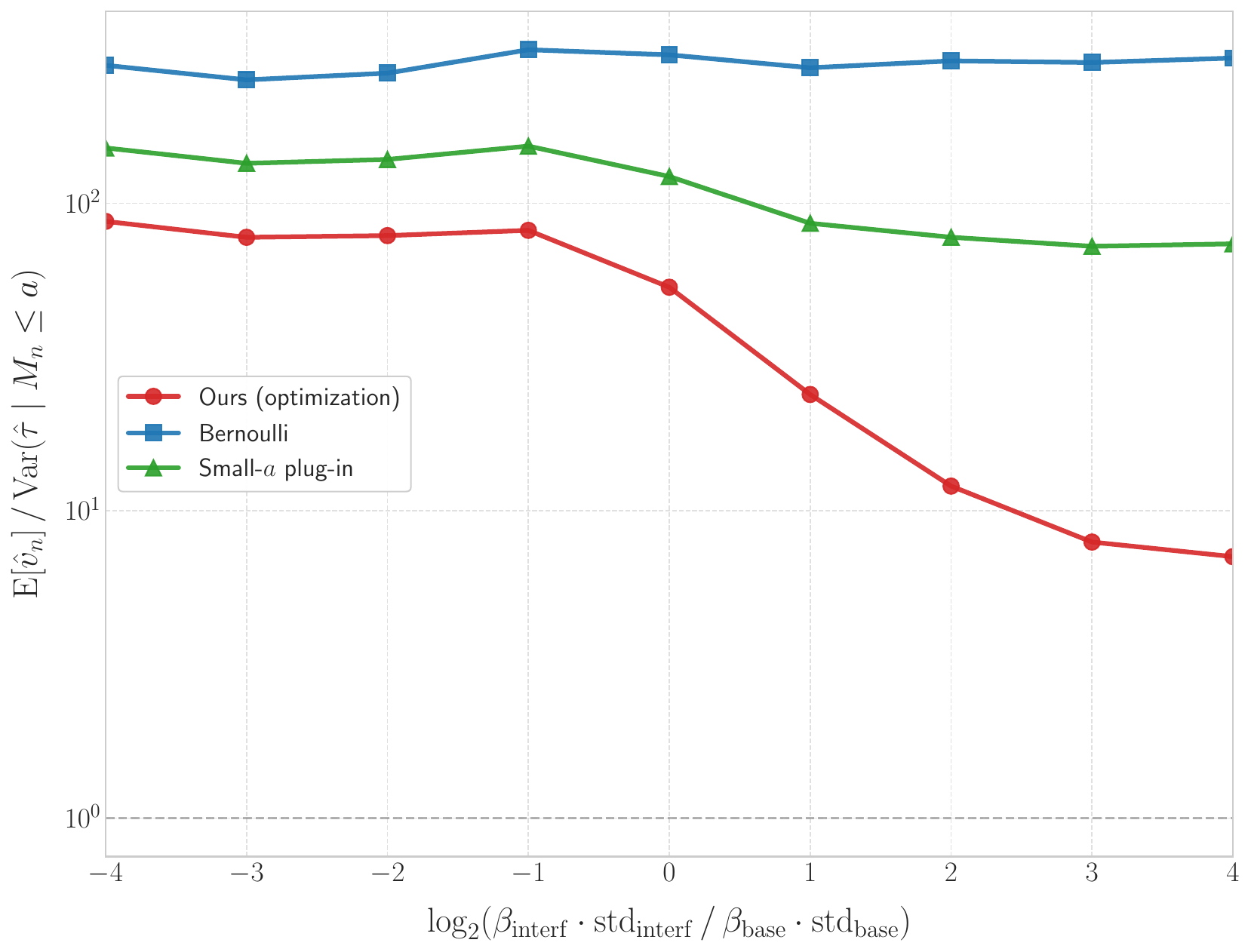"}
\caption{Nonlinear (\texttt{exp-prop})}
\end{subfigure}
\caption{Conservativeness of the three conditional-variance estimators, $\mathrm{E}[\hat v]/\Var(\hat\tau \mid M_n \le a)$ (log scale; values close to $1$ are tighter), as a function of the log signal ratio $\log_2\kappa$. Red: our optimization-based estimator; blue: Bernoulli $\hat U_{11}$; green: small-$a$ plug-in.}
\label{fig:altinf_variance}
\end{figure}

Figure~\ref{fig:altinf_variance} compares the three variance estimators. Our optimization-based estimator is by far the tightest, while the Bernoulli estimator is severely conservative; the small-$a$ plug-in lies in between.

\begin{figure}[ht]
\centering
\begin{subfigure}[t]{0.48\textwidth}
\includegraphics[width=\linewidth]{"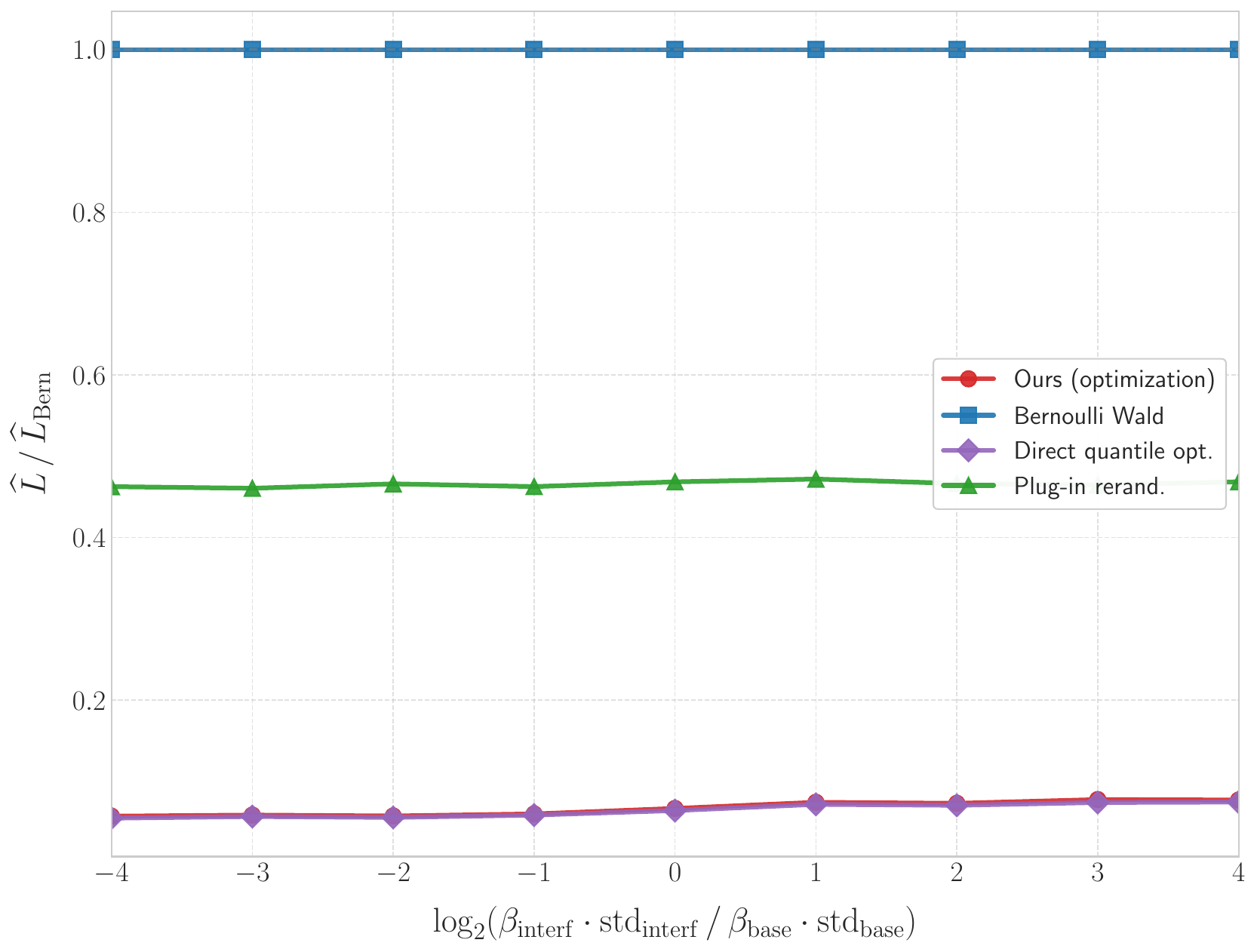"}
\caption{Linear (\texttt{lin-prop})}
\end{subfigure}
\hfill
\begin{subfigure}[t]{0.48\textwidth}
\includegraphics[width=\linewidth]{"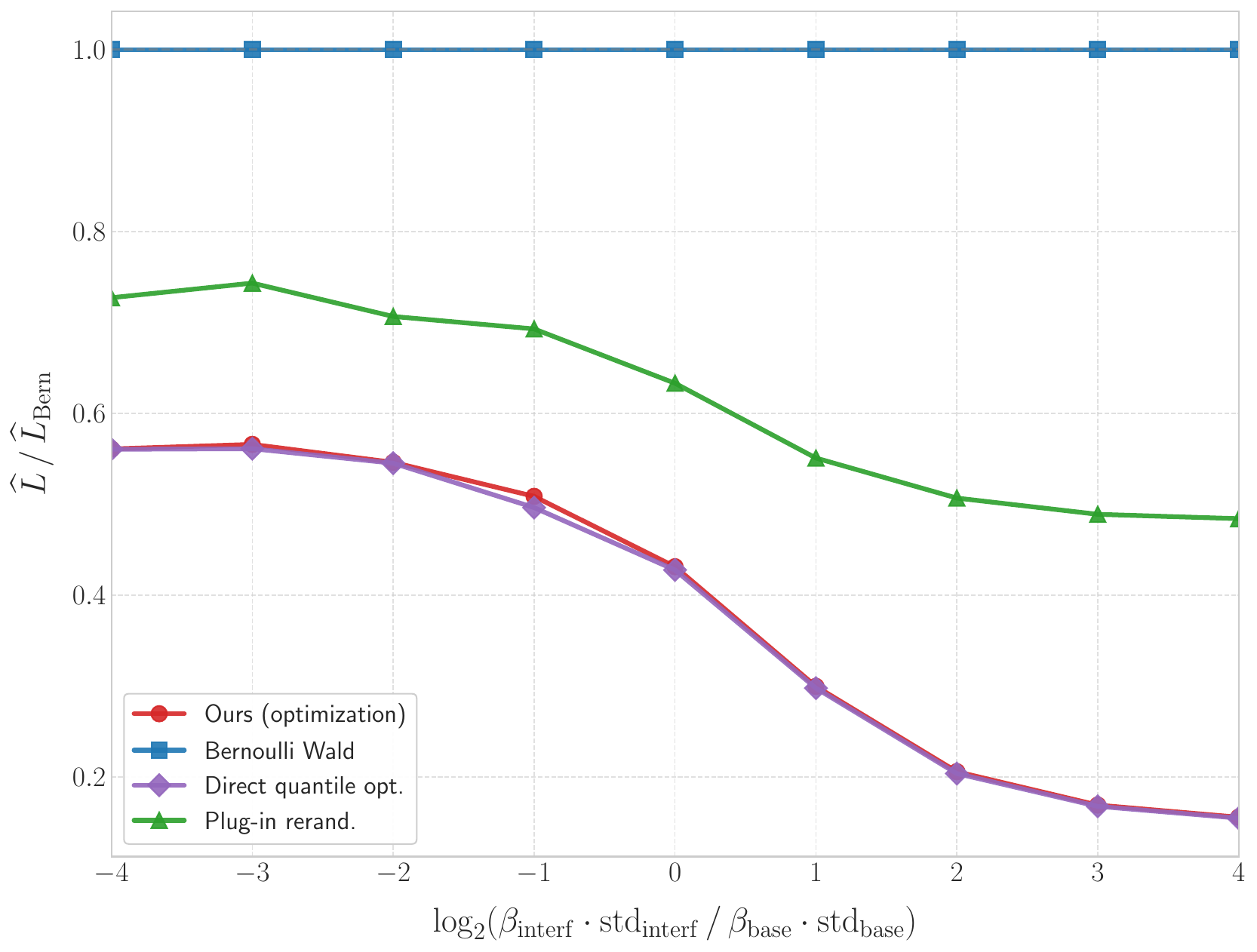"}
\caption{Nonlinear (\texttt{exp-prop})}
\end{subfigure}
\caption{Confidence interval half-length relative to the Bernoulli Wald interval, $\hat L / \hat L_{\mathrm{Bern}}$, as a function of the log signal ratio $\log_2\kappa$. Red: ours; blue: Bernoulli Wald; purple: direct quantile optimization; green: plug-in rerandomization interval.}
\label{fig:altinf_cilength}
\end{figure}

Figure~\ref{fig:altinf_cilength} compares the four confidence intervals. Our interval and the direct quantile optimization interval are nearly indistinguishable, indicating that our computationally simple procedure essentially attains the sharper—but much more expensive—direct optimization. Both are substantially shorter than the plug-in rerandomization interval, which is in turn shorter than the Bernoulli Wald interval. As expected, the gains over the Bernoulli benchmark are largest when rerandomization reduces the variance the most.

\end{document}